\documentclass[11pt]{article}

\usepackage{amsthm}
\usepackage{graphicx} 
\usepackage{array} 

\usepackage{amsmath, amssymb, amsfonts, verbatim}
\usepackage{hyphenat,epsfig,subcaption,multirow}
\usepackage[font=footnotesize,labelfont=bf]{caption}
\usepackage{nicefrac}
\usepackage{paralist}
\usepackage{wrapfig}

\usepackage[usenames,dvipsnames]{xcolor}
\usepackage[ruled]{algorithm2e}

\DeclareFontFamily{U}{mathx}{\hyphenchar\font45}
\DeclareFontShape{U}{mathx}{m}{n}{
      <5> <6> <7> <8> <9> <10>
      <10.95> <12> <14.4> <17.28> <20.74> <24.88>
      mathx10
      }{}
\DeclareSymbolFont{mathx}{U}{mathx}{m}{n}
\DeclareMathSymbol{\bigtimes}{1}{mathx}{"91}

\usepackage{tcolorbox}
\tcbuselibrary{skins,breakable}
\tcbset{enhanced jigsaw}

\usepackage[normalem]{ulem}
\usepackage[compact]{titlesec}

\definecolor{DarkRed}{rgb}{0.5,0.1,0.1}
\definecolor{DarkBlue}{rgb}{0.1,0.1,0.5}

\usepackage{nameref}
\definecolor{ForestGreen}{rgb}{0.1333,0.5451,0.1333}
\definecolor{Red}{rgb}{0.9,0,0}
\usepackage[linktocpage=true,
	pagebackref=true,colorlinks,
	linkcolor=DarkRed,citecolor=ForestGreen,
	bookmarks,bookmarksopen,bookmarksnumbered]
	{hyperref}
\usepackage[noabbrev,nameinlink]{cleveref}
\crefname{property}{property}{Property}
\creflabelformat{property}{(#1)#2#3}
\crefname{equation}{eq}{Eq}
\creflabelformat{equation}{(#1)#2#3}

\usepackage{bm}
\usepackage{url}
\usepackage{xspace}
\usepackage[mathscr]{euscript}

\usepackage{tikz}
\usetikzlibrary{arrows}
\usetikzlibrary{arrows.meta}
\usetikzlibrary{shapes}
\usetikzlibrary{backgrounds}
\usetikzlibrary{positioning}
\usetikzlibrary{decorations.markings}
\usetikzlibrary{decorations.pathreplacing} 
\usetikzlibrary{patterns}
\usetikzlibrary{calc}
\usetikzlibrary{fit}
\usetikzlibrary{decorations}

\usepackage[framemethod=TikZ]{mdframed}

\usepackage[noend]{algpseudocode}
\makeatletter
\def\BState{\State\hskip-\ALG@thistlm}
\makeatother

\usepackage{cite}
\usepackage{enumitem}
\setlist[itemize]{leftmargin=20pt}
\setlist[enumerate]{leftmargin=20pt}

\usepackage[margin=1in]{geometry}

\usepackage{thmtools}
\usepackage{thm-restate}

\newtheorem{theorem}{Theorem}
\newtheorem{lemma}{Lemma}[section]
\newtheorem{proposition}[lemma]{Proposition}
\newtheorem{corollary}[lemma]{Corollary}
\newtheorem{claim}[lemma]{Claim}
\newtheorem{fact}[lemma]{Fact}

\newtheorem{definition}[lemma]{Definition}
\newtheorem{problem}{Problem}

\newtheorem*{claim*}{Claim}
\newtheorem*{assumption*}{Assumption}
\newtheorem*{proposition*}{Proposition}
\newtheorem*{lemma*}{Lemma}

\newtheorem{observation}[lemma]{Observation}

\newtheorem*{theorem*}{Theorem}

\crefname{lemma}{Lemma}{Lemmas}
\crefname{claim}{claim}{claims}
\crefname{property}{Property}{Properties}
\crefname{invariant}{Invariant}{Invariants}

\newtheorem{mdresult}{Result}
\newenvironment{result}{\begin{mdframed}[backgroundcolor=lightgray!40,topline=false,rightline=false,leftline=false,bottomline=false,innertopmargin=2pt]\begin{mdresult}}{\end{mdresult}\end{mdframed}}

\theoremstyle{definition}

\newtheorem*{mdproblem*}{Problem}
\newenvironment{Problem*}{\begin{mdframed}[hidealllines=false,innerleftmargin=10pt,backgroundcolor=gray!10,innertopmargin=5pt,innerbottommargin=5pt,roundcorner=10pt]\begin{mdproblem*}}{\end{mdproblem*}\end{mdframed}}
\newtheorem{mddefinition}[lemma]{Definition}
\newenvironment{Definition}{\begin{mdframed}[hidealllines=false,innerleftmargin=10pt,backgroundcolor=white!10,innertopmargin=5pt,innerbottommargin=5pt,roundcorner=10pt]\begin{mddefinition}}{\end{mddefinition}\end{mdframed}}
\newtheorem*{mddefinition*}{Definition}
\newenvironment{Definition*}{\begin{mdframed}[hidealllines=false,innerleftmargin=10pt,backgroundcolor=white!10,innertopmargin=5pt,innerbottommargin=5pt,roundcorner=10pt]\begin{mddefinition*}}{\end{mddefinition*}\end{mdframed}}
\newtheorem{mdremark}{Remark}
\newenvironment{Remark}{\begin{mdframed}[hidealllines=false,innerleftmargin=10pt,backgroundcolor=gray!10,innertopmargin=5pt,innerbottommargin=5pt,roundcorner=10pt]\begin{mdremark}}{\end{mdremark}\end{mdframed}}

\newenvironment{ourbox}{\begin{mdframed}[hidealllines=false,innerleftmargin=10pt,backgroundcolor=white!10,innertopmargin=2pt,innerbottommargin=5pt,roundcorner=10pt]}{\end{mdframed}}

\newtheorem{mdalgorithm}{Algorithm}
\newenvironment{Algorithm}{\begin{ourbox}\begin{mdalgorithm}}{\end{mdalgorithm}\end{ourbox}}

\newtheoremstyle{restate}{}{}{\itshape}{}{\bfseries}{~(restated).}{.5em}{\thmnote{#3}}
\theoremstyle{restate}
\newtheorem*{restate}{}

\allowdisplaybreaks

\renewcommand{\qed}{\nobreak \ifvmode \relax \else
      \ifdim\lastskip<1.5em \hskip-\lastskip
      \hskip1.5em plus0em minus0.5em \fi \nobreak
      \vrule height0.75em width0.5em depth0.25em\fi}

\newcommand{\Qed}[1]{\rlap{\qed$_{\textnormal{~~\Cref{#1}}}$}}

\renewcommand{\leq}{\leqslant}
\renewcommand{\geq}{\geqslant}

\renewcommand{\ge}{\geq}

\newcommand{\rs}{{{Ruzsa-Szemerédi}}\xspace}

\newcommand{\Leq}[1]{\ensuremath{\underset{\textnormal{#1}}\leq}}

\newcommand{\Eq}[1]{\ensuremath{\underset{\textnormal{#1}}=}}

\newcommand{\tvd}[2]{\ensuremath{\norm{#1 - #2}_{\mathrm{tvd}}}}
\newcommand{\Ot}{\ensuremath{\widetilde{O}}}
\newcommand{\eps}{\ensuremath{\varepsilon}}
\newcommand{\Paren}[1]{\Big(#1\Big)}
\newcommand{\Bracket}[1]{\Big[#1\Big]}
\newcommand{\bracket}[1]{\left[#1\right]}
\newcommand{\paren}[1]{\ensuremath{\left(#1\right)}\xspace}
\newcommand{\card}[1]{\left\vert{#1}\right\vert}

\newcommand{\IR}{\ensuremath{\mathbb{R}}}

\newcommand{\IN}{\ensuremath{\mathbb{N}}}

\newcommand{\norm}[1]{\ensuremath{\|#1\|}}

\newcommand{\expect}[1]{\Exp\bracket{#1}}
\newcommand{\var}[1]{\textnormal{Var}\bracket{#1}}
\newcommand{\cov}[1]{\textnormal{Cov}\bracket{#1}}

\newcommand{\set}[1]{\ensuremath{\left\{ #1 \right\}}}
\newcommand{\poly}{\mbox{\rm poly}}
\newcommand{\polylog}{\mbox{\rm  polylog}}

\newcommand{\alg}{\ensuremath{\mathcal{A}}\xspace}

\DeclareMathOperator*{\Exp}{\ensuremath{{\mathbb{E}}}}
\DeclareMathOperator*{\Prob}{\ensuremath{\textnormal{Pr}}}
\renewcommand{\Pr}{\Prob}

\newcommand{\Ex}{\Exp}

\newenvironment{tbox}{\begin{tcolorbox}[
		enlarge top by=5pt,
		enlarge bottom by=5pt,
		 breakable,
		 boxsep=0pt,
                  left=4pt,
                  right=4pt,
                  top=10pt,
                  arc=0pt,
                  boxrule=1pt,toprule=1pt,
                  colback=white
                  ]
	}
{\end{tcolorbox}}

\newcommand{\event}{\ensuremath{\mathcal{E}}}

\newcommand{\rv}[1]{\ensuremath{{\mathsf{#1}}}\xspace}
\newcommand{\rA}{\rv{A}}
\newcommand{\rB}{\rv{B}}
\newcommand{\rC}{\rv{C}}
\newcommand{\rD}{\rv{D}}

\newcommand{\rX}{\rv{X}}
\newcommand{\rY}{\rv{Y}}
\newcommand{\rZ}{\rv{Z}}

\newcommand{\rProt}{\rv{\Prot}}
\newcommand{\rR}{\rv{R}}

\newcommand{\supp}[1]{\ensuremath{\textnormal{\text{supp}}(#1)}}
\newcommand{\distribution}[1]{\ensuremath{\textnormal{dist}(#1)}\xspace}

\newcommand{\kl}[2]{\ensuremath{\mathbb{D}(#1~||~#2)}}
\newcommand{\II}{\ensuremath{\mathbb{I}}}
\newcommand{\HH}{\ensuremath{\mathbb{H}}}
\newcommand{\mi}[2]{\ensuremath{\def\mione{#1}\def\mitwo{#2}\mireal}}
\newcommand{\mireal}[1][]{
  \ifx\relax#1\relax%
    \II(\mione \,; \mitwo)%
  \else%
    \II(\mione \,; \mitwo\mid #1)%
  \fi
}
\newcommand{\en}[1]{\ensuremath{\HH(#1)}}

\newcommand{\itfacts}[1]{\Cref{fact:it-facts}-(\ref{part:#1})\xspace}

\newcommand{\cc}[1]{\ensuremath{\textsc{CC}(#1)}}

\newcommand{\ic}[2]{\ensuremath{\textsc{IC}(#1,#2)}}

\newcommand{\PP}{\ensuremath{\mathbb{P}}}

\newcommand{\GG}{\ensuremath{\mathcal{G}}}
\newcommand{\cH}{\ensuremath{\mathcal{H}}}

\newcommand{\prot}{\ensuremath{\pi}}
\newcommand{\Prot}{\ensuremath{\Pi}}

\newcommand{\jstar}{\ensuremath{j^{\star}}}
\newcommand{\istar}{\ensuremath{i^{\star}}}

\newcommand{\kstar}{\ensuremath{k^{\star}}}

\newcommand{\cX}{\ensuremath{\mathcal{X}}}
\newcommand{\cY}{\ensuremath{\mathcal{Y}}}

\newcommand{\player}[1]{\ensuremath{{P}_{#1}}}

\newcommand{\AHM}{\ensuremath{\mathsf{AHM}}\xspace}
\newcommand{\ahm}{\ensuremath{\AHM}}
\newcommand{\suc}[1]{\ensuremath{\textnormal{suc}(#1)}}

\newcommand{\DD}{\ensuremath{\mathcal{D}}}
\newcommand{\Aa}[1]{\ensuremath{A}^{(#1)}}
\newcommand{\rAa}[1]{\ensuremath{\rv{A}}^{(#1)}}
\newcommand{\Bb}[1]{\ensuremath{B}^{(#1)}}
\newcommand{\rBb}[1]{\ensuremath{\rv{B}}^{(#1)}}
\newcommand{\Xx}{\ensuremath{X}}
\newcommand{\Yy}{\ensuremath{Y}}
\newcommand{\rRa}{\ensuremath{\rv{R}}_a}
\newcommand{\rRcomp}{\ensuremath{\rv{R}}_\text{comp}}

\newcommand{\sigmar}{\ensuremath{\sigma_{\textnormal{\textsc{r}}}}}
\newcommand{\sigmac}{\ensuremath{\sigma_{\textnormal{\textsc{c}}}}}

\newcommand{\rsigmar}{\ensuremath{\bm{\sigma}}_{\textnormal{\textsc{r}}}}
\newcommand{\rsigmac}{\ensuremath{\bm{\sigma}}_{\textnormal{\textsc{c}}}}

\newcommand{\edgeset}{\ensuremath{\textnormal{\textsc{edges}}}}
\newcommand{\Eadd}{\ensuremath{E_\textnormal{ins}}}
\newcommand{\Edel}{\ensuremath{E_\textnormal{del}}}

\newcommand{\allone}{\ensuremath{\bm{1}}}
\newcommand{\allzero}{\ensuremath{\bm{0}}}

\newcommand{\expo}[2]{\ensuremath{\paren{#1} \uparrow \paren{#2}}}

\newcommand{\jpyconst}{\ensuremath{c_{\textsc{jpy}}}}

\newcommand{\alicepart}[1]{\ensuremath{A^{(#1)}}}
\newcommand{\bobpart}[1]{\ensuremath{B^{(#1)}}}

\newcommand{\speckstar}[1]{\ensuremath{k^{\star}_{#1}}}
\newcommand{\Bcommon}{\ensuremath{B^{\textnormal{\textsc{off}}}}}
\newcommand{\Acommon}{\ensuremath{A^{\textnormal{\textsc{off}}}}}
\newcommand{\Brest}{\ensuremath{B^{\textnormal{\textsc{rest}}}}}
\newcommand{\advconst}{\ensuremath{c_{\textsc{adv}}}}

\newcommand{\protnum}[1]{\ensuremath{\prot^{(#1)}}}
\newcommand{\Aastar}{\ensuremath{A^{\star}}}
\newcommand{\Bbstar}{\ensuremath{B^{\star}}}
\newcommand{\Bspec}{\ensuremath{B^{\textnormal{\textsc{spec}}}}}
\newcommand{\Aspec}{\ensuremath{A^{\textnormal{\textsc{spec}}}}}
\newcommand{\rAstar}{\ensuremath{\rv{A}^{\star}}}
\newcommand{\rBstar}{\ensuremath{\rv{B}^{\star}}}
\newcommand{\rBcommon}{\rv{B}^{\textnormal{\textsc{off}}}}
\newcommand{\rAcommon}{\rv{A}^{\textnormal{\textsc{off}}}}
\newcommand{\rkstar}{\rv{k}^{\star}}
\newcommand{\rAspec}{\rv{A}^{\textnormal{\textsc{spec}}}}
\newcommand{\rBspec}{\rv{B}^{\textnormal{\textsc{spec}}}}

\newcommand{\rBrest}{\ensuremath{\rv{B}^{\textnormal{\textsc{rest}}}}}

\newcommand{\DDreal}{\ensuremath{\DD^{\textnormal{\textsc{real}}}}}
\newcommand{\DDfake}{\ensuremath{\DD^{\textnormal{\textsc{fake}}}}}

\newcommand{\rhor}{\ensuremath{\rho_\textnormal{\textsc{r}}}}
\newcommand{\rhoc}{\ensuremath{\rho_{\textnormal{\textsc{c}}}}}

\newcommand{\rrhor}{\ensuremath{\bm{\rho}_\textnormal{\textsc{r}}}}
\newcommand{\rrhoc}{\ensuremath{\bm{\rho}_\textnormal{\textsc{c}}}}

\newcommand{\Srow}{\ensuremath{T_{\textnormal{\textsc{row}}}}}
\newcommand{\Scol}{\ensuremath{T_{\textnormal{\textsc{col}}}}}

\newcommand{\rSrow}{\ensuremath{\rv{T}_{\textnormal{\textsc{row}}}}}
\newcommand{\rScol}{\ensuremath{\rv{T}_{\textnormal{\textsc{col}}}}}

\newcommand{\leftv}[2]{\ensuremath{\ell^{#2}_{#1}}}
\newcommand{\rightv}[2]{\ensuremath{r^{#2}_{#1}}}

\newcommand{\graphed}{\ensuremath{\textnormal{\textsc{graph}}}}

\newcommand{\apx}{\ensuremath{\beta}}

\newcommand{\Lspec}{\ensuremath{L^{\textnormal{\textsc{spec}}}}}
\newcommand{\Rspec}{\ensuremath{R^{\textnormal{\textsc{spec}}}}}
\newcommand{\Bbase}{\ensuremath{B^{\textnormal{\textsc{base}}}}}
\newcommand{\Abase}{\ensuremath{A^{\textnormal{\textsc{base}}}}}
\newcommand{\Lbase}{\ensuremath{L^{\textnormal{\textsc{base}}}}}
\newcommand{\Rbase}{\ensuremath{R^{\textnormal{\textsc{base}}}}}

\newcommand{\Hspec}{\ensuremath{H^{\textnormal{\textsc{spec}}}}}

\newcommand{\Hbase}{\ensuremath{H^{\textnormal{\textsc{base}}}}}

\renewcommand{\Abase}{\ensuremath{A^{*}}}
\renewcommand{\Lbase}{\ensuremath{L^{*}}}
\renewcommand{\Bbase}{\ensuremath{B^{*}}}
\renewcommand{\Rbase}{\ensuremath{R^{*}}}

\title{Settling the Pass Complexity of Approximate Matchings \\ in Dynamic Graph Streams} 
\author{Sepehr Assadi\footnote{(sepehr@assadi.info) Cheriton School of Computer Science, University of Waterloo. 
Supported in part by a  Sloan Research Fellowship, an NSERC
Discovery Grant, a University of Waterloo startup grant, and a Faculty of Math Research Chair grant. \smallskip} 
\\ {\small University of Waterloo} \and
Soheil Behnezhad\footnote{(s.behnezhad@northeastern.edu) Khoury College of Computer Sciences, Northeastern University. \smallskip}
 \\ {\small Northeastern University} \and 
Christian Konrad\footnote{(christian.konrad@bristol.ac.uk) School of Computer Science, University of Bristol, UK. 
Supported by EPSRC New Investigator Award EP/V010611/1. \smallskip} \\ {\small University of Bristol} \and
Kheeran K. Naidu\footnote{(kheeran.naidu@bristol.ac.uk) School of Computer Science, University of Bristol, UK.
Supported by EPSRC Doctoral Training Studentship EP/T517872/1. \smallskip} \\ {\small University of Bristol} \and
Janani Sundaresan\footnote{(jsundaresan@uwaterloo.ca) Cheriton School of Computer Science, University of Waterloo. Supported in part by a David R. Cheriton Scholarship from Cheriton School of Computer Science, Faculty of Math Graduate Research Excellence Award, and Sepehr Assadi's NSERC Discovery Grant. \smallskip} \\ {\small University of Waterloo}
}

\date{}

\begin{document}
\maketitle


\begin{abstract}

\medskip

A semi-streaming algorithm in \textbf{dynamic graph streams} processes any $n$-vertex graph by making one or multiple passes over a stream of insertions and deletions to edges of the graph and
using $O(n \cdot \polylog{(n)})$ space. Semi-streaming algorithms for dynamic streams were first obtained 
in the seminal work of Ahn, Guha, and McGregor in 2012, alongside the introduction of the \emph{graph sketching} technique, 
which remains the de facto way of designing algorithms in this model and a highly popular technique for designing graph algorithms in general. 

\medskip

We settle the pass complexity of approximating \textbf{maximum matchings} in dynamic streams via semi-streaming algorithms by improving the state-of-the-art in \emph{both} upper and lower bounds: 

\begin{itemize}[leftmargin=10pt]
	\item We present a randomized sketching based semi-streaming algorithm for $O(1)$-approximation of maximum matching in dynamic streams using $O(\log\log{n})$ passes. 
	The approximation ratio of this algorithm can be improved to $(1+\eps)$ for any fixed $\eps > 0$ even on weighted graphs using standard techniques.

	This exponentially improves upon
	several $O(\log{n})$ pass algorithms developed for this problem since the introduction of the dynamic graph streaming model.

	\item We prove that any semi-streaming algorithm (not only sketching based) for $O(1)$-approximation of maximum matching in dynamic streams requires $\Omega(\log\log{n})$ passes. 
	
	This presents the first multi-pass lower bound for this problem, which is already also optimal, settling a longstanding open question in this area.  
\end{itemize}

\end{abstract}

\pagenumbering{roman}

\clearpage

\setcounter{tocdepth}{3}
\tableofcontents
\clearpage
\pagenumbering{arabic}
\setcounter{page}{1}

\clearpage


\section{Introduction}

In the \textbf{dynamic graph streaming} model, we have a graph $G=(V,E)$ with vertices $V := [n]$. 
The edges in $E$ are defined by a sequence of insertions and deletions in a stream $\sigma := (\sigma_1,\ldots,\sigma_N)$ of length $N$ which
is often assumed to be some $\poly{(n)}$. Each entry $\sigma_i$ of the stream is either inserting a \emph{new} edge $(u_i,v_i)$ to $E$ or deleting an \emph{already inserted} edge from it. 
The goal is to make one or a few passes over the stream, use a limited memory---ideally, $O(n \cdot \poly\!\log{(n)})$ bits, referred to as \textbf{semi-streaming space}---and
compute the answer to a given problem on the graph $G$ at the end of the last pass. We focus on the \emph{maximum matching} problem in this model. 

Maximum matching is arguably the most studied problem in the graph streaming model at this point (including both dynamic and insertion-only streams); we refer the interested reader to~\cite{AssadiS23} that lists various lines of work on this problem. 
The history of this problem, focusing solely on $O(1)$-approximation algorithms and in \emph{dynamic} graph streams, is as follows: 

\begin{itemize}[itemsep=2pt, leftmargin=10pt]
	\item The first such algorithms for matchings were obtained in~\cite{AhnGM12a} alongside the introduction of the dynamic graph streaming model itself. 
	The authors in~\cite{AhnGM12a} observed that the prior techniques of~\cite{LattanziMSV11} (in the MapReduce/MPC model) also imply an $O(\log{n})$-pass semi-streaming algorithm for $2$-approximation of maximum matching in dynamic streams.%
	\footnote{Technically, the algorithms of~\cite{LattanziMSV11,AhnGM12a}, and some subsequent ones, use $n^{1+1/p}$-space in $O(p)$ passes. This translates 
	to an $O(\frac{\log\!{(n)}}{\log\!\log\!{(n)}})$-pass algorithm in semi-streaming space. Yet, to keep the focus on the bigger picture, we 
	ignore this lower-order term improvement and refer to these algorithms as $O(\log{n})$ passes still.}
	\item In the same work,~\cite{AhnGM12a}, building on~\cite{AhnG11}, further improved the approximation ratio of the algorithm of~\cite{LattanziMSV11} to $(1+\eps)$-approximation for any \emph{fixed} $\eps > 0$ using $O(\log^2\!{(n)})$ passes
	on general graphs and $O(\log{n})$ passes on bipartite graphs. 
	\item The algorithms of~\cite{AhnGM12a} were subsequently improved in~\cite{AhnG15} to an $O(\log{n})$-pass algorithm for $(1+\eps)$-approximation even on weighted (general) graphs. 
	Very recently, this algorithm was simplified and slightly improved in~\cite{Assadi24}. Yet another algorithm with similar guarantees for unweighted bipartite graphs was obtained in~\cite{AssadiJJST22}. 
	
	\item In addition, some generic reductions from general to bipartite~\cite{McGregor05,Tirodkar18} or weighted to unweighted matchings~\cite{GamlathKMS19,BernsteinDL21} developed over the years 
	can be applied to the algorithms of~\cite{LattanziMSV11} to obtain $O(\log{n})$-pass algorithms for $(1+\eps)$-approximation for fixed $\eps > 0$.\footnote{These algorithms generally have a (much) worse dependence on the parameter $\eps$ 
	compared to the ones in the bullet point above, but for constant $\eps > 0$, their guarantees are still asymptotically the same.} 

	 \item In parallel to the line of work on multi-pass algorithms, a series of work studied single-pass algorithms for this problem~\cite{Konrad15,ChitnisCHM15,AssadiKLY16,ChitnisCEHMMV16,AssadiKL17,DarkK20,AssadiS22}. In particular,~\cite{AssadiKLY16}
	 proved that any $O(1)$-approximation of matchings via single-pass algorithms requires $n^{2-o(1)}$ space and~\cite{DarkK20} improved this to an optimal $\Omega(n^2)$ space lower bound.  
\end{itemize}
This constitutes the state-of-the-art for matchings in dynamic graph streams: 
\begin{quote}
\emph{For semi-streaming algorithms on dynamic streams, $O(1)$-approximation to maximum matching is possible in $\approx \log{n}$ passes and not possible in a single pass.} 
\end{quote}
Closing this huge gap between upper and lower bounds for dynamic streaming matchings has been a longstanding open question in the graph streaming literature. 
This is precisely the contribution of our work: we fully settle the pass complexity of $O(1)$-approximation of maximum matching in dynamic streams by improving \textbf{both} the upper and lower bounds for this problem.

\subsection{Our Contributions}

Our first main result shows that surprisingly---at least to the authors---the \emph{right} answer to the problem is \emph{not} even close to $\approx \log{n}$ passes:
one can exponentially improve the pass complexity of 
different algorithms developed for this problem in~\cite{LattanziMSV11,AhnGM12a,AhnG15,AssadiJJST22,Assadi24}. 

\begin{result}\label{res:upper}
	There is a randomized $O(\log\log{n})$-pass $O(1)$-approximation semi-streaming algorithm for the maximum matching problem in dynamic streams. 
	The result continues to hold even for $(1+\eps)$-approximation of weighted (general) matching for any constant $\eps > 0$. 
\end{result}

We find our main contribution in~\Cref{res:upper} to be the $O(1)$-approximation algorithm, which relies on different sets of techniques compared to the prior work on this problem.
The improvement to $(1+\eps)$-approximation and weighted graphs follows from this novel algorithm using the existing reductions developed in~\cite{McGregor05,GamlathKMS19}. 
We emphasize that previously, no better than $O(\frac{\log{n}}{\log\log{n}})$-pass dynamic semi-streaming algorithms were known even for $\poly\!\log\!{(n)}$-approximation of matchings in its simplest form, namely, for 
unweighted bipartite matching. 


Prior to our work, the only other problem with a similar pass complexity in dynamic streams that we are aware of is the maximal independent set (MIS) problem, which also admits an $O(\log\log{n})$ pass algorithm~\cite{AhnCGMW15} (this 
result is related to ours as we will discuss in~\Cref{sec:techniques}). For the maximum matching problem itself, the best approximation ratio achievable by $O(\log\log{n})$-pass 
algorithms was $n^{O(1/\!\log\log{n})}$-approximation that follows from~\cite[Theorem 4.6]{DobzinskiNO14}. 

Furthermore, in addition to our algorithmic improvement, we can also improve the single-pass lower bounds of~\cite{AssadiKLY16,DarkK20} all the way to $\Omega(\log\log{n})$ passes. 

\begin{result}\label{res:lower}
	Any randomized semi-streaming algorithm for $O(1)$-approximation of maximum matching in dynamic streams with constant probability of success 
	requires $\Omega(\log\log{n})$ passes. The lower bound holds even on (unweighted) bipartite graphs. 
\end{result}

The only other semi-streaming lower bounds of similar nature are the very recent  $\Omega(\log\log{n})$-pass and $\Omega(\log{n})$-pass lower bounds for, respectively, MIS in insertion-only streams~\cite{AssadiKNS24} 
and exact minimum spanning tree (MST) in dynamic streams~\cite{AssadiKZ24} (these work are related to ours technique-wise and we shall discuss them in~\Cref{sec:techniques}). 
For the maximum matching problem itself, we only knew $\Omega(\log{n})$-pass lower bounds for finding \emph{exact} maximum matchings~\cite{GuruswamiO13,ChenKPSSY21a,AssadiS23} and a conditional $\Omega(\log{(1/\eps)})$-pass lower bound for $(1+\eps)$-approximation for small constant $\eps \in (0,1)$~\cite{AssadiS23};
see also~\cite{KonradN21,Assadi22,KonradN24} for two-pass lower bounds for small approximation ratios (way) below $2$; all these lower bounds for matchings hold even for insertion-only streams\footnote{The focus of these results is qualitatively different than ours; in insertion-only stream, obtaining a $2$-approximation is trivial in a single pass, whereas in dynamic streams, 
the whole question is on obtaining \emph{some} $O(1)$-approximation.}. 

Proving multi-pass semi-streaming lower bounds has been generally a challenging question (compared to the wealth of single-pass lower bounds; see the short survey in~\cite{Assadi23} for some discussion of this topic). 
With a few notable exceptions~\cite{FeigenbaumKMSZ08,GuruswamiO13}, ``strong'' multi-pass semi-streaming lower bounds have only been obtained very recently for different problems, starting from two-pass algorithms~\cite{AssadiR20,ChenKPSSY21b,Assadi22,KonradN24} and now even for multi-pass ones~\cite{AssadiCK19a,ChenKPSSY21a,ChakrabartiGMV20,AssadiS23,AssadiGLMM24}. 
\Cref{res:lower} also contributes to this line of work and is among the very few \emph{optimal} lower bounds (together with~\cite{AssadiKNS24,AssadiKZ24}). 

In conclusion,~\Cref{res:upper} and~\Cref{res:lower} together establish that the optimal pass complexity of approximate matchings in dynamic graph streams 
is $\Theta(\log\log{n})$ passes.

\begin{Remark}\label{rem:mpc}
Before moving on from our results, a quick detour is in order.  Similar to \emph{all} other dynamic graph streaming algorithms, our algorithm is based on the \textbf{graph sketching} technique (see, e.g.~\cite{AhnGM12a} for the definition). 
	Our~\Cref{res:upper}, put differently, states that:
	\begin{quote}
	\emph{There is an adaptive sketching algorithm that in $O(\log\log{n})$-rounds and $\Ot(n)$-size sketches can recover an $O(1)$-approximate matching with high probability.}
	\end{quote} 
	Such a result is interesting on its own given the generality of graph sketching and its implications to other models as well. 
	
	For instance, this implies a \emph{Massively Parallel Computation (MPC)} algorithm for approximating matchings in $O(\log\log{n})$ rounds with
	machines of $\Ot(n)$ memory (even $O(n/\poly\log{(n)})$ memory) and only $\Ot(n)$ working memory. Prior work in~\cite{CzumajLMMOS18,GhaffariGKMR18,AssadiBBMS19,BehnezhadHH19} achieved MPC algorithms with similar guarantees 
	using various other techniques (and to our knowledge, all with $n^{1+\Omega(1)}$ working memory). Our result shows that graph sketching technique itself, which is one of the oldest techniques in the MPC model as well, 
	can achieve such bounds in a conceptually simpler way (and with the additional benefit of using a smaller working memory). As this is not the focus of the paper, we omit the definition and details of the model and instead
	refer the interested reader to the aforementioned papers for more details. 
\end{Remark}

\subsection{Our Techniques}\label{sec:techniques}

Our upper and lower bounds are intimately connected to each other by looking at matching through the lens of \emph{maximal independent sets}. Specifically:
\begin{itemize}[itemsep=2pt, leftmargin=10pt]
\item Our upper bound uses the $O(\log\log{n})$-pass semi-streaming algorithm of~\cite{AhnCGMW15} for MIS as a subroutine (in a non blackbox way) and borrows ideas and inspiration from the recent work of~\cite{Veldt24} that relates MIS to the vertex cover problem (the dual problem of maximum matching);
\item Our lower bound builds on and adapts the recent communication complexity techniques developed for proving an $\Omega(\log\log{n})$-pass semi-streaming lower bound for MIS in~\cite{AssadiKNS24}. 
\end{itemize}
We briefly discuss the techniques behind our work in this subsection and postpone a more elaborate discussion to our technical overview in~\Cref{sec:overview}. 

\paragraph{Upper bound.} The first main technical ingredient of our algorithm in~\Cref{res:upper} is a model-independent reduction from $O(1)$-approximate \emph{fractional} matchings to the \emph{randomized greedy MIS} algorithm\footnote{This is the algorithm that iterates over the vertices in a random order and greedily adds a vertex to the MIS as long as none of its neighbors that appear before it in the ordering are already chosen in the MIS.}. This reduction is inspired by the brilliant recent
work of~\cite{Veldt24} that showed that the complement of the randomized greedy MIS is a $2$-approximate vertex cover in expectation! 
On the other hand, we show that in every step of the randomized greedy MIS, we can assign a fractional value to the edges in the \emph{2-hop} neighborhood of vertices that join the MIS, to instead form
a large fractional matching in expectation. 
It is worth noting that~\cite{Veldt24}, similar to us, relies on a primal-dual analysis and exhibits a $2$-approximate fractional matching in the reduction; nevertheless, that fractional matching 
is only an analytical tool and in fact is a function of the randomized greedy MIS over \emph{all} possible ordering of vertices and cannot be found by an algorithm\footnote{We shall note that while~\cite{Veldt24} has been
an important source of inspiration for us---and in the first place suggested to us that randomized greedy MIS might also be relevant to approximate matchings---our specific reduction and the techniques in its analysis
are almost entirely disjoint from~\cite{Veldt24}; see~\Cref{sec:overview} for a more detailed comparison.}.

The second main technical ingredient of our algorithm is a \emph{partial} implementation of the above reduction in $O(\log\log{n})$ passes of dynamic streams. To do this, 
we rely on the semi-streaming implementation of the randomized greedy MIS in $O(\log\log{n})$ passes by~\cite{AhnCGMW15}. We show that we can run this algorithm and additionally collect enough auxiliary information 
to also be able to somewhat recover the associated fractional matching defined in the reduction as well. The challenge is that this fractional matching, quite crucially, works with edges in the
 \emph{2-hop} neighborhood of vertices that join the MIS; these edges however are not even visited by the randomized greedy MIS algorithm and its simulation in~\cite{AhnCGMW15}. As a result, we are only able to work
 with them through certain ``proxy'' edges that we can sample algorithmically, and then delegate some part of the computation of this fractional matching to the analysis instead (the algorithm itself only returns a maximum matching
 of the sampled edges). 
 
 \paragraph{Lower bound.} Our lower bound follows the recently-developed \emph{hierarchical embedding} technique of~\cite{AssadiKNS24} (inspired by~\cite{KonradN24}) that $(i)$ creates hard instances for $p$-pass streaming algorithms for a problem $P$
 by embedding \emph{many} $(p-1)$-pass hard instances of $P$ in a single graph $G$; and, $(ii)$ applies a \emph{generalized round elimination} argument to prove the lower bound (see~\cite{MiltersenNSW95} for the original round elimination
 and~\cite{AssadiKNS24} for its generalization). There are two main differences in implementing this strategy in our work compared to~\cite{AssadiKNS24} however. 
 
 The first key difference is in the \emph{combinatorial} construction of hard instances.~\cite{AssadiKNS24} designed a family of extremal graphs, based on a generalization of \rs (RS) graphs~\cite{RuzsaS78}, that pack \underline{many} \emph{induced} collections of vertex-disjoint 
 \emph{``small''} graphs inside a single \emph{``base''} graph. This allows them to embed the $(p-1)$-pass hard instances as small graphs inside a single $p$-pass hard instance as the base graph. The inducedness guarantee of the base 
 graph now ensures that these embedded $(p-1)$-pass hard instances do not interfere with each other (e.g., do not add edges between vertices of each other) and thus remain hard even inside a single graph. 
 Instead, we create our hard instances by exploiting the power of edge deletions following the approach of~\cite{DarkK20} for proving 
 \emph{single-pass} dynamic streaming lower bounds for approximate matchings. This way, our hard instances consists of a stream that inserts many $(p-1)$-pass hard instances together, not  necessarily with any induced subgraph collections (unlike~\cite{AssadiKNS24}), followed by deletions of edges in many of these instances so that effectively only \underline{one} large induced collection of $(p-1)$-pass instances remain. We show that in the context of $O(1)$-approximate matchings, 
 this is enough to force any algorithm for the $p$-pass instance to also solve many $(p-1)$-pass hard instances. 
 
 The second key difference is in the \emph{information-theoretic} arguments. The new round elimination argument established in~\cite{AssadiKNS24} crucially relies on
 the \emph{independence} of the inputs of players in the corresponding communication game used to establish the streaming lower bound. In contrast, such an independence cannot hold for us given that 
 we need to ensure the input of one player is only deleting edges already inserted by another player (otherwise, the stream may delete edges that have not been inserted). Addressing this issue requires
 a careful sharing of the input of players with each other to guarantee that no not-inserted edge gets deleted, while making sure there is also not too much correlation between their inputs (correlation generally makes the task of proving these lower bounds
 harder or even impossible). 
 This part borrows ideas from the recent work of~\cite{AssadiKZ24} in proving multi-pass dynamic streaming lower bounds for MST. Finally, proving the general round elimination argument 
 with these restrictions also requires a different \emph{direct-sum} result
 based on \emph{\underline{internal} information complexity}~\cite{BarakBCR10} and a corresponding \emph{message compression} argument for internal information~\cite{JainPY16} (in contrast to the \underline{external} information complexity direct sum~\cite{ChakrabartiSWY01} used in~\cite{AssadiKNS24} and its own message compression technique from~\cite{HarshaJMR07}).

\clearpage


\section{Technical Overview}\label{sec:overview}

We use this section to unpack the main ideas behind our work and give a streamlined overview
of our approach. This section oversimplifies many details and the discussions
will be informal for the sake of intuition. Thus, while this section provides ample intuition and introduction to our approach, the rest of the paper is written in an independent way, and the reader 
can entirely skip this section and directly jump to technical arguments. Moreover, given the disjoint sets of techniques used in our upper bounds versus lower bounds, 
the following two subsections are entirely independent of each other and can be read in any order (this is also true of the rest of the paper). 

\subsection{Overview of Upper Bound}

The starting point of our algorithm is a recent reduction of~\cite{Veldt24} from \emph{vertex cover} to \emph{randomized greedy MIS}, and a decade-old result of~\cite{AhnCGMW15} that finds
the randomized greedy MIS in dynamic streams in $O(\log\log{n})$ passes. Let us start with a quick overview of these works.

\subsubsection{Prior Work in~\cite{AhnCGMW15,Veldt24}}

\paragraph{Dynamic-streaming MIS algorithm of~\cite{AhnCGMW15}} Recall that in the randomized greedy MIS, we go over vertices in a random order, pick the first vertex in the MIS $\mathcal{I}$, remove all its neighbors from consideration from now on, 
and continue this way until we have visited all vertices.  

The algorithm of~\cite{AhnCGMW15} is based on the following key observation: the \emph{effective} degree of vertices, their degrees to not-yet-removed vertices, drops quite rapidly as we go through 
the random ordering of the vertices (this is often referred to as the ``residual sparsity property'' of the greedy algorithm~\cite{AhnCGMW15,GhaffariGKMR18,Konrad18,AssadiOSS19}). Specifically, by the time we are processing
the $k$-th vertex, we expect the degree of each remaining vertex to be $\lesssim n/k$. Intuitively, this is because a high degree vertex has a high chance of becoming a neighbor to one of the first $k$ random
vertices in the beginning of this ordering and thus be removed itself (the actual argument is more nuanced because not all of the first $k$ vertices of the ordering actually join the MIS; see~\Cref{lem:residual}).  

\cite{AhnCGMW15} uses this property to simulate running the randomized greedy MIS in \textbf{batches}: pick a random ordering $\sigma$ of vertices and let $U_1$ be the first batch of $\simeq n^{1/2}$ vertices of this ordering. 
Store all edges between them in a single pass using sparse-recovery (see~\Cref{prop:sparse-recovery}) since this subgraph can only have $\simeq n$ edges. 
Using these edges, we can identify which vertices in $U_1$ will join the MIS in the algorithm, say, set $\mathcal{I}_1 \subseteq U_1$. Go over the stream one more time and this time mark each vertex 
that is neighbor to $\mathcal{I}_1$ as removed (this can be done by maintaining a counter for each vertex, to count the total number of insertions and deletions of its incident edges to $\mathcal{I}_1$). 

At this point, we have simulated the first $\simeq n^{1/2}$ iterations of the algorithm in $O(1)$ passes. The residual sparsity property implies that degree of remaining vertices is only $\lesssim n^{1/2}$. 
This means that we can now consider the next batch of $\simeq n^{3/4}$ vertices of the ordering as the set $U_2$ and with high probability still be able to store all their edges in $\simeq n$ space\footnote{Sampling $\simeq n^{3/4}$ vertices
randomly or alternatively sampling each vertex w.p. $\simeq n^{-1/4}$ (implied by the random ordering of $\sigma$) means each sampled vertex only has $\lesssim n^{-1/4} \cdot n^{1/2} = n^{1/4}$ neighbors in the sample.}, and compute the 
independent set $\mathcal{I}_2 \subseteq U_2$. 
We can thus continue like this with batches $U_2,U_3,\ldots,U_t$ and by the time $t \simeq \log\log{n}$, all vertices are processed. This leads to an $O(\log\log{n})$ pass semi-streaming algorithm that with high probability 
 simulates the randomized greedy MIS faithfully and outputs the same MIS. 

\paragraph{Model-independent reduction of~\cite{Veldt24}.} A straightforward fact about any independent set of any graph $G$ is that its complement must be a vertex cover. \cite{Veldt24} made a beautiful discovery that 
the complement of the randomized greedy MIS on any graph $G$ is in fact a \emph{$2$-approximate} vertex cover of $G$ in expectation! This allows for ``translating'' many of the nice properties of the randomized greedy MIS
for obtaining a $2$-approximation of vertex cover as well. 

The proof of this result is an elegantly simple application of LP duality. For a random order $\sigma$ of vertices, we say an edge $e \in E$ is \textbf{blocking} if one of its endpoints 
belongs to the MIS $\mathcal{I}$ and it is the first neighbor (in ordering $\sigma$) of the other endpoint that joins the vertex cover $\mathcal{C} = V \setminus \mathcal{I}$. I.e., $e$  is ``blocking'' this endpoint from joining $\mathcal{I}$ and places it in the vertex cover instead. 
Let $p_e$ denote the probability that $e$ is blocking where the probability is over the randomness of $\sigma$. We have: 
\begin{align}
	\begin{split}
	&\Exp\card{\mathcal{C}} = \sum_{e} p_e \\
	&\text{for all $v \in V$:} \quad \sum_{e \ni v} \frac{1}{2} \cdot p_e = \Pr\paren{v \in \mathcal{C}} \leq 1. \label{eq:Veldt24} 
	\end{split}
\end{align}
The first equation holds because for each vertex $v$ that joins $\mathcal{C}$, there is exactly one blocking edge incident on $v$. 
The second equation is more tricky and roughly holds because of the following: for any edge $e=(u,v)$, conditioned on $e$ being blocking, 
the probability that each of $u$ or $v$ belongs to the MIS is exactly half. This in turn implies that
\[
	\Pr\paren{v \in \mathcal{C}}  = \sum_{e \ni v} \Pr\paren{\text{$e$ is blocking} \wedge v \in \mathcal{C}} = \sum_{e \ni v} \frac{1}{2} \cdot p_e,
\]
where the first equality holds because these events are mutually exclusive (a vertex that joins $\mathcal{C}$ can only have one incident blocking edge; this is \emph{not} true of vertices that join $\mathcal{I}$). 
The upshot is that the assignment $p_e/2$ to every edge $e \in E$ is a \emph{fractional matching} of $G$ 
with total value exactly half the size of the vertex cover $\mathcal{C}$ of $G$. By duality of matching and vertex cover, these imply that $\mathcal{C}$ is a $2$-approximate vertex cover and $\set{p_e/2}_{e \in E}$ is a $2$-approximate fractional matching.

It is worth pointing out that a direct combination of the above two works implies a semi-streaming algorithm that finds  a $2$-approximate vertex cover in dynamic streams. We now discuss the challenges of extending these ideas to 
matching and how we address these challenges. 

\subsubsection{A Model-Independent Reduction from Matching to MIS}

\paragraph{A similar reduction as in~\cite{Veldt24} for matching?} While the reduction of~\cite{Veldt24} \emph{explicitly} finds a $2$-approximate vertex cover (in expectation), the fractional matching $\set{p_e/2}_{e\in E}$ introduced above 
is only an \emph{analysis} tool: the algorithm itself is not actually finding this fractional matching; while the algorithm can identify the set of blocking edges used in the definition of $p_e$'s for a \emph{single} run of the randomized greedy MIS, these
edges are quite far from being any matching (see~\Cref{fig:clique-alg}). 

\begin{figure}[h!]
	\centering
	 \begin{subfigure}[b]{0.3\textwidth}
	 \centering
	 \includegraphics[scale=0.35]{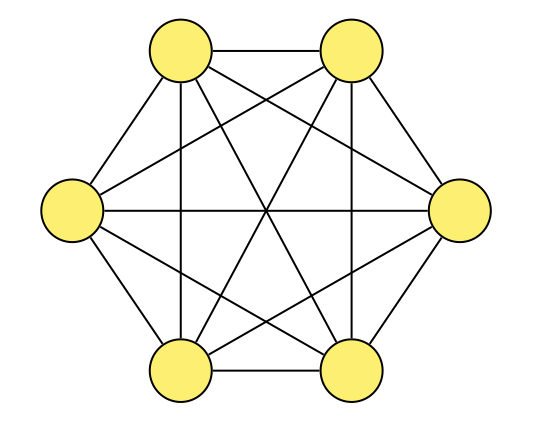}
	\caption{}
	 \end{subfigure}
	 \begin{subfigure}[b]{0.3\textwidth}
	 \centering
	 \includegraphics[scale=0.35]{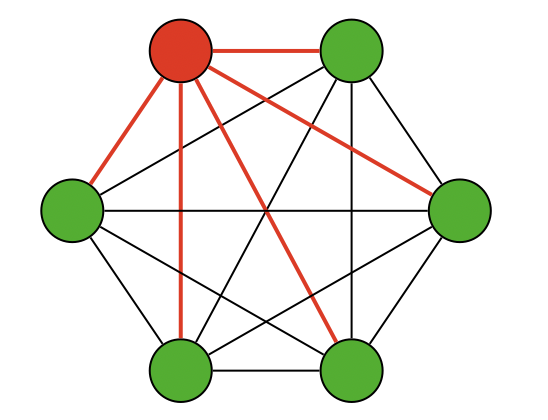}
	\caption{}
	 \end{subfigure}	
	 \begin{subfigure}[b]{0.3\textwidth}
	 \centering
	 \includegraphics[scale=0.35]{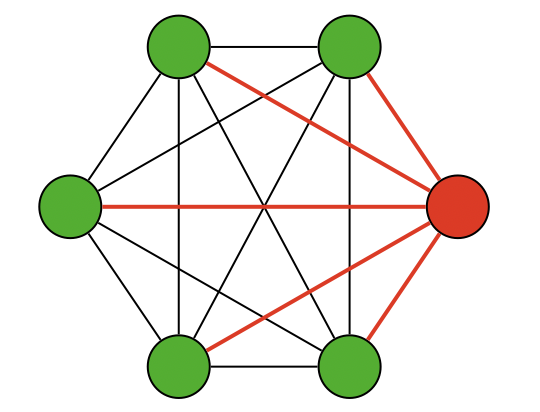}
	 \caption{}
	 \end{subfigure}
	 \caption{A clique in $(a)$ and two different sets of blocking edges in $(b)$ and $(c)$ based on different ordering of vertices in the randomized greedy MIS (the red vertex joins the MIS, green vertices join the vertex cover,
	 and red edges are blocking). While the \emph{probabilities} of edges becoming blocking are $2/n$ in a clique and can form a $2$-approximate fractional matching,  blocking edges in each \emph{single} run form stars and are very  far from matchings themselves.}\label{fig:clique-alg}
\end{figure}


\paragraph{Our approach.} We design a new scheme for finding a fractional matching from a \emph{single} run of the randomized greedy MIS. We start with the following assignment $x \in \IR^E$ to the edges: 
\begin{itemize}
	\item Whenever a vertex $v$ joins the vertex cover $\mathcal{C}$ (i.e., becomes incident to the MIS $\mathcal{I}$ for the first time), assign
	a value of $1/\deg(v)$ to every edge $(v,w)$ with $\deg{(w)} \leq \deg{(v)}$. Here, the graph considered (including number of vertices or their degrees) is the one obtained
	by removing all vertices and their incident edges added in the previous iterations to $\mathcal{I}$ and $\mathcal{C}$. See~\Cref{fig:graph-alg}.
\end{itemize}

\begin{figure}[h!]
	\centering
	 \begin{subfigure}[b]{0.22\textwidth}
	 \centering
	 \includegraphics[scale=0.35]{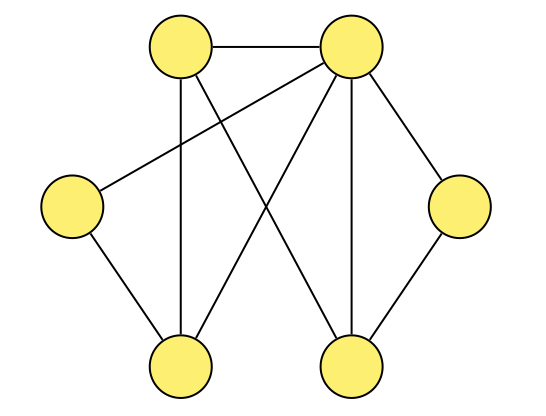}
	\caption{}
	 \end{subfigure}
	 \begin{subfigure}[b]{0.22\textwidth}
	 \centering
	 \includegraphics[scale=0.35]{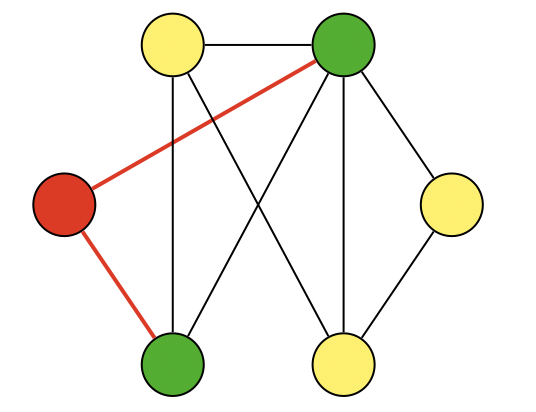}
	\caption{}
	 \end{subfigure}	
	 \begin{subfigure}[b]{0.22\textwidth}
	 \centering
	 \includegraphics[scale=0.35]{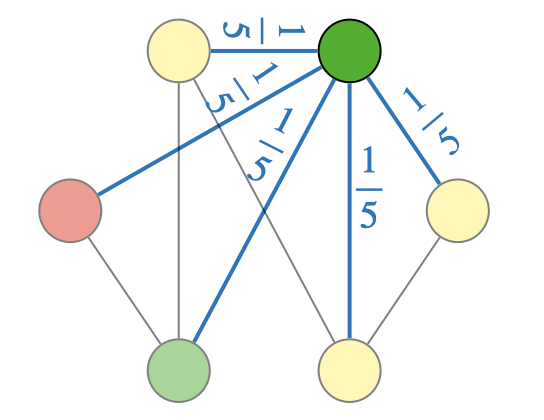}
	 \caption{}
	 \end{subfigure}
	 \begin{subfigure}[b]{0.22\textwidth}
	 \centering
	 \includegraphics[scale=0.35]{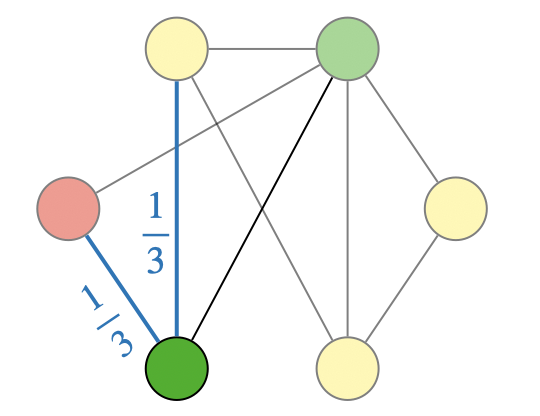}
	 \caption{}
	 \end{subfigure}
	 \caption{An illustration of our new fractional matching assignment $x$ in a single run of the randomized greedy MIS on the graph in $(a)$. Figure $(b)$ shows the blocking edges, plus the vertex that joins the MIS in this iteration (red), 
	 and its neighbors that join the vertex cover (green). In Figure $(c)$, we have the assignment of $1/5$ over all edges of a vertex that joins the vertex cover, but in Figure $(d)$, the assignment of $1/3$ misses one edge, since its other endpoint 
	 has a higher degree (unlike~\Cref{fig:clique-alg}, this figure shows a single iteration of the algorithm and not multiple runs).}\label{fig:graph-alg}
\end{figure}

For now, let us ignore the technical (but necessary) condition that we only assign a value to edges $(v,w)$ incident on $v$ with $\deg(w) \leq \deg(v)$ and instead assume we write $1/\deg(v)$ over all remaining edges of $v$ (we will  get back to this later). 
An easy observation now is that 
\begin{align}
	\sum_{e \in E} x_e = \card{\mathcal{C}}, \label{eq:tech-x-value}
\end{align}
 because whenever a vertex $v$ joins $\mathcal{C}$, we increase total $x$-values by $\deg(v) \cdot 1/\deg(v) = 1$. We want to establish that $x$ is also not ``far from'' a fractional matching, meaning
that the $x$-value incident on any one vertex is not (much) more than one (thus a rescaling turns it into a fractional matching).

Consider some vertex $w \in V$. There are two sources of $x$-value on edges of $w$: the ones that \emph{originate} from $w$ 
if and when it joins $\mathcal{C}$, and the ones \emph{borrowed} by $w$ from its neighbors $v$ that join $\mathcal{C}$ before $w$ is removed from the graph. The contribution of the first assignment is exactly $1$ as argued for~\Cref{eq:tech-x-value} and is thus bounded. As such, we only 
need to focus on the second assignment.

This is where our crucial observation lies: in each iteration of the algorithm, if the expected borrowed assignment by a vertex is ``high'', then this vertex also has an equally ``high'' chance of being removed from consideration after this iteration. 
Let us formalize this. Consider any iteration wherein we pick some vertex $u$ in the random ordering to join $\mathcal{I}$ in the algorithm. Firstly, 
\begin{align*}
	\Exp\bracket{\text{borrowed assignment of $w$ in this iteration}} &= \sum_{(v,w) \ni w} \Pr\paren{\text{$v$ joins $\mathcal{C}$ in this iteration}} \cdot \frac{1}{\deg(v)} \tag{as $v$ joining $\mathcal{C}$ results in assigning $1/\deg(v)$ value to one edge of $w$} \\
	&= \sum_{(v,w) \ni w} \frac{\deg(v)}{n} \cdot \frac{1}{\deg(v)} \tag{$v$ joins $\mathcal{C}$ if $u \in N(v)$ and by the choice of $\sigma$, vertex $u$ is chosen uniformly from remaining vertices} = \frac{\deg(w)}{n}.
\end{align*}
Secondly, 
\begin{align*}
	\Pr\paren{\text{$w$ is removed in this iteration}} &= \Pr\paren{\text{$w$ joins $\mathcal{I}$ or $\mathcal{C}$ in this iteration}} = \frac{\deg(w)+1}{n}, 
\end{align*}
where the second equality holds because $w$ joins $\mathcal{I}$ if $u=w$ and joins $\mathcal{C}$ if $u \in N(w)$. Combining the above two equations gives us that
for any vertex $w \in V$, 
\begin{align}
	\Exp\bracket{\text{borrowed assignment of $w$ in this iteration}} \leq \Pr\paren{\text{$w$ is removed in this iteration}} \label{eq:tech-borrowed}. 
\end{align}

This brings us to an interesting probabilistic question. Fix a vertex $w \in V$ and let $X^w_1,X^w_2,\ldots,X^w_n$ be $n$ random variables where $X^w_i$ is the borrowed assignment of $w$ in iteration $i$ of 
the randomized greedy MIS algorithm. We are interested in upper bounding: 
\begin{align}
\begin{split}
	&\expect{\sum_{i=1}^{n} X^w_i} \quad \text{and} \quad \var{\sum_{i=1}^n X^w_i}  \quad \text{subject to}\\
	\text{for all $i \in [n]$}: ~&\expect{X^w_i \mid X^w_1,\ldots,X^w_{i-1}} \leq \Pr\paren{X^w_{i+1}=\cdots=X^w_n = 0 \mid X^w_1,\ldots,X^w_{i-1}}.
\end{split}\label{eq:tech-prob-game1}
\end{align}
In words, we have a probabilistic experiment wherein the expected \emph{loss} incurred in each step, no matter the history, is upper bounded by the probability of terminating the experiment at this step, and
we want to upper bound the total loss of the experiment. 

Suppose we could bound the expectation in~\Cref{eq:tech-prob-game1} by $O(1)$ and the variance by an $O(1)$-factor of the expectation. Then, using a somewhat careful application 
of Chebyshev's inequality, we can bound the expected ``overflow'' of the assignment $x$ in total. In other words, we can say that $(i)$ if we remove a constant fraction (less than one) of the $x$-value from 
vertices that have ``too much'' $x$-value on their edges, and $(ii)$ further scale down $x$ by some constant factor, we will end up with a fractional matching\footnote{While this step is non-trivial, it is mostly a careful calculation
and  there is not much more illuminating information that can be provided about it in this overview; so, we postpone more details of it to the actual proof.}. In conclusion, we can obtain a true fractional matching by ``trimming down'' $x$ by a constant factor; combined with \Cref{eq:tech-x-value} and the duality of matchings and vertex covers,
this implies that the resulting fractional matching is an $O(1)$-approximation in expectation. 

Unfortunately however, we actually cannot achieve the desired bounds to the problem in~\Cref{eq:tech-prob-game1} since the variables $X^w_i$ can be unboundedly large. This is where our technical condition 
in the assignment of $x$ comes in handy. In the actual definition of $x$, we are only assigning a value of $1/\deg(v)$ to an edge $(v,w)$ of a vertex $v$ that joins $\mathcal{C}$ if $\deg(w) \leq \deg(v)$. Thus, for each vertex $w$, 
the total borrowed assignment in \emph{each iteration} will be at most 
\[
	\sum_{\substack{(v,w) \ni w \\ \deg(w) \leq \deg(v)}} \hspace{-10pt}\frac{1}{\deg(v)} \leq \sum_{\substack{(v,w) \ni w \\ \deg(w) \leq \deg(v)}} \hspace{-10pt}\frac{1}{\deg(w)} \leq \deg(w) \cdot \frac1{\deg(w)} = 1. 
\]
In other words, in the problem of~\Cref{eq:tech-prob-game1}, we additionally have that for every $i \in [n]$, $0 \leq X^w_i \leq 1$ holds deterministically. This extra condition is now enough to 
bound the expectation and the variance of this problem as desired, using a careful probabilistic analysis. 

Unfortunately, we now have one other problem. Under this actual definition of $x$,~\Cref{eq:tech-x-value} no longer 
holds since the $x$-value we assign to edges of a vertex joining $\mathcal{C}$ is no longer $1$. However, we can analyze this step more carefully, and obtain an approximate version of this equation in expectation, i.e., prove that
\[
	\expect{\sum_{e \in E} x_e} \geq \frac{1}{2} \cdot \Exp\card{\mathcal{C}}.
\]
This is sufficient to perform the above primal-dual analysis. 

In conclusion, we designed a ``light weight'' reduction that given a \emph{single} execution of the randomized greedy MIS, finds an $O(1)$-approximate fractional matching in expectation. 

\subsubsection{Our Dynamic Streaming Algorithm for Matchings}

The next step is to incorporate our reduction from fractional matchings to randomized greedy MIS in the semi-streaming implementation of~\cite{AhnCGMW15}. 
We will be running the algorithm of~\cite{AhnCGMW15} and then maintain enough auxiliary information along the way to be able to implement our own reduction
for finding a fractional matching. 
The issue here is to figure out which edges should be assigned a fractional matching, and with what value. Addressing this issue requires bypassing several challenges that we outline below, but we should right away note that we will \emph{not} be able 
to  implement this reduction in a black-box way, and need to settle for some relaxations.  

\paragraph{Challenge 1: Large support in fractional matchings.} An obvious but easy-to-address issue is that the support of the fractional matching $x \in \IR^E$ returned by our reduction can be quite large (e.g., on a clique, it involves all the edges). 
Hence, we simply cannot hope to recover it with a semi-streaming algorithm. However, given that our original goal was not to recover this particular fractional matching, but rather find \emph{some} large matching in the input graph, 
we can use a standard trick: we only need to sample each edge of the graph \emph{independently} with probability $\simeq x_e \cdot \ln{(n)}$. Then, one can use a standard analysis\footnote{We emphasize that the independence in sampling is crucial here and is the key difference between our fractional matching and the one used in the analysis of~\cite{Veldt24}. Running the randomized greedy MIS and picking the blocking edges does indeed sample each edge with probability proportional to some fractional matching; however, the choice of edges are \emph{positively correlated}, hence forming many stars instead of a large matching; see~\Cref{fig:clique-alg}.} to argue 
that the set of sampled edges contains an integral matching with size within a constant factor of the original fractional matching $x$. 

As such, our goal in implementing the reduction is to be able to determine the value of $x_e$ for each edge in the stream at the time when an update to this edge happens (either insertion or deletion); then, using 
standard sparse-recovery primitives, we will be able to perform the sampling step above and recover a large enough matching. 

\paragraph{Challenge 2: Determining neighbors of a vertex $v$ joining the vertex cover.} In the fractional matching $x \in \IR^E$, 
whenever a vertex $v$ joins the vertex cover $\mathcal{C}$, it will assign a value of $1/\deg(v)$ to (a subset of) its neighbors that are \emph{still} present in the graph. 
But this requires the semi-streaming implementation to be able to determine the neighborhood of every vertex at the time it joins the vertex cover, despite the fact that~\cite{AhnCGMW15} processes the input in large batches 
of vertices without looking at the entire graph. See~\Cref{fig:implementation-alg}. 

\begin{figure}[h!]
	\centering
	 \begin{subfigure}[b]{0.3\textwidth}
	 \centering
	 \includegraphics[scale=0.35]{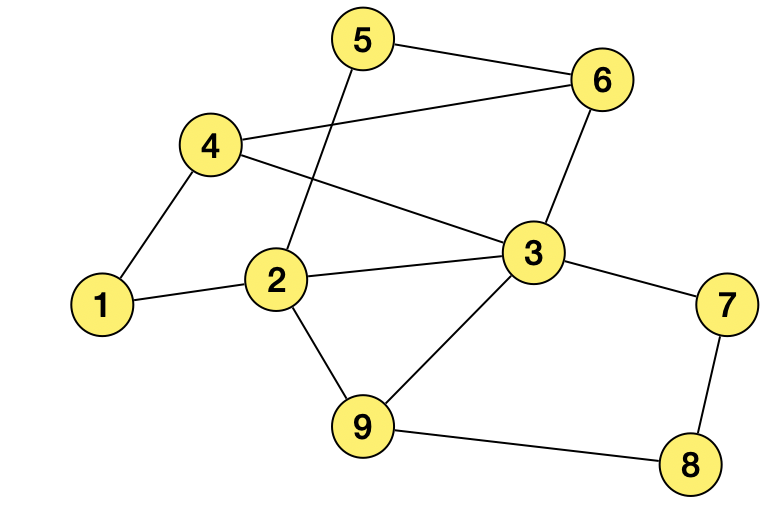}
	\caption{}
	 \end{subfigure} 
	 \hspace{0.5cm}
	 \begin{subfigure}[b]{0.3\textwidth}
	 \centering
	 \includegraphics[scale=0.35]{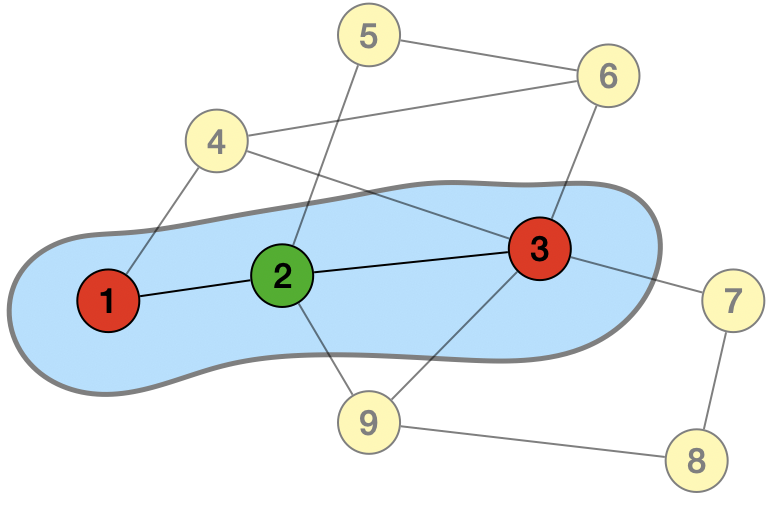}
	\caption{}
	 \end{subfigure}	
	 \hspace{0.5cm}
	 \begin{subfigure}[b]{0.3\textwidth}
	 \centering
	 \includegraphics[scale=0.35]{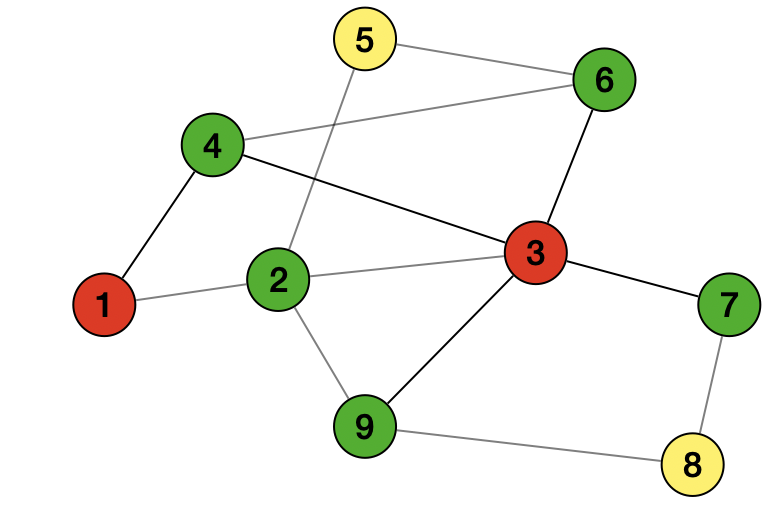}
	 \caption{}
	 \end{subfigure}
	 \caption{Consider running the randomized greedy MIS algorithm on the graph in $(a)$ using the specified ordering. The semi-streaming algorithm of~\cite{AhnCGMW15} processes vertices in large batches, say, a batch of $3$ vertices in part $(b)$. This allows the algorithm to determine which vertices in the batch are in the MIS (red) or the vertex cover (green) in a single pass. Then, in part $(c)$, 
	 using another pass, the algorithm identifies all the other vertices in the graph that also join the vertex cover (again, green vertices). However, for the fractional matching reduction, there are more considerations: for instance, the 
	 edge $(4,6)$ receives a fractional matching from vertex $4$ but not vertex $6$, as vertex $4$ is already removed by the time $6$ is added to the vertex cover (even though, the algorithm of~\cite{AhnCGMW15} treats both vertices
	 $4$ and $6$ the same way). In terms of \emph{time stamps}, time stamp of $4$ is $t(4)=1$ while for $6$ it is $t(6)=3$.}\label{fig:implementation-alg}
\end{figure}

To address this, we augment the semi-streaming algorithm to add a \textbf{time stamp} $t(v)$ to each vertex $v$ that joins the vertex cover: this is the index of the vertex $u \in N(v)$ in the random ordering whose choice
in the MIS led the algorithm to add $v$ to the vertex cover. If we compute these time stamps, then, for any vertex $v$, the set of its neighbors that are still present in the graph when $v$ joins the vertex cover is all vertices $w \in N(v)$ with $t(w) \geq t(v)$. 

But we now have another challenge: how can we compute these time stamps? As argued earlier, in~\cite{AhnCGMW15}, figuring out if a vertex $v$ needs to join the vertex cover after processing a batch $U_i$, amounts to finding 
if $v$ has \emph{any} neighbor to the newly added vertices $\mathcal{I}_i \subseteq U_i$ to the MIS. This can be handled using a simple counter. But, figuring out the time stamps of $v$ requires finding the \emph{smallest-index} neighbor of $v$ in $\mathcal{I}_i$.
In general, finding the minimum entry of an array of length $m$ (corresponding to vertices in $U_i$) undergoing insertion and deletions is not possible in $m^{o(1)}$ space and $o(\log{m})$ passes which is way above our budget 
(indeed, this fact is the basis of the $\Omega(\log{n})$-pass lower bound of~\cite{AssadiKZ24} for findings MSTs in dynamic streams). 

Fortunately, it turns out the minimum-entry problem we need to solve here has a special form. To see this, we once again rely on the residual sparsity property of the randomized greedy MIS.
This time, we further partition each batch $U_i$ into $O(\log{n})$ \textbf{sub-batches} of geometrically increasing sizes. We can then prove that, with high probability, for every vertex $v$, 
the first sub-batch that contains at least one neighbor of $v$ in $\mathcal{I}_i$, can also only have $\lesssim \log{n}$ neighbors of $v$ (in $\mathcal{I}_i$) in total. In other words, except for a negligible probability, 
it is not possible that an ``empty'' sub-batch with no neighbor of $v$ is followed by a ``full'' sub-batch with many neighbors of $v$. Using this property and standard sparse-recovery ideas, we show that 
we can run~\cite{AhnCGMW15} as is, and then spend $O(1)$ passes and $O(n \cdot \poly\log{(n)})$ space and recover the time stamps of all vertices as well. 

\paragraph{Challenge 3: Determining which edges receive fractional matchings.} Recall that the fractional matching $x \in \IR^E$ in our reduction, for each vertex $v$ joining the vertex cover,  adds $1/\deg(v)$ to each edge $(v,w)$
when $\deg(w) \leq \deg(v)$ and not all neighbors of $v$; this was the crucial technical modification we needed for the correctness of the reduction. What we described so far allows us to determine the value of $\deg(v)$ at the time $v$ joins the vertex cover. 
But to know which edges can receive this fractional value, we additionally need to compute $\deg(w)$ at the time $v$ joins the vertex cover and not $w$ itself. This is simply way too much information to even store, yet alone compute,
and we cannot hope to achieve this in the semi-streaming space at all. 

The solution to this challenge is to \emph{delegate} some computation of the fractional matching in the reduction to the analysis instead. Specifically, in our semi-streaming implementation, we are simply going to
compute the time stamps as described before. Then, for every vertex $v$, we can figure out the remaining neighbors of $v$ at the time $v$ joins the vertex cover and sample $\simeq \log{n}/\!\deg(v)$ fraction of remaining neighbors of $v$. 
The total number of edges sampled this way will be $\simeq n\log{n}$ edges still and we can find them using a sparse-recovery approach (albeit, for technical reasons, this part of the argument is more subtle and needs to work with limited-independence
hash functions). The rest of the algorithm is to simply return a maximum matching of these sampled edges. For the analysis \emph{only}, we perform a simple \emph{rejection sampling} idea to recover a sub-sampled version of the fractional matching $x$ 
in the reduction, using the already sampled set of edges, and use this to argue about existence of a large matching among the sampled edges. 

In conclusion, our dynamic semi-streaming algorithm involves: $(i)$ running the $O(\log\log{n})$-pass semi-streaming algorithm of~\cite{AhnCGMW15} for simulating the randomized greedy MIS; $(ii)$ spending $O(1)$ passes to 
recover the time stamps of all vertices and using them to sample $\simeq n\log{n}$ edges from the graph guided by our reduction of matching to greedy MIS; and finally $(iii)$ returning
a maximum matching of the sampled edges. This gives an $O(1)$-approximate semi-streaming algorithm for maximum matching in $O(\log\log{n})$ passes. As stated earlier, we can then use this algorithm in existing boosting frameworks 
for matchings to improve the approximation ratio to $(1+\eps)$ for any constant $\eps > 0$ within the same asymptotic space and number of passes.


\subsection{Overview of Lower Bound}\label{sec:overview-lb}

Our lower bound uses the very recently developed technique of \emph{hierarchical embeddings} in~\cite{AssadiKNS24} and adapts it to proving {dynamic} streaming lower bounds, as opposed to {insertion-only} ones in~\cite{AssadiKNS24}. 
We start by providing a quick overview of this technique and the challenges along the way in adapting it for our purpose. We then describe our fixes for them which also involve borrowing ideas from 
prior work in~\cite{DarkK20}---for single-pass dynamic streaming matching---to construct hard instances of the problem and~\cite{AssadiKZ24}---for multi-pass dynamic streaming MST---as part of the information-theoretic 
analysis of these instances.  


Before starting the rest of this section, a quick comment in order. Our lower bounds (as well as prior work in~\cite{AssadiKNS24,DarkK20,AssadiKZ24} and almost all other streaming lower bounds) rely heavily on \emph{communication complexity} arguments. 
As such, we go back and forth freely between communication and streaming in our discussions. A reader unfamiliar with communication complexity may want to review~\Cref{sec:cc} for a basic introduction and its connection to streaming before proceeding. 

\subsubsection{The Hierarchical Embedding Technique of~\cite{AssadiKNS24}}

The $\Omega(\log\log{n})$ pass lower bound proven in~\cite{AssadiKNS24} for finding MIS in insertion-only streams is based on a new communication/streaming lower bound technique termed \emph{hierarchical embeddings}. 
At a (very) high level, this works as follows: A hard instance for $p$-pass 
streaming algorithms that solve MIS on $n_p$-vertices is constructed by `embedding' a collection of $a \times b$ hard instances 
\[
\GG := \set{G_{i,j}}_{i \in [a],j \in [b]}
\]
 for $(p-1)$-pass algorithms of MIS on smaller $n_{p-1}$-vertex graphs. All these smaller instances in $\GG$ are put together inside a \emph{single} graph $G$ in such a way that: $(i)$ finding MIS of $G$ requires solving a \emph{special subset} $\GG_{\istar} := \set{\GG_{\istar,j}}_{j \in [b]}$ of $b$ of these $(p-1)$-hard 
instances for some $\istar \in [a]$; but, $(ii)$ these specific instances are \emph{hidden} in the first pass of the algorithm, i.e., $\istar$ is not known to the algorithm, until effectively the end of its first pass. See~\Cref{fig:lower-AKNS24}
for a rough illustration. 

\begin{figure}[h!]
	\centering
	 \begin{subfigure}[b]{0.35\textwidth}
	 \centering
	 \includegraphics[scale=0.25]{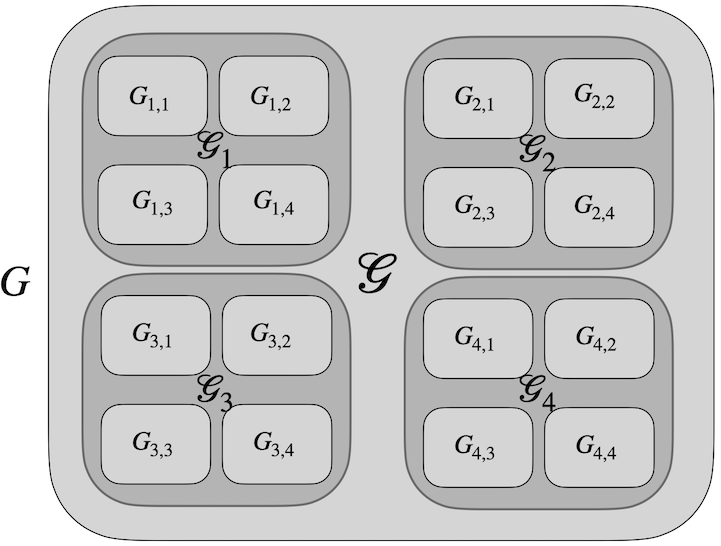}
	\caption{First part of the stream is $G$.}
	 \end{subfigure} 
	 \begin{subfigure}[b]{0.6\textwidth}
	 \centering
	 \includegraphics[scale=0.25]{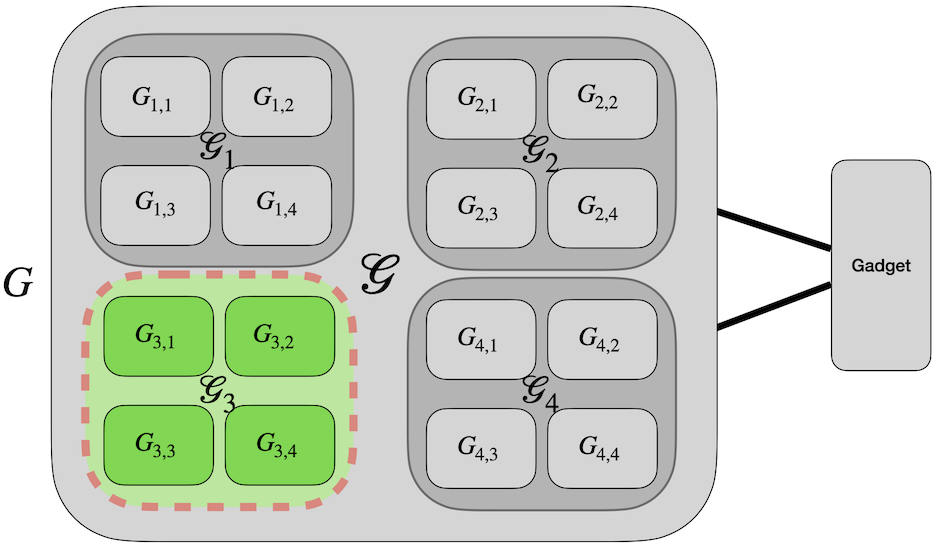}
	\caption{Second part of the stream is some ``gadget'' connected to $G$.}
	 \end{subfigure}	
	 \caption{An illustration of the hierarchical embedding with a collection of $(4 \times 4)$ many $(p-1)$-pass hard instances inside a single $p$-pass hard instance which is the graph $G$. The special hidden instances
	 here correspond to $\istar=3$ which need to be solved to solve the entire problem. We emphasize that, unlike this figure, in the actual construction the sub-instances for different $\GG_i$ and $\GG_j$ 
	 need to necessarily share some vertices so they all \emph{fit} in the same graph.}\label{fig:lower-AKNS24}
\end{figure}

\noindent
Based on this construction,~\cite{AssadiKNS24} proves that any algorithm now effectively has two choices: 
\begin{itemize}
	\item Either spends $\Omega(a \cdot b)$ space to learn some \emph{non-trivial} information about all the sub-instances in $\GG$ with the hope of learning something non-trivial about the $b$ special hidden ones; 
	\item Or, spends its first pass simply to learn $\istar$ only and then in the remaining $p-1$ passes, solve {all} the $b$ special $(p-1)$-hard instances in $\GG_{\istar}$ \emph{independently}, each on a different $n_{p-1}$-vertex graph. 
\end{itemize}
Let us use $s_p(n)$ to denote the minimum space needed for solving MIS on $p$-pass instances of this distribution on graphs with $n$ vertices. Then, the above two choices imply that 
\[
	s_p(n_p) \gtrsim \min\Paren{a \cdot b~,~b \cdot s_{p-1}(n_{p-1})}; 
\]
the second term is because in the second scenario above, the algorithm has effectively learned nothing non-trivial about the special hidden instances and now only has $(p-1)$ passes to solve $b$ independent $(p-1)$-pass hard instances; 
thus, by induction and a direct-sum type argument, we can expect this task to require $b \cdot s_{p-1}(n_{p-1})$ space. 

Working out the {right} parameters and using existing lower bound of $s_1(n) = \Omega(n^2)$ previously established~\cite{AssadiCK19a,CormodeDK19} for the induction base implies that 
\[
	s_p(n) \gtrsim n^{1+1/2^{\Omega(p)}},
\]
which in turn implies an $\Omega(\log\log{n})$-pass lower bound for semi-streaming algorithms. 

Despite the clean plan above, formalizing this idea in~\cite{AssadiKNS24} is quite challenging and involves handing the following two disjoint aspects of this approach: 

\paragraph{Combinatorial aspects.} We need to \emph{pack} many subgraphs of $\GG$ inside a \emph{single} graph $G$ while ensuring that they do not interfere with 
each other so as to not ``corrupt'' their $(p-1)$-pass hardness properties. At the same time, the graph $G$ should also force the algorithm---that is only tasked with solving MIS on a single graph $G$---to also solve MIS of all the special hidden instances in $\GG_{\istar}$, 
without revealing the identity of $\istar \in [a]$ in the first pass.  
This in particular requires the \emph{induced} subgraph\footnote{This is actually only true about the graph appearing in the first half of the stream (which is oblivious to the choice of $\istar$) and not the entire graph; 
while this is a crucial aspect in~\cite{AssadiKNS24}, we skip it in this informal discussion.}  of $G$ on \emph{each} of the sets $\GG_i$ for $i \in [a]$ to consists of solely the edges of subgraphs in $\GG_i$. This is a quite stringent combinatorial requirement and 
is handled in~\cite{AssadiKNS24} by the introduction of a new family of extremal graphs that generalize \rs (RS) graphs~\cite{RuzsaS78}; as this will be too much of a detour, we refer the reader to~\cite{AssadiKNS24} for more details here.

\paragraph{Information-theoretic aspects.} While the two scenarios for different types of semi-streaming algorithms outlined above provide a good intuition about  natural strategies, 
in reality, algorithms are not forced to follow such strategies. For instance, they may decide to \emph{correlate} the $b$ special instances in the remaining $(p-1)$ passes 
instead of treating them independently. Analyzing arbitrary algorithms requires maintaining different (conditional) independence properties between different parts of the inputs and then using various information theory tools. 
This in particular includes \emph{information complexity direct sum} arguments~\cite{ChakrabartiSWY01,BarakBCR10} and \emph{message compression} techniques~\cite{HarshaJMR07} on one hand, 
and generalizing standard \emph{round elimination} arguments in~\cite{MiltersenNSW95} on the other. We again refer the reader to~\cite{AssadiKNS24} for more details. 

We now mention the additional challenges of applying this technique for our purpose of proving a dynamic streaming lower bound for $O(1)$-approximation of matchings. 

\subsubsection{Challenges of Applying Hierarchical Embeddings for Approximate Matchings}

To apply the hierarchical embedding technique of~\cite{AssadiKNS24} for our purpose, we need to change both the above aspects of this technique entirely, due to the following challenges. 

\paragraph{Challenge 1: Combinatorial aspects of $O(1)$-approximate matchings.} In the language of~\cite{AssadiKNS24}, to force a $p$-pass \emph{$O(1)$-approximation algorithm for matchings} on a graph $G$ to have to solve
the $(p-1)$-pass instances of the matchings in the special hidden instances $\GG_{\istar}$, we need the majority of edges in \emph{all} large matchings in $G$ to come solely from instances $\GG_{\istar}$. This means that except for the edges of $\GG_{\istar}$, all other edges of $G$ should admit a vertex cover with $o(n)$ vertices. If we attempt to create our hard instances in the same way as~\cite{AssadiKNS24}, 
this, at the very least, requires every $\GG_i$ for $i \in [a]$ to contain an \emph{induced} matching of size $n-o(n)$ in the graph. This is way too stringent of a combinatorial requirement, and it is easy to argue that any graph admitting such a 
property can only contain $O(n)$ edges in total (see~\cite[Theorem 1.2]{FoxHS17}). Thus, the semi-streaming algorithm can simply store this graph in its memory and solve the problem exactly. 

Of course, it is no surprise that this approach does \emph{not} work for matchings, since it is tailored toward  insertion-only streams (which admit a simple $2$-approximation single-pass algorithm). 
For our purpose, we instead use an idea due to~\cite{DarkK20} for proving single-pass $O(1)$-approximation lower bounds for matchings in \emph{dynamic} streams. This effectively allows us let go of the \emph{group-structure} 
of sub-instances (i.e., their partitioning into $\GG_1,\ldots,\GG_a$) and create an instance that involves only one \emph{induced} set of $(p-1)$-pass instances (instead of many {groups} of them but then picking one group to be special as in~\cite{AssadiKNS24}). 
We discuss our fix here in~\Cref{sec:lower-challenge-1}. 

\paragraph{Challenge 2: Information-theoretic aspects of $O(1)$-approximate matchings.} Addressing the previous challenge using edge deletions creates a significant hurdle in applying the information-theoretic arguments of~\cite{AssadiKNS24}. 
On one hand, to ensure that we are working with \emph{correct} dynamic streams, we need to ensure that the edges deleted in the second half of the stream already have appeared in its first half. 
But this means that the inputs in the second half and the first half of the stream have to be highly correlated, which breaks the independence properties 
used in~\cite{AssadiKNS24} for their information-theoretic arguments. 

To address this challenge, we further borrow ideas from arguments of~\cite{DarkK20} for their single-pass lower bounds using \emph{augmented Index} communication problem -- roughly speaking, augmented Index
can be seen as the \emph{base case} of the ``augmented'' round elimination technique of~\cite{MiltersenNSW95} (which~\cite{AssadiKNS24} generalizes in the ``non-augmented'' case). Fortunately for us,~\cite{AssadiKZ24} 
have recently showed a way to properly augment the input of players even for multi round/pass algorithms on dynamic streams, namely, which parts of the inputs of players to share, and which to keep hidden, 
in a way that generalizes the augmented Index the ``right way''. Thus, to obtain our result, we need to generalize the augmented round elimination lemma of~\cite{MiltersenNSW95} and its dynamic streaming version in~\cite{AssadiKZ24}, 
 in a similar manner that~\cite{AssadiKNS24} did for the non-augmented round elimination in~\cite{MiltersenNSW95}. We discuss our approach in~\Cref{sec:lower-challenge-2}. 

\subsubsection{Addressing Challenge 1: Our Combinatorial Construction}\label{sec:lower-challenge-1}

We now discuss our fix to challenge 1 above, starting with an overview of the prior work of~\cite{DarkK20}. 

\subsubsection*{Single-pass lower bound construction of~\cite{DarkK20}} 

The authors in \cite{DarkK20} proved an $\Omega(n^2)$-space lower bound for $O(1)$-approximation of matchings in single-pass dynamic streams. Prior to~\cite{DarkK20}, all semi-streaming lower bounds for approximate matchings 
in insertion-only streams~\cite{GoelKK12,Kapralov13,Kapralov21} or dynamic streams~\cite{Konrad15,AssadiKLY16,AssadiKL17} relied on various combinatorially complex constructions 
based on RS graphs. On the other hand~\cite{DarkK20} designed an elegant method of replacing these combinatorial constructions in dynamic streams with much simpler graphs by crucially exploiting the edge \emph{deletions} in the stream. 

Consider the following communication problem. 
\begin{itemize}
\item Alice receives an $n \times n$ matrix $A$ encoding the bipartite adjacency matrix of a \emph{random} bipartite graph $G$ between two sets of vertices $L$ and $R$, each of size $n$. 
Each entry of $A$ is independently chosen to be $0$ or $1$ with probability half. 

We think of the input to Alice as the insertions of edges $\Eadd$ at the beginning of the stream. 
\item The input to Bob is random permutations $\sigma_r$ and $\sigma_c$ of $[n]$ 
that reorder rows and columns of $A$, respectively. Bob is also given the following submatrix $B$ of $A$: reorder the rows and columns of $A$ according to $\sigma_r$ and $\sigma_c$; then, 
for some previously-fixed $\istar = n-o(n)$, consider the $(\istar \times \istar)$ submatrix of $A$ at the bottom right corner of $A$ (after the reordering). Provide this sub-matrix except for its \emph{diagonal} entries as $B$, and let 
diagonal entries of $B$ be $0$'s instead. 

We think of the input to Bob as the deletion of edges $\Edel$ in the second half of the stream (based on the bipartite adjacency matrix). Since $B$ is a submatrix of $A$, all edges deleted by Bob are already inserted by Alice. 
The input graph $G$ at the end of the stream is then the bipartite graph on bipartition $L$ and $R$ and edges $\Eadd \setminus \Edel$. 
\end{itemize}

It is not hard to see that the graph $G$ created this way has an induced matching $M$ of size $n/2-o(n)$ with high probability (inside the submatrix $B$ of $A$ with edges $A - B$). 
Moreover, the graph $G \setminus M$ has a vertex cover of size $o(n)$ (by picking the $o(n)$ vertices not part of $B$). This implies that any $O(1)$-approximate matching algorithm
on the graph $\Eadd \setminus \Edel$ needs to recover many edges of $M$. See~\Cref{fig:lower-DK20} for an illustration. 

\begin{figure}[h!]
	\centering
	 \begin{subfigure}[b]{0.45\textwidth}
	 \centering
	 \includegraphics[scale=0.3]{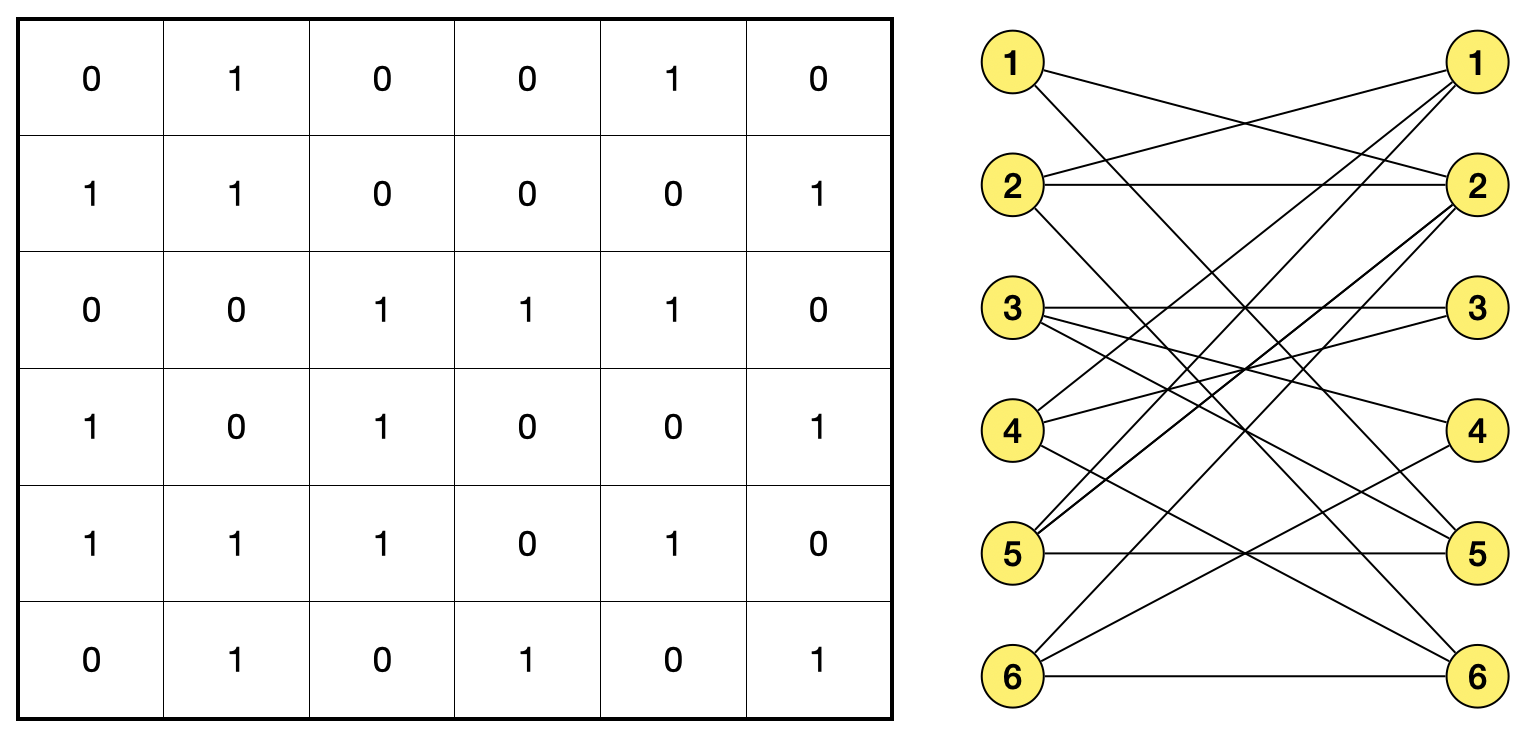}
	\caption{Alice's input $A$ is the bipartite adjacency matrix of the graph $\Eadd$. \\ ~ \\ ~}
	 \end{subfigure} 
	 \hspace{25pt}
	 \begin{subfigure}[b]{0.45\textwidth}
	 \centering
	 \includegraphics[scale=0.3]{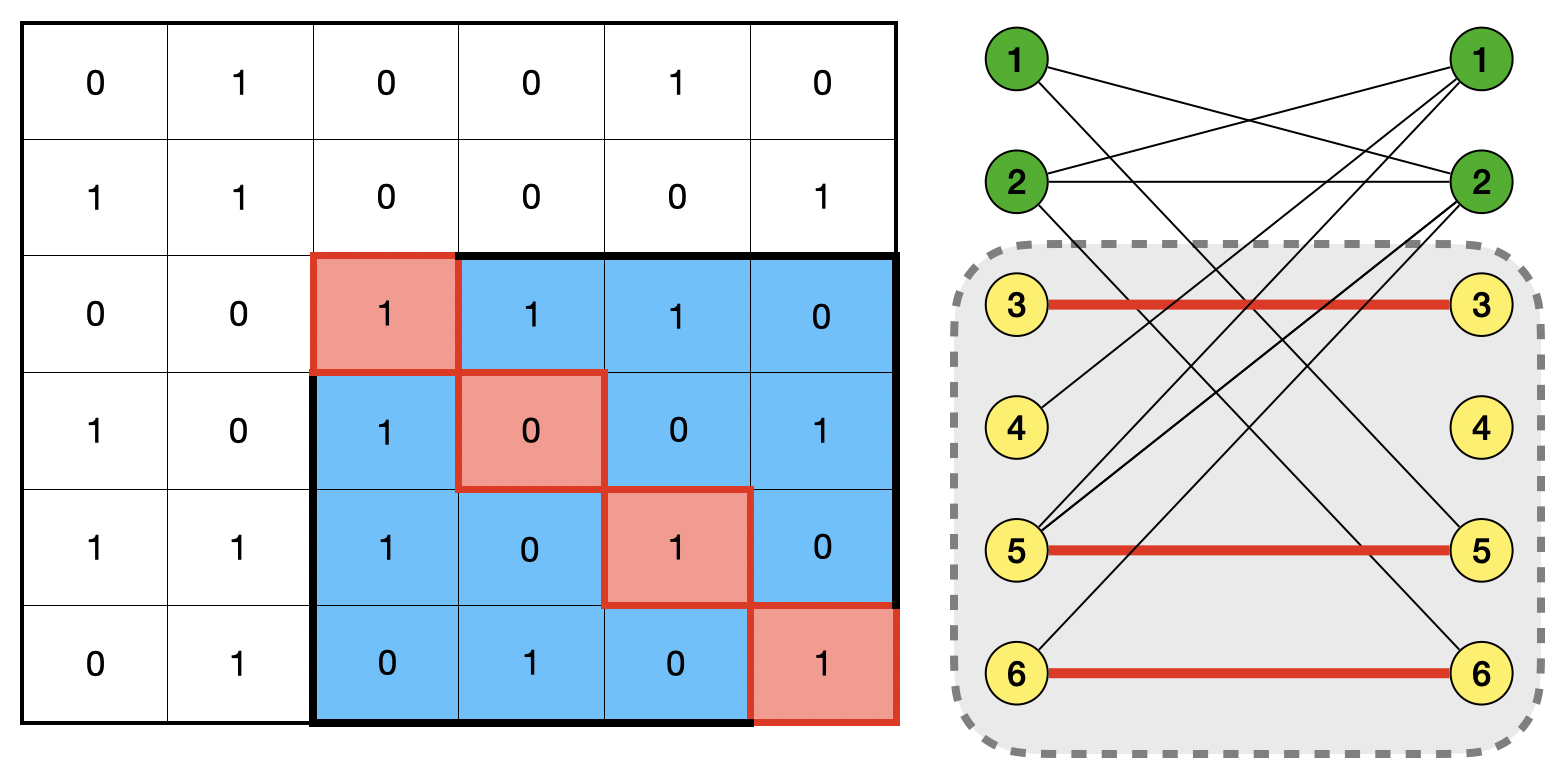}
	\caption{Bob's input $B$ is the sub-matrix with blue entries, with diagonal (red) entries replaced with $0$. 
	The graph here is $\Eadd \setminus \Edel$ with red edges being the induced matching $M$ and green vertices being vertex cover of $G \setminus M$.}
	 \end{subfigure}	
	 \caption{An illustration of the inputs $A$ and $B$ to Alice and Bob in the lower bound of~\cite{DarkK20}. Here,  $n=6$ and $\istar=4$ and both $\sigma_r$ and $\sigma_c$ are identity permutations. In general, the submatrix of $B$ corresponds to a combinatorial rectangle in $A$, the induced matching is of size $n/2-o(n)$ with high probability, and the vertex cover is of size $o(n)$.}\label{fig:lower-DK20}
\end{figure}

In~\cite{DarkK20}, it was shown that given only $A$ and without the knowledge of $\sigma_r$ and $\sigma_c$, majority of edges in $\Eadd$ have an equal chance of appearing as the edges of the induced matching $M$. 
Thus, from Alice's perspective, most edges in $\Eadd$ can become important in finding a large matching in $\Eadd \setminus \Edel$. Hence, Alice needs to communicate $\Omega(n^2)$ bits before Bob can output
an $O(1)$-approximate matching in the graph (this argument misses the crucial information-theoretic aspect that Bob has the knowledge of $\Edel$ and thus knows ``a lot'' about $A$ already; but we postpone this part to the discussion in~\Cref{sec:lower-challenge-2}). 

\subsubsection*{Our multi-pass lower bound construction} 

We are now finally ready to state how we create our hard instances. We define a new two-player communication problem, called \textbf{Augmented Hidden Matrices (AHM)}, 
inspired by the construction of~\cite{DarkK20} described earlier. We then use the standard connection between streaming and communication complexity (\Cref{prop:cc-stream}) to turn $r$-round communication lower bounds 
for AHM into $\Theta(r)$-pass lower bounds for streaming algorithms. 

An $r$-round instance of AHM consists of two matrices $A$ and $B$ and is denoted by AHM$_r(A,B)$. 
As before, matrix $A$ (resp. $B$) corresponds to bipartite adjacency matrix of a bipartite graph on $n$ vertices (resp. $n-o(n)$ vertices) on each side. 
Moreover, $A$ will determine edges $\Eadd$ to be added to the graph and $B$ will determine $\Edel \subseteq \Eadd$ to be deleted\footnote{\label{footnote:add-before-delete} This is actually only true when $r$ is an odd number but we will ignore this subtlety for this discussion; however, we 
note that in our streaming lower bounds, $r$ will always be chosen to be odd.}. These edges are  such that any $O(1)$-approximation of matchings on the graph $\Eadd \setminus \Edel$ 
allows for solving the instance AHM$_r(A,B)$. Crucially different from~\cite{DarkK20} however, it is \emph{not} the case that $A$ is entirely an input to Alice and $B$ is input to Bob (but rather both players receive different parts of both matrices). 

Following the hierarchical embedding technique, an instance of AHM$_r(A,B)$ is created from \emph{many} instances of AHM$_{r-1}$ as follows (see~\Cref{fig:lower-ours} for an illustration): 
\begin{itemize}
	\item Let $n_r$ and $n_{r-1}$ denote, respectively, the number of vertices in instances AHM$_r$ and AHM$_{r-1}$, and $t_r := n_r/n_{r-1}$. 
	Start with a $(t_r \times t_r)$ \emph{matrix} of $(r-1)$-hard instances 
	\[
		\mathcal{I} := \set{(\alicepart{r-1}_{i,j}, \bobpart{r-1}_{i,j})}_{i \in [t_r], j \in [t_r]}, 
	\]
	where each $(\alicepart{r-1}_{i,j},\bobpart{r-1}_{i,j})$ is an instance of AHM$_{r-1}$ on $n_{r-1}$ vertices. 
	
	\item Matrix $A$ is an $(n_r \times n_r)$ matrix 
	defined as follows. Partition the rows and columns of $A$ into $t_r$ \textbf{groups} of size $n_{r-1}$ (say, the first $n_{r-1}$ are in group one, the second in group two, and so on and so forth). 
	Then, matrix $A$ is obtained by plugging matrix $\bobpart{r-1}_{i,j}$ between group $i$ of rows and group $j$ of columns.\footnote{It is worth making two remarks. Firstly, in a typical round elimination argument, we often remove the first round of 
	a protocol that Alice speaks first to obtain an $(r-1)$-round protocol wherein Bob speaks first. But, this means that for the underlying instances, the role of Alice and Bob keeps switching, hence, here the $B_{i,j}$-matrices (which one typically associated 
	with Bob's input) are instead defining $A$ (which again is typically associated with Alice's input). Secondly, even though the matrices $\alicepart{r-1}_{i,j}$ are not used at the moment, some information about them actually will be revealed to Alice and Bob 
	which is crucial for the lower bound; this will be discussed later on.}
	
	In terms of the underlying graph, we think of each $\bobpart{r-1}_{i,j}$ inside $A$ as defining edges between two distinct sets of $n_{r-1}$ vertices in the bipartite graph $G$ (over the same grouping of vertices as rows and columns of $A$). 
	The set $\Eadd$ of edges to be inserted to the graph are defined this way by the matrix $A$. 
	
	\item Matrix $B$ is an $(n_r - o(n_r) \times n_r - o(n_r))$ matrix defined as follows. We pick two random permutations $\sigma_r$ and $\sigma_c$ of $[t_r]$; these correspond
	to reordering the row-groups and column-groups of the matrix $A$, or equivalently rows and columns of the instance-matrix $\mathcal{I}$.
	For some previously-fixed $\istar = t_r - o(t_r)$, define the instance-matrix $\mathcal{I}_{induced} \subseteq \mathcal{I}$ as the $(\istar \times \istar)$ submatrix of $\mathcal{I}$ 
	 in the bottom right corner \emph{after} we apply the reorderings $\sigma_r$ and $\sigma_c$. 
	 Similarly, define $\mathcal{I}_{special}$ as the collection of $\istar$ instances in the diagonals of $\mathcal{I}_{induced}$. Matrix $B$ consists of submatrices $\bobpart{r-1}_{i,j}$ on 
	 the off-diagonal entries of $\mathcal{I}_{induced}$, i.e., instances in $\mathcal{I}_{induced} \setminus \mathcal{I}_{special}$ and submatrices $\alicepart{r-1}_{i,j}$ on the diagonal entries, i.e., instances in $\mathcal{I}_{special}$. 
	
	In terms of the underlying graph, the sub-matrices in $B$ are again mapped to the edges of the graph $G$ in a similar way as matrix $A$. However, this time edges in $B$ are being deleted from the graph $G$ 
	and thus define the set $\Edel$ instead. 
	 	
\end{itemize}

The conclusion of this construction is that in the graph $G$, \emph{all} edges inserted by $A$ corresponding to instances in $\mathcal{I}_{induced} \setminus \mathcal{I}_{special}$ are being deleted by $B$, 
while for the instances in $\mathcal{I}_{special}$, we will end up having a graph corresponding to an instance of AHM$_{r-1}$ (their $B$-side is inserted and their $A$-side is deleted here). This way, the graph $G$ consists of 
a collection of $\istar = t_r - o(t_r)$ \emph{induced} hard $(r-1)$-round instances in $\mathcal{I}_{induced}$ on a set of $n_r - o(n_r)$ vertices; if we ignore edges of these instances, then the entire graph 
has a vertex cover of size $o(n_r)$.

\begin{figure}[h!]
	\centering
	 \begin{subfigure}[b]{0.45\textwidth}
	 \centering
	 \includegraphics[scale=0.25]{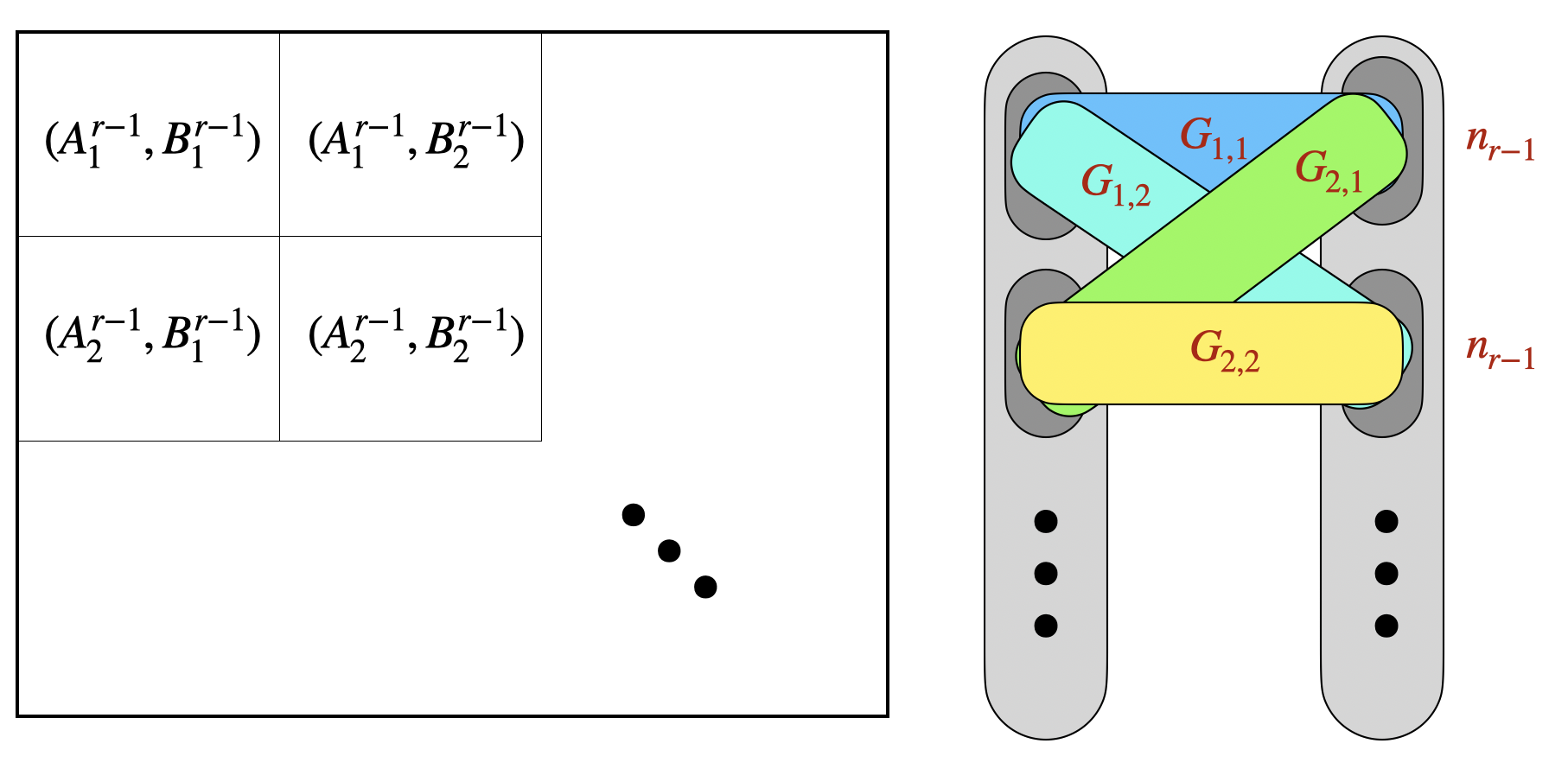}
	\caption{Matrix $A$ is the bipartite adjacency matrix of the graph $\Eadd$. It is constructed from the $(t_r \times t_r)$ instance-matrix $\mathcal{I}$
	whose entries are hard instances of AHM$_{r-1}$. Each of those instances creates a bipartite subgraph on $n_{r-1}$ vertices on each side inside $G$.}
	 \end{subfigure} 
	 \hspace{25pt}
	 \begin{subfigure}[b]{0.45\textwidth}
	 \centering
	 \includegraphics[scale=0.25]{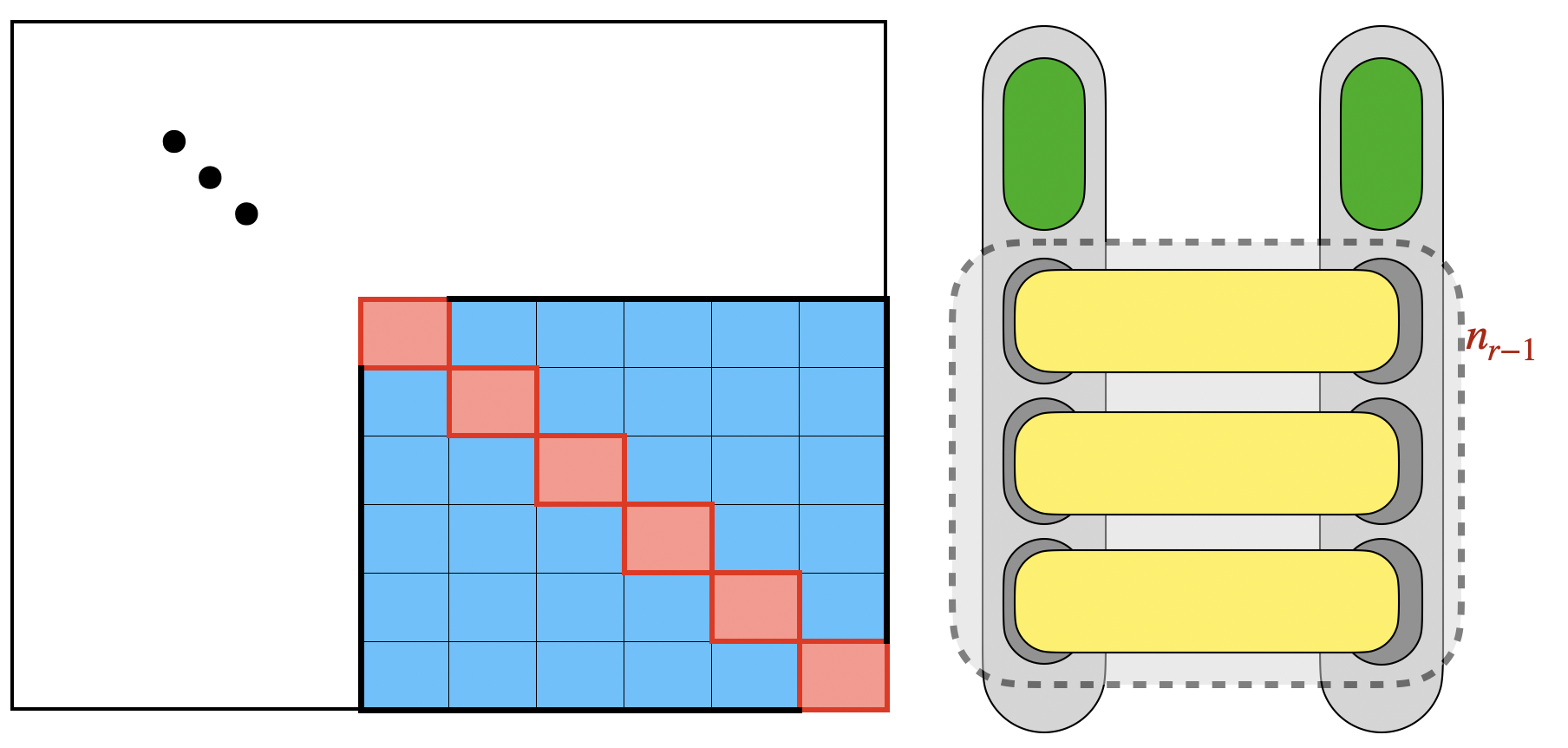}
	\caption{Matrix $B$ is the same as $A$ on the off-diagonal entries (blue) and consists of $\alicepart{r-1}$-matrices on the diagonal entries (red). 
	The graph here is $\Eadd \setminus \Edel$ with yellow subgraphs being special induced instances and green vertices being vertex cover of $G$ minus edges of these instances.}
	 \end{subfigure}	
	 \caption{An illustration of the matrices $A$ and $B$ in our lower bound construction for AHM$_r$. While graph $G$ consists of many instances of AHM$_{r-1}$, the only subgraphs that correspond
	 to \emph{complete} (and thus hard) instances are the special induced ones in $\mathcal{I}_{special}$ (yellow subgraphs in $(b)$). 
	 }\label{fig:lower-ours}
\end{figure}

To see the intuition behind why such instances should be hard consider the following. Graph $G$ created this way contains a very large induced subgraph consisting of vertex-disjoint copies of hard $(r-1)$-round instances, namely, 
the ones in $\mathcal{I}_{special}$. Moreover, all other edges of the graph can only form matchings of size $o(n_r)$ given they are incident on $o(n_r)$ vertices (i.e., the vertex cover argument above). 
This means that any $O(1)$-approximate matching algorithm needs to find \emph{large} matchings from many of the AHM$_{r-1}$ instances in $\mathcal{I}_{special}$; in other words, it needs to ``solve'' these instances as well. 
In the language of hierarchical embeddings, we reduced AHM$_r$ to $t_r$ instances of AHM$_{r-1}$ (albeit with some loss in the parameters, which will be handled carefully in the proof). 

\paragraph{What about the parameters?} Let us again consider the two natural choices of the algorithms: 
\begin{itemize}
	\item Either spend $\Omega(t_r^2)$ communication/space to learn something non-trivial about the $t_r$ hidden instances in $\mathcal{I}_{special}$; the reason behind this ``discounting'' of information by a factor of $t_r$ 
	is that, similar to the construction of~\cite{DarkK20},  it is the case that the $t_r$ instances in $\mathcal{I}_{special}$ are effectively uniform (although only \emph{marginally}) among the majority of the $t_r^2$ instances in $\mathcal{I}$; 
	\item Or, spend communication/space proportional to what is needed to solve (most of) the $t_r$ instances in $\mathcal{I}_{special}$ in only $(r-1)$-rounds and independently. 
\end{itemize}
By using $s_{r}(n_r)$ to denote the communication/space needed for solving AHM$_{r}$ on $n_r$-size instances, the above suggests that 
\begin{align}
	s_r(n_r) \gtrsim \min\Paren{t_r^2~,~t_r \cdot s_{r-1}(n_{r-1})} = \min\Paren{\paren{\frac{n_r}{n_{r-1}}}^2~,~\paren{\frac{n_r}{n_{r-1}}} \cdot s_{r-1}(n_{r-1})}. \label{eq:tech-proof}
\end{align}
Using $s_1(n_1) = \Omega(n_1^2)$ as the base case by the lower bound of~\cite{DarkK20}, we obtain 
\[
	s_2(n_2) \gtrsim \min\Paren{\paren{\frac{n_2}{n_{1}}}^2~,~n_2 \cdot n_{1}} = n_2^{4/3}, 
\]
by letting $n_2 = ({n_1})^3$. Repeating this inductively using~\Cref{eq:tech-proof} allows us to prove that 
\[
	s_r(n) \gtrsim n^{1+1/2^{\Omega(r)}}.
\]
This implies that $\Omega(\log\log{n})$ rounds are needed by $\Ot(n)$-communication protocols for AHM$_r$ and in turn $\Omega(\log\log{n})$ passes are needed by semi-streaming algorithms for $O(1)$-approximation of matchings. 


\paragraph{How is the input partitioned?} One thing that we have neglected so far is specifying what the input to Alice and Bob precisely are in the AHM$_r(A,B)$ problem. 
Naturally, Alice receives the matrix $A$ and Bob receives the submatrix $B$. But since these matrices have quite a lot of overlap with each other on one hand, and the recursive nature of the instances, on the other hand,
we need to specify this more accurately (for instance, in the special induced instances in $\mathcal{I}_{special}$, Alice again receives some parts of $B$, and recursively like this). 
Specifying full details of this partitioning recursively is beyond the scope of this overview, but roughly speaking this is how the partitioning of inputs looks like (see~\Cref{fig:input-partition}): \\~
\vspace{-20pt}
\begin{wrapfigure}[20]{r}{0.45\textwidth}
    \centering
\includegraphics[scale=0.45]{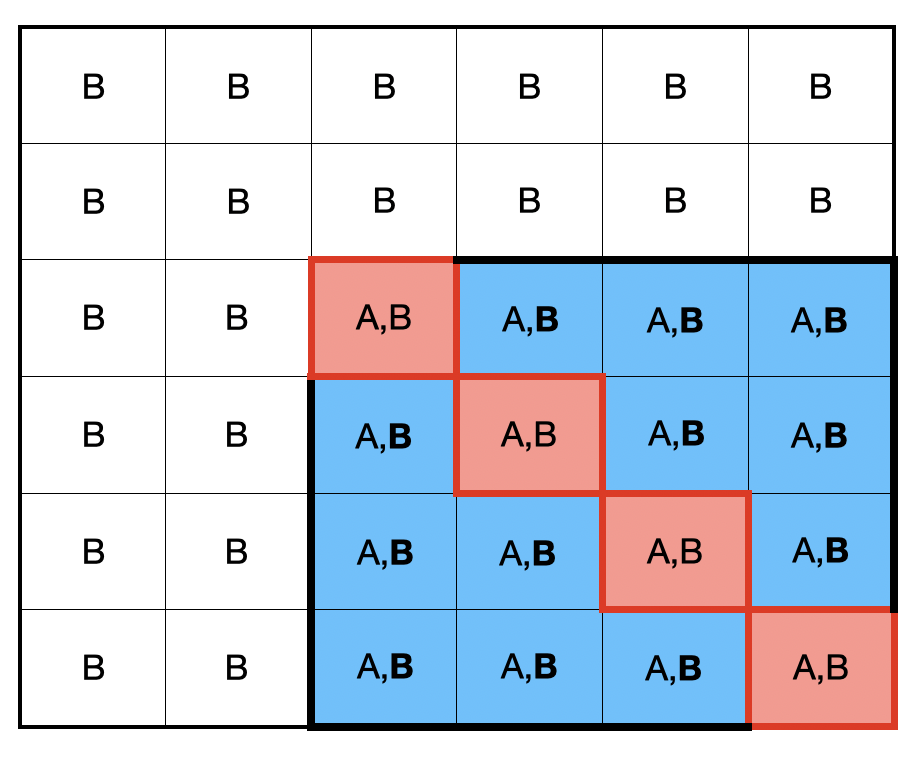}
\caption{An illustration of the input partitioning of Alice and Bob in an instance of AHM$_r$ with the sub-matrix of Bob shown (in blue). 
Here, $A$ refers to some $\alicepart{r-1}_{i,j} \in \mathcal{I}$ given to $(r-1)$-round Alice and $B$ similarly is for some $\bobpart{r-1}_{i,j} \in \mathcal{I}$
given to $(r-1)$-round Bob. In the AHM$_r$ instance, Alice receives all the $B$'s and Bob receives all the $A$'s, and additionally the $B$'s on the off-diagonal sub-matrix (shown in bold).}
\label{fig:input-partition}
\end{wrapfigure}
\begin{itemize}[leftmargin=10pt]
	\item The $r$-round Alice in AHM$_r(A,B)$ receives \emph{all} the inputs that $(r-1)$-round Bob receives (recursively) in all instances of $\mathcal{I}$.

	\item The $r$-round Bob in AHM$_{r}(A,B)$ receives: 
	\begin{itemize}[leftmargin=10pt]
	\item all the inputs that \emph{both} $(r-1)$-round Alice and $(r-1)$-round Bob (recursively defined) in all non-special (= off-diagonal) induced instances, i.e., in $\mathcal{I}_{induced} \setminus \mathcal{I}_{special}$; 
	\item and, all the inputs of only $(r-1)$-round Alice (recursively defined) in the special (= diagonal) induced instances, i.e., in $\mathcal{I}_{special}$. 
	\end{itemize}
\end{itemize}

With this input partitioning, Alice can insert all edges in $B$ and Bob is able to delete all edges inserted by Alice in the off-diagonal instances to create the ``inducedness property'', 
and delete just enough edges from the diagonal instances to make them truly a hard $(r-1)$-round instance (we  caution the reader that this view is only true for \emph{odd} values of $r$; see~\Cref{footnote:add-before-delete}).

It is also worth mentioning that edges entirely out of the sub-matrix $B$ of Bob do not form any real instance of AHM$_r$, given that the other parts of their inputs is never added to the graph. 
However, these are the edges that are incident on the $o(n_{r})$ size vertex cover of the graph $G$ minus edges of special instances in $\mathcal{I}_{special}$; our lower bound construction does not require any property
from these edges beside that they do not form a large matching which is guaranteed by the vertex cover argument. As such, not having real instances here is not a problem from the \emph{combinatorial} aspect. At the same, 
\emph{information-theoretically} it is quite crucial that Bob has no knowledge of this part of the graph, and we shall use this crucially in the next part. We remark that this way of partitioning the input in our work
has also appeared recently in the multi-pass dynamic streaming lower bound of~\cite{AssadiKZ24} for MSTs to facilitate their information-theoretic arguments. 

The remaining and the most technical part of the argument is to formalize the intuitions in this subsection for the hardness of our instances and prove~\Cref{eq:tech-proof} information-theoretically. 



\subsubsection{Addressing Challenge II: Information-Theoretic Arguments}\label{sec:lower-challenge-2}

We now discuss the proof ideas in the communication lower bound in~\Cref{eq:tech-proof}. Let us go back to the one-round lower bound of~\cite{DarkK20} and the part we explicitly left out for here. Already in the one-round problem, 
why it should be the case that Alice's $o(n^2)$-size message cannot reveal much information about diagonal entries Bob needs to output, despite Bob knowing a great deal about Alice's matrix $A$, i.e., 
all the off-diagonal entries of $B$, which form the vast majority of $A$. 

The proof in~\cite{DarkK20} actually showed a weaker guarantee than what we advertised before (and need in our proofs). Specifically,~\cite{DarkK20} only proves that the entry $A[\sigma_r(\istar)][\sigma_c(\istar)]$ (the top left entry of the diagonal) 
cannot be revealed by a small message from Alice even given Bob's matrix $B$. They then used this fact in a separate randomized reduction to extend this lower bound to $O(1)$-approximation of matchings. 
This proof however does not work in the  plan we outlined above, which requires the lower bound for solving \emph{many} lower-round instances instead of just one. 

To handle this, we are going to provide a different proof here compared to~\cite{DarkK20} which is in fact a technical contribution of our work. In particular, this already shows that none of the diagonal entries of $B$ (to be more accurate $A$ on rows and columns of $B$) can be revealed by Alice's message (this is in a \emph{marginal} sense, i.e., the information revealed about \emph{any one} entry on the diagonal is $o(1)$ bits, even conditioned on all the other diagonal values).   

\paragraph{A new one-round lower bound.} Consider the input permutations $\sigma_r$ and $\sigma_c$ of $[n]$ given to Bob and suppose our goal is to return the $j$-th diagonal entry, i.e., 
return $A[\sigma_r(n-\istar+j)][\sigma_c(n-\istar+j)]$. Define the following two sets of rows and columns (see~\Cref{fig:tech-chainrule} for an illustration):
\begin{itemize}
	\item $T_{row}$: the rows in the \emph{entire} $A$ that cover rows of $B$ except for the $j$-th one, i.e., 
	\[
		T_{row} := \set{\sigma_r(n),\sigma_r(n-1),\ldots,\sigma_r(n-\istar)} \setminus \set{\sigma_r(n-\istar+j)}.
	\]
	\item $T_{col}$: the columns in the \emph{entire} $A$ that cover columns of  $B$ except for the $j$-th one, i.e., 
	\[
		T_{col} := \set{\sigma_c(n),\sigma_r(n-1),\ldots,\sigma_c(n-\istar)} \setminus \set{\sigma_c(n-\istar+j)}.
	\] 
\end{itemize}

\begin{figure}[h!]
	\centering
	 \begin{subfigure}[b]{0.45\textwidth}
	 \centering
	 \includegraphics[scale=0.3]{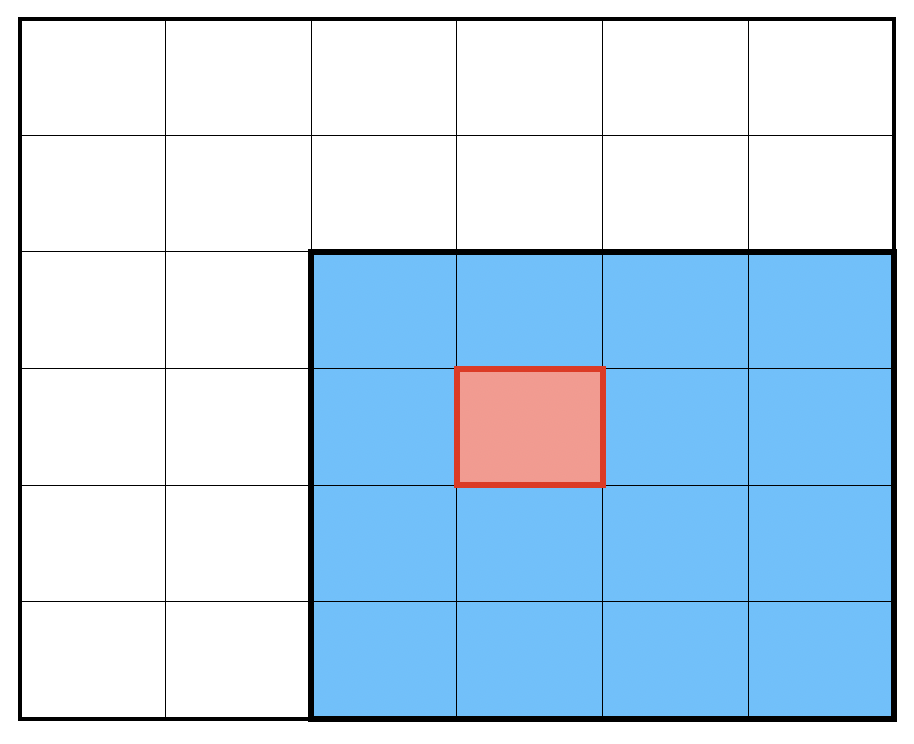}
	\caption{}
	 \end{subfigure} 
	 \hspace{25pt}
	 \begin{subfigure}[b]{0.45\textwidth}
	 \centering
	 \includegraphics[scale=0.3]{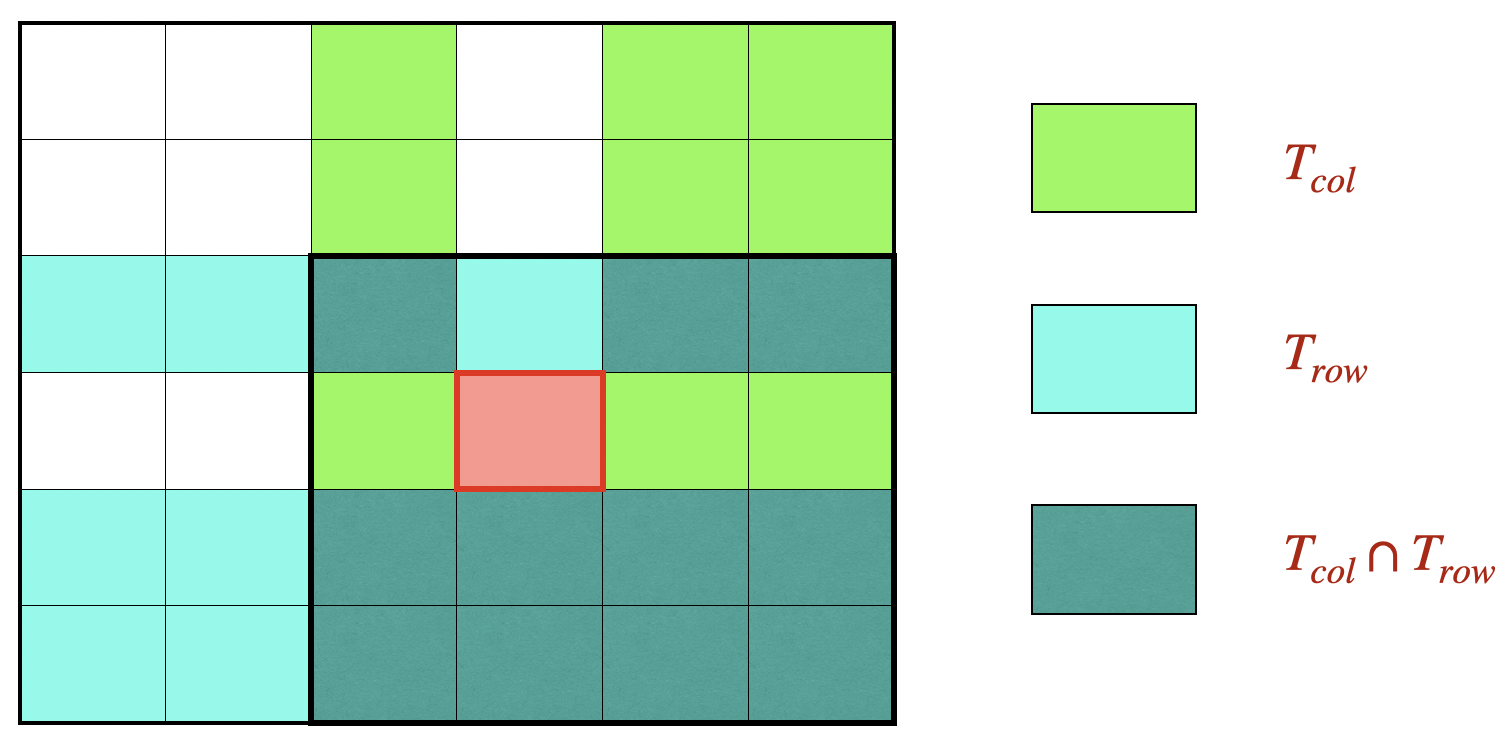}
	\caption{}
	 \end{subfigure}	
	 \caption{An illustration of $T_{row}$ and $T_{col}$ and the chain rule argument for the new one-round lower bound proof. Here, the matrices are shown after applying 
	 the reorderings by $\sigma_r$ and $\sigma_c$ (alternatively, think of them as being identity permutations). Part $(a)$ shows the submatrix $B$ and the special entry we are interested in recovering. 
	 Part $(b)$ shows the fixing of $T_{row}$ and $T_{col}$ (with different colors for each of $T_{row}$, $T_{col}$, and $T_{row} \cap T_{col}$) that fixes all of $B$ (and possibly more of $A$) and leaves the special entry on the diagonal of $B$ unfixed. 
	 }\label{fig:tech-chainrule}
\end{figure}

It is easy to see that if we provide Bob with $A[T_{row}][T_{col}]$ (namely, $A$ on the entries in $T_{row}$ and $T_{col}$) we have provided him with more information than his original input in the off-diagonal entries of $B$. 
On the other hand, a critical observation is that conditioned on the choices of $T_{row},T_{col}$  (but crucially not $\sigma_r$ and $\sigma_c$), the special entry $A[\sigma_r(n-\istar+j)][\sigma_c(n-\istar+j)]$
is chosen uniformly from all \emph{unfixed} entries of $A$, namely, $A \setminus A[T_{row}][T_{col}]$. This is sufficient (using a simple application of chain rule of mutual information) for proving that a 
message of size $\ll (n-\card{T_{row}}) \cdot (n-\card{T_{col}})$ cannot reveal more than $o(1)$ bits about the special entry Bob needs to output, even conditioned on $A[T_{row}][T_{col}]$ (which fixes off-diagonal entries of $B$). 
This proves the one-round lower bound. 

We note that these types of chain rule arguments (by conditioning on partial parts of permutations) have also appeared in somewhat similar contexts elsewhere, e.g., in~\cite{BravermanRWY13,NelsonY19,KunKoW20}. 

\subsubsection*{Our multi-round lower bound}
 At last, we are ready to review our approach for proving our multi-round lower bound for the AHM$_r$ problem introduced in~\Cref{sec:lower-challenge-1}, namely, prove~\Cref{eq:tech-proof}.  
This follows the three-step approach put forward by~\cite{AssadiKNS24} in proving their general round elimination argument, but with a different argument for each step, as we outline below. 

\paragraph{Step 1.} Turn an $r$-round protocol $\prot_r$ for AHM$_r$ with communication cost $s$ to an $r$-round protocol $\protnum{1}_r$ for solving AHM$_{r-1}$ with communication cost $s$ but \emph{information cost} (see~\Cref{sec:ic}) $\lesssim s/t_r$. 
	
	The intuition behind this step is that any protocol for AHM$_r$ needs to solve $\istar = t_r-o(t_r)$ \emph{independent} instances of AHM$_{r-1}$. As such, we expect that the information it reveals about a {random} such instance
	to be $1/t_r$ times the total information cost of the protocol which is upper bounded by its communication cost. This is basically a \emph{direct sum} argument. 
	
	The difference between our approach and~\cite{AssadiKNS24} is in the nature of this direct sum argument. Specifically since our input distributions are \emph{not} product, 
	we use \emph{internal} information cost to implement this step (as opposed to external information in~\cite{AssadiKNS24}). 
	A quick description of this step for a reader already
	familiar with internal information cost direct sum arguments (e.g.~\cite[Lemma 3.1]{Weinstein15}) is the following\footnote{This is not the exact way we sample the input distribution and we have been slightly flexible in the overview.}: $(i)$ Using \underline{public} randomness, Alice and Bob sample $j \in [\istar]$ and all off-diagonal induced instances, 
	$(ii)$ embed the given instance of AHM$_{r-1}$ at the $j$-th diagonal entry of the the special induced instances, and $(iii)$ for every remaining instances (diagonal or entirely out of Bob's submatrix), 
	sample one part of the input \underline{publicly} and the other \underline{privately} (this way, we can sample from non-product distributions correctly). The exact choice of this public-private sampling is a function of the input sharing of the players
	and is similar in nature to~\cite{AssadiKZ24}. 
	
\paragraph{Step 2.} Turn the $r$-round protocol $\protnum{1}_r$ for AHM$_{r-1}$ with information cost $s/t_r$ to an $r$-round protocol $\protnum{2}_r$ for solving AHM$_{r-1}$ with communication cost $\lesssim s/t_r \cdot \poly(r)$. 
	
	This step is based on existing \emph{message compression} arguments that allow for compressing the communication of a \emph{limited} round protocol, down to its information cost. 
	The difference in this step between our approach and~\cite{AssadiKNS24} is that we use the compression technique of~\cite{JainPY16} for internal information cost whereas~\cite{AssadiKNS24} uses~\cite{HarshaJMR07} 
	for external information. 
	
\paragraph{Step 3.} Turn the $r$-round protocol $\protnum{2}_r$ for AHM$_{r-1}$ with communication cost $s/t_r \cdot \poly(r)$ to an $(r-1)$-round protocol $\prot_{r-1}$ for solving AHM$_{r-1}$ with the same communication, 
while incuring an additional $o(1)$ additive factor on error probability \emph{as long as} $s \ll t_r^2$. 
	
	This step is the \emph{real} round elimination step, wherein we finally obtain an $(r-1)$-round protocol. This part of the argument is quite problem-specific and is entirely disjoint from~\cite{AssadiKNS24}. Specifically, 	
	the additive factor in the error probability comes from a similar argument as our one-round lower bound outlined above. The choice of a 
	random special instance in $\mathcal{I}_{special}$ among the $t_r^2$ instances in $\mathcal{I}$ used in creating an instance of AHM$_r$, is random 
	among $\simeq t_r^2$ instances even conditioned on the input of Bob (namely, the off-diagonal entries in $\mathcal{I}_{induced} \setminus \mathcal{I}_{special}$). 
	As argued earlier, Alice's \emph{first} message should only be able to reveal $o(s/t_r^2)$ information about a random diagonal AHM$_{r-1}$ instance $\prot_{r-1}$ is solving even conditioned on Bob's input in the entire AHM$_r$ instance (and not only AHM$_{r-1}$). Thus, even if we ignore the first message of $\protnum{2}_r$ and instead run it from its second run onwards, as long as $s \ll t_r^2$, we should expect 
	a very similar outcome for the underlying instance AHM$_{r-1}$ as if we run the whole protocol. 
	
	Given the technical nature of this step, we postpone more details of this step to the actual proof, and only mention that this step, is the heart
	of our information theoretic arguments.

\paragraph{Concluding the proof.} 

After these steps, we obtain that as long as $s \ll t_r^2$, the resulting $(r-1)$-round protocol $\prot_{r-1}$ for AHM$_{r-1}$ succeeds with a non-trivial probability. But, given that communication cost of $\prot_{r-1}$ is $s/t_r \cdot \poly(r)$, 
we obtain that 
\[
	\frac{s}{t_r} \cdot \poly(r) \gtrsim s_{r-1}(n_{r-1}), 
\]
where $s_{r-1}(n_{r-1})$ is the communication cost of protocols for AHM$_{r-1}$ on $n_{r-1}$-size inputs we have inductively established. Stated differently, this implies that either 
\[
	s \gtrsim t_r^2 \qquad \text{or} \qquad s \gtrsim \frac{1}{\poly(r)} \cdot t_r \cdot s_{r-1}(n_{r-1}),
\]
which establishes~\Cref{eq:tech-proof} (the $\poly(r)$-term is negligible for our purpose as $r=O(\log\log{n})$ always). 

\bigskip

In conclusion, this way we can establish a lower bound of $n^{1+1/2^{\Omega(r)}}$ communication for $r$-round protocols for AHM$_r$. By our reduction to the matching problem, this in turn
implies that any $p$-pass streaming algorithm for $O(1)$-approximation of matchings requires $n^{1+1/2^{\Omega(p)}}$ space; in particular, semi-streaming algorithms require $\Omega(\log\log{n})$ passes as desired.

\clearpage


\section{Preliminaries}

\paragraph{Notation.} Throughout, for any integer $n \geq 1$, we define $[n] := \set{1,\ldots,n}$. 
For a vector $x \in \IR^n$, we use $\norm{x}_0$ to denote the number of non-zero entries in $x$. 
We use $S_n$ to denote the set of all permutations over $[n]$ for any $n \geq 1$. 
For any permutation $\sigma \in S_n$, we use $\sigma(i)$ to denote the element in $[n]$ that $\sigma$ maps $i $ to for each $i \in [n]$. 

For any tuple $(x_1,\ldots,x_n)$ and $i \in [n]$,  
we define $x_{<i} := (x_1,\ldots,x_{i-1})$. We define $x_{>i}$ and $x_{-i}$, analogously. For a set of tuples $\set{(x,y) \mid x \in X, y \in Y}$ for some 
sets $X$ and $Y$, and any $x \in X$, we define $(x,*) := \set{(x,y) \mid y \in Y}$. 

For any matrix $Z$ of dimensions $m \times n$ for integers $m,n\geq 1$, we use $Z[i,j]$ to denote the value at row $i$ and column $j $ for $i \in [m], j \in [n]$. 

For a graph $G=(V,E)$ and any vertex $v \in V$, we use $N(v)$ to denote the neighbors of $v$ and $\deg{v} := \card{N(v)}$ to denote its degree.  For any $U \subseteq V$, we use $G[U]$ to denote the subgraph of $G$ induced on vertices in $U$. 
We further use $\mu(G)$ to denote the maximum matching size in $G$. 

\paragraph{Random variables and information theory notation.} When there is room for confusion, we use sans-serif letters for random
variables (e.g. $\rA$) and normal letters for their realizations (e.g. $A$). We use $\distribution{\rA}$ and $\supp{\rA}$ to denote the distribution and support of $\rA$, respectively. 

For random variables $\rA,\rB$, we use $\en{\rA}$ to denote the \emph{Shannon entropy}
and $\mi{\rA}{\rB}$ to denote the \emph{mutual information}. For two distributions $\mu$ and $\nu$ on the same support, $\tvd{\mu}{\nu}$ denotes their \emph{total variation distance} 
and $\kl{\mu}{\nu}$ is their \emph{KL-divergence}.~\Cref{app:info} contains the definitions of these notions and standard
information theory facts that we use in this paper. 

\paragraph{Fractional matchings.}  A \emph{fractional} matching in a graph $G=(V,E)$ is any assignment $x \in \IR^E$ to the edges $E$ of $G$ with the following properties: 
\begin{align*}
	&\text{for all $e \in E$:} \quad 0 \leq x_e \leq 1; \quad \text{and} \quad \text{for all $v \in V$:} \quad x_v := \sum_{e \ni v} x_e \leq 1. 
\end{align*}
We use $\card{x} := \sum_{e \in E} x_e$ to denote the \emph{size} of the fractional matching $x$. 

It is easy to see that incidence vector of a matching $M$ is also a fractional matching of size $\card{M}$. The following standard fact (see, e.g.~\cite{LovaszP09}) provides the other direction as well. 

\begin{fact}\label{fact:fractional-to-integral-matching}
	For any graph $G=(V,E)$, and any fractional matching $x \in \IR^E$, we have $\mu(G) \geq 2/3 \cdot \card{x}$. Moreover, if $G$ is bipartite, then $\mu(G) \geq \card{x}$. 
\end{fact}

\subsection{Concentration Inequalities}\label{sec:concentration}

We use the following standard form of Chernoff bound and its extension to negatively correlated variables in~\cite{PanconesiS97}. 

\begin{proposition}[Chernoff Bound; cf.~\cite{DubhashiP09}]\label{prop:chernoff}
	Let $X_1,\ldots,X_n$ be $n$ independent random variables in $[0,1]$ and $X := \sum_{i=1}^{n} X_i$. For any $\delta > 0$ and $\mu_{min} \leq \expect{X} \leq \mu_{max}$, 
	\begin{align*}
		&\Pr\paren{X \geq (1+\delta) \cdot \mu_{max}} \leq \exp\paren{-\frac{\delta^2 \cdot \mu_{max}}{2+\delta}} \quad , \quad \Pr\paren{X \leq (1-\delta) \cdot \mu_{min}} \leq \exp\paren{-\frac{\delta^2 \cdot \mu_{min}}{2}}.
	\end{align*} 	
	Moreover, the upper tail bound continues to hold as long $X_1,\ldots,X_n$ are negatively correlated, i.e., for every $S \subseteq [n]$, 
	\[
		\expect{\prod_{i \in S} X_i} \leq \prod_{i \in S} \expect{X_i}. 
	\]
\end{proposition}

We also need the extension of Chernoff-Hoeffding bounds to limited independence random variables. We define limited independence hash functions first. 

\begin{definition}[Limited Independence Hash Functions]\label{def:k-wise-ind-hash}
For integers $r, t, k \geq 1$, a family $\cH$ of hash functions from $[r]$ to $[t]$ is called a $k$-wise independent hash function iff for any two $k$-subsets $a_1, a_2, \ldots, a_k \subseteq [r]$ and $b_1, b_2, \ldots, b_k \subseteq [t]$, 
\[
	\underset{h \sim \cH}{\Pr}(h(a_1) = b_1, h(a_2) = b_2, \ldots, h(a_k)  = b_k) = \frac1{t^k}.
\]
\end{definition}

$k$-wise independent hash functions behave like random functions, as long as sets of size at most $k$ are considered. 
We know that we can store and access these functions in limited space. 

\begin{proposition}[\!\!\cite{MotwaniR95}]\label{prop:storing-sampling-hash}
	For any integers $r, t, k \geq 1$, there is a $k$-wise independent hash function family $\cH = \set{h : [r] \rightarrow [t]}$ such that sampling and storing a function $h \in \cH$ takes $O(k \cdot (\log (r\cdot t)))$ space. 
\end{proposition}

We can now state the extension of Chernoff bounds for limited independence hash functions.

\begin{proposition}[\!\!\cite{SchmidtSS95}]\label{prop:chernoff-hash}
	Let $X_1, X_2, \ldots, X_n$ be $k$-wise independent random variables in $[0,1]$ and $X := \sum_{i=1}^{n} X_i$. Then, for any $\delta \geq 1$ and $\expect{X} \leq \mu_{max}$, we have,
	\begin{align*}
				&\Pr(X \geq (1+\delta) \cdot \mu_{max}) \leq \exp \paren{-\min  \left\{\frac{k}2, \frac{\delta}3 \cdot\mu_{max}\right\}  }
			\end{align*}.
		We also have the following lower tail bound for any $\delta'\leq 1$ and $\mu_{min} \leq \expect{X}$.
	\begin{align*}
				&\Pr(X \leq (1-\delta')\cdot \mu_{min}) \leq \exp\paren{-\min  \left\{\frac{k}2, \frac{(\delta')^2}3 \cdot\mu_{min}\right\}  }.
	\end{align*}
\end{proposition}

\subsection{Sketching and Streaming Toolkit}\label{sec:sketch-toolkit}

\subsubsection*{Sparse-Recovery Algorithms} 


We use standard sparse-recovery algorithms combined with a simple randomized test to ensure it can also detect non-sparse inputs; see, e.g.~\cite[Propositions 3.6 and 3.7]{AssadiKM23} that construct this using Vandermonde matrices and an equality test 
(see also~\cite[Proposition A.16]{AshvinkumarADGW23} that explicitly shows how to use a PRGs for degree-2 polynomials in~\cite{Lovett09,BogdanovV10} to ensure that the latter algorithm does not need to store many random bits either). 

\begin{proposition}[Sparse Recovery; cf.~\cite{DasS13,AssadiKM23}]\label{prop:sparse-recovery}
	There exists a single-pass deterministic algorithm that given integers $q,n,m \geq 1$ and a dynamic stream defining a vector $\phi \in [m]^n$, 
	uses $O(q \cdot \log{(m \cdot n)})$ bits of space and recovers $\phi$ as long as vector $\phi$ is $q$-sparse (meaning $\norm{\phi}_0 \leq q$). 
	
	Moreover, there is a single-pass randomized algorithm that given $\delta > 0$ can test if the vector $\phi \in \IN^n$ is $q$-sparse with probability at least $1-\delta$ 
	using $O(q \cdot \log{(m \cdot n)} + \log{(n/\delta)})$ bits of space. 
\end{proposition}

We note that the use of sparse-recovery algorithms in graph sketching already dates back to the seminal work of~\cite{AhnGM12a}. Moreover, many dynamic streaming algorithms 
use a particular application of sparse-recovery in form of \emph{$\ell_0$-samplers}~\cite{JowhariST11} that allows for sampling a single element of a given vector $\phi$ in a dynamic stream
in $\poly\!\log{(n)}$ space (instead of recovering the entire vector only if it is sparse). However, for our applications, working with sparse-recovery algorithms directly is more convenient and thus 
we opted to skip using $\ell_0$-samplers altogether. 


\subsubsection*{Prior Sketching and Streaming Tools for Matchings} 

We use a vertex sampling approach due to~\cite{AssadiKLY16}  that allows for reducing the number of vertices in a graph in an \emph{oblivious} manner, while preserving its largest matching approximately (see also~\cite{ChitnisCEHMMV16} for a related but slightly different result). This is a key subroutine used for finding matchings in dynamic streams in a single pass also. We use the following formulation from~\cite[Lemma 3.8]{AssadiKL16} that presents this result explicitly. 

\begin{proposition}[\!\!{\cite[Lemma 3.8]{AssadiKL16}}]\label{prop:vertex-sampling}
	Let $G=(V,E)$ be any graph with maximum matching size $\mu(G) = \omega(\log{n})$. For any $\eps \in (0,1)$, suppose we partition the vertices of $G$ randomly into $t \geq 8\mu(G)/\eps$ groups $U_1,\ldots,U_t$ 
	by sending each vertex to one group chosen independently and uniformly at random. Let $E_U \subseteq E$ be a subset of edges such that for any pair of groups $U_i,U_j$ for $i \neq j \in [t]$, we pick one arbitrary edge
	$(x,y)$ with $x \in U_i$ and $y \in U_j$ (if at least one such edge exists). Then, with high probability, $\mu(E_U) \geq (1-\eps) \cdot \mu(G)$. 
\end{proposition}

We also use the standard framework of boosting $O(1)$-approximation algorithms for maximum matching to $(1+\eps)$-approximation algorithms in dynamic streams for any $\eps > 0$. 
The original version of this framework is due to~\cite{McGregor05} which was de-randomized in~\cite{Tirodkar18} and extended to weighted graphs in~\cite{GamlathKMS19}; for bipartite graphs, more efficient reductions are shown in~\cite{AhnG11,AssadiLT21}. 


\begin{proposition}[\!\!\cite{McGregor05,AhnG11,GamlathKMS19,AssadiLT21}]\label{prop:boosting-approximation}
	For any $\eps \in (0,1)$ and any integer $n \geq 1$, suppose we have an $O(1)$-approximation algorithm for finding a maximum matching on $n$-vertex \underline{unweighted} graphs in dynamic streams using $s(n)$-space and $p(n)$-passes with high probability. 
	Then, there is a dynamic streaming algorithm with $O(f(\eps) \cdot s(n))$-space and $O(f(\eps) \cdot p(n))$-passes that with high probability finds a $(1+\eps)$-approximation to maximum matching even in \underline{weighted} graphs, for some function $f(\eps)$ that depends only on $\eps$. 
	
	Furthermore, for general weighted graphs, we can set $f(\eps) = (1/\eps)^{\Theta(1/\eps^2)}$, for general unweighted graphs, set $f(\eps) = (1/\eps)^{\Theta(1/\eps)}$, and for (un)weighted bipartite graphs, set $f(\eps) = \poly(1/\eps)$. 
\end{proposition}

\subsection{Two-Party Communication Complexity}\label{sec:cc}

We work in the standard two party communication model; we provide some basic definitions here and refer the interested reader to the excellent textbooks~\cite{KushilevitzN97,RaoY20} for more details. 

There are two players Alice and Bob who receive input from $\cX$ and $\cY$ respectively. The players jointly compute some function $f$ with domain $\cX\times \cY$. The players follows some protocol $\prot$ to compute $f$.
They have access to a shared tape of randomness, referred to as \emph{public randomness}, in addition to their own \emph{private randomness}.

Alice first sends a message to Bob, followed by a message from Bob to Alice and so on. The last player who receives a message has to output the answer. The total number of rounds is the total number of messages passed between Alice and Bob. 
Moreover, the message sent by any player only depends on their private input, the communicated messages so far, the public randomness, and the private randomness.



\begin{Definition}\label{def:cc}
	For any protocol $\prot$, the \textbf{communication cost} of $\prot$, denoted by $\cc{\prot}$, is defined as the worst-case (maximum) total length of messages, measured in bits, communicated by the players
	on any input. We assume that all \textbf{transcripts}, i.e., the set of all messages sent by any player, in $\prot$ have the same worst-case length (by padding).
\end{Definition}

The following standard result relates communication protocols and streaming algorithms. 

\begin{proposition}[cf.~\cite{AlonMS96}]\label{prop:cc-stream}
	For any $p \geq 1, s \geq 1$, and $\delta \in (0,1)$, suppose there is a $p$-pass $s$-space streaming algorithm $A$ for some problem $\PP$ that succeeds with probability at least $\delta$. Then, there also exists a two-party protocol $\prot$ with $2p-1$ rounds, communication cost $\cc{\prot} = O(p \cdot s)$, and success probability at least $\delta$ for the same problem $\PP$. 
\end{proposition}
\begin{proof}
	Consider the stream $\sigma = \sigma_A \circ \sigma_B$ where $\sigma_A$ (resp. $\sigma_B$) is the input to Alice (resp. Bob) in $\prot$ (ordered arbitrarily in the stream). 
	Alice runs $A$ on $\sigma_A$ and sends the memory content to Bob, which allows Bob to continue running $A$ on $\sigma_B$, and send
	the memory content back to Alice to continue running the next pass. This allows the players to run one pass of $A$ using communication cost at most $O(s)$ and $2$ rounds of communication.  
	The players can continue this, faithfully simulating the $p$ passes of the algorithm, and at the end of the last pass, Bob can output the answer of $A$, instead of replying back to Alice. This 
	requires $2p-1$ rounds of communication and $O(p \cdot s)$ communication, and has the same success probability as the algorithm $A$. 
\end{proof}
\Cref{prop:cc-stream} allows us to translate communication lower bounds into streaming ones.

\subsection{Information Cost and Message Compression}\label{sec:ic}
We also work with the notion of \emph{information cost} of protocols that originated in~\cite{ChakrabartiSWY01} and has since
found numerous applications (see, e.g., \cite{Weinstein15} for an excellent survey of this topic). There are various definitions of information cost that have been considered depending on the application.  
The following definition due to~\cite{BarakBCR10} is best suited for our purpose. 
\begin{Definition}\label{def:int-info}
	For any two-party protocol $\prot$ whose inputs are distributed according to some distribution $\mu$, the \textbf{(internal) information cost} is defined as: 
	\[
		\ic{\prot}{\mu} := \mi{\rX}{\rProt \mid \rR, \rY} + \mi{\rY}{\rProt \mid \rR, \rX},
	\]
	where $\rX, \rY$ denote the random variable for the inputs $X, Y$ sampled from $\mu$, $\rProt$ denotes the random variable corresponding to the communicated messages, and $\rR$ is the public randomness. 
\end{Definition}

Since a single bit of communication can only carry one bit of information, we can upper bound information cost by the communication cost. 
 \begin{proposition}[cf.~\cite{BarakBCR10}]\label{prop:ic-cc}
 	For any protocol $\prot$ on any distribution $\mu$, 
 	$
 		\ic{\prot}{\mu} \leq \cc{\prot}. 
 	$
 \end{proposition}

We also use \emph{message compression} to reduce communication cost of limited-round protocols close to their information cost. 
The following is a simplified version of~\cite[Theorem 3.4]{JainPY16}.   

\begin{proposition}[\!{\cite[Theorem 3.4]{JainPY16}}]\label{prop:msg-compress}
For $r \geq 1$, any $\eps \in (0,1)$, and input distribution $\mu$, any $r$-round protocol $\prot$ can be simulated with error at most $\eps$ in $r$-rounds by a protocol $\prot'$ with communication at most
\[
\jpyconst \cdot \paren{r/\eps \cdot \ic{\prot}{\mu} + r^2/\eps}
\]
for some absolute constant $\jpyconst \geq 1$. The simulation by $\prot'$ uses a source of public randomness that is independent from any randomness used in $\prot$.
\end{proposition}


\newcommand{\VC}{\ensuremath{V_{\textnormal{\tiny \textsc{cover}}}}}
\newcommand{\MIS}{\ensuremath{V_{\textnormal{\tiny \textsc{mis}}}}}

\newcommand{\Gt}{\ensuremath{G^{(t)}}}
\newcommand{\Gat}[1]{\ensuremath{G^{(#1)}}}
\newcommand{\nt}{\ensuremath{n^{(t)}}}
\newcommand{\nat}[1]{\ensuremath{n^{(#1)}}}
\newcommand{\mt}{\ensuremath{m^{(t)}}}
\newcommand{\mat}[1]{\ensuremath{m^{(#1)}}}
\newcommand{\ut}{\ensuremath{u^{(t)}}}
\newcommand{\uat}[1]{\ensuremath{u^{(#1)}}}
\newcommand{\degt}[1]{\ensuremath{\deg^{(t)}(#1)}}
\newcommand{\degat}[1]{\ensuremath{\deg^{(#1)}}}
\newcommand{\Nt}[1]{\ensuremath{N^{(t)}(#1)}}
\newcommand{\Nat}[1]{\ensuremath{N^{(#1)}}}
\newcommand{\VCt}{\ensuremath{\VC^{(t)}}}
\newcommand{\VCat}[1]{\ensuremath{\VC^{(#1)}}}
\newcommand{\xt}[1]{\ensuremath{x^{(t)}_{#1}}}
\newcommand{\xat}[2]{\ensuremath{x^{(#1)}_{#2}}}
\newcommand{\outt}[1]{\ensuremath{out^{(t)}_{#1}}}
\newcommand{\outat}[2]{\ensuremath{out^{(#1)}_{#2}}}
\newcommand{\intt}[1]{\ensuremath{in^{(t)}_{#1}}}
\newcommand{\intat}[2]{\ensuremath{in^{(#1)}_{#2}}}
\newcommand{\Rt}{\ensuremath{R^{<t}}}
\newcommand{\Rat}[1]{\ensuremath{R^{<#1}}}

\newcommand{\extra}[1]{\ensuremath{\textnormal{extra}{(#1)}}}
\clearpage

\section{The Upper Bound} 

We present our algorithms in this section. 
We start with a randomized greedy algorithm for finding an approximate fractional matching by growing a random maximal independent set (MIS). 
This algorithm is the power horse of our approach, and, given its general nature, we provide it in an entirely model-independent manner. 
We then show how this algorithm---in a non black-box way---can be implemented in dynamic graph streams as an $O(\log\log{n})$-pass semi-streaming algorithm.  
Finally, we show how to improve this algorithm to yield an approximation factor of $1+\eps$ and prove the following theorem, which formalizes~\Cref{res:upper}. 

\begin{theorem}\label{thm:upper}
	There is a randomized algorithm that, given any $n$-vertex graph $G=(V,E)$ presented in a dynamic stream, finds an $O(1)$-approximate maximum matching in $G$ using $O(n \cdot \polylog{(n)})$ space
	and $O(\log\log{n})$ passes with high probability. 
	
	Moreover, for any $\eps > 0$, the algorithm can be improved to finding a $(1+\eps)$-approximation even in weighted graphs by increasing the space and number of passes with some $f(\eps)$ factor,
	for some function $f$ that only depends on $\eps$ (and takes different forms depending on whether or not the graph is bipartite or weighted). 
\end{theorem}

The rest of this section is dedicated to the proof of this theorem. 


\subsection{A Random-Order Greedy Algorithm for Matching}\label{sec:random-order} 

We present a model-independent algorithm for computing an approximate fractional matching. 
The algorithm follows the strategy of~\cite{Veldt24} by computing a random order greedy MIS and letting its complement be a vertex cover. In addition to this, whenever the algorithm inserts a new vertex into the vertex cover, it also puts a certain fractional mass distributed uniformly on some \emph{subset} of the not-yet-covered edges incident to this vertex. At the end, the algorithm further ``trims down'' these fractional values to turn them into a fractional matching, by reducing the mass on every vertex to become at most one explicitly. Formally, the algorithm is as follows. 


\begin{Algorithm}\label{alg:random-order} An input graph $G=(V,E)$ and an approximation parameter\footnote{We note that $\beta$ is not exactly the approximation ratio that this algorithm achieves, but rather a quantity that governs this ratio.} $\beta \in (0,1/8)$. 

\vspace{-10pt}
	\begin{enumerate}[leftmargin=15pt]
		\item Let $\VC = \emptyset$ and $x_e = 0$ for all $e \in E$. 
		\item While $G$ is not an empty graph: 
		\begin{enumerate}
			\item\label{line:choose} Pick a vertex $u$ uniformly at random from $G$. 
			\item Add all neighbors $v \in N(u)$ to $\VC$.  
			\item\label{line:w-increase} For any $v \in N(u)$ and $e=(v,w)$ with $\deg(w) \leq \deg(v)$ set $x_e = \beta/\deg(v)$. 
			\item Remove $u \cup N(u)$ and all their edges from $G$. 
		\end{enumerate}
		\item\label{line:update} Create $y \in \IR^E$ initialized with $y = x$. Then, for any vertex $v \in V$ with $y_v := \sum_{e \ni v} y_e > 1$, reduce $y_v$ to $1$ by arbitrarily reducing the fractional values on incident edges of $v$. 
		\item Return $\VC$ as a vertex cover, $x$ as an intermediate solution, and $y$ as a fractional matching. 
	\end{enumerate}
\end{Algorithm}

The following theorem captures the main properties of this algorithm that are important for our purposes. Given how we use this result in the subsequent sections, we provide separate guarantees for the assignments $x$ and $y$. 

\begin{theorem}\label{thm:random-order}
	Given any graph $G=(V,E)$ and parameter $\beta < 1/8$, \Cref{alg:random-order} outputs an integral vertex cover $\VC$, an intermediate solution $x \in \IR^E$, and a fractional matching $y \in \IR^E$ such that 
	\begin{align}
		&\expect{\sum_{e \in E} x_e} \geq \frac{\beta}{2} \cdot \Exp\card{\VC}, \label{eq:random-order-1} \\
		&\expect{\sum_{e \in E} y_e} \geq \paren{\frac{1-8\beta}{1-2\beta}} \cdot \expect{\sum_{e \in E} x_e}. \label{eq:random-order-2}
	\end{align}
\end{theorem}

Notice that since $\VC$ is always a vertex cover, the inequality $\card{\VC} \geq \mu(G)$ always holds, and, thus, combining~\Cref{eq:random-order-1} and~\Cref{eq:random-order-2} in~\Cref{thm:random-order} also directly implies that 
\begin{equation}\label{eq:final-apx-ye}
	\expect{\sum_{e \in E} y_e} \geq \paren{\frac{1-8\beta}{1-2\beta}} \cdot \frac{\beta}2 \cdot \mu(G).
\end{equation}
In other words, the fractional matching $y$ returned by the algorithm is a multiplicative approximation in terms of $\beta$. However, 
we need the separate guarantees of~\Cref{thm:random-order} for the analysis of our dynamic streaming algorithm. 

We now start the proof of~\Cref{thm:random-order}. 
Let $T$ denote the number of iterations of the while-loop in~\Cref{alg:random-order} (which is a random variable). For each iteration $t \in [T]$ of the while-loop, we define: 
\begin{align}
\begin{split}
	\Gt &: \text{the graph $G$ in this iteration}; \\
	\nt &: \text{the number of vertices in \Gt}; \\ 
	\mt &: \text{the number of edges in \Gt}; \\ 
	\ut &: \text{the vertex $u$ chosen from $G$ in this iteration in Line~\eqref{line:choose}}; \\
	\VCt &: \text{vertices added to $\VC$ in this iteration;} \\
	\Nt{v} &: \text{for each $v \in \Gt$,  the neighbors of $v$ in $\Gt$}; \\
	\degt{v} &: \text{for each $v \in \Gt$,  the degree of $v$ in $\Gt$}; \\
	\xt{e} &: \text{for each $e \in \Gt$,  the value {added} to $x_e$ in this iteration}; \\
	\Rt &: \text{the entire random choices made in iterations $1,\ldots,t-1$ by the algorithm}. 
\end{split}\label{eq:all-parameters-t}
\end{align}

The first lemma relates the expected size of the vertex cover computed by~\Cref{alg:random-order} and 
the intermediate solution $x$ that it computes, thus proving~\Cref{eq:random-order-1} in~\Cref{thm:random-order}. 

\begin{lemma}\label{lem:x-vc}
	For the intermediate solution $x \in \IR^E$ computed in~\Cref{alg:random-order}, we have
	\[
		\expect{\sum_{e \in E} x_e} \geq \frac{\beta}{2} \cdot \Exp\card{\VC}.
	\]
\end{lemma}
\begin{proof}
	Fix any iteration $t \geq 1$ in the algorithm and condition on all randomness $\Rt$ up until this iteration. This fixes the graph $\Gt$. 
	Over the randomness of iteration $t$, we have, 
	\[
		\expect{\card{\VCt} \mid \Rt} = \sum_{v \in \Gt} \Pr\paren{v \in \VCt \mid \Rt} = \sum_{v \in \Gt} \frac{\degt{v}}{\nt} = \frac{2\mt}{\nt},
	\]
	where the first equality is by the linearity of expectation, the second is because $v$ joins $\VCt$ if $\ut$ is chosen from $\Nt{v}$, and the last is by the handshaking lemma. 
	
	On the other hand, we have, 
	\begin{align*}
		&\expect{\sum_{e \in G} \xt{e} \mid \Rt} \\
		&=  \hspace{-4pt}\sum_{u \in \Gt} \hspace{-4pt} \left( \Pr\paren{\ut = u \mid \Rt} \cdot \hspace{-4pt}\sum_{v \in \Nt{u}}\hspace{-4pt} \card{\set{w \in \Nt{v}: \degt{w} \leq \degt{v}}} \cdot \frac{\beta}{\degt{v}} \right) \tag{by the definition of the algorithm} \\
		&=  \frac{1}{\nt} \sum_{v \in \Gt} \sum_{u \in \Nt{v}} \card{\set{w \in \Nt{v}: \degt{w} \leq \degt{v}}} \cdot \frac{\beta}{\degt{v}} \tag{as $\ut$ is chosen uniformly from $\Gt$ and by re-ordering the sum for each edge} \\
		&= \frac{1}{\nt} \sum_{v \in \Gt} \sum_{w \in \Nt{v}} \mathbb{I}\bracket{\degt{w} \leq \degt{v}} \cdot \beta \tag{as $\degt{v} = \card{\Nt{v}}$ by definition and $\mathbb{I}\bracket{\cdot}$ is the indicator function} \\  
		&\geq \beta \cdot \frac{\mt}{\nt},
	\end{align*}
	since for each $(v,w)$ at least one of $\mathbb{I}\bracket{\degt{w} \leq \degt{v}}$ or $\mathbb{I}\bracket{\degt{v} \leq \degt{w}}$ is true. 
	
	Combining the above two equations implies that, for any choice of $\Rt$, 
	\[
		\expect{\sum_{e \in G} \xt{e} \mid \Rt} \geq \frac{\beta}{2} \cdot \expect{\card{\VCt} \mid \Rt}, 
	\]
	which, by the law of total expectation, implies that 
	\[
		\expect{\sum_{e \in G}\xt{e}} \geq \frac{\beta}{2} \cdot \Exp{\card{\VCt}}.
	\]
	Noting that 
	\begin{align*}
		\sum_{e \in G} x_e &= \sum_{t \geq 1} \sum_{e \in G} \xt{e} \qquad \text{and }\qquad \card{\VC} = \sum_{t \geq 1} \card{\VCt}, 
	\end{align*}
	and linearity of expectation concludes the proof. 
\end{proof}

While the size of $x$ by~\Cref{lem:x-vc} is sufficiently large, unfortunately, $x$ 
is not necessarily a fractional matching. As such, we do need to run the update in Line~\eqref{line:update} to ``trim down'' $x$ into $y$, which becomes a fractional matching. 
The main step of the proof is to show that this step is not going to reduce the size of $x$ dramatically, proving~\Cref{eq:random-order-2} in~\Cref{thm:random-order}.  

\begin{lemma}\label{lem:x-update}
	For the intermediate solution $x \in \IR^E$ and fractional matching $y \in \IR^E$ in~\Cref{alg:random-order}, 
	\[
		\expect{\sum_{e \in E} y_e} \geq \paren{\frac{1-8\beta}{1-2\beta}} \cdot \expect{\sum_{e \in E} x_e}. 
	\]
\end{lemma}
\noindent
To continue, we need a couple of more notation. For any iteration $t \geq 1$, define: 
\begin{align*}
	\xt{v} &: \text{for each $v \in \Gt$, $\xt{v} := \sum_{e \ni v} \xt{e}$}; \\
	\outt{v} &: \text{for each $v \in \Nt{\ut}$, the value {added} to $\xt{v}$ by $v$ itself}; \\
	\intt{v} &: \text{for each $v \in \Gt$, the value {added} to $\xt{v}$ by any vertex other than $v$}.
\end{align*}
This way, for every $v \in V$ and $t \geq 1$, we have, 
\[
	\xt{v} = \outt{v} + \intt{v}. 
\]
Notice that $\outt{v}$ is non-zero in at most one iteration, wherein $v$ joins $\VC$. For that iteration $t$, 
\[
	\xt{v} = \card{\set{w \in \Nt{v}: \degt{w} \leq \degt{v}}} \cdot \frac{\beta}{\degt{v}} \leq \beta. 
\]
This implies that 
\begin{align}
	x_v := \sum_{t \geq 1} \xt{v} \leq \beta + \sum_{t \geq 1} \intt{v}. \label{eq:beta-intt}
\end{align}
Thus, we only need to focus on bounding $\intt{v}$ across the iterations. We do this in the following claims. Note that in~\Cref{alg:random-order}, 
the vertices $w$ in Line~\eqref{line:w-increase} are the ones that have their $\intt{w}$ increase in this iteration; thus, to avoid confusion, we use $w$ 
in place of $v$ in the following. 

We first show that the total sum that can be assigned to $w$ across the iterations is upper bounded by $\beta$ in expectation. For technical reasons that will become clear shortly, 
we prove this bound in a more general form (in the following claim, the statement above corresponds to setting $t' = 0$). 
\begin{claim}\label{clm:xw-intt-expectation}
	For any iteration $t' \geq 1$, any $w \in V$, and choice of all the random bits $\Rat{t'}$ before iteration $t'$, we have,
	\[
		\expect{\sum_{t \geq t'} \intt{w} \mid \Rat{t'}} \leq \beta. 
	\]
\end{claim}
\begin{proof}
	We know that the given choice of random bits before iteration $t'$, $\Rat{t'}$ fixes the graph $\Gat{t'}$. If $w$ is no longer present in $\Gat{t'}$, $\intat{t'}{w} = 0$, and it remains zero for every iteration after $t'$ also. Hence, proving the statement is trivial. Thus, we assume that the vertex $w$ is present in $\Gat{t'}$.

For any iteration $t \geq t'$, and any choice of random bits $\Rt$ for all preceding iterations,
	\begin{align*}
		\expect{\intt{w} \mid \Rt} &= \hspace{-4pt}\sum_{v \in \Nt{w}} \hspace{-4pt}\mathbb{I}\bracket{\degt{v} \geq \degt{w}} \cdot \Pr\paren{v \in \VCt \mid \Rt} \cdot \frac{\beta}{\degt{v}} \tag{by the definition of the algorithm} \\
		&\leq \degt{w} \cdot \frac{\degt{v}}{\nt} \cdot \frac{\beta}{\degt{v}} \tag{as $v$ belongs to $\VCt$ if $\ut$ belongs to $\Nt{v}$} \\
		&= \beta \cdot \frac{\degt{w}}{\nt}. 
	\end{align*}
	On the other hand, for every vertex $w \in \Gt$ define the event: 
	\begin{align*}
		\event_w(t): \text{vertex $w$ is removed from $\Gt$ by being picked as $\ut$ or joining $\VCt$}. 
	\end{align*}
	We have, 
	\[
		\Pr\paren{\event_w(t) \mid \Rt} = \Pr\paren{\ut \in \set{w} \cup \Nt{w}} = \frac{\degt{w}+1}{\nt}. 
	\]
	Thus, we have that 
	\begin{align}
		\expect{\intt{w} \mid \Rt} \leq \beta \cdot \Pr\paren{\event_w(t) \mid \Rt}. \label{eq:cost<pr}
	\end{align}
	Intuitively, this means that if $\intt{w}$ is expected to be ``large'', then there is also a ``good chance'' that $w$ is removed from $\Gt$ at this iteration and thus no longer receives $\intt{w}$ in the subsequent iterations. 
	We formalize this as follows: 
	\begin{align*}
		&\expect{\sum_{t \geq t'} \intt{w} \mid \Rat{t'}}  \\
	 &= \sum_{t \geq t'} \sum_{\Rt} \mathbb{I}\bracket{\text{$w$ belongs to $\Gt$ after $\Rt \mid\Rat{t'}$}} \cdot \Pr\paren{\Rt \mid \Rat{t'}} \cdot \expect{\intt{w} \mid \Rt, \Rat{t'}} \tag{by the law of total expectation} \\
		&\leq  \sum_{t \geq t'} \sum_{\Rt} \mathbb{I}\bracket{\text{$w$ belongs to $\Gt$ after $\Rt \mid \Rat{t'}$}} \cdot \Pr\paren{\Rt \mid \Rat{t'}} \cdot \beta \cdot \Pr\paren{\event_w(t) \mid \Rt, \Rat{t'}} \tag{by~\Cref{eq:cost<pr}} \\
		&= \sum_{t \geq t'} \Pr\paren{\text{$w$ belongs to $G^{(t')},\ldots,G^{(t)}$ but not $G^{(t+1)} \mid \Rat{t'}$}} \cdot \beta  \tag{by the definition of $\event_w(t)$ and the probability in the inner sum} \\ 
		&= \Pr\paren{\text{$w$ is removed from some $\Gt$ with $t \geq t'$}} \cdot \beta \tag{as the events in the sum are mutually exclusive and partition the space} \\
		&= \beta, \tag{as $w$ will be removed at some iteration after $t'$} 
	\end{align*}
	concluding the proof. 
\end{proof}

We can also observe that the value of $\intt{w}$ in any iteration $t$ is upper bounded by $\beta$. 

\begin{observation}\label{obs:xw-intt-ub}
	For every $t \geq 1$, regardless of the choice of $\Rt$, $\intt{w} \leq \beta$. 
\end{observation}
\begin{proof}
	We have,
		\[
	\intt{w} \leq \sum_{v \in \Nt{w}} \mathbb{I}\bracket{\degt{v} \geq \degt{w}} \cdot \frac{\beta}{\degt{v}} \leq \degt{w} \cdot \frac{\beta}{\degt{w}} = \beta. \qedhere
	\]
\end{proof}

Lastly, using \Cref{obs:xw-intt-ub} and \Cref{clm:xw-intt-expectation}, we can show that the variance of the total value assigned to $w$ is also low, enabling us to prove the required concentration bounds later. 

\begin{claim}\label{clm:xw-intt-variance}
	For every $w \in V$, 
	\[
	\var{\sum_{t \geq 1} \intt{w}} \leq 3\beta\cdot \expect{\sum_{t \geq 1} \intt{w}} .  
	\]
\end{claim}
\begin{proof}
	We can write the variance as,
	\begin{align*}
		\var{\sum_{t \geq 1} \intt{w}} &= \sum_{t \geq 1} \var{\intt{w}} + 2  \cdot \sum_{t_1 < t_2 }  \cov{\intat{t_1}{w}, \intat{t_2}{w}}  \tag{as $\var{\sum_{i}X_i} = \sum_{i} \var{X_i} + \sum_{i \neq j} \cov{X_i,X_j}$ for any random variables $X_1,X_2,\cdots$}\\
		&\leq \sum_{t \geq 1} \expect{(\intt{w})^2} + 2 \cdot \sum_{t_1 < t_2} \expect{\intat{t_1}{w} \cdot \intat{t_2}{w}}  \tag{as $\var{X} \leq \expect{X^2}$ for any random variable $X$} \\
		&\leq \beta \cdot \sum_{t \geq 1} \expect{\intt{w}} + 2 \cdot \sum_{t_1 < t_2} \expect{\intat{t_1}{w} \cdot \intat{t_2}{w}}  \tag{since by \Cref{obs:xw-intt-ub}, we have $\intt{w} \leq \beta $}
	\end{align*}
We will bound the second term by $2 \beta \cdot \expect{\sum_{t \geq 1} \intt{w}}$ to complete the proof. 
\begin{align*}
 \sum_{t_1 < t_2} \expect{\intat{t_1}{w} \cdot \intat{t_2}{w}} &= \sum_{t_1} \sum_{t_2 \geq t_1+1} \expect{\intat{t_1}{w} \cdot \intat{t_2}{w}} \\
 &= \sum_{t_1} \sum_{\Rat{t_1+1}} \Pr\paren{\Rat{t_1+1}} \cdot  \paren{\sum_{t_2 \geq t_1+1} \expect{\intat{t_1}{w} \cdot \intat{t_2}{w} \mid \Rat{t_1+1}}} \tag{by the law of total expectation} \\
  &= \sum_{t_1} \sum_{\Rat{t_1+1}} \Pr\paren{\Rat{t_1+1}} \cdot \expect{\intat{t_1}{w} \mid \Rat{t_1+1}} \cdot \paren{ \sum_{t_2 \geq t_1+1} \expect{\intat{t_2}{w} \mid \Rat{t_1+1}}} \tag{as the value of $\intat{t_1}{w}$ is fixed when conditioned on $\Rat{t_1+1}$} \\
  &\leq  \sum_{t_1} \sum_{\Rat{t_1+1}} \Pr\paren{\Rat{t_1+1}} \cdot \expect{\intat{t_1}{w} \mid \Rat{t_1 +1}} \cdot \beta \tag{by \Cref{clm:xw-intt-expectation}} \\
  &= \beta \cdot \sum_{t_1} \expect{\intat{t_1}{w}} \tag{by the law of total expectation} \\
  &= \beta \cdot \expect{\sum_{t \geq 1} \intt{w}}, \tag{by the linearity of expectation}
\end{align*}
finishing the proof. 
\end{proof}

We use~\Cref{clm:xw-intt-variance} to bound the probability that $\intt{w}$ reaches a ``too large'' value. 

\begin{claim}\label{clm:xw-intt-concentration}
	For every $w \in V$ and every $\theta \geq 1-\beta$, 
	\[
		\Pr\paren{\sum_{t \geq 1} \intt{w} > \theta} \leq \frac{3\beta}{(\theta-\beta)^2} \cdot \expect{\sum_{t \geq 1} \intt{w}}. 
	\]
\end{claim}
\begin{proof}
Let $X := \sum_{t \geq 1} \intt{w}$. Note that $X$ is a sum of at most $n$ random variables as there are at most $n$ iterations in \Cref{alg:random-order} (in each iteration, at least one vertex is removed). 
We have,
\begin{align*}
		\Pr\paren{\sum_{t \geq 1} \intt{w} > \theta} &= \Pr\paren{X> \theta} \leq \frac{\var{X}}{(\theta-\expect{X})^2}  \tag{by Chebyshev's inequality}\\
		&\leq \frac{3\beta \cdot \expect{X}}{(\theta-\expect{X})^2} \tag{by \Cref{clm:xw-intt-variance}} \\
		&\leq  \frac{3\beta}{(\theta-\beta)^2} \cdot \expect{\sum_{t \geq 1} \intt{w}}. \tag{as $\expect{X} = \expect{\sum_{t \geq 1} \intt{w}} \leq \beta < \theta$, by \Cref{clm:xw-intt-expectation} and since $\beta < 1/8$}
	\end{align*}
This proves the claim. 
\end{proof}
We are now ready to prove~\Cref{lem:x-update}. 

\begin{proof}[Proof of~\Cref{lem:x-update}]
For every vertex $v \in V$, define 
\[
	\extra{v} := \min\paren{\sum_{t \geq 1} \intt{v} - (1-\beta), 0}. 
\]
We first have
\begin{align*}
	\expect{\extra{v}} &= \sum_{\theta = (1-\beta)}^{\infty} \Pr\paren{\sum_{t \geq 1} \intt{v} = \theta} \cdot (\theta-(1-\beta)) \\
	&\leq \sum_{\theta = (1-\beta)}^{\infty} \Pr\paren{\sum_{t \geq 1} \intt{v} = \theta} \cdot \theta \tag{as $\beta < 1/8$ and thus removing $-(1-\beta)$ can only increase the sum}\\
	&= \int_{\theta=(1-\beta)}^{\infty} \Pr\paren{\sum_{t \geq 1} \intt{v} \geq \theta} \cdot d\theta \tag{as for any non-negative variable $X$, $\expect{X} = \int_{0}^{\infty} (1-F_X(x))dx$ where $F_X(\cdot)$ is the CDF function}\\
	&\leq \int_{\theta=(1-\beta)}^{\infty}\frac{3\beta}{(\theta-\beta)^2} \cdot \expect{\sum_{t \geq 1} \intt{v}} \cdot d\theta \tag{by~\Cref{clm:xw-intt-concentration}} \\
	&=3\beta \cdot \expect{\sum_{t \geq 1} \intt{v}} \cdot \int_{\theta=(1-\beta)}^{\infty} \frac1{(\theta-\beta)^2} d\theta \\
	&= 3\beta \cdot \expect{\sum_{t \geq 1} \intt{v}} \cdot \frac1{1-2\beta}.
	  \tag{as $\int 1/x^2 dx = -1/x +$ constant}
\end{align*}
Finally, we have, 
\begin{align*}
	\expect{\sum_{e \in E} y_e} &\geq \expect{\sum_{e \in E} x_e}  - \sum_{v \in V} \expect{\extra{v}} \tag{by the update in Line~\eqref{line:update} and \Cref{eq:beta-intt}}. \\
	&\geq \expect{\sum_{e \in E} x_e} - \frac{3\beta}{(1-2\beta)} \cdot \sum_{v \in V} \expect{\sum_{t \geq 1} \intt{v}} \\
	&\geq  \paren{\frac{1-8\beta}{1-2\beta}} \cdot \expect{\sum_{e \in E} x_e} \tag{as $\sum_{v \in V}\expect{\sum_{t \geq 1} \intt{v}} \leq 2\cdot \expect{\sum_{e \in E} x_e}$}
\end{align*}
concluding the proof. 
\end{proof}

The proof of~\Cref{thm:random-order} now follows from~\Cref{lem:x-vc} for~\Cref{eq:random-order-1} and~\Cref{lem:x-update} for~\Cref{eq:random-order-2}. 

\clearpage



\newcommand{\Gbef}[1]{\ensuremath{G^{\textnormal{\textsc{bef}}}_{#1}}}
\newcommand{\Gduring}[1]{\ensuremath{G_{#1}}}
\newcommand{\nbef}[1]{\ensuremath{n_{ #1}}}
\newcommand{\Ubef}[1]{\ensuremath{U_{\leq #1}}}
\newcommand{\Nbef}[1]{\ensuremath{N_{#1}}}
\newcommand{\Nduring}[1]{\ensuremath{N_{#1}}}
\newcommand{\Mstar}{\ensuremath{M^{\star}}}
\newcommand{\verte}{\ensuremath{\textnormal{\textsc{vert}}}}
\newcommand{\ystar}{\ensuremath{y^*}}
\newcommand{\wfix}{\ensuremath{y^{\textnormal{fix}}}}
\newcommand{\wstar}{\ensuremath{y^*}}

\subsection{The Dynamic Streaming Implementation of~\Cref{alg:random-order}}\label{sec:dynamic-alg}

We now show how to implement~\Cref{alg:random-order} in dynamic streams. For this, we follow the approach of~\cite{AhnCGMW15} for implementing the randomized greedy MIS algorithm, 
which also forms the backbone of~\Cref{alg:random-order}. The main new step here is to find the assignments $x$ and $y$ to the edges of the graph in the algorithm (although we will not be able to explicitly find these, but rather a ``proxy'' to them). 
To do this, we need a procedure that can determine the \emph{exact} iteration each vertex is being removed from the graph. This is done via~\Cref{alg:time-stamp} that we design. 

\Cref{alg:time-stamp} finds the set $\VC$ in~\Cref{alg:random-order}, and assigns a \textbf{time stamp} to each
vertex that indicates in which iteration of the while-loop this vertex was removed, namely, was \textbf{settled}. This algorithm, similar to~\cite{AhnCGMW15}, processes the 
graph in $O(\log\log{n})$ \textbf{batches} of vertices with growing sizes. A key new subroutine allows us to determine the time stamp of \emph{all} vertices. 

\begin{Algorithm}\label{alg:time-stamp} An input graph $G=(V,E)$ in a dynamic stream. 
	
	\vspace{-10pt}
	\begin{enumerate}
		\item Set the \textbf{time} $t=0$ and let $\VC = \emptyset$ and $\MIS=\emptyset$. Let $\sigma$ be a random permutation of $V$. 
		\item For $i=1$ to $\log\log{n}$ \textbf{batches}: 
		\begin{enumerate}
			\item Let $k_i := 2 \cdot (n^{1-1/2^{i}} - n^{1-1/2^{i-1}})$ and $U_i := (u_1,\ldots,u_{k_i})$ be the next $k_i$ vertices in $\sigma$ to be processed\footnote{If the remaining graph has $<k_i$ vertices, we let $U_i$ be all remaining vertices.}.  
			\item\label{line:space-bound} In a \textbf{single pass} over the stream, store $G[U_i \setminus (\VC \cup \MIS)]$ using~\Cref{clm:store-subgraph} below. 
			\item\label{line:for-loop-2} At the end of the pass: \textbf{for} $j=1$ to $k_i$ \textbf{do} the following: 
			\begin{enumerate}[leftmargin=5pt]
				\item If $u_j$ is \textbf{settled} already, move to the next vertex of the for-loop in Line~\eqref{line:for-loop-2}. 
				\item\label{line:time-stamp1} Else, increase $t \leftarrow t+1$ and add $u_j$ to $\MIS$ with $t(u_j) = t$ and mark neighbors of $u_j$ in $U_i \setminus (\VC \cup \MIS)$ as settled. 
			\end{enumerate}
			\item\label{line:time-stamp2} In a \textbf{single pass} over the stream, for every $v \in V \setminus \MIS$, find the vertex $u \in N(v) \cap \MIS$ with the smallest value of $t(u)$ 
			using~\Cref{alg:clm-find-time} (to be defined in~\Cref{sec:find-time-stamps}); if such a vertex $u$ is found for $v$, add $v$ to $\VC$, set $t(v) = t(u)$, and mark $v$ as {settled}. 
		\end{enumerate}
	\end{enumerate}
\end{Algorithm}

We start by arguing that this algorithm faithfully simulates~\Cref{alg:random-order}. 

\begin{observation}\label{obs:time-simulates-random}
	For any graph $G=(V,E)$, the set $\set{(v,t(v)) \mid v \in V}$ in~\Cref{alg:time-stamp} is sampled from the same distribution as in~\Cref{alg:random-order} where $t(v)$ in the latter refers 
	to the iteration of the while-loop wherein $v$ joins $\VC$ or is picked as the vertex $u$ (i.e., is removed from the graph). 
\end{observation}
\begin{proof}
	In~\Cref{alg:random-order}, we can think of sampling a vertex $u$ uniformly from the remaining graph in each iteration of the while-loop, as first sampling a random permutation $\sigma$, and 
	then picking the remaining vertex of $G$ with the smallest value of $\sigma$ in each iteration of the while-loop. Moreover, we have that,
	\[
		\sum_{i=1}^{\log \log {n}} k_{i} = 2 \cdot n^{1-1/2^{\log\log{n}}} = \frac{2n}{n^{1/\log{n}}} = n,
	\]
	and thus any remaining vertex will be sampled in the last iteration of the algorithm and will be processed. Then, it is easy to verify that the two algorithms are performing the same 
	exact computation, finalizing the proof.  
\end{proof}

An important remark is in order here. Given the greedy nature of~\Cref{alg:random-order}, for the purpose of the analysis (and by using~\Cref{obs:time-simulates-random}), we can also consider the choice of vertices $\VC$ right after 
we process a vertex $u$ in~\Cref{alg:time-stamp} (to be added to $\MIS$), even though in reality, vertices in $\VC$ are only added after the batch is fully processed. In other words, in the analysis, we can add neighbors of $u$ to $\VC$ 
right at that point even though these vertices will be added to $\VC$ at the end of the batch. 

Let us define some useful notation about the random variables in \Cref{alg:time-stamp} before we proceed. 
\begin{align*}
	\Gduring{i} &: \text{the graph $G[U_i \setminus (\VC \cup \MIS)]$ that is stored in batch $i$}; \\
	\nbef{i} &:= 2n^{1-1/2^i} = \sum_{j=1}^i k_j,\text{ which is the number of vertices in all of $U_1 \cup U_2 \cup \ldots \cup U_i$}; \\
	\Ubef{i} &:  \text{the set of $\nbef{i}$ vertices in $U_1 \cup U_2 \cup \ldots \cup U_i$};  \\
	\Nduring{i}(w) &: \text{for any vertex $w \in \Gduring{i}$, this is the neighborhood of $w$ in graph $\Gduring{i}$}. 
\end{align*}
(We emphasize that in the above notation, $n_i$ is \emph{not} the number of vertices in $G_i$ and is larger.) 

We prove a helper lemma that allows us to bound the space complexity of this algorithm in Line~\eqref{line:space-bound}. This is a standard result at this point---originally due to~\cite{AhnCGMW15}---and is often
referred to as the ``residual sparsity property'' of the greedy (MIS) algorithm~\cite{GhaffariGKMR18,Konrad18,AssadiOSS19}. We thus provide the proof only for completeness. 

\begin{lemma}[cf.~\cite{AhnCGMW15}]\label{lem:residual}
	For $1 < i \leq \log \log n$, let $\Delta(\Gduring{i})$ denote the maximum degree of $\Gduring{i}$. Then, with high probability, 
	\[
		\Delta(\Gduring{i}) \leq 100 \cdot (\nbef{i}/\nbef{i-1}) \cdot \ln{n}. 
	\]
\end{lemma}
\begin{proof}
	We prove the statement for $i+1$ with $i \geq 1$.
	Fix the $\nbef{i+1}$ vertices in $\Ubef{i+1}$, but \emph{not} their order in the permutation $\sigma$. 
	We can think of sampling the set $\Ubef{i}$ as picking $\nbef{i}$ vertices $u_1,\ldots,u_{\nbef{i}}$ uniformly at random from the set $\Ubef{i+1}$ one at a time (without replacement). Now, consider a vertex $v$ in $\Ubef{i+1}$ 
	and for $j \in [\nbef{i}]$, let $d_j(v)$ denote the degree of $v$ in $G[\Ubef{i+1} \setminus (\VC^j \cup \MIS^j)]$ where $\VC^j \cup \MIS^j$ includes the part of $\VC$ and $\MIS$ that will be added due to the 
	choices of vertices $u_1,\ldots,u_{j-1}$ and their neighbors. Finally, let $d(v)$ 
	denote the degree of $v$ in $G_{i+1}$. 
	We have,
	\begin{align*}
		\Pr\paren{v \in \Gduring{i+1} \wedge d(v) \geq 100 \cdot \frac{\nbef{i+1}}{\nbef{i}} \cdot \ln{n}} &\leq \prod_{j=1}^{\nbef{i}} \Pr\paren{\text{$u_j$ is not in current neighbors of $v$} \mid u_{1},\ldots,u_{j-1}} \tag{as $v$ remaining in $G_{i+1}$ means
		none of its neighbors are sampled in $\MIS$} \\
		&\leq \prod_{j=1}^{\nbef{i}} \paren{1-\frac{d_j(v)}{\nbef{i+1}}} \tag{as each $u_j$ is chosen without replacement from $\nbef{i+1}-1-(j-1)$ vertices at this point} \\
		&\leq \paren{1-\frac{100 \cdot (\nbef{i+1}/\nbef{i}) \cdot \ln{n}}{\nbef{i+1}}}^{\nbef{i}} \tag{as $d_j(v) \geq d(v)$ since the degrees drop monotonically as $\VC$ and $\MIS$ grow} \\
		&\leq \exp\paren{-\frac{\nbef{i} \cdot 100  \cdot \ln{n}}{\nbef{i}}} \tag{as $1-x \leq e^{-x}$} = n^{-100}. 
	\end{align*}
	 A union bound over all the vertices concludes the proof. 
\end{proof}


\begin{claim}\label{clm:store-subgraph}
	Line~\eqref{line:space-bound} of~\Cref{alg:time-stamp} can be implemented in $O(n\log^2{n})$ space with high probability. 
\end{claim}
\begin{proof}
	The statement holds vacuously for $i = 1$, because the total number of vertices in $G_1$ is $k_1 = n_1 = 2\sqrt{n}$; hence, we can store a counter between all pairs of vertices in $G_1$ 
	during the stream and recover all the edges in $O(n \log n)$ space trivially. 
	
	Consider each batch $1 <i \leq \log\log{n}$ of the algorithm. By~\Cref{lem:residual}, with high probability, we can bound the maximum degree of the graph $\Gduring{i}$. We 
	We also know that the total number of vertices stored in batch $i$ is $k_i \leq \nbef{i} =  2n^{1-1/2^i}$. Thus, the total number of edges in $\Gduring{i}$ is at most,
	\begin{equation}\label{eq:dummy-2}
		\nbef{i} \cdot \Delta(\Gduring{i}) \leq 100 \cdot (\nbef{i}^2/\nbef{i-1})\cdot \ln{n}  = 100 \cdot (4 n^{2-2/2^i} / 2 n^{1-2^{i-1}}) \cdot \ln{n} = 200 \cdot n \cdot \ln{n},
	\end{equation}
	where the first inequality is by~\Cref{lem:residual} and the second equality is by the choice of $\nbef{i}$ and $\nbef{i-1}$.

	We now run a sparse recovery algorithm to recover all edges of $G[U_i \setminus (\VC \cup \MIS)]$ in Line~\eqref{line:space-bound}. Specifically, define $\phi \in \set{0,1}^m$ to be 
	the indicator vector of edges of this subgraph. Given that the algorithm explicitly stores $U_i, \VC,$ and $\MIS$, 
	we can define $\phi$ on the fly when seeing the updates to the edges of $G$. 
	
	The total number of non-zero elements of $\phi$ is $O(n \log n)$, by \Cref{eq:dummy-2}. Thus, we can run the deterministic sparse-recovery of \Cref{prop:sparse-recovery} with a single pass 
	over $\phi$, and the choice of $q = O(n\log{n})$ using $O(n\log^2{n})$ bits of space. This concludes implementation of Line~\eqref{line:space-bound}. 
\end{proof}

\subsubsection{Finding Time Stamps in~\Cref{alg:time-stamp}}\label{sec:find-time-stamps}

We now show how to find the time stamps in Line~\eqref{line:time-stamp2} of~\Cref{alg:time-stamp}. We devise the following algorithm for each batch $i \in [\log\log{n}]$. 
Note that at this point, the algorithm has computed $\MIS$ and their time stamp for the current batch $i$ but have not done so for $\VC$ and that is the task of the following algorithm.

	\begin{Algorithm}\label{alg:clm-find-time} 
	For implementing Line~\eqref{line:time-stamp2} of~\Cref{alg:time-stamp} in each batch $i \in [\log\log{n}]$. 
	
	\begin{enumerate}
		\item Partition the set $U_i$ into $b = \log{k_i}$ groups based on geometrically increasing sizes, 
		namely, for every $j \in [b]$, $U_{i,j}$ contains the next $2^j$ elements of $U_i$ in the order of permutation $\sigma$:\footnote{The last group $b$ may contain less than $2^b$ elements, as it only has all the remaining elements of $U_i$.}
		\[
			U_{i,j} = \set{u_{2^j-1}, u_{2^i}, \ldots, u_{2^{j+1}-2}}.
		\]
		\item For every $v \in V$ and $j \in [b]$, define the vector $\phi(v,j) \in \set{0,1}^V$ as the indicator vector of $N(v) \cap U_{i,j} \cap \MIS$. 
		Note that by the end of Line~\eqref{line:for-loop-2} of~\Cref{alg:time-stamp}, the set $\MIS \cap U_{i,j}$ is known and each update $(u,v)$ to the dynamic stream for edges of the input graph $G$ can be used to also update the vectors $\phi(u,j)$ and $\phi(v,j)$ for all $j \in [b]$. 
		
		\item Let $q := 200\ln{n}$ and for every $v \in V$ and $j \in [b]$, run a randomized $q$-sparse recovery algorithm of~\Cref{prop:sparse-recovery} on $\phi(v,j)$ and test if $\phi(v,j)$ is $q$-sparse or not, using $\delta = 1/n^{200}$. 
		
		\item Let $j \in [b]$ be the smallest index such that~\Cref{prop:sparse-recovery} declares $0 < \norm{\phi(v,j)}_0 \leq q$ and returns $\phi(v,j)$; let $u \in \MIS$ be the vertex with the minimum $t(u)$ in the support of $\phi(v,j)$. 
		
		Return $(u,t(u))$ as the choice for the vertex $v$ in Line~\eqref{line:time-stamp2} of~\Cref{alg:time-stamp} (if no such $j$ is found for $v$, return no $u$ exists for $v$). 
	\end{enumerate}
\end{Algorithm}

\begin{lemma}\label{lem:find-time}
	With high probability, \Cref{alg:clm-find-time} finds the time-stamps in 
	Line~\eqref{line:time-stamp2} of~\Cref{alg:time-stamp} and can be implemented in $O(n\log^3{n})$ space. 
\end{lemma}
\begin{proof}
We will show that for any vertex $v \in V \setminus \MIS$, \Cref{alg:clm-find-time} finds the required time-stamp $t(v)$ with high probability.
	
 Consider the smallest index $\jstar \in [b]$ where $\norm{\phi(v,\jstar)}_0 > 0$; if no such $\jstar$ exists then $v$ is not a neighbor to any vertex in $\MIS$ and thus will not be added to $\VC$. Otherwise, 
	the neighbor $u$ of $v$ in $\MIS$ with the smallest value of $t(u)$ belongs to $U_{i,\jstar}$. Thus, returning such $u$ is the correct answer. 

Let us see how the algorithm performs for any $j \in [b]$, and enumerate the sources of error.
	\begin{enumerate}[label=$(\roman*)$]
	 \item For any $j < \jstar$, we know that $\norm{\phi(v, j)}_0 = 0$, but the randomized sparse-recovery may return that $\norm{\phi(v,j)}_0 > q$ with probability at most $\delta$. 
	 \item For $j = \jstar$, it may be the case that $\norm{\phi(v, j)}_0 > q$ or if $\norm{\phi(v,j)}_0 \leq q$, but randomized sparse-recovery returns otherwise. 
	\end{enumerate}
	
	Firstly, let us condition on the fact that randomized sparse-recovery performs correctly and retrieves $\phi(v,j)$ for all $v \in V \setminus \MIS$ and $j \in [b]$. 
	Using union bound over all vertices and $j \in [b]$, we get that this event happens with probability at least $1-n^{-198}$ for the choice of $\delta = n^{-200}$.

	Now, we only need to bound the probability that the case $\norm{\phi(v, \jstar)_0} > q$ happens.  We do so using the following intermediate claim.

\begin{claim}\label{clm:finding-label-min}
		For any $j > 1$ and any vertex $v \in V \setminus \MIS$, if $\norm{\phi(v,j-1)}_0 = 0$, then $\norm{\phi(v,j)}_0 \leq q$ with high probability. 
\end{claim}
\begin{proof}
	 Let us fix the vertices $U_i$ which are picked to be in $\sigma$, but not the order in which they are processed, i.e., we know the $k_i$ vertices that belong to $U_i$ but not the vertex groupings of $U_{i,j}$ for $j \in [b]$.
	We define some useful notation to prove this claim. 
	\begin{align*}
	G_{i,j} &: \text{the graph $G[U_i \setminus (\VC \cup \MIS)]$ at the beginning of sampling $U_{i,j}$, where, $\MIS \cup \VC$ } \\
	&\text{\hspace{4mm}include all the vertices in $U_i$ settled by  all the vertices in $\sigma$ till $U_{i,j-1}$};\\
	d_{j}(v)&:\text{the degree of $v$ in $G_{i,j}$}; \\
	k_{i,j}&:k_i - (2^1 + 2^2 + \ldots + 2^{j-1}), \text{ the number of vertices left in $U_i$ after}\\
	&\text{\hspace{4mm}$U_{i,1}, U_{i,2}, \ldots, U_{i,j-1}$ are sampled} 
	\end{align*}

 For any $j < b$ with $2^j \leq k_{i,j}$ (group $j$ is not the last group), we can think of picking $U_{i,j}$ as sampling $2^j$ vertices from the $k_{i,j}$ vertices of $U_i \setminus (U_{i,1} \cup U_{i,2} \cup \ldots \cup U_{i,j-1})$ without replacement. Each vertex  is picked in $U_{i,j}$ with probability $p_j$, where,
	\begin{equation}\label{eq:dummy-pj}
	p_j := 2^j/k_{i,j}.
	\end{equation}

	Suppose $d_j(v) < 100\ln{n}/p_j$. Then, in expectation, we know that there are at most $100 \ln n$ neighbors of $v$ in $G \setminus (\MIS \cup \VC)$. We use Chernoff bound for negatively correlated random variables in \Cref{prop:chernoff} (because of sampling without replacement) to argue that 
	with high probability, we cannot have more than $200\ln{n} = q$ neighbors of $v$ in $G \setminus (\MIS \cup \VC)$ that are sampled in $U_{i,j}$; even if all of those neighbors join $\MIS$, 
	we will still have that sparsity of $\phi(v,j)$ is at most $q$ as desired. 
	
	Now we need to argue the case when $d_j(v) \geq 100\ln{n}/p_j$. This implies that, 
	\begin{align*}
	d_{j}(v)& \geq 100\ln{n}/p_j  \tag{as the set $\MIS \cup \VC$ only grows} \\
	 &= 100\ln{n} \cdot \frac{k_{i,j}}{2^j} \tag{by value of $p_j$ in \Cref{eq:dummy-pj}}\\
	 &= 50 \ln{n} \cdot \frac{k_{i,j}}{2^{j-1}} \\
	 &= 50 \ln{n} \cdot \frac1{p_{j-1}} \cdot \frac{k_{i,j}}{k_{i,j-1}}. \tag{given that $p_{j-1} = 2^{j-1}/k_{i,j-1}$} \\
	 &= 50 \ln{n} \cdot \frac1{p_{j-1}} \cdot \frac{k_{i,j-1}-2^{j-1}}{k_{i,j-1}} \tag{by definition of $k_{i,j}$} \\
	 &= 50 \ln{n} \cdot \frac1{p_{j-1}} \cdot (1-p_{j-1}) \\
	 &\geq 25 \ln{n} \cdot (1/p_{j-1}). \tag{$k_{i,j-1} \geq 2^j + 2^{j-1} \geq 2 \cdot 2^{j-1}$, we get $p_{j-1} \leq 1/2$}
	\end{align*}
When group $j$ is the last group with $2^j \geq k_{i,j}$ (here $p_j = 1$), we have,
\[
	d_{j}(v) \geq 100 \ln{n} \geq 25 \ln {n} \cdot (1/p_{j-1}),
\]
 where we have used that $k_{i,j-1} < 2^{j-1} + 2^j < 4 \cdot 2^{j-1}$, and  $p_{j-1} \geq 1/4$.

We have argued that $d_{j}(v) \geq 25 \ln{n}/p_{j-1}$ in all cases. Among these $d_j(v)$ vertices, none of them are sampled in $U_{i,j-1}$ as $v$ is not settled in $U_{i,j}$. However, in expectation, $25 \ln {n}$ of these vertices must have been sampled in $U_{i,j-1}$. 
We use the same argument as in~\Cref{lem:residual} to say that at least one of these $d_j(v)$ vertices will be sampled in $U_{i,j-1}$ and $v$ will have a neighbor in $\MIS$ already inside $U_{i,j-1}$, with high probability. This contradicts our assumption that $\norm{\phi(v,j-1)}_0 = 0$.

	Thus, in this case $j$ cannot be the first index with non-empty support in $\phi(v,j)$. \Qed{clm:finding-label-min}
	
\end{proof}
	
	To finalize the proof, we have that if $\jstar = 1$, $U_{i,1}$ has only 2 vertices and $\phi(v, 1)$ is $2$-sparse; for any $\jstar>1$, we use \Cref{clm:finding-label-min} to argue that for any vertex $v$, $\phi(v, \jstar)$ is $q$-sparse with high probability. 
	This proves the correctness of the algorithm as argued earlier. 

We can now bound the space of the algorithm. For every $v \in V$, we are maintaining $O(\log{n})$ randomized sparse-recovery algorithms, each for recovering an $O(\log{n})$-sparse vector from a domain of size $\leq n$ and error probability $\delta=1/\poly(n)$; thus, 
by~\Cref{prop:sparse-recovery}, this needs $O(\log^3{n})$ space per vertex, and $O(n\log^3{n})$ space in total, concluding the proof of~\Cref{lem:find-time}. 
\end{proof}

\subsubsection{Finding a Large Matching from~\Cref{alg:time-stamp} via the Reduction of~\Cref{alg:random-order}}

\Cref{alg:time-stamp} allows us to recover the time stamps of all vertices in a single run of the randomized greedy MIS algorithm in $O(\log\log{n})$ passes over a dynamic stream. 
We now use this information alongside our reduction of approximate matching to randomized greedy MIS in~\Cref{alg:random-order}, to recover a large matching from the input graph. 

Let us first recall some notation about some variables in \Cref{alg:random-order}: 
\begin{align*}
	\Gt &: \text{the graph $G$ at iteration $t$}; \\
	\degt{v} &: \text{for each $v \in \Gt$, denotes the degree of $v$ in $\Gt$}; \\
	x_e &:\text{the value given to edge $e$ in $x \in \IR^{E}$ from \Cref{eq:random-order-1}} \\
	y_e &: \text{the value given to edge $e$ in fractional matching $y \in [0,1]^{E}$ from \Cref{eq:final-apx-ye}}.
\end{align*}
As we have stated earlier, given the support of the fractional matching $y$ can be too large, we cannot hope to recover it explicitly. Instead, 
our goal is to \emph{sample} the edges of the graph with probabilities proportional to their $y$-values, and then use the sampled edges to find a large matching. 
Specifically, for every $e \in E$, define: 
\begin{equation}\label{eq:pe-values}
	p_e:= \min(1, y_e \cdot 200 \log n).
\end{equation}
We would like to sample each edge of the graph with probability $p_e$. The challenge however is that we will not be able to actually recover the values of $y$ (or even $x$) explicitly and learn $p_e$'s, and thus need to use a ``proxy'' for them algorithmically.

To start with, we have the time stamps of all the vertices. Let us show that this also gives us the degrees of all the vertices at the time in which they are settled. 
\begin{observation}\label{obs:deg-at-removal}
	Given the time stamps $t(v)$ for all $v \in V$, using one pass and $O(n\log{n})$ space, we can find $\degat{t(v)}(v)$ for all vertices $v \in V$. 
\end{observation}
\begin{proof}
	For any vertex $v$, $\degat{t(v)}(v)$ only includes the edges of $v$ to vertices $w \in V$ which have $t(w) \geq t(v)$. We know the labels of all vertices explicitly, so it is easy to count the total number of edges from $v$ to vertices with labels after $v$
	using a counter per vertex. 
\end{proof}

Equipped with~\Cref{obs:deg-at-removal}, we can assume we also have the remaining degree of every vertex $v \in V$ at the time it is settled. This fixes the value $x$ will assign to the edges of $v$. 
The problem however is that we still do not have sufficient information to perform the check in Line~(\ref{line:w-increase}) of~\Cref{alg:random-order} to know which incident edges of $v$ receive
a non-zero fractional matching, i.e., which edges $(v,w)$ satisfy $\deg(w) \leq \deg(v)$; this is because for this check, we need to know the degree of $w$ at the time $v$ is being settled not $w$ itself. 

We side step this issue in the following by defining an intermediate assignment $z \in \IR^E$ which we can explicitly find for any pair of vertices. 
To formally define the vector $z$, we describe the notion of assigning a vertex pair to one of the vertices. 

\begin{Definition}[Assignment of vertex pairs to vertices]\label{def:assign-edges} For any pair of vertices $u,v$,  we say the pair $u,v$ is assigned to vertex $u$, denoted by $\verte(u,v)$, iff:
\begin{enumerate}[label=$(\roman*)$]
\item \label{line:assign-1}  $u$ is settled before vertex $v$, i.e., $t(u) < t(v)$, or else,
\item  \label{line:assign-2} $u, v$ are removed at the same time with $t(u) = t(v)$ and $u \in \VC$ while $v \in \MIS$, or else,
\item  \label{line:assign-3} $t(u) = t(v)$, $u, v \in \VC$, and $\degat{t(u)}(u) \geq \degat{t(v)}(v)$ (breaking the ties between $u$ and $v$ in case of the equality consistently, say, based on whichever appear in $\sigma$ first).
\end{enumerate}
\end{Definition}

Note that by~\Cref{def:assign-edges}, every pair $u,v$ of vertices, regardless of whether or not is an edge, is assigned to one of its endpoints. 
We then define $z_{uv}$ for a pair $u,v\in V$ as, 
\begin{equation}\label{eq:final-z}
	z_{uv} := \frac{200 \log n}{\degat{t(\verte(u,v))}(\verte(u,v))};
\end{equation}
that is $z_{uv}$ is proportional to the inverse of the degree of its assigned vertex, at the time this assigned vertex was settled. 
When vertices $u,v$ do contain an edge $e = (u,v)$ in the input graph, we use $z_e$ and $\verte(e)$ also to denote $z_{uv}$ and $\verte(u,v)$ interchangeably. 

We observe that for any vertex pair $u,v$, the value of $z_{uv}$ can be found easily.
\begin{observation}\label{obs:find-z}
	For any vertex pair $u,v$, the value of $z_{uv}$ can be determined given the information collected by the algorithm.  
\end{observation}
\begin{proof}
	The only information required to find $\verte(e)$ is the time stamps, the sets $\VC$ and $\MIS$, and the degrees at the time of removal, and lastly the permutation $\sigma$, all of which is stored in the memory explicitly (using~\Cref{lem:find-time}
	and~\Cref{obs:deg-at-removal}). 
	Thus, we can find for each vertex pair what $\verte(u,v)$ is, and then, determining $z_{uv}$ is trivial.
\end{proof}

The following claims shows the relevance of $z$-values for us. 

\begin{claim}\label{clm:ye-ze-ub}
	For every edge $e \in E$, $p_e \leq z_e$ where $p_e \in [0,1] $ is from \Cref{eq:pe-values}.
\end{claim}
\begin{proof}
	It is sufficient to show that $x_e \cdot 200\log n \leq z_e$ given that $y_e \leq x_e$ always ($y$ is obtained from $x$ by reducing some of its values) and by the definition of $p_e$ based on $y_e$ in~\Cref{eq:pe-values}. 
	
	For any edge $e = (u,v) \in E$, the value of $x_e $ is updated in at most one iteration of the loop in Line~\eqref{line:for-loop-2} of \Cref{alg:random-order}, i.e., the iteration where either of $u$ or $v$ are settled. We will argue that
	\[
	x_e \leq \frac{\beta}{\degat{t(\verte(e))}(\verte(e))}.
	\]
	We have multiple cases to consider:
	\begin{itemize}
		\item When $v$ is settled (strictly) before $u$, that is $t(v) < t(u)$, we have that $\verte(e) = v$ from Line~\ref{line:assign-1} of~\Cref{def:assign-edges}. We know that $x_e$ is set as $\beta/\degat{t(v)}(v)$ if $\degat{t(v)}(u) \leq \degat{t(v)}(v)$,  otherwise $x_e$ is set as zero. In either case, $x_e $ is less than $\beta/\degat{t(\verte(e))}(\verte(e))$. This is similarly true when $u$ is settled (strictly) before $v$. 
		\item When $v$ and $u$ are settled at the same time, with $t(u) = t(v)$, and if $u$ goes to $\MIS$ at this iteration, then $v \in \VC$, and $x_e$ is again updated as $\beta/\degat{t(v)}(v)$. We have that $\verte(e) = v$ by Line~\ref{line:assign-2} of~\Cref{def:assign-edges} and can apply the above reasoning.
		\item When both $u,v $ go to $\VC$ and have the same timestamps, the value of $x_e$ is set as $\beta/\degat{t(v)}(v)$ if $v$ is the higher degree vertex among $u,v$ at time $t(u) = t(v)$. Again, we have that $\verte(e) = v$ by Line~\ref{line:assign-3} of~\Cref{def:assign-edges} and can apply the above reasoning.
	\end{itemize}
	The claim follows by observing that, 
	\[
	x_e \cdot 200 \cdot \log n \leq \frac{\beta \cdot 200 \cdot \log n}{\degat{t(\verte(e))}(\verte(e))} \leq \frac{200 \cdot \log n}{\degat{t(\verte(e))}(\verte(e))} =  z_e,
	\]
	where we have used that $\beta < 1$.
\end{proof}

Given~\Cref{clm:ye-ze-ub}, for algorithmic purposes, we can instead sample the edges with probability proportional to $z$-values, which we \emph{can} compute. 
We now need to check this can be done efficiently in dynamic streams (and that the size of the sampled edges are not too large, given $z$ is  \emph{not} a fractional matching).
Our sampling is done in two steps: we first sample \emph{all} pairs $u,v$ based on their value $z_{u,v}$ to obtain a vector $\phi_z \in \IR^{{{V}\choose{2}}}$ and then maintain a sparse-recovery algorithm over $\phi_z$ to 
recover the actual edges inside this sample. Furthermore, we need to sample $\phi_z$ itself with \emph{limited independence hash functions} (defined in~\Cref{sec:concentration}) 
and work with $\phi_z$ \emph{implicitly} given that its size is larger than the allowed space. We now formalize this: 

\paragraph{Sampling process.} Let $\kappa := 200\log{n}$ and sample a $\kappa$-wise independent hash function from \Cref{def:k-wise-ind-hash} with domain of the set of vertex pairs and the range of $[0,1]$%
\footnote{\label{footnote:hash} For the simplicity of exposition, we let the range of the hash function to be the interval $[0,1]$. For algorithmic purposes, one needs to further discretize this range to integers. Specifically, use the range as $\{1, 2, \ldots, \poly(n)\}$ for some large polynomial in $n$ and consider a mapping of $[0,1]$ to this range in a standard way. The distribution when sampling from the hash family with this large discrete domain varies from the range $[0,1]$ 
by at most $1/\poly(n)$ in total variation distance, and this negligible difference can be added to the probability of error of the algorithm.}:
\begin{equation}\label{eq:hash-family-sample}
	h : \binom{n}2 \rightarrow [0,1]. 
\end{equation}
Let $\Phi_z \subseteq {{V}\choose{2}}$ be the set of pairs such that $h(u,v) \leq \min(1,z_{u,v})$ and $\phi_z \in \set{0,1}^{\Phi_z}$ be the incidence vector of $E_z:= E \cap \Phi_z$. This way, the elements in the support of $\phi_z$ 
are obtained from $E$ by sampling each edge $e \in E$, using a $\kappa$-wise independent hash function, with probability $z_e$. 

\medskip

We now analyze this sampling process. The size of $\Phi_z$ may still be too large for us but we show that the total number of \emph{actual edges} sampled, i.e., the support of $\phi_z$ will not be large
and can be found and stored explicitly by the algorithm. 

\begin{claim}\label{clm:edges-storage}
	The set $E_z$ can be found and stored in $O(n \log^2 n)$ space explicitly with high probability. 
\end{claim}
\begin{proof}
	We will prove that size of $E_z$ is $O(n\log{n})$ with high probability (over the randomness of $h$ alone). 
	This will be sufficient to prove the claim since we can implicitly maintain the set $\Phi_z$ using the hash function $h$, which itself can be stored in $O(\kappa \cdot \log{n}) = O(\log^2{n})$ space using~\Cref{prop:storing-sampling-hash} (see also~\Cref{footnote:hash}); then, we will store a randomized sparse-recovery of~\Cref{prop:sparse-recovery} for the vector $\phi_z$ and each relevant update on the graph $G$ can be passed to this vector given the implicit access to $\Phi_z$. This way, 
	we can recover  $E_z$, the support of $\phi_z$, in $O(n\log^2{n})$ space with high probability. 
	
	It thus suffices to prove the bound on the size of $E_z$.  For any vertex $v \in V$, let $A(v)$ be the set of edges $e \in E$ that are assigned to $v$, i.e., $\verte(e) = v$. 
	We claim that $A(v) \cap E_z$ is of size $O(\log{n})$ with high probability which implies the bound on $E_z$ immediately since the set $\set{A(v) \mid v \in V}$ covers all edges of the graph. 
	
	Let us fix a vertex $v \in V$. Let $d_v = \degat{t(v)}(v)$ be the degree of $v$ at the time it is settled. 
	We argue that $\card{A(v)} \leq d_v$. This is because any edge $e = (v,w)$ assigned to $v$ satisfies $t(w) \geq t(v)$ and thus belongs to the graph at the time $v$ is settled. 
	If $d_v \leq 200\log{n}$, we are already done, so in the following we assume $d_v > 200\log{n}$. 
	
	For $e \in A(v)$, define the indicator random variable $X_e \in \set{0,1}$ which is $1$ iff $e$ is sampled in $A_z$. By~\Cref{eq:final-z}, we have $\expect{X_e} = 200\log{n}/d_v$
	and thus by the argument above
	\[
		\Exp\card{E_z \cap A(v)} = \card{A(v)} \cdot \frac{200\log{n}}{d_v} \leq d_v \cdot \frac{200\log{n}}{d_v} = 200\log{n}. 
	\]
	Moreover, $\set{Z_e}_{e \in A(v)}$ are chosen via a $\kappa$-wise independent hash function. Thus, by~\Cref{prop:chernoff-hash}, 
	\[
		\Pr\paren{\card{E_z \cap A(v)} \geq 400\log{n}} \leq \exp\paren{-\min\paren{\frac{\kappa}{2}, \frac{2}{3} \cdot 200\log{n}}} \leq n^{-100},
	\]
	by the choice of $\kappa = 200\log{n}$. A union bound over all the vertices concludes the proof.
\end{proof}

We have proved that our sampling process does not violate the space constraint and we can obtain the set $E_z$
as well. We now finalize the argument by showing that $E_z$ contains a large matching in expectation by relating it to the fractional matching $y$ of~\Cref{alg:random-order}.  
This is a combination of the standard fact that sampling a fractional matching leads to a large matching as well (slightly modified to work even with limited-independence sampling), and a rejection sampling argument 
in the analysis to relate sampling with probability $z_e$ to the one with $p_e$ needed for the fractional matching sampling. 

\begin{lemma}\label{lem:good-matching}
	There exists a matching $M$ in the sampled edges $E_z$ with 
	\[
		\Exp\card{M} \geq \paren{\frac{1-8\beta}{1-2\beta}} \cdot \frac{\beta}{10} \cdot \mu(G). 
	\]
\end{lemma}
\begin{proof}
	Fix a choice of random ordering $\sigma$ in~\Cref{alg:random-order} and its simulation in~\Cref{alg:time-stamp} and let $y \in \IR^E$ be the fractional matching defined by this process. 
	Let $R$ denote the entire randomness of these algorithms, which we condition on in the following proof. Thus, the only randomness remained at this point is in the sampling process of $E_z$.

	Consider any choice of the hash function $h$ in the sampling process from earlier. Define the following alternative sampling process for obtaining a set $E_p$ as opposed to $E_z$: 
	add any edge $e \in E$ to $E_p$ iff $h(e) \leq p_e$ (recall $p_e = \min(1,200\log{n} \cdot y_e)$ by~\Cref{eq:pe-values}). Since $p_e \leq z_e$ by~\Cref{clm:ye-ze-ub}, for any choice of $h$, the set $E_p \subseteq E_z$. We can thus think of $E_p$ as 
	obtained from $E_z$ by rejecting the edges with $p_e < h(e) \leq z_e$. We emphasize that we are not algorithmically finding $E_p$ but rather only define it for the analysis. Thus, picking the random hash function $h$ 
	defines a sampling process for picking $E_p$ as well.

	We now argue that 
	\begin{align}
		\expect{\mu(E_p) \mid R} \geq \frac{1}{5} \card{y}, \label{eq:last-step-of-the-proof}
	\end{align}
	where $\card{y} = \sum_{e \in E} y_e$. This is sufficient to conclude the proof since, 
	\begin{align*}
		\expect{\mu(E_z)} &= \Exp_R \Bracket{\expect{\mu(E_z) \mid R}} \geq \Exp_R \Bracket{\expect{\mu(E_p) \mid R}} \tag{by the law of total expectation and since $\mu(E_z) \geq \mu(E_p)$ given $E_p \subseteq E_z$} \\
		&\geq \frac{1}{5} \cdot \Exp_R \Bracket{\expect{\card{y} \mid R}} = \frac{1}{2} \cdot \Exp\card{y} \tag{by~\Cref{eq:last-step-of-the-proof} and the law of total expectation again}\\
		&\geq \frac{1}{5} \cdot \paren{\frac{1-8\beta}{1-2\beta}} \cdot \frac{\beta}{2} \cdot \mu(G) \cdot \paren{1-1/\poly{(n)}},
	\end{align*}
	where the final inequality holds by~\Cref{thm:random-order} for the choice of $y$ in~\Cref{alg:random-order} with the additional $1-1/\poly{(n)}$ accounting from the correctness probability  
	of~\Cref{alg:time-stamp} in~\Cref{lem:find-time}. 
	
	It thus remains to prove~\Cref{eq:last-step-of-the-proof}. 
	Define the following assignment $\ystar_e$ for each edge $e \in E$: 
	\begin{equation*}
		\ystar_e:= \begin{cases}
			y_e \qquad &\textnormal{if $e \in E_p$ and $y_e \geq 1/(200 \log n)$,} \\
			y_e/p_e \qquad &\textnormal{if $e \in E_p$ and $y_e < 1/(200\log n)$,} \\
			0 \qquad &\textnormal{otherwise.}
		\end{cases}
	\end{equation*}
	This way, $\expect{\ystar_e \mid R} = y_e$ for all $e \in E$: if $y_e \geq 1/(200\log{n})$, we have $p_e = 1$ and thus $h(e) \leq p_e$ always meaning $e \in E_p$ and $\ystar_e = y_e$ deterministically. 
	Otherwise, $e$ belongs to $E_p$ with probability $p_e < 1$ and thus $\expect{\ystar_e} = p_e \cdot y_e/p_e = p_e$. 
	
	We have 
	\[
		\expect{\sum_{e \in E} \ystar_e \mid R}  = \sum_{e \in E} y_e = \card{y},
	\]
	and at the same time $\ystar$ is only supported on the edges in $E_p$. We are going to prove that, with a proper scaling, $\ystar$ can actually become a fractional matching. This means $E_p$ contains a 
	fractional matching with expected size $\approx \card{y}$ which will be sufficient to prove~\Cref{eq:last-step-of-the-proof}. 
	
	Fix a vertex $v \in V$ and define,
\begin{align*}
	&\wfix_v := \sum_{e \ni v} y_e \cdot \II(y_e \geq 1/(200 \log n)) \quad \text{and} \quad 
	\wstar_v  := \sum_{e \ni v} \ystar_e \cdot \II(y_e < 1/(200\log n)).
\end{align*}
It is easy to see that $\wfix_v \leq 1$ as $y \in \IR^E$ is a fractional matching. Define $E_p(v)$ as the edges in $E_p$ that are incident on $v$. 
We can also observe that,
\begin{align}\label{eq:inter-dummy4}
\wstar_v = \sum_{e \ni v} \ystar_e \cdot \II(y_e < 1/(200\log n)) &\leq \frac1{200 \log n} \cdot \card{E_p(v)}. 
\end{align}
because, when $y_e < 1/(200\log n)$, we know that $y_e/p_e = 1/(200 \log n)$. We have, 
\[
	\expect{\card{E_p(v)} \mid R} = \sum_{e \ni v} p_e = \sum_{e \ni v} y_e \cdot (200\log{n}) \leq 200\log{n},
\]
again using the fact that $y$ is a fractional matching. Moreover, $\card{E_p(v)}$ is a sum of $\kappa$-wise independent random variables (by the choice of $h$) and thus, by~\Cref{prop:chernoff-hash}, 
\[
	\Pr\paren{\card{E_p(v)} \geq 400\log{n} \mid R} \leq \exp\paren{-\min\paren{\frac{\kappa}{2}, \frac{2}{3} \cdot 200\log{n}}} \leq n^{-100},
\]
by the choice of $\kappa = 200\log{n}$. Plugging in this in~\Cref{eq:inter-dummy4} implies that with high probability, 
\[
	\wstar_v \leq \frac1{200 \log n} \cdot 400\log{n} = 2. 
\]
All in all, we have that for every vertex $v \in V$, with high probability
\[
	\sum_{e \ni v} \ystar_e \leq 3.
\]
A union bound over all the vertices, implies that $\ystar/3$ is with high probability a fractional matching. We thus have, 
\[
	\expect{\card{\frac{\ystar}{3}} \mid \text{$\frac{\ystar}{3}$ is a fractional matching, $R$}} \geq \frac{\card{y}}{3} \cdot (1-1/\poly{(n)}).
\]
To conclude the proof of~\Cref{eq:last-step-of-the-proof}, we have, 
\begin{align*}
	\expect{\mu(E_p) \mid R} &\geq \frac{2}{3} \cdot \Pr\paren{\text{$\frac{\ystar}{3}$ is a fractional matching} \mid R} \cdot \expect{\card{\frac{\ystar}{3}} \mid \text{$\frac{\ystar}{3}$ is a fractional matching}, R} \\
	&\geq \frac{2\card{y}}{9} \cdot (1-1/\poly(n)) > \frac{\card{y}}{5};
\end{align*}
the first inequality holds by the integrality gap of fractional matchings in~\Cref{fact:fractional-to-integral-matching} for the first term, and the trivial lower bound of $\mu(E_p) \geq 0$ when $\ystar/3$ is not a fractional matching; 
the second inequality holds by the fact that $\ystar/3$ is with high probability a fractional matching and the expectation calculated above. 
\end{proof}

\subsubsection{The Final Dynamic Streaming Algorithm: Concluding the Proof of~\Cref{thm:upper}}

We are almost done with the proof of~\Cref{thm:upper}. Up until now, we have already established that for any given graph $G=(V,E)$, 
\begin{itemize}
\item We can run~\Cref{alg:time-stamp} in $O(n\log^3{n})$ space and recover the time stamps of all vertices with high probability (\Cref{lem:find-time});
\item We can then extract the vector $z \in \IR^E$ from the information collected above in an additional $O(n\log{n})$ space and $O(1)$ passes deterministically (\Cref{obs:deg-at-removal} and \Cref{obs:find-z});
\item We can further sample the edges $E_z \subseteq E$ in another $O(n\log^2{n})$ space and $O(1)$ passes with high probability (\Cref{clm:edges-storage});
\item Finally, we find a maximum matching in $E_z$; conditioned on the high probability events above, this gives us a matching $M$ of \emph{expected} size (\Cref{lem:good-matching}):
\[
\Exp\card{M} \geq \paren{\frac{1-8\beta}{1-2\beta}} \cdot \frac{\beta}{10} \cdot \mu(G). 
\]
\end{itemize}

The main missing pieces are to $(i)$ boost the probability of success of this algorithm to a high probability bound instead of an expectation-guarantee, and, subsequently, $(ii)$ boost the approximation ratio of the algorithm 
to a $(1+\eps)$-approximation (and extend it even to weighted graphs). We show that both these tasks can be achieved using the existing sketching and streaming results listed in~\Cref{sec:sketch-toolkit}. 
The proofs in the following are pretty standard. 

\paragraph{Boosting the probability of success.} Let $G=(V,E)$ be the input graph. We can assume without loss of generality that $\mu(G) \geq \log^2\!{(n)}$; the total number of edges in any graph is at most $2n \cdot \mu(G)$ 
and thus we can run a randomized sparse-recovery in~\Cref{prop:sparse-recovery} for recovering $2n\log^2\!{(n)}$ edges in $O(n\log^3{n})$ space with high probability; if $\mu(G) < \log^2\!{(n)}$, this algorithm succeeds in recovering the entire graph
and finding a maximum matching of $G$ exactly. Otherwise, we continue with the main algorithm. 

Let $\eps = 1/2$ and run the vertex-sampling approach of~\Cref{prop:vertex-sampling} for $O(\log{n})$ guesses of $\mu(G)$ as powers of $2$ between $\log^2\!{(n)}$ to $n$ and run our algorithm for each guess separately in parallel. 
We  specify this more details in the following. 

For a given guess $\tilde{\mu}$,
at the beginning of the stream, we partition the vertices randomly into groups $U_1,\ldots,U_t$ for $t = 8\tilde{\mu}/\eps = 16\tilde{\mu}$. This allows us to define a graph $H := H(\tilde{\mu})$ on new vertices $\set{u_1,\ldots,u_t}$ (representing the groups $U_1,\ldots,U_t$) and edges between $u_i$ and $u_j$ iff there exists at least one edge $(x,y) \in E$ with $x \in U_i$ and $y \in U_j$ (note that multiple edges in $E$ can be mapped to a single edge in $H$ but $H$ does \emph{not} have multi-edges). Moreover, 
an update to the dynamic stream defining $E$ can be also directly translated to an update to the graph $H$ and thus we can run our algorithms over $H$.\footnote{We should mention a minor point here. 
There is a slight difference between the dynamic stream defining $E(H)$ versus for $E$. In particular, we should treat the dynamic stream on $E(H)$ as a vector $\phi \in \IN^{{t}\choose{2}}$ where $\phi_{u_i,u_j} > 0$ (but not necessarily $\phi_{u_i,u_j} = 1$) is interpreted as the the existence of the edge $(u_i,u_j) \in E(H)$ (whereas for the dynamic stream defining $E$, the vector $\phi \in \set{0,1}^{{V}\choose{2}}$ will be the characteristic vector of $E$). Nevertheless, our algorithms work the same exact 
way over $H$. This is because the only access 
of these algorithms to the edges of the graph was via sparse-recovery algorithms and these algorithms return the support of the vector $\phi$ (or whichever subsets of it they are being run on).
Moreover, since the entries of $\phi$ are still $\poly{(n)}$ bounded, the space complexity of the algorithm is exactly as before (as a function of $n$ itself, i.e., we do not aim for a better dependence based on $t \leq n$).}

Now consider, the guess $\tilde{\mu}$ such that $\mu(G) \leq \tilde{\mu} \leq 2 \cdot \mu(G)$. Firstly, by~\Cref{prop:vertex-sampling}, we have that with high probability $\mu(H) \geq \mu(G)/2$.
As such, for this choice and by picking $\beta = 1/16$ in our algorithm, for the returned matching $M$, we have that 
\[
	\Exp\card{M} \geq \paren{\frac{1-8\beta}{1-2\beta}} \cdot \frac{\beta}{10} \cdot \mu(H) \geq \frac{1}{40} \cdot \mu(G). 
\]
On the other hand, in $H$, we deterministically have that $\card{M} \leq t/2$ since it is a matching and $H$ has at most $t$ vertices. This means that we also always have 
\[
	\card{M} \leq 16 \mu(G).
\]
This implies that 
\[
	\Pr\paren{\card{M} \geq \frac{1}{50} \cdot \mu(G)} \geq \frac{1}{3200}
\]
as otherwise 
\[
	\Exp\card{M} < \frac{1}{50} \cdot \mu(G) + \frac{1}{3200} \cdot 16\mu(G) < \paren{\frac{1}{50}+\frac{1}{200}} \cdot \mu(G) = \frac{1}{40} \cdot \mu(G), 
\]
a contradiction. Repeating this algorithm now $O(\log{n})$ times and returning the largest matching implies that with high probability, we will find a matching of size at least $\mu(G)/50$. 
Note that given a matching in $H$, we can spend $O(1)$ more passes and $O(n\log^2{n})$ space to find a matching in $G$ of the same size: for any edge $(u_i,u_j)$ in the matching of $H$, we need
to just find a single edge between the vertices of the group $U_i$ and $U_j$ which we know exists. This can be done using a standard sampling trick and sparse-recovery algorithms: basically, we first count the number $c_{i,j}$ of 
edges between $U_i$ and $U_j$ in a single pass and in the second pass, we sample $\simeq \log{(n)}/c_{i,j}$ fraction of vertex-pairs between $U_i$ and $U_j$ and run a sparse-recovery algorithm for $\simeq (\log{n})$-sparse
vectors on this subsampled pairs. With high probability, the number of sampled edges is indeed $\lesssim \log{n}$ and the recovery algorithm finds at least one of those edges. 

In conclusion, by running our main algorithm $O(\log^2{(n)})$ times in parallel for different guesses of $\mu(G)$ and different repetitions for each guess, we 
obtain an $O(1)$-approximation algorithm to the maximum matching of any given graph $G=(V,E)$ in a dynamic stream with high probability\footnote{To simplify the proof, we have been quite cavalier with the choice
of constants in the analysis of the algorithm and only prove a $50$-approximation guarantee. Given that we are going to boost the approximation ratio of this algorithm in a black-box way in the next step, 
any constant-approximation works for us and thus we did not optimize the approximation ratio of this base algorithm. A more careful analysis of the bounds can reduce this ratio dramatically, but we suspect the limit will be 
close to $6$-approximation (or possibly $4$ for bipartite graphs).}. Moreover, this algorithm
uses $O(n\log^5{n})$ space and $O(\log\log{n})$ passes. 

\paragraph{Boosting the approximation ratio.} Now that we have an $O(1)$-approximation algorithm that succeeds with high probability, this step becomes a black-box application of~\Cref{prop:boosting-approximation}. 

\medskip

This concludes the proof of~\Cref{thm:upper}.


\section{The Lower Bound}

We provide our multi-pass lower bound for approximate matchings in the dynamic streaming model in this section, formalizing~\Cref{res:lower}. 

\begin{theorem}\label{thm:main-lb}
Any semi-streaming algorithm that given any $n$-vertex graph in a dynamic streams, outputs an $O(1)$-approximation to the maximum matching problem with any constant probability of success requires $\Omega(\log \log n)$ passes.
\end{theorem}

We prove~\Cref{thm:main-lb} using the connection between the streaming and communication models (\Cref{prop:cc-stream}). Our proof consists of two separate parts proven in the subsequent sections.
\begin{description}
	\item \textbf{Part 1.} In \Cref{subsec:ahm-def}, we define a new two-party communication problem for $r$-round protocols called Augmented Hidden Matrices ($\ahm_r$, see \Cref{def:ahm}) and prove the following lower bound on its communication complexity in \Cref{subsubsec:proof-ahm-lb}.

	\begin{lemma}\label{lem:ahm-lb}
		For any sufficiently large $n_r \in \IN$, positive integer $r = O(\log\log\!{(n_r)}) $, positive number $\alpha \leq 1 - n_r^{-1/2r}$ and constant $\advconst \geq 1$, any $r$-round protocol that successfully solves $\ahm_r(n_r, \alpha)$ with
		probability of success at least \[
		\frac12 \cdot \paren{1 + \frac{r}{20 \cdot \advconst \cdot (r+1)}}
		\] must have a communication cost of at least 
		\[
		s_r = \expo{n_r}{1+\frac1{2^{r}-1}} \cdot \expo{1-\alpha}{\frac{r\cdot2^r}{2^r-1}} \cdot \frac{(\alpha/(1-\alpha))^2}{(\jpyconst \cdot \advconst^2 \cdot 80^2)^r \cdot ((r+1)!)^3} .
		\]
	\end{lemma}

	\item \textbf{Part 2.} In \Cref{subsec:red-to-graph}, we construct graph instances corresponding to the $\ahm_r(n_r, \alpha)$ problem. We prove the following connection between any $p$-pass dynamic streaming algorithm for approximate maximum matching and $\ahm_r$ where $r = 2p -1$ in \Cref{subsec:proof-ahm-to-graph}.
	
	\begin{lemma}\label{lem:redToMatching}
		For any sufficiently large $n_r \in \IN$, integers $p,s \geq 1$, and number $\apx \geq 1$, let $\alg$ by any $p$-pass $s$-space
		dynamic streaming algorithm that computes a $\apx$-approximate maximum matching on a $(4n_r)$-vertex bipartite graphs with probability of success at least $1-1/\poly{(n)}$.
		Then, for $r = 2p-1$ and $\alpha = 1/(4 \cdot \apx \cdot r)$, there exists an $r$-round protocol $\prot$ for $\ahm_{r}(n_r, \alpha)$ with 
		probability of success 
		\[
		\suc{\prot} \geq \frac{1}{2} \cdot \left(1 + \frac{1}{3\apx}\right)
		\]
		and
		communication cost 
		\[
		\cc{\prot} \leq r \cdot s + O(r \cdot n_r \log{(n_r)}). 
		\]
	\end{lemma}	
\end{description}

We now prove~\Cref{thm:main-lb} using the above two lemmas.

\begin{proof}[Proof of \Cref{thm:main-lb}]
	Without loss of generality, we can solely focus on proving the lower bound for semi-streaming algorithms that succeed with high probability instead of constant probability. 
	This is due to the following standard reduction.
	 
	Simply run a constant probability of success streaming algorithm in parallel $O(\log{n})$ times 
	and consider all the outputs; then, spend one more pass to filter out all the ones that output an edge that does not belong to the graph (because they failed), and among the rest, output the largest matching. 
	This way, we still obtain a semi-streaming algorithm with just one more pass (which is negligible in this context) and a high probability of success. As such, in the rest of the proof, we focus on high-probability-of-success algorithms. 
	
	For sufficiently large $n \in \IN$, integers $p,s \geq 1$ and any constant $\apx \geq 1$,
	let $\alg$ be any $p$-pass $s$-space dynamic streaming algorithm that computes a $\apx$-approximate maximum matching on any $(4n)$-vertex graph with probability of success at least $1-1/\poly{(n)}$. 
	
	By \Cref{lem:redToMatching} with $r = 2p-1$ and $\alpha = 1/(4 \cdot \apx \cdot r)$, there exists an $r$-round protocol $\prot$ for $\ahm_r(n, \alpha)$ where 
	\[
	\suc{\prot} \geq \frac{1}{2} \cdot \left(1 + \frac{1}{3\apx}\right) \qquad \text{and} \qquad \cc{\prot} \leq r \cdot s + O(r^2 \cdot n \log n).
	\]
	By \Cref{lem:ahm-lb} with $n_r = n$ and $\advconst = \apx$, we have that $\cc{\prot} \geq s_r$ since 
	\[
	\suc{\prot} \geq \frac12 \cdot \paren{1+\frac1{3\apx}} \geq \frac12 \cdot \paren{1 + \frac{r}{20 \cdot \apx \cdot (r+1)}}.
	\]
	Then, by combining the bounds on $\cc{\prot}$ and re-arranging the inequality, we obtain the following lower bound on the space of $\alg$ using the fact that $\jpyconst$ and $\beta$ are constants (and thus $\alpha = \Theta(1/r)$):
	\begin{align}
		s &\geq \frac{s_r}{r} - O( n \log n) \notag \\ 
		&= {\expo{n}{1+\frac1{2^{r}-1}} \cdot \expo{1-\alpha}{\frac{r\cdot2^r}{2^r-1}} \cdot \frac{(\alpha/(1-\alpha))^2}{(\jpyconst \cdot \apx^2 \cdot 80^2)^r \cdot ((r+1)!)^3}} \cdot \frac1r - O(n \log n) \notag \\
		&\geq {\expo{n}{1+\frac1{2^{r}-1}}} \cdot r^{-\Theta(r)} - O(n \log n). \label{eq:boundspacealg}
	\end{align}
	
	Overall, by considering $\alg$ with $p = o(\log \log n)$ and thus $r = 2p -1 = o(\log \log n)$, we have that the space used by the algorithm is $s \gg n \cdot \polylog (n)$ using \Cref{eq:boundspacealg}.
	Therefore, any dynamic streaming algorithm that uses $n \cdot {\polylog} (n)$ space, namely, a semi-streaming algorithm, to compute an $O(1)$-approximate maximum matching must use $\Omega(\log \log n)$ passes.
\end{proof}

\subsection{Augmented Hidden Matrices}\label{subsec:ahm-def}

For any $r \geq 1$, the $\ahm_r(n_r, \alpha)$ problem is defined recursively using $b_r^2$ many instances of the $\ahm_{r-1}(n_{r-1}, \alpha)$ problem where $n_r, b_r, n_{r-1} \in \mathbb{N}$ and $\alpha \in (0,1)$.
The output of any instance is a single bit, pointed to by a \emph{search sequence} (to be defined soon).

\begin{Definition}[Augmented Hidden Matrices]\label{def:ahm}
The Augmented Hidden Matrices problem, denoted by $\ahm_r(n_r, \alpha)$, is defined as follows (see \Cref{fig:insdel_harddist} for an illustration):
\begin{itemize}
	\item
	\textbf{For $\mathbf{r = 0}$.} Alice does not receive any input, i.e., $\alicepart{0} = \emptyset$, and Bob receives a single 
	bit $\bobpart{0} \in \{0,1\}$. There is no communication and Alice has to output $\bobpart{0}$. 
	\item
	\textbf{For $\mathbf{r \geq 1}$.} Start with $b_r^2$ instances of the players' inputs in $\ahm_{r-1}(n_{r-1}, \alpha)$, denoted by $(\alicepart{r-1}_{i,j}, \bobpart{r-1}_{i,j})$ for $i, j \in [b_r]$.
	Alice receives $\alicepart{r} = \Xx$ where $\Xx$ is a $b_r \times b_r$ matrix such that $\Xx[i,j] = \bobpart{r-1}_{i,j}$ for $i, j \in [b_r]$. 
	Bob receives the input $\bobpart{r} = (\sigmar, \sigmac, \Yy)$ where $\sigmar, \sigmac \in S_{b_r}$ are permutations of $[b_r]$ and $\Yy$ is a $b_r \times b_r$ matrix such that, for $i, j \in [b_r]$,
	 \begin{align*}
		\Yy[i,j] =
		\begin{cases}
			\left(\Aa{r-1}_{i,j}, \Bb{r-1}_{i,j}\right) & \text{if } b_r \cdot \alpha <\sigmar(i) \neq \sigmac(j) \leq b_r, \\
			\Aa{r-1}_{i,j} & \text{if } b_r \cdot \alpha < \sigmar(i) = \sigmac(j) \leq b_r, \\
			\emptyset & \text{otherwise.}
		\end{cases}
	\end{align*}
	Alice starts the communication. After all the messages are sent, 
	the player receiving the last message is given a \textbf{search sequence}, which is a tuple of $r$ integers, $(\speckstar{r}, \speckstar{r-1}, \ldots, \speckstar{1})$ with $\speckstar{i} \in [b_i \cdot (1-\alpha)]$. This player must output the solution to $\ahm_{r-1}$ instance $(\alicepart{r-1}_{i,j}, \bobpart{r-1}_{i,j})$ where $\sigmar(i) = \sigmac(j) = b_r \cdot \alpha + \speckstar{r}$ using the search sequence $(\speckstar{r-1}, \speckstar{r-2} \ldots \speckstar{1})$ .

\end{itemize}
\end{Definition}
We expound more on the defintion of $\ahm_r$ problem, and give some useful terminology next. 

For $r = 0$ case, we set $n_0 = 1$, and $\ahm_0(1, \alpha)$ is also denoted by  $\ahm_0(1)$. This is a base case that is trivially hard for $0$-round protocols. 

For each of the $b_r^2$ many $\ahm_{r-1}$ instances $(\alicepart{r-1},\bobpart{r-1})$ used in an instance of $\ahm_r$, the roles of Alice and Bob are swapped, i.e., Alice holds $\bobpart{r-1}$ and Bob holds $\alicepart{r-1}$.
Out of these instances, we call the set of the $\bm{k_r} := b_r \cdot (1- \alpha)$ many $\ahm_{r-1}$ instances $(\alicepart{r-1}_{i,j}, \bobpart{r-1}_{i,j})$ where $b_r \cdot \alpha < \sigmar(i) = \sigmac(j) \leq b_r$ as the \textbf{special sub-instances} in an instance of $\ahm_r$.
Since each special sub-instance identifies another set of $k_{r-1}$ many $\ahm_{r-2}$ special (sub-)sub-instances, recursively applying this for all special sub-instances of $\ahm_r$ identifies a set of $k_r \cdot k_{r-1} \cdot \ldots \cdot k_1$ many $\ahm_0$ \textbf{special base instances}, where each one corresponds to a single bit $\bobpart{0}$. 

It is useful to point to the number of special base instances explicitly for later. 
\begin{observation}\label{obs:numspecial}
	In any instance of $\ahm_r(n_r, \alpha)$, there are $k_r \cdot k_{r-1} \cdot \ldots \cdot k_1$ many special base instances where $k_i = b_i \cdot (1- \alpha)$ for each $i \in [r]$.
\end{observation}

Furthermore, we call the set of $k_r^2 - k_r$ many $\ahm_{r-1}$ instances $(\alicepart{r-1}_{i,j}, \bobpart{r-1}_{i,j})$ where $b_r \cdot \alpha < \sigmar(i) \neq \sigmac(j) \leq b_r$ as the \textbf{off-diagonal sub-instances} in an instance of $\ahm_r$. We define some more notation to point to the sub-instances. 

\paragraph*{Notation.}
In any instance of $\ahm_r$, denoted by $(\Aa{r}, \Bb{r})$, we use the following notation.
We denote each of the $k_r$ many special sub-instances as $(\Aspec_k, \Bspec_k) = (\alicepart{r-1}_{i,j}, \bobpart{r-1}_{i,j})$ where $\sigmar(i) = \sigmac(j) = b_r \cdot \alpha + k$ and $i,j \in [b_r]$ for $k \in [k_r]$.
We collectively denote these special sub-instances as $(\Aspec,\Bspec)$ where 
\[
\Aspec = (\Aspec_1, \Aspec_2, \ldots, \Aspec_{k_r}) \quad \text{and} \quad \Bspec = (\Bspec_1, \Bspec_2, \ldots, \Bspec_{k_r}).
\] 
We denote the collection of the $k_r^2 - k_r$ many off-diagonal sub-instances as $(\Acommon, \Bcommon)$ where 
\[
\Acommon = (\alicepart{r-1}_{i,j} : b_r \cdot \alpha < \sigmar(i) \neq \sigmac(j) \leq b_r)\quad \text{and} \quad \Bcommon = (\bobpart{r-1}_{i,j} : b_r \cdot \alpha < \sigmar(i) \neq \sigmac(j) \leq b_r).
\]
We let $\Brest$ be the remaining $b_r^2 - k_r^2$ sub-instances  $\bobpart{r-1}_{i,j}$ where $\sigmar(i) \leq b_r \cdot \alpha$ or $\sigmac(j) \leq b_r \cdot \alpha$.

\paragraph{Search Sequences.}
After all the messages are sent, the search sequence given to the player who receives the last message is always \textbf{uniformly random} and \textbf{independent of the players' inputs}. That is, a search sequence $(\speckstar{r}, \speckstar{r-1}, \ldots, \speckstar{1})$ is chosen where for $i \in [r]$, $\speckstar{i} $ is chosen uniformly at random and independently from $[b_i \cdot (1-\alpha)]$. 
The solution to any instance of $\ahm_r$ given $(\speckstar{r}, \speckstar{r-1}, \dots, \speckstar{1})$ is the same as the solution to the uniformly chosen special sub-instance $(\Aspec_{\speckstar{r}}, \Bspec_{\speckstar{r}})$ on the search sequence $(\speckstar{r-1}, \dots, \speckstar{1})$.
Continuing this until the last index $\speckstar{1}$ in the search sequence ultimately identifies a uniform random special base instance, i.e., an instance of $\ahm_0$, whose solution is the solution to the $\ahm_r$ instance.

Our definition of the $\ahm_r(n_r, \alpha)$ communication problem is non-standard as the uniformly random search sequence, which is a part of the input, is given to the final player \emph{at the very end of a protocol}.
In particular, the search sequence is crucial to defining the output of the problem, but does not form part of the input held by the players throughout the protocol.

Although the search sequence defines the communication problem in a non-standard way, the easy direction of Yao's minmax theorem is still applicable.
Thus, 
we prove \Cref{lem:ahm-lb} by considering deterministic protocols for $\ahm_r(n_r, \alpha)$ where the players' inputs $(\Aa{r}, \Bb{r})$ are sampled from the hard distribution (which we define shortly). 

\begin{Remark}
	The definition of $\ahm_r(n_r, \alpha)$ (\Cref{def:ahm}) always includes the uniform random search sequence, which is independent of the players' inputs and how they are distributed.
	However, when referring to the $\ahm_r(n_r, \alpha)$ problem, we do not explicitly mention the given search sequence if this is clear from context.
\end{Remark}


\bigskip

\begin{figure}[h!]
	\begin{center}
	\begin{tikzpicture}[scale=1.1, yscale=-0.75]
		\foreach \i in {1, 2, ..., 7}{
			\foreach \j in {1, 2, ..., 7}{
				\ifthenelse{\i < 3 \OR \j < 3 }{
					\draw[draw=black] (\i -1, \j -1) rectangle (\i,\j);
					\node[font=\small] () at (\i - 0.5, \j - 0.5) {$B$};
				}{
					\ifthenelse{\NOT \i = \j}{
						\draw[draw=black, fill=none] (\i -1, \j -1) rectangle (\i,\j);
						\node[font=\small] () at (\i - 0.5, \j - 0.5) {$A, \underline{\bm{B}}$};
					}{
					\ifthenelse{\NOT \i = 4} {
						\draw[draw=black, fill=gray!30] (\i -1, \j -1) rectangle (\i,\j);
						\node[font=\small] () at (\i - 0.5, \j - 0.5) {$A, B$};
					}
				{
					\draw[draw=black, fill=green!80!black!70] (\i -1, \j -1) rectangle (\i,\j);
					\node[font=\small] () at (\i - 0.5, \j - 0.5) {$A, B$};
				}
					}
				}
			}
		}
		\draw[ultra thick] (2,2) rectangle (7,7);
		\draw[decorate, decoration={brace}] (7.25,0.05) -- (7.25,1.95);
		\node[font=\small, anchor=west] () at (7.5,1) {$b_r \cdot \alpha$}; 
		\node[font=\small, anchor=east, white] () at (-0.5,1) {$b_r/\alpha$};  
		
		\draw[decorate, decoration={brace}] (7.25,2.05) -- (7.25,6.95);
		\node[font=\small, anchor=west] () at (7.5,4.5) {$b_r \cdot (1 - \alpha)$}; 
		
		\draw[decorate, decoration={brace}] (1.95,7.25) -- (0.05,7.25);
		\node[font=\small] () at (1,7.75) {$b_r \cdot \alpha$};
		\draw[decorate, decoration={brace}] (6.95,7.25) -- (2.05,7.25);
		\node[font=\small] () at (4.5,7.75) {$b_r \cdot (1 - \alpha)$};
		
		\foreach \i in {1,2}{
			\node[font=\small] () at (-0.5,\i -0.5) {$\sigmar(\i)$};
			\node[font=\small] () at (\i -0.5, -0.5) {$\sigmac(\i)$};
		}
		\node[font=\small] () at (-0.5,7 -0.5) {$\sigmar(b_r)$};
		\node[font=\small] () at (7 -0.5, -0.5) {$\sigmac(b_r)$};
		
		\node[font=\small] () at (-0.5,4) {$\vdots$};
		\node[font=\small] () at (4, -0.5) {$\dots$};
		
		\foreach \i in {3,4, ..., 7}{
			\draw[ultra thick, draw=black, fill=none] (\i -1, \i -1) rectangle (\i,\i);
		}
	\end{tikzpicture}   
\end{center}
\caption{
	An illustration of $\ahm_r(n_r,\alpha)$ where $b_r = 7$, $\speckstar{r} = 2$, and $\alpha = 2/7$, i.e., $b_r \cdot \alpha = 2$. 
	In this figure, the $b_r \times b_r$ matrix is presented after applying the respective permutations of rows and columns using $\sigmar$ and $\sigmac$, which are held by Bob (alternatively, think of $\sigmar$ and $\sigmac$ as identity permutations here). 
	In each position of the matrix, the corresponding input is either a complete instance of $\ahm_{r-1}(n_{r-1}, \alpha)$ denoted by $(A,B) = (\Aa{r-1}, \Bb{r-1})$ or a partial instance of $\ahm_{r-1}(n_{r-1}, \alpha)$ denoted by $B = \Bb{r-1}$ due to the asymmetry of the construction. 
	All of the $B$'s are held by Alice, all of the $A$'s are held by Bob, and Bob further holds the $B$'s corresponding to $\Bcommon$, which are bolded and underlined.
	The shaded instances are the special sub-instances $(\Aspec,\Bspec)$, i.e., the ones in the diagonal of the bolded box.
	The remaining $B$'s belong to $\Brest$, i.e., the ones not in the bolded box.
	Given the search sequence $(\speckstar{r}, \speckstar{r-1}, \ldots, \speckstar{1})$, the goal is to solve the the instance identified by ${\speckstar{r}} = 2$ (shaded in green, and denoted by $(\Aspec_2, \Bspec_2)$), using the search sequence $(\speckstar{r-1}, \ldots, \speckstar{1})$. 
}\label{fig:insdel_harddist}
\end{figure}

Next, we fix the parameters of the problem, and then define our hard distribution $\DD_r(n_r, \alpha)$. 

\paragraph{Parameter Choices.}
For all $r \geq 1$, it is sufficient for us to consider the 
$\ahm_r$ problem for parameters that satisfy the following conditions:
\begin{equation}\label{eq:params-nb}
	\begin{split}
	n_r &= n_{r-1} \cdot b_r \\
	b_r &= \expo{n_r}{\frac{2^{r-1}}{2^r-1}} \cdot \expo{1-\alpha}{\frac{r\cdot 2^{r-1}}{2^r-1}-1} \\
	n_{r-1} &= \expo{n_r}{\frac{2^{r-1}-1}{2^r-1}} / \expo{1-\alpha}{\frac{r\cdot2^{r-1}}{2^r-1} - 1}.
	\end{split}
\end{equation}
From now on, we only consider $\ahm_r(n_r, \alpha)$ with these parameter choices.

\begin{ourbox}
	\vspace{0.5em}
	\textbf{Hard Distribution $\DD_r(n_r, \alpha)$ for $\ahm_r(n_r, \alpha)$}:
	\begin{itemize}
		\item \textbf{For $\mathbf{r=0}$.} 
		Sample $\bobpart{0} \in \{0,1\}$ uniformly and independently. We also use $\DD_0(1)$ to denote this distribution.
		\item \textbf{For $\mathbf{r \geq 1}$.} 
		Sample the $b_r^2$ instances $(\alicepart{r-1}_{i,j}, \bobpart{r-1}_{i,j})$ from $\DD_{r-1}(n_{r-1}, \alpha)$ independently for each $i, j \in [b_r]$ where $b_r, n_{r-1}$ are as defined in \Cref{eq:params-nb}. Sample the permutations $\sigmar, \sigmac \in S_{b_r}$ uniformly and independently.
	\end{itemize}
\end{ourbox}

Recall that the distribution of the search sequence was already defined to be uniformly random, and independent of all other inputs. 

We have the following obvious observation. 

\begin{observation}\label{obs:round-0-trivial}
    Any $0$-round protocol (wherein Alice outputs the answer) for solving $\ahm_0(1)$ succeeds with probability at most $1/2$ when the input is sampled from $\DD_0(1)$.
\end{observation}

	

The proof of \Cref{lem:ahm-lb} relies on a round-elimination argument that shows that if there is a good deterministic protocol for $\ahm_r(n_r, \alpha)$ under the input distribution $\DD_r(n_r, \alpha)$, then there must be a good deterministic protocol for $\ahm_{r-1}(n_{r-1}, \alpha)$ under the input distribution $\DD_{r-1}(n_{r-1}, \alpha)$. 
We prove the following key lemma in \Cref{subsec:analysis-ahm-lb}.

\begin{lemma} \label{lem:round-elim}
	For any $r \geq 1$, sufficiently large $n_r \in \IN$, $\alpha,\delta \in (0,1)$, and integer $s \geq 1$,
	suppose there exists a deterministic $r$-round protocol $\prot_r$ for $\ahm_r(n_r, \alpha)$ with
	\[
	\cc{\prot_r} \leq s
	\]
	and probability of success at least $\delta$ when the input is sampled from $\DD_r(n_r, \alpha)$.  
	Then, for any $\eps \in (0,1)$, there exists a deterministic $(r-1)$-round protocol $\prot_{r-1}$ for $\ahm_{r-1}$ with 
	\[
	\cc{\prot_{r-1}} \leq\jpyconst \cdot (r/\eps) \cdot \paren{ \frac{s}{k_r} + r},
	\]
	 and probability of success 
	 \[
	 \suc{\prot_{r-1}} \geq \delta - \eps -\sqrt{\frac{s}{2 \cdot b_r^2 \cdot \alpha^2}}
	 \] 
	 when the players' input is sampled from $\DD_{r-1}(n_{r-1}, \alpha)$.
\end{lemma}

With that, we have the main tools required to prove of our lower bound for $\ahm_r(n_r, \alpha)$.
\subsubsection{The Lower Bound for Augmented Hidden Matrices: Proof of \Cref{lem:ahm-lb}}\label{subsubsec:proof-ahm-lb}
\begin{proof}[Proof of \Cref{lem:ahm-lb}]
	We prove this result by first assuming the existence of a deterministic $r$-round protocol $\prot_r$ for $\ahm_r$ that is `too good to be true' (w.r.t.\ the statement of the lemma) on inputs sampled from $\DD_r(n_r, \alpha)$. Then, by iteratively applying the round-elimination argument in \Cref{lem:round-elim}, we ultimately obtain a deterministic $0$-round protocol $\prot_0$ for $\ahm_0$ that is trivially `too good to be true' on inputs sampled from $\DD_0(1)$ (given~\Cref{obs:round-0-trivial}), a contradiction.
	Formally, we prove this by induction on $r$.
	
\textbf{Base case for $\mathbf{r=1}$.}
Suppose that there exists a deterministic $1$-round protocol $\prot_1$ for $\ahm_1(n_1, \alpha)$ with communication cost 
\[\cc{\prot_1} < s_1 = \frac{n_1^2 \cdot  \alpha^2}{\jpyconst \cdot \advconst^2 \cdot 80^2 \cdot 8}\]
and probability of success
\[\suc{\prot_1} \geq \frac{1}{2} \cdot \left(1 + \frac1{40 \cdot \advconst}\right) \]
when the input is sampled from $\DD_1(n_1, \alpha)$.
Then, using \Cref{lem:round-elim} with $\varepsilon = 1/(160 \cdot \advconst)$, we obtain a deterministic $0$-round protocol $\prot_0$ for $\ahm_0(1)$ which has no communication cost (since no messages are communicated) and succeeds with probability
\begin{align*}
	\suc{\prot_0} &\geq \frac{1}{2} \cdot \left(1 + \frac{1}{40 \cdot \advconst}\right) - \frac{1}{160 \cdot \advconst} - \sqrt{\frac{n_1^2 \cdot \alpha^2}{\jpyconst \cdot \advconst^2 \cdot 80^2 \cdot 16 \cdot b_1^2 \cdot  \alpha^2} } \\
	&\geq \frac{1}{2} \cdot \left(1 + \frac{1}{80 \cdot \advconst} - \sqrt{\frac{n_1^2}{\advconst^2 \cdot 160^2 \cdot b_1^2}}\right) \tag{since $\jpyconst \geq 1$} \\
	& = \frac{1}{2} \cdot \left(1 + \frac{1}{160 \cdot \advconst} \right) \tag{since $b_1 = n_1$ by \Cref{eq:params-nb}}
\end{align*}
when the input is sampled from $\DD_0(1)$.
This, however, contradicts \Cref{obs:round-0-trivial} and thus proves the result for the base case when $r = 1$.

\textbf{Inductive step for $\mathbf{r \geq 2}$.} Having proven the base case, we now prove the result for any $r \geq 2$ given the result holds for $r-1$, i.e., the inductive hypothesis. Suppose that there exists a deterministic protocol $\prot_r$ for $\ahm_r(n_r, \alpha) $ with communication cost 
\begin{equation*}
\cc{\prot_r} < s_r = \expo{n_r}{1+\frac1{2^{r}-1}} \cdot \expo{1-\alpha}{\frac{r\cdot2^r}{2^r-1}} \cdot \frac{(\alpha/(1-\alpha))^2}{(\jpyconst \cdot \advconst^2 \cdot 80^2)^r \cdot ((r+1)!)^3} 
\end{equation*}
and probability of success
\[
	\suc{\prot_r} \geq \frac12 \cdot \left(1+\frac{r}{20 \cdot \advconst \cdot (r+1)}\right)
\]
when the input is sampled from $\DD_r(n_r, \alpha)$.
Then, using \Cref{lem:round-elim} with
$\varepsilon = 1/{(40 \cdot \advconst \cdot (r+1)^2)}$,
we obtain a deterministic $(r-1)$-round protocol $\prot_{r-1}$ for $\ahm_{r-1}(n_{r-1}, \alpha)$ with cost
\begin{align}
	\cc{\prot_{r-1}} < \jpyconst \cdot (40 \cdot \advconst \cdot (r+1)^2 \cdot r) \cdot \left( s_r/k_r + r \right) \label{eq:ccprot_r-1}
\end{align}
and probability of success
\begin{align}
	\suc{\prot_{r-1}} \geq \frac12 \cdot \left(1+\frac{r}{20 \cdot \advconst \cdot (r+1)}\right) - \frac{1}{40 \cdot \advconst \cdot (r+1)^2} - \sqrt{\frac{s_r}{2 \cdot b_r^2 \cdot \alpha^2}}
	\label{eq:sucprot_r-1}
\end{align}
when the input is sampled from $\DD_{r-1}(n_{r-1}, \alpha)$.

We now argue that protocol $\prot_{r-1}$ contradicts the result for $r-1$, namely, the inductive hypothesis.
In particular, we show that $\cc{\prot_{r-1}} < s_{r-1}$ and 
\[
\suc{\prot_{r-1}} \geq \frac12 \cdot \paren{1 + \frac{(r-1)}{20 \cdot \advconst \cdot r}}.
\]
The rest of the proof is a careful but rather tedious calculation of the values we obtain using the above reduction for the communication cost and probability of success of the protocol. 

To bound the communication cost, we first reiterate the relation between $n_r$ and $n_{r-1}$ from \Cref{eq:params-nb} as follows:
\begin{align}
	\expo{n_{r}}{\frac{1}{2^{r}-1}} &= \expo{n_{r-1}}{\frac{1}{2^{r-1}-1}} \cdot  \expo{1-\alpha}{\paren{\frac{r\cdot2^{r-1}}{2^r-1} - 1} \cdot \frac1{2^{r-1}-1}} \notag \\
	&= \expo{n_{r-1}}{\frac{1}{2^{r-1} -1}} \cdot  \expo{1-\alpha}{\frac{(r-2)\cdot2^{r-1}+1}{(2^r-1) \cdot (2^{r-1}-1)}}. \label{eq:n-nr-1-inter-1}
\end{align}
Next, we use this to simplify the expression for the term $s_r/k_r$ in \Cref{eq:ccprot_r-1}.
\begin{align*}
	\frac{s_r}{k_r} &= \underbrace{\expo{n_r}{1+\frac1{2^{r}-1}} \cdot \expo{1-\alpha}{\frac{r\cdot2^r}{2^r-1}} \cdot \frac{(\alpha/(1-\alpha))^2}{(\jpyconst \cdot \advconst^2 \cdot 80^2)^r \cdot ((r+1)!)^3}}_{s_r} \cdot \frac1{\underbrace{b_r \cdot (1-\alpha)}_{1/k_r}} \\
	&= \frac{n_r}{b_r} \cdot \expo{n_r}{\frac1{2^{r}-1}} \cdot \expo{1-\alpha}{\frac{r \cdot 2^r}{2^r - 1} - 1} \cdot \frac{(\alpha/(1- \alpha))^2 }{(\jpyconst \cdot \advconst^2 \cdot 80^2)^r \cdot ((r+1)!)^3} \\
	&= n_{r-1} \cdot \expo{n_r}{\frac{1}{2^{r}-1}} \cdot \expo{1-\alpha}{\frac{r\cdot2^r}{2^r-1}-1} \tag{as $n_r = n_{r-1} \cdot b_r$ by \Cref{eq:params-nb}} \cdot \frac{(\alpha/(1-\alpha))^2}{(\jpyconst \cdot \advconst^2 \cdot 80^2)^r \cdot ((r+1)!)^3} \\
	&= \expo{n_{r-1}}{1 + \frac{1}{2^{r-1}-1}} \cdot \expo{1-\alpha}{x} \cdot   \frac{(\alpha/(1-\alpha))^2}{(\jpyconst \cdot \advconst^2 \cdot 80^2)^r \cdot ((r+1)!)^3}. \tag{by \Cref{eq:n-nr-1-inter-1}}
	\end{align*}
	where the exponent of $(1-\alpha)$ in the above equation is
\begin{align*}
	x &= \frac{(r-2)\cdot2^{r-1}+1}{(2^r-1) \cdot (2^{r-1}-1)} + \frac{r\cdot2^r}{2^r-1} - 1 \\
	&=\frac{(r-2)\cdot(2^{r-1}-1)+1 + r-2}{(2^r-1) \cdot (2^{r-1}-1)} + \frac{(r-1)\cdot2^r+1}{2^r-1} \\
	&=\frac{1}{2^{r}-1} \cdot \paren{(r-1)2^{r} + 1 + (r-2) + \frac{r-1}{2^{r-1}-1}} \\
	&= \frac{(r-1)}{2^r-1} \cdot (2^r + 1 + \frac1{2^{r-1}-1}) \\
	&= \frac{(r-1)}{2^r-1} \cdot (\frac{2^{2r-1}-2^r + 2^{r-1}-1 + 1}{2^{r-1}-1}) \\
	&= \frac{(r-1)}{2^r-1} \cdot (\frac{2^{2r-1}-2^{r-1}}{2^{r-1}-1}) \\
	&= \frac{(r-1) \cdot 2^{r-1}}{2^{r-1}-1} \ .
\end{align*}
Therefore, we have that
\begin{align*}
	\frac{s_r}{k_r} &=  \expo{n_{r-1}}{1 + \frac{1}{2^{r-1}-1}} \cdot \expo{1-\alpha}{\frac{(r-1) \cdot 2^{r-1}}{2^{r-1}-1}} \cdot   \frac{(\alpha/(1-\alpha))^2}{(\jpyconst \cdot \advconst^2 \cdot 80^2)^r \cdot ((r+1)!)^3} \\
	&= s_{r-1} \cdot \frac1{\jpyconst \cdot \advconst^2 \cdot 80^2 \cdot (r+1)^3}.
\end{align*}
By plugging this back into \Cref{eq:ccprot_r-1}, we obtain
\begin{align*}
\cc{\prot_{r-1}} 
&< \jpyconst \cdot (40 \cdot \advconst \cdot (r+1)^2 \cdot r)  \cdot \paren{s_{r-1} \cdot \frac1{\jpyconst \cdot \advconst^2 \cdot 80^2 \cdot (r+1)^3} + r} \\
&\leq \frac1{\advconst^2 \cdot 160} \cdot s_{r-1} + 40 \cdot \jpyconst \cdot \advconst \cdot (r+1)^4 \tag{using $r \leq r +1$} \\
&\leq s_{r-1} \ , \tag{as $r^4 \ll s_{r-1}$ for $r = O(\log \log n)$ and $\alpha \leq 1 - n_r^{-1/2r}$}
\end{align*}
which is our desired bound on the communication cost of $\prot_{r-1}$.

To bound the probability of success, we first reiterate the following useful relation from \Cref{eq:params-nb}:
\begin{align}
	b_r^2 &= \expo{n_r}{\frac{2 \cdot 2^{r-1}}{2^r-1}} \cdot \expo{1-\alpha}{\frac{r\cdot 2 \cdot 2^{r-1}}{2^r-1}-2} \notag \\
	&= \expo{n_r}{1 + \frac{1}{2^r -1}} \cdot \expo{1-\alpha}{\frac{r\cdot 2^{r}}{2^r-1}} / (1- \alpha)^2 \ . \label{eq:params-br2}
\end{align}
Using this, we now simplify the term $s_r/(2\cdot b_r^2 \cdot \alpha^2)$ in \Cref{eq:sucprot_r-1}.
\begin{align*}
	\frac{s_r}{2 \cdot b_r^2 \cdot \alpha^2} &= \underbrace{\expo{n_r}{1+\frac1{2^{r}-1}} \cdot \expo{1-\alpha}{\frac{r\cdot2^r}{2^r-1}} \cdot \frac{(\alpha/(1-\alpha))^2}{(\jpyconst \cdot \advconst^2 \cdot 80^2)^r \cdot ((r+1)!)^3}}_{s_r} \cdot \frac{1}{2 \cdot b_r^2 \cdot \alpha^2} \\
	&= \frac{b_r^2 \cdot \alpha^2}{(\jpyconst \cdot \advconst^2 \cdot 80^2)^r \cdot ((r+1)!)^3} \cdot \frac{1}{2 \cdot b_r^2 \cdot \alpha^2} \tag{by \Cref{eq:params-br2}} \\
	&= \frac{1}{2 \cdot (\jpyconst \cdot \advconst^2 \cdot 80^2)^r \cdot ((r+1)!)^3} \  .
\end{align*}
By plugging this back into \Cref{eq:sucprot_r-1}, we have that
\begin{align*}
	\suc{\prot_{r-1}} 
	&\geq \frac12 \cdot \left(1+\frac{r}{20 \cdot \advconst \cdot (r+1)}\right) - \frac{1}{40 \cdot \advconst \cdot (r+1)^2} - \sqrt{\frac{1}{2 \cdot (\jpyconst \cdot \advconst^2 \cdot 80^2)^r \cdot ((r+1)!)^3}} \\
	&\geq \frac12 \cdot \left(1+\frac{r}{20 \cdot \advconst \cdot (r+1)}\right) - \frac{1}{40 \cdot \advconst \cdot (r+1)^2} - \sqrt{\frac{1}{\advconst^2 \cdot 80^2 \cdot (r+1)^4}} \tag{since $\jpyconst \geq 1$ and $2 \cdot ((r + 1)!)^3 \geq (r + 1)^4$ for $r \geq 1$} \\
	&= \frac12 \cdot \left(1+\frac{r}{20 \cdot \advconst \cdot (r+1)}\right) - \frac{1}{\advconst \cdot 80 \cdot (r+1)^2} \\
	&= \frac12 \cdot \left(1+\frac{r}{20 \cdot \advconst \cdot (r+1)} - \frac{1}{40 \cdot \advconst \cdot (r+1)^2}\right) \\
	&\geq \frac12 \cdot \left(1+\frac{r-1}{20 \cdot \advconst \cdot r}\right) ,
\end{align*}
which is our desired bound.

Overall, we have that the deterministic protocol $\prot_{r-1}$ constructed from $\prot_r$ using \Cref{lem:round-elim} contradicts the lemma for $r-1$ (inductive hypothesis) and thus the lemma must hold for deterministic $r$-round protocols when the input is sampled from the distribution $\DD_{r}$ for every $r \geq 1$.
Finally, by the easy direction of Yao's minmax principle, the same result holds for any (even randomized) $r$-round protocol for any instance of $\ahm_r$ (given a uniform random search sequence).
\end{proof}


\subsubsection{The Round Elimination Argument: Proof of~\Cref{lem:round-elim}}\label{subsec:analysis-ahm-lb}
We now get to the main part of the argument, which is the proof of~\Cref{lem:round-elim}, restated below. 

\begin{restate}[\Cref{lem:round-elim}]
	For any $r \geq 1$, sufficiently large $n_r \in \IN$, $\alpha,\delta \in (0,1)$, and integer $s \geq 1$,
	suppose there exists a deterministic $r$-round protocol $\prot_r$ for $\ahm_r(n_r, \alpha)$ with
	\[
	\cc{\prot_r} \leq s
	\]
	and probability of success at least $\delta$ when the input is sampled from $\DD_r(n_r, \alpha)$.  
	Then, for any $\eps \in (0,1)$, there exists a deterministic $(r-1)$-round protocol $\prot_{r-1}$ for $\ahm_{r-1}$ with 
	\[
	\cc{\prot_{r-1}} \leq\jpyconst \cdot (r/\eps) \cdot \paren{ \frac{s}{k_r} + r},
	\]
	 and probability of success 
	 \[
	 \suc{\prot_{r-1}} \geq \delta - \eps -\sqrt{\frac{s}{2 \cdot b_r^2 \cdot \alpha^2}}
	 \] 
	 when the players' input is sampled from $\DD_{r-1}(n_{r-1}, \alpha)$.
\end{restate}

 Starting with a deterministic $r$-round protocol $\prot_r$ for $\ahm_r(n_r, \alpha)$ where $\cc{\prot_r} \leq s$ and $\suc{\prot_r} \geq \delta$ under $\DD_r$ as in the statement of \Cref{lem:round-elim}, we construct a deterministic $(r-1)$-round protocol $\prot_{r-1}$ for $\ahm_{r-1}(n_{r-1}, \alpha)$ that proves the lemma in three steps:
\begin{enumerate}[label=$(\roman*)$]
	\item \label{step:i}
	 \textbf{(Input Embedding)} Construct a randomized $r$-round protocol $\protnum{1}_r$ under $\DD_{r-1}$ with success probability the same as $\prot_r$ under $\DD_{r}$, i.e., $\suc{\protnum{1}_r} \geq \delta$, and internal information cost less than the communication cost of $\prot_r$ by a $k_r$ multiplicative factor, i.e., 
	\[
	\ic{\protnum{1}_r}{\DD_{r-1}} \leq \frac{s}{k_r}.
	\]
	\item\label{step:ii} \textbf{(Message Compression)} Construct a randomized $r$-round protocol $\protnum{2}_r$ under $\DD_{r-1}$ with success probability less than that of $\protnum{1}_r$ under $\DD_{r-1}$ by at most an additive $\varepsilon$ factor, i.e., $\suc{\protnum{2}_r} \geq \delta - \varepsilon$, and with communication cost similar to the internal information cost of $\protnum{1}_r$ under $\DD_{r-1}$, i.e., 
	\[
	\cc{\protnum{2}_r} \leq \jpyconst \cdot \frac{r}{\varepsilon} \cdot \paren{\ic{\protnum{1}_r}{\DD_{r-1}} + r} \leq  \jpyconst \cdot \frac{r}{\varepsilon} \cdot \paren{\frac{s}{k_r} + r}. 
	\]
	\item\label{step:iii} \textbf{(Guessing the First Message)} Construct a deterministic $(r-1)$-round protocol $\prot_{r-1}$ under $\DD_{r-1}$ with success probability less than that of $\protnum{2}_r$ under $\DD_{r-1}$ by at most an additive $\simeq \sqrt{s/(b_r^2 \cdot \alpha^2)}$ factor, i.e, $\suc{\prot_{r-1}} \geq \delta - \varepsilon - \sqrt{s/(2 \cdot b_r^2 \cdot \alpha^2)}$, and with communication cost no more than that of $\protnum{2}_r$, i.e., 
	\[
	\cc{\prot_{r-1}} \leq \cc{\protnum{2}_r} \leq \jpyconst \cdot \frac{r}{\varepsilon} \cdot \paren{\frac{s}{k_r} + r}.
	\] 
\end{enumerate}

\subsubsection*{Step \ref{step:i}: A Low Information Cost Protocol via Input Embedding}

We construct a low information cost $r$-round protocol $\protnum{1}_r$ for $\ahm_{r-1}$ by considering an input $(\alicepart{r-1}, \bobpart{r-1}) \sim \DD_{r-1}(n_{r-1}, \alpha)$ and  embedding it in a simulated input instance $(\alicepart{r}, \bobpart{r}) \sim \DD_{r}(n_r, \alpha)$ for $\ahm_r$ such that its output can be used to solve $(\alicepart{r-1}, \bobpart{r-1})$.
Then, the players simulate a run of the protocol $\prot_r$ in $r$ rounds using the simulated input. 

To simplify the exposition, we disambiguate the players in the construction of $\protnum{1}_r$ (and $\prot_{r-1}$ later in step \ref{step:iii}).
We use Alice and Bob to denote the players in the input instance of $\ahm_{r-1}$ where Alice holds $\Aastar = \alicepart{r-1}$ and Bob holds $\Bbstar = \bobpart{r-1}$. 
On the other hand, we use $\player{X}$ and $\player{Y}$ to refer to the players in the simulated input instance of $\ahm_r$ where
player $\player{X}$ holds $\alicepart{r} = \Xx$ and player $\player{Y}$ holds $\bobpart{r} = (\sigmar, \sigmac, \Yy)$.

\vspace{0.25em}
\begin{ourbox}
	\vspace{0.25em}
	\textbf{An $r$-round protocol $\protnum{1}_r$ for $\ahm_{r-1}(n_{r-1},\alpha)$ on input $(\Aastar,\Bbstar) \sim \DD_{r-1}(n_{r-1}, \alpha)$ where $(\speckstar{r-1}, \speckstar{r-2}, \dots, \speckstar{1})$ is the given uniform random search sequence}: 
	\begin{enumerate}[label=$(\alph*)$]
		\item Using \underline{public randomness}, independently sample a uniform random search index $\kstar \in [k_r]$. Then, jointly sample the following from $\DD_r(n_r, \alpha)$: \label{step:a}
		\begin{itemize}
			\item The uniform random permutations $\sigmar,\sigmac$ of $[b_r]$;
			\item Bob's part of all the $k_r^2 - k_r$ many off-diagonal sub-instances 
			\[
			\Bcommon = (\bobpart{r-1}_{i,j} : b_r \cdot \alpha < \sigmar(i) \neq \sigmac(j) \leq b_r);
			\]
			\item Alice's part of $\kstar-1$ many special sub-instances $\Aspec_{< \kstar} = (\Aspec_k : 1 \leq k < \kstar)$;
			\item Bob's part of $k_r-\kstar$ many special sub-instances $\Bspec_{> \kstar} = (\Bspec_k : \kstar < k \leq k_r)$.
		\end{itemize}
		Let $\istar, \jstar \in [b_r]$ be such that $\sigmar(i) = \sigmac(j) = b_r \cdot \alpha + \kstar$, which both players can compute since $\sigmar,\sigmac, \kstar$ are sampled using public randomness.
		\item Bob takes on the role of player $\player{X}$ and sets $\Xx[\istar,\jstar] = \Bbstar$. 
		Bob \underline{privately samples} $\Xx$ from $\DD_r(n_r,\alpha)$ conditioned on all the random variables from step \ref{step:a} and $\Xx[\istar,\jstar]$, i.e., the remainder of $\player{X}$'s input $\alicepart{r}$.
		\item Alice takes on the role of player $\player{Y}$ and sets $\Yy[\istar,\jstar] = \Aastar$. 
		Alice \underline{privately samples} $\Yy$ from $\DD_r(n_r, \alpha)$ conditioned on all the variables from step \ref{step:a} and $\Yy[\istar,\jstar]$, i.e., the remainder of player $\player{Y}$'s input $\bobpart{r}$.
		\item The players $\player{X}$ ($= \text{Bob}$) and $\player{Y}$ ($= \text{Alice}$) simulate a run of protocol $\prot_r$ in $r$ rounds using their respective inputs $\alicepart{r} = \Xx$ and $\bobpart{r} = (\sigmar,\sigmac,\Yy)$ and starting with player $\player{X}$.
		\item At the end of the protocol, the player that receives the final message gets the uniform random search sequence $(\speckstar{r-1}, \speckstar{r-2}, \ldots, \speckstar{0})$ and then returns the answer of $\prot_r$ on the search sequence $(\kstar, \speckstar{r-1}, \speckstar{r-2}, \ldots, \speckstar{0})$.
	\end{enumerate}
\end{ourbox}

\begin{observation}\label{obs:input-prot-1}
	In protocol $\protnum{1}_r$, $(\alicepart{r}, \bobpart{r})$ in the simulation of $\prot_r$ is sampled from $\DD_r(n_r, \alpha)$ and $(\kstar, \speckstar{r-1}, \speckstar{r-2}, \ldots, \speckstar{0})$ is a uniform random search sequence. 
\end{observation}

\begin{proof}
	The input to the simulation of $\prot_r$ is $(\alicepart{r}, \bobpart{r})$ where $\alicepart{r} = \Xx$ and $\bobpart{r} = (\sigmar, \sigmac, \Yy)$.
	All variables of the input are jointly sampled from $\DD_r(n_r, \alpha)$ except for Alice and Bob's embedded instance $(\Aastar, \Bbstar) \sim \DD_{r-1}(n_{r-1}, \alpha)$ where $\Xx[\istar, \jstar] = \Bbstar$ and $\Yy[\istar,\jstar] = \Aastar$.
	Since the instance of $\ahm_{r-1}$ corresponding to $(\Xx[\istar, \jstar], \Yy[\istar,\jstar])$ is independent of all other variables in the distribution $\DD_{r}(n_r, \alpha)$, all other variables are jointly distributed according to $\DD_r(n_r, \alpha)$.
	It remains to show that $(\Xx[\istar, \jstar], \Yy[\istar,\jstar])$ distributed according to $\DD_{r-1}(n_{r-1},\alpha)$, which is required in $\DD_r(n_r, \alpha)$.
	This is immediate from the embedding as $(\Xx[\istar, \jstar], \Yy[\istar,\jstar]) = (\Aastar, \Bbstar) \sim \DD_{r-1}(n_{r-1}, \alpha)$

	Each search index in $(\speckstar{r-1}, \speckstar{r-2}, \ldots, \speckstar{0})$ is uniformly sampled. Then, since $\kstar \in [k_r]$ is also uniformly sampled, we have $(\kstar, \speckstar{r-1}, \speckstar{r-2}, \ldots, \speckstar{0})$ is a random search sequence as required.	
\end{proof}

With this observation, we have that Alice and Bob have successfully simulated $\prot_r$ for $\ahm_r(n_r, \alpha)$ on the correct distribution, as the search sequence $(\kstar, \speckstar{r-1}, \speckstar{r-2}, \ldots, \speckstar{0})$ is uniformly random and the players' inputs are distributed according to $\DD_r(n_r, \alpha)$.
Since the messages communicated by the players are exactly the messages of the simulated protocol $\prot_r$, the communication cost of $\protnum{1}_r$ is the same as $\prot_r$. Although the communication cost is large, we now show that its internal information cost (about input $(\Aastar, \Bbstar)$) is smaller by a multiplicative $k_r$ factor.

\paragraph{Notation.}
We use $\rAstar, \rBstar$ to denote the random variables corresponding to the input instance of $\ahm_{r-1}$ given to Alice and Bob, which are distributed according to $\DD_{r-1}$. We use $\rProt$ to denote the random variable corresponding to the set of messages sent by both Alice and Bob. For $i \in [r]$, $\Prot_i$ denotes the message sent in round $i$ and $\rProt_i$ denotes the random variable corresponding to $\Prot_i$. 
We use $\rAcommon, \rBcommon, \rkstar, \rAspec, \rBspec$ to denote the random variables corresponding to $\Acommon, \Bcommon$, $\kstar, \Aspec = (\Aspec_1, \Aspec_2, \ldots, \Aspec_{k_r})$, $\Bspec = (\Bspec_1, \ldots, \Bspec_{k_r})$, respectively. 

The following claim proves that information cost of $\protnum{1}_r$ about its input sampled from $\DD_{r-1}$ is $1/k_r$ times smaller than the information cost of $\prot_r$ about its own input sampled from $\DD_r$. 

\begin{claim}\label{clm:ic-lower-1}
	In protocol $\protnum{1}_r$, we have that
	\[
	\ic{\protnum{1}_r}{\DD_{r-1}(n_{r-1}, \alpha)} \leq \frac{1}{k_r} \cdot \ic{\prot_r}{\DD_r(n, \alpha)}.
	\]
\end{claim}

\begin{proof}
	By the definition of internal information cost (\Cref{def:int-info}),
	\begin{align}\label{eq:int-info-prot-one}
		\begin{split}
		\ic{\protnum{1}_r}{\DD_{r-1}(n_{r-1}, \alpha)} &= \mi{\rAstar}{\rProt \mid  \underbrace{\rsigmar, \rsigmac, \rBcommon, \rkstar, \rBspec_{>\rkstar}, \rAspec_{<\rkstar}}_{\text{public randomness of $\protnum{1}_r$}}, \rBstar} \\
		&\hspace{5mm}+ \mi{\rBstar}{\rProt \mid  \underbrace{\rsigmar, \rsigmac, \rBcommon, \rkstar, \rBspec_{>\rkstar}, \rAspec_{<\rkstar}}_{\text{public randomness $\protnum{1}_r$}}, \rAstar}. 
		\end{split}
	\end{align}
Each mutual information term corresponds to the amount of information communicated by each player about their input, and we bound them separately.

For Alice's mutual information term in \Cref{eq:int-info-prot-one}, we have that
\begin{align*}
&\hspace{-20pt}\mi{\rAstar}{\rProt \mid  \rsigmar, \rsigmac, \rBcommon, \rkstar, \rBspec_{>\rkstar}, \rAspec_{<\rkstar}, \rBstar} \\
	&= \frac1{k_r} \cdot \sum_{k \in [k_r]} \mi{\rAstar}{\rProt \mid \rsigmar, \rsigmac, \rBcommon, \rkstar = k, \rBspec_{> k}, \rAspec_{<k}, \rBstar} \tag{by the definition of conditional mutual information and the uniform distribution of $\rkstar$} \\
		&= \frac1{k_r} \cdot \sum_{k \in [k_r]} \mi{\rAspec_k}{\rProt \mid \rsigmar, \rsigmac, \rBcommon, \rkstar = k, \rBspec_{\geq k}, \rAspec_{<k}}  \tag{as $\rAstar = \rAspec_{k}, \rBstar = \rBspec_{k}$ in protocol $\protnum{1}_r$ when $\rkstar = k$}\\
		&= \frac1{k_r} \cdot \sum_{k \in [k_r]} \mi{\rAspec_k}{\rProt \mid \rsigmar, \rsigmac, \rBcommon, \rBspec_{\geq k}, \rAspec_{<k}}; 
	\end{align*}
 the last part holds because the joint distribution of $(\rAspec_{\leq k},\rBspec_{\geq k}, \rsigmar, \rsigmac, \rBcommon, \Prot)$ is independent of the value of $\rkstar$. 
In particular, if we sample any input $(\Aa{r}, \Bb{r})$ from $\DD_r(n, \alpha)$ (which does not include $\kstar = \speckstar{r}$), the values of $\rsigmar, \rsigmac, \rBcommon, \rAspec_{\leq k}, \rBspec_{\geq k}$ are fixed and the value of $\rProt$ is also fixed as $\prot_r$ is deterministic. Hence, the distribution of $\rkstar$ remains uniform over $[k_r]$ irrespective of the instance $(\Aa{r}, \Bb{r})$.
Continuing the bounding of Alice's mutual information term, we have that
\begin{align*}
	&\frac1{k_r} \cdot \sum_{k \in [k_r]} \mi{\rAspec_k}{\rProt \mid  \rsigmar, \rsigmac, \rBcommon, \rBspec_{\geq k}, \rAspec_{<k}} \\
	&\hspace{20pt}\leq \frac1{k_r} \cdot \sum_{k \in [k_r]} \mi{\rAspec_k}{\rProt \mid  \rsigmar, \rsigmac, \rBcommon, \rBspec, \rAspec_{<k}} \tag{by \Cref{prop:info-increase} as $\rBspec_{<k} \perp \rAspec_k \mid  \rsigmar, \rsigmac, \rBcommon, \rAspec_{<k}$
	by the definition of distribution $\DD_r$}  \\
	&\hspace{20pt}= \frac1{k_r} \cdot \mi{\rAspec}{\rProt \mid  \rsigmar, \rsigmac, \rBcommon, \rBspec} \tag{by the chain rule of mutual information, \itfacts{chain-rule}} \\
	&\hspace{20pt}\leq \frac1{k_r} \cdot \mi{\rAspec}{\rProt \mid \rsigmar, \rsigmac, \rBcommon, \rBspec, \rBrest} \tag{by \Cref{prop:info-increase}, as $\rAspec \perp \rBrest \mid \rBspec, \rsigmar, \rsigmac, \rBcommon$ by
	the definition of distribution $\DD_r$} \\
	&\hspace{20pt}= \frac1{k_r} \cdot \mi{\rAspec, \rBcommon}{\rProt \mid \rsigmar, \rsigmac, \rBcommon, \rBspec, \rBrest} \tag{as $\rBcommon$ is fixed} \\
	&\hspace{20pt}\leq \frac1{k_r} \cdot \mi{\rAspec, \rAcommon, \rBcommon, \rsigmar, \rsigmac}{\rProt \mid \rBcommon, \rBspec, \rBrest} \tag{by the chain rule (\itfacts{chain-rule}) and non-negativity of mutual information (\itfacts{info-zero})}\\
	&\hspace{20pt}= \frac1{k_r} \cdot \mi{\rBb{r}}{\rProt \mid \rAa{r}} \tag{recall that input of $\player{Y}$ (resp. $\player{X}$) simulated by Alice (resp. Bob) in $\protnum{1}_r$ is $\Bb{r}$ (resp. $\Aa{r}$)}. 
\end{align*}

For Bob's mutual information term in \Cref{eq:int-info-prot-one}, following a similar argument, we have
\begin{align*}
	&\mi{\rBstar}{\rProt \mid  \rsigmar, \rsigmac, \rBcommon, \rkstar, \rBspec_{>\rkstar}, \rAspec_{<\rkstar}, \rAstar} \\
	&= \frac1{k_r} \cdot \sum_{k \in [k_r]} \mi{\rBspec_{k}}{\rProt \mid  \rsigmar, \rsigmac, \rBcommon, \rkstar = k, \rBspec_{>k}, \rAspec_{\leq k}} \tag{by definition of conditional mutual information and $\rAstar,\rBstar$ and the uniform distribution of $\rkstar$} \\
	&= \frac1{k_r}  \cdot \sum_{k \in [k_r]} \mi{\rBspec_{k}}{\rProt \mid  \rsigmar, \rsigmac, \rBcommon, \rBspec_{>k}, \rAspec_{\leq k}} \tag{as in the previous case for Alice, for any $k \in [k_r]$, $(\rkstar = k) \perp (\rBspec_{\geq k}, \rAspec_{\leq k}, \rsigmar, \rsigmac,\rBcommon)$} \\
	&\leq \frac1{k_r} \cdot \sum_{k \in [k_r]}  \mi{\rBspec_{k}}{\rProt \mid  \rsigmar, \rsigmac, \rBcommon, \rBspec_{>k}, \rAspec} \tag{by \Cref{prop:info-increase}, as $\rAspec_{>k} \perp \rBspec_k \mid \rsigmar,\rsigmac,\rBcommon,\rBspec_{> k}$ by the definition of distribution $\DD_r$} \\
	&= \frac1{k_r} \cdot \mi{\rBspec}{\rProt \mid  \rsigmar, \rsigmac, \rBcommon, \rAspec} \tag{by chain rule of mutual information, \itfacts{chain-rule}} \\
	&\leq \frac1{k_r} \cdot \mi{\rBspec}{\rProt \mid \rsigmar, \rsigmac, \rAcommon, \rBcommon, \rAspec} \tag{by \Cref{prop:info-increase}, as $\rAcommon \perp \rBspec \mid \rsigmar, \rsigmac, \rBcommon, \rAspec$ by the definition of distribution $\DD_r$} \\
	&\leq \frac1{k_r} \cdot \mi{\rBspec, \rBcommon, \rBrest}{\rProt \mid \rsigmar, \rsigmac, \rAcommon, \rBcommon, \rAspec} \tag{by the chain rule (\itfacts{chain-rule}) and non-negativity of mutual information (\itfacts{info-zero})} \\
	&= \frac1{k_r} \cdot \mi{\rAa{r}}{\rProt \mid \rBb{r}} \tag{recall that input of $\player{X}$ (resp. $\player{Y}$) simulated by Bob (resp. Alice) in $\protnum{1}_r$ is $\Aa{r}$ (resp. $\Bb{r}$)}. 
\end{align*}

With both these bounds, we can upper bound the LHS of \Cref{eq:int-info-prot-one} and obtain
\begin{align*}
	\ic{\protnum{1}_r}{\DD_{r-1}} &\leq \frac1{k_r} \cdot \left( \mi{\rBb{r}}{\rProt \mid \rAa{r}} + \mi{\rAa{r}}{\rProt \mid \rBb{r}} \right) \\
	&= \frac1{k_r} \cdot \ic{\prot_r}{\DD_r}, 
\end{align*}
where the last step follows since the messages communicated in $\protnum{1}_r$ are exactly the messages communicated in the simulation of $\prot_r$ and since $(\rAa{r}, \rBb{r}) \sim \DD_r$ in $\protnum{1}_r$ by \Cref{obs:input-prot-1}. 
\end{proof}

We now obtain the following main lemma of this step of the argument:
\begin{lemma}\label{lem:prot-1-details}
	Protocol $\protnum{1}_r$ is an $r$-round protocol for $\ahm_{r-1}(n_{r-1}, \alpha)$ where Bob speaks first with communication cost at most $s$, probability of success at least $\delta$, and internal information cost at most $s/k_r$ when the input is sampled from $\DD_{r-1}(n_{r-1}, \alpha)$.
\end{lemma}

\begin{proof}
	In protocol $\protnum{1}_r$, the instance $(\Aa{r}, \Bb{r})$ used to simulate $\prot_r$ is constructed without any communication, so the total number of bits communicated in $\protnum{1}_r$ is at most the total number of bits communicated in $\prot_r$, which is $\cc{\prot_r} \leq s$ as in the statement of \Cref{lem:round-elim}. 
	Bob takes on the role of $\player{X}$ in the simulation and thus Bob speaks first.

	By construction, protocol $\protnum{1}_r$ outputs the same answer as $\prot_r$ since the solution to $(\Aa{r}, \Bb{r}) \sim \DD_r$ on search sequence $(\kstar, \speckstar{r-1}, \speckstar{r-2}, \ldots, \speckstar{0})$ is the same as the solution to the input $ (\Aspec_{\kstar}, \Bspec_{\kstar}) = (\Aastar, \Bbstar) \sim \DD_{r-1}$ on the search sequence $(\speckstar{r-1}, \speckstar{r-2}, \ldots, \speckstar{0})$. Therefore, $\protnum{1}_r$ succeeds with probability at least $\delta$. 
	Finally, the bound on the internal information cost of $\protnum{1}_r$ follows from \Cref{clm:ic-lower-1} and the fact that the communication cost of $\prot_r$, namely, $s$, is an upper bound of its internal information cost.
\end{proof}

\subsubsection*{Step \ref{step:ii}: A Low Communication Cost Protocol via Message Compression}

In this step, we compress the communication cost of $\protnum{1}_r$ using standard message compression techniques. As a direct corollary of \Cref{prop:msg-compress}, we get the following:

\begin{corollary}\label{cor:prot2-details}
	For any $0 < \eps<1$, there exists an $r$-round protocol $\protnum{2}_r$ for $\ahm_{r-1}(n_{r-1},\alpha)$ where Bob speaks first with 
	\[
		\cc{\protnum{2}_r} \leq \jpyconst  \cdot (r/\eps) \cdot\paren{   \frac{s}{k_r} +  r}. 
	\]
	and probability of success at least $\delta - \eps$ when the input is sampled from $\DD_{r-1}(n_{r-1},\alpha)$.
\end{corollary}

\begin{proof}
	We use \Cref{prop:msg-compress} on protocol $\protnum{1}_r$ to get protocol $\protnum{2}_r$ such that the simulation fails with probability at most $\eps$. Then, we have that
	\begin{align*}
		\cc{\protnum{2}_r} &\leq \jpyconst \cdot \paren{r/\eps \cdot \ic{\protnum{1}_r}{\DD_{r-1}} + r^2/\eps} \leq \jpyconst  \cdot (r/\eps) \cdot\paren{   \frac{s}{k_r} +  r}. \tag{by \Cref{lem:prot-1-details}}
	\end{align*}
The simulation of $\protnum{1}_r$ using compressed messages fails with probability at most $\eps$, so $\protnum{2}_r$ succeeds with probability at least $\delta-\eps$, as desired. 
\end{proof}

Since this step only compresses the messages sent in $\protnum{1}_r$, protocol $\protnum{2}_r$ is constructed in the same way as $\protnum{1}_r$ with an additional message compression step that uses an independent source of public randomness. 
In particular, $\protnum{2}_r$ is a simulation of $\prot_r$ on an instance of $\ahm_r$ sampled from $\DD_r$, but with communication cost smaller by a $\simeq k_r$ multiplicative factor and with a slightly smaller probability of success.

\subsubsection*{Step \ref{step:iii}: An $(r-1)$-Round Protocol via Guessing the First Message}

In this final and the most important step, we take the $r$-round protocol $\protnum{2}_r$ for $\ahm_{r-1}$ where Bob communicates first and we eliminate the first message by making the players \emph{guess} it using public randomness, thus beginning the protocol from the second round. 
Therefore, we obtain an $(r-1)$-round protocol $\prot_{r-1}$ for $\ahm_{r-1}$ where Alice communicates first as required.
This alters the joint distribution of the input and thus affects the guarantees of the protocol (as in \Cref{cor:prot2-details}), but we will show that the effect is small.

\vspace{0.25em}
\begin{ourbox}
	\vspace{0.25em}
	\textbf{An $(r-1)$-round protocol $\prot_{r-1}$ for $\ahm_{r-1}$ on inputs $(\Aa{r-1},\Bb{r-1}) \sim \DD_{r-1}(n_{r-1}, \alpha)$ where $(\speckstar{r-1}, \speckstar{r-2}, \dots, \speckstar{1})$ is the given uniform random search sequence}: 
	\begin{enumerate}[label=$(\alph*)$]
		\item Using \underline{public randomness}, jointly sample the first message $\Prot_1$ and the search index $\kstar$ from the protocol $\protnum{2}_r$. Then, jointly sample the following from $\DD_{r}(n_r, \alpha)$ conditioned on $\Prot_1, \kstar$:\label{step:aa_pub}
		\begin{itemize}
			\item The uniform random permutations $\sigmar,\sigmac$ of $[b_r]$;
			\item Bob's part of all the $k_r^2 - k_r$ many off-diagonal sub-instances 
			\[
			\Bcommon = (\bobpart{r-1}_{i,j} : b_r \cdot \alpha < \sigmar(i) \neq \sigmac(j) \leq b_r);
			\]
			\item Alice's part of the $\kstar-1$ many special sub-instances $\Aspec_{< \kstar} = (\Aspec_k : 1 \leq k < \kstar)$;
			\item Bob's part of the $k_r-\kstar$ many special sub-instances $\Bspec_{> \kstar} = (\Bspec_k : \kstar < k \leq k_r)$.
		\end{itemize}
		Let $\istar, \jstar \in [b_r]$ be such that $\sigmar(i) = \sigmac(j) = b_r \cdot \alpha + \kstar$, which both players can compute since $\sigmar,\sigmac, \kstar$ are sampled using public randomness.

		\item Bob takes on the role of player $\player{X}$ and sets $\Xx[\istar,\jstar] = \Bbstar$. 
		Bob \underline{privately samples} $\Xx$ from $\DD_r(n_r,\alpha)$ conditioned on all the random variables from step \ref{step:aa_pub} and $\Xx[\istar,\jstar]$, i.e., the remainder of $\player{X}$'s input $\alicepart{r}$.
		\item Alice takes on the role of player $\player{Y}$ and sets $\Yy[\istar,\jstar] = \Aastar$. 
		Alice \underline{privately samples} $\Yy$ from $\DD_r(n_r, \alpha)$ conditioned on all the variables from step \ref{step:aa_pub} and $\Yy[\istar,\jstar]$, i.e., the remainder of player $\player{Y}$'s input $\bobpart{r}$.
		\item The players $\player{X}$ ($= \text{Bob}$) and $\player{Y}$ ($= \text{Alice}$) simulate protocol $\protnum{2}_r$ using their respective inputs $\alicepart{r} = \Xx, \bobpart{r} = (\sigmar,\sigmac,\Yy)$ and the search sequence $(\kstar, \speckstar{r-1}, \speckstar{r-2}, \dots, \speckstar{1})$.
		They use $\Prot_1$ as a guess of the first message sent from $\player{X}$ to $\player{Y}$ and thus only simulate $r-1$ rounds of $\protnum{2}_r$ starting with the second round, i.e., player $\player{Y}$ communicates first.
	\end{enumerate}
\end{ourbox}

Observe that, similar to the construction of $\protnum{2}_r$ (which is identical to $\protnum{1}_r$, except for the message compression), protocol $\prot_{r-1}$ embeds the input for $\ahm_{r-1}$ sampled from $\DD_{r-1}$ into a simulated instance $(\alicepart{r}, \bobpart{r})$ of $\ahm_{r}$.
Then, the players simulate a run of $\protnum{2}_r$ using $(\alicepart{r}, \bobpart{r})$ where, to remove the first round of communication, the first message is guessed using public randomness.
However, due to the random guessing of the first message, the simulated instance is no longer distributed according to $\DD_{r}(n_r, \alpha)$ (recall \Cref{obs:input-prot-1}), i.e., the guarantees of $\protnum{2}_r$ do not hold for $\prot_{r-1}$. 
Despite this, we show that the simulated instance and first message are statistically close  to being distributed as they are in $\protnum{2}_r$ and thus similar guarantees hold.

Let $\DDfake$ denote the joint distribution of $(\rAa{r}, \rBb{r}, \rProt_1, \rkstar)$ as it is in $\prot_{r-1}$ and let $\DDreal$ denote the joint distribution of $(\rAa{r}, \rBb{r}, \rProt_1, \rkstar)$ as it is in $\protnum{2}_r$.
Let 
\[
\rRa := (\rProt_1, \rsigmar, \rsigmac, \rBcommon, \rkstar, \rBspec_{>\rkstar}, \rAspec_{<\rkstar})
\]
denote the public randomness used in step \ref{step:aa_pub} of $\prot_{r-1}$. 
We show in \Cref{clm:DDrealfake-input} that the distributions $\DDfake$ and $\DDreal$ only differ in the way that the input $(\rAstar, \rBstar) = (\rAspec_{\rkstar}, \rBspec_{\rkstar})$ is sampled.
\begin{claim} \label{clm:DDrealfake-input}
	We have that
	\begin{align*}
		\DDfake 
		&= \distribution{\rRa} 
		\times ~~~ \distribution{\rAspec_{\rkstar}, \rBspec_{\rkstar}} ~~~ \!
		\times  \distribution{\rAcommon, \rAspec_{> \rkstar} \mid \rRa, \rAspec_{\rkstar}}
		\times  \distribution{\rBrest, \rBspec_{<\rkstar} \mid \rRa, \rBspec_{\rkstar}}, \\
		\DDreal 
		&= \distribution{\rRa} 
		\times \distribution{\rAspec_{\rkstar}, \rBspec_{\rkstar} \mid \rRa} 
		\times  \distribution{\rAcommon, \rAspec_{> \rkstar} \mid \rRa, \rAspec_{\rkstar}}
		\times  \distribution{\rBrest, \rBspec_{<\rkstar} \mid \rRa, \rBspec_{\rkstar}}.
	\end{align*}
\end{claim}
	
\begin{proof}
	We first consider $(\rAa{r}, \rBb{r}, \rProt_1, \rkstar)$ directly as it is sampled in $\prot_{r-1}$.
	The random variable $\rRa$ in $\prot_{r-1}$ is jointly sampled in step \ref{step:aa_pub}.
	Then, Alice and Bob embed $(\rAstar, \rBstar) \sim \DD_{r-1}$ as $\rAspec_{\rkstar} = \rAstar$ and $\rBspec_{\rkstar} = \rBstar$, i.e., $(\rAspec_{\rkstar}, \rBspec_{\rkstar}) \sim \DD_{r-1}$, which is \emph{not} conditioned on $\rRa$.
	Finally, Bob and Alice privately sample the remainder of $\Xx$ and $\Yy$ (and thus the remainder of $\alicepart{r}$ and $\bobpart{r}$), respectively, which is only conditioned on $\rRa$ and their respective embedded inputs. This defines $\DDfake$.

	We then obtain $\DDreal$ from $\DDfake$ by including all previously sampled random variables in the conditioning of the subsequently sampled random variables. That is, we obtain
	\begin{align*}
		\DDreal &= \distribution{\rRa}  \\
		&\times \distribution{\rAspec_{\rkstar}, \rBspec_{\rkstar} \mid \rRa}  \\
		&\times  \distribution{\rAcommon, \rAspec_{> \rkstar} \mid \rRa, \rAspec_{\rkstar}, \rBspec_{\rkstar}} \\
		&\times  \distribution{\rBrest, \rBspec_{<\rkstar} \mid \rRa, \rBspec_{\rkstar}, \rAcommon, \rAspec_{\geq \rkstar}}.
	\end{align*}
	It is easy to verify that this is exactly the distribution of $(\rAa{r}, \rBb{r}, \rProt_1, \rkstar)$ as it is in $\protnum{2}_r$.

	Finally, we simplify the third and fourth terms in $\DDreal$ to show that they are identical to the corresponding terms in $\DDfake$, which follows directly from the claims that 
	\[
	(\rAcommon, \rAspec_{>\rkstar}) \perp \rBspec_{\rkstar} \mid \rRa, \rAspec_{\rkstar} \quad \text{and} \quad \rBspec_{<\rkstar}, \rBrest \perp \rAcommon, \rAspec_{\geq \rkstar} \mid \rRa, \rBspec_{\rkstar},
	\]
	respectively. 
	To prove these claims, it is sufficient to bound the corresponding mutual information terms as follows:
	\begin{align*}
		&\mi{\rAcommon, \rAspec_{>\rkstar}}{\rBspec_{\rkstar} \mid \rRa, \rAspec_{\rkstar}} \\ 
		&\leq \mi{\rAcommon, \rAspec_{>\rkstar}}{\rBspec_{\leq \rkstar}, \rBrest \mid \rRa, \rAspec_{\rkstar}}\tag{by the chain rule (\itfacts{chain-rule}) and non-negativity of mutual information (\itfacts{info-zero})}  \\
			&=  \mi{\rAcommon, \rAspec_{>\rkstar}}{\rBspec_{\leq \rkstar}, \rBrest \mid \rProt_1, \rsigmar, \rsigmac, \rBcommon, \rkstar, \rBspec_{>\rkstar}, \rAspec_{\leq \rkstar}} \tag{by the definition of $\rRa$}\\
			&\leq \mi{\rAcommon, \rAspec_{>\rkstar}}{\rBspec_{\leq \rkstar}, \rBrest, \rProt_1 \mid  \rsigmar, \rsigmac, \rBcommon, \rkstar, \rBspec_{>\rkstar}, \rAspec_{\leq \rkstar}} \tag{by the chain rule (\itfacts{chain-rule}) and non-negativity of mutual information (\itfacts{info-zero})}   \\
			&= \mi{\rAcommon, \rAspec_{>\rkstar}}{\rBspec_{\leq \rkstar}, \rBrest \mid  \rsigmar, \rsigmac, \rBcommon, \rkstar, \rBspec_{>\rkstar}, \rAspec_{\leq \rkstar}} \tag{as $\rProt_1$ is fixed by $\rBcommon, \rBrest, \rBspec$} \\
			&= 0, \tag{by~\itfacts{info-zero} and construction of distribution $\DD_r$}
	\end{align*}
	and
	\begin{align*}
		&\mi{\rBspec_{<\rkstar}, \rBrest}{\rAcommon, \rAspec_{\geq \rkstar} \mid \rRa, \rBspec_{\rkstar}} \\
		&\leq \mi{\rBspec_{<\rkstar}, \rBrest, \rProt_1}{\rAcommon, \rAspec_{\geq \rkstar} \mid \rsigmar, \rsigmac, \rBcommon, \rkstar, \rBspec_{\geq \rkstar}, \rAspec_{\leq \rkstar}}\tag{by the chain rule (\itfacts{chain-rule}) and non-negativity of mutual information (\itfacts{info-zero})}  \\
		&= \mi{\rBspec_{<\rkstar}, \rBrest}{\rAcommon, \rAspec_{\geq \rkstar} \mid \rsigmar, \rsigmac, \rBcommon, \rkstar, \rBspec_{\geq \rkstar}, \rAspec_{\leq \rkstar}} \tag{as $\rProt_1$ is fixed by $\rBrest, \rBcommon, \rBspec$} \\
		&= 0 \tag{by~\itfacts{info-zero} and construction of distribution $\DD_r$}.
	\end{align*}
	This completes the proof.
\end{proof}

We note here that the simulated instance $(\alicepart{r}, \bobpart{r})$ in $\DDreal$ is sampled according to $\DD_r(n_r, \alpha)$ (recall \Cref{obs:input-prot-1}). 
However, this is not true in $\DDfake$ since $(\rAstar, \rBstar) = (\rAspec_{\rkstar}, \rBspec_{\rkstar})$ is sampled according to $=\DD_{r-1}$ instead of $\DD_r \mid \rRa$.
Thus, to bound the total variation distance between $\DDreal$ and $\DDfake$, we show that the information revealed by the random variable 
$\rRa$
about the input $(\rAstar, \rBstar)$ in $\protnum{2}_r$, i.e., under distribution $\DDreal$, is not too large.
We need an independence claim (\Cref{clm:tvd-bound-help-1}) and then a crucial bound on the information carried by the first message (\Cref{clm:dir-sum-B}).

\begin{claim}\label{clm:tvd-bound-help-1}
	\[
		\rAspec_{<\rkstar} \perp \rBspec_{\rkstar} \mid \rProt_1, \rsigmar, \rsigmac, \rBcommon, \rkstar, \rBspec_{>\rkstar}, \rBspec_{<\rkstar}.
	\]
\end{claim}
\begin{proof}
	Again, we bound the mutual information between the terms. Let $\rRcomp$ represent the independent source of public randomness used in the message compression in step \ref{step:ii}. Then,
\begin{align*}
	&\mi{\rAspec_{<\rkstar}}{\rBspec_{\rkstar} \mid \rProt_1, \rsigmar, \rsigmac, \rBcommon, \rkstar, \rBspec_{>\rkstar}, \rBspec_{<\rkstar}} \\
	&\leq \mi{\rAspec_{<\rkstar}}{\rBspec_{\rkstar}, \rProt_1 \mid \rsigmar, \rsigmac, \rBcommon, \rkstar, \rBspec_{>\rkstar}, \rBspec_{<\rkstar}} \tag{moving $\rProt_1$ by chain rule in \itfacts{chain-rule}, and non-negativity of mutual information \itfacts{info-zero}} \\
	&\leq \mi{\rAspec_{<\rkstar}}{\rBspec_{\rkstar}, \rBrest, \rRcomp \mid \rsigmar, \rsigmac, \rBcommon, \rkstar, \rBspec_{>\rkstar}, \rBspec_{<\rkstar}} \tag{by data processing inequality in \itfacts{data-processing}, $\rProt_1$ is fixed by $\rBrest, \rBspec, \rBcommon, \rRcomp$} \\
	&= \mi{\rAspec_{<\rkstar}}{\rBspec_{\rkstar}, \rBrest \mid \rsigmar, \rsigmac, \rBcommon, \rkstar, \rBspec_{>\rkstar}, \rBspec_{<\rkstar}} \tag{since $\rRcomp$ is an entirely independent source of randomness} \\
	&\leq 0. \qedhere \tag{by~\itfacts{info-zero} and construction of $\DD_r$}
\end{align*}
\end{proof}

The following claim is the heart of the entire proof. 

\begin{claim}\label{clm:dir-sum-B}
	For any $k \in [k_r]$, 
	\[
		\mi{\rBspec_k}{\rProt_1 \mid \rsigmar, \rsigmac, \rBcommon, \rBspec_{<k}, \rBspec_{>k}}  \leq \frac{s}{b_r^2 \cdot \alpha^2}.
	\]
\end{claim}
\begin{proof}
	Let $\rhor, \rhoc$ be the inverse of $\sigmar, \sigmac$, and let $\rrhor, \rrhoc$ denote the random variables corresponding to $\rhor, \rhoc$ respectively.
	For ease of exposition, we represent any permutation $\rho$ on $[b_r]$ as an ordered set $S := \{\rho(\ell) : \ell \in [b_r] \}$ with the ordering that $\rho(\ell) < \rho(\ell')$ for $\ell < \ell'$.
	Using this representation, we allow for partial definitions of permutations, i.e., ordered subsets $T \subset S$.
	In particular, we  require the following partial definitions of $\rhor$ and $\rhoc$:
	\begin{itemize}
		\item 
		Let $\Srow$ be an ordered set $$\Srow = \set{ \rhor(\ell) \mid b_r \cdot \alpha  <\ell \leq b_r, \ell \neq b_r\cdot \alpha +k}$$
		of size $b_r \cdot (1-\alpha) -1$, which is the set of rows that corresponds to $\Bcommon$ except for the row corresponding to $\rBspec_k$. Let $\rSrow$ denote its random variable.
		\item 
		Let $\Scol$ be an ordered set 
		\[\Scol =  \set{ \rhoc(\ell) \mid b_r \cdot \alpha  <\ell \leq b_r, \ell \neq b_r\cdot \alpha +k}\]
		of size $b_r \cdot (1-\alpha) -1$, which is the set of columns that corresponds to $\Bcommon$ except for the column corresponding to $\rBspec_k$. Let $\rScol$ denote its random variable.  
	\end{itemize}
	We give an illustration of these definitions in \Cref{fig:TrowTcol}.

	\begin{figure}[ht]
		\centering
		\begin{tikzpicture}[scale=1.1, yscale=-0.75]


    \foreach \i in {3, 5, 6, 7}{
        \foreach \j in {1, ..., 7}{
            \draw[draw=none, fill=blue, opacity = 0.5] (\i -1, \j -1) rectangle (\i,\j);
        }
    }

    \foreach \i in {3, 5, 6, 7}{
        \foreach \j in {1, ..., 7}{
            \draw[draw=none, fill=red, opacity = 0.5] (\j -1, \i -1) rectangle (\j,\i);
        }
    }

    \foreach \i in {1, 2, ..., 7}{
        \foreach \j in {1, 2, ..., 7}{
            \draw[draw=black] (\i -1, \j -1) rectangle (\i,\j);
            \node[font=\small] () at (\i - 0.5, \j - 0.5) {$B$};
        }
    }

    \draw[ultra thick] (2,2) rectangle (7,7);
    \draw[decorate, decoration={brace}] (7.25,0.05) -- (7.25,1.95);
    \node[font=\small, anchor=west] () at (7.5,1) {$b_r \cdot \alpha$}; 
    \node[font=\small, anchor=east, white] () at (-0.5,1) {$b_r/\alpha$};  
    
    \draw[decorate, decoration={brace}] (7.25,2.05) -- (7.25,6.95);
    \node[font=\small, anchor=west] () at (7.5,4.5) {$b_r \cdot (1 - \alpha)$}; 
    
    \draw[decorate, decoration={brace}] (1.95,7.25) -- (0.05,7.25);
    \node[font=\small] () at (1,7.75) {$b_r \cdot \alpha$};
    \draw[decorate, decoration={brace}] (6.95,7.25) -- (2.05,7.25);
    \node[font=\small] () at (4.5,7.75) {$b_r \cdot (1 - \alpha)$};
    
    \foreach \i in {1}{
        \node[font=\small] () at (-0.5,\i -0.5) {$\sigmar(\i)$};
        \node[font=\small] () at (\i -0.5, -0.5) {$\sigmac(\i)$};
    }
    \node[font=\small] () at (-0.5,7 -0.5) {$\sigmar(b_r)$};
    \node[font=\small] () at (7 -0.5, -0.5) {$\sigmac(b_r)$};
    
    \node[font=\small] () at (-0.5,3.5) {$\vdots$};
    \node[font=\small] () at (3.5, -0.5) {$\dots$};
    

    \draw[ultra thick, draw=green!80!black, fill=none] (4 -1, 4 -1) rectangle (4,4);

\end{tikzpicture}   
		\caption{An illustration of the partial definitions $\Srow$ and $\Scol$ for an instance of $\ahm_r(n_r,\alpha)$ where $b_r = 7$, $\alpha = 2/7$, and $k = 2$. In this figure, the $b_r \times b_r$ matrix is after applying the respective permutations of rows and columns using $\sigmar$ and $\sigmac$ held by Bob.
		In each position of the matrix, the corresponding input is the $B's$ from each instance of $\ahm_{r-1}(n_{r-1}, \alpha)$ that is held by Alice.
		$\Srow$ and $\Scol$ are partial definitions of the permutations $\rhor$ and $\rhoc$, i.e., the inverses of $\sigmar$ and $\sigmac$, respectively. 
		$\Srow$ corresponds to fixing the rows filled red and $\Scol$ corresponds to fixing the columns filled blue, which implies fixing positions filled purple. The unfilled (white) positions remain unrestricted by $\Srow,\Scol$, i.e., the green outlined special position $\Bspec_k$ could still (uniformly) be in any of the unfilled positions when conditioned on $\Srow,\Scol$ (instead of $\sigmar, \sigmac$).
		} \label{fig:TrowTcol}
	\end{figure}
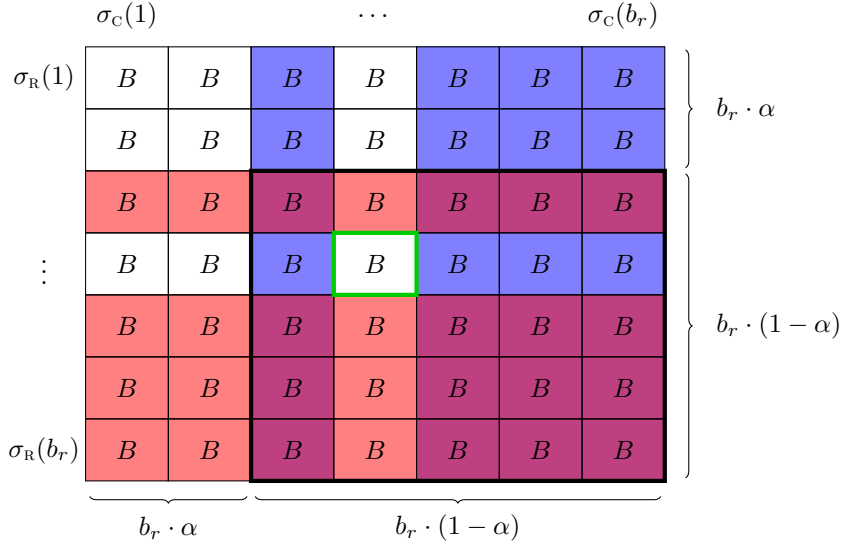

	This now allows us to write the indices $i,j \in [b_r]$ of the instances corresponding to $\rBspec$ and $\rBcommon$ using $\rhor, \rhoc, \Srow, \Scol$:
	\begin{enumerate}[label=$(\roman*)$]
		\item 	For any $\ell \in [k_r]$, $\rBspec_{\ell}$ is the same as the random variable $\rB_{\rhor(b_r\cdot \alpha + \ell), \rhoc(b_r \cdot \alpha + \ell)}$
		where we omit the superscript $r-1$ as it is clear from context; 
		\item $\rBcommon, \rBspec_{<k}, \rBspec_{>k}$ is the union of combinatorial rectangles formed by 
		\[
		(\rSrow \times \rScol), \quad (\rSrow \times \set{\rrhoc(b_r\cdot \alpha +k)}),  \quad (\set{\rhor(b_r \cdot \alpha + k)} \times \rScol).
		\]
	\end{enumerate}
	Finally, let $\rZ$ be defined as, 
	\[
		\rZ := (\rBcommon, \rBspec_{<k}, \rBspec_{>k}). 
	\]
	With these at hand, we can start bounding the mutual information term in the claim. Firstly, it is easy to see that,
\begin{align*}
		&\mi{\rBspec_k}{\rProt_1 \mid \rsigmar, \rsigmac, \rBcommon, \rBspec_{<k}, \rBspec_{>k}} = \mi{\rBspec_k}{\rProt_1 \mid \rsigmar, \rsigmac, \rZ},
\end{align*}
by definition of $\rZ$. 

We need the following intermediate claim for our proof. 
\begin{claim}\label{clm:ind-rest-of-sigma}\[
	\distribution{\rBspec_k, \rProt_1, \rZ \mid \rsigmar, \rsigmac} = \distribution{\rBspec_k, \rProt_1, \rZ \mid \rSrow, \rScol,  \rrhor(b_r \cdot \alpha + k), \rrhoc(b_r \cdot \alpha + k)}.
	\]
\end{claim}
\begin{proof}
	Let $\rrhor([b_r \cdot \alpha])$ and $\rrhoc([b_r \cdot \alpha])$ denote the random variables $\rrhor(i)$ and $\rrhoc(i)$ for all $i \in [b_r \cdot \alpha]$ respectively. 
We know that random variable $\rsigmar$ is fixed by $\rrhor$, as they are inverses of each other. Random variable $\rrhor$ is in turn fixed by all the values each element of $[b_r]$ is mapped to, i.e., $\rrhor(i)$ for $i \in [b_r]$. 

Hence, it is sufficient to argue that,
\[
(\rBspec_k, \rProt_1, \rZ) \perp \rrhor([b_r \cdot \alpha]), \rrhoc([b_r \cdot \alpha])  \mid  \rSrow, \rScol, \rrhor(b_r \cdot \alpha +k), \rrhoc(b_r \cdot \alpha + k),
\]
to prove the claim. 
We know that $\rBspec_k, \rZ$ is the random variable corresponding to the instances in $\Aa{r}$ inside the combinatorial rectangle 
\[
	(\rSrow \cup \{\rrhor(b_r \cdot \alpha +k)\}) \times (\rScol \cup \{\rrhoc(b_r \cdot \alpha +k)\}).
\]
The entirety of this rectangle is fixed by the conditioning. 

Proving the claim amounts to proving that the joint distribution of $\Bb{r-1}_{i,j}$ for a fixed set of positions of $\Aa{r}$ (the ones inside the rectangle fixed by conditioning), and the transcript $\rProt_1$ is independent of the \emph{permutation} of the rest of the positions of $\Aa{r}$. This is true because $\player{X}$, who sends  $\rProt_1$, does not have access to either of the permutations $\rrhor$ or $\rrhoc$ (held by $\player{Y}$).\Qed{clm:ind-rest-of-sigma}

\end{proof}
We can continue to bound the mutual information now. 
\begin{align*}
&\mi{\rBspec_k}{\rProt_1 \mid \rsigmar, \rsigmac, \rZ} \\
&=  \mi{\rBspec_k}{\rProt_1 \mid \rZ, \rSrow, \rScol, \rrhor(b_r \cdot \alpha + k), \rrhoc(b_r\cdot \alpha + k)} \tag{by \Cref{clm:ind-rest-of-sigma}} \\
&= \Exp_{\substack{(\Srow,\Scol) \\ \sim (\rSrow , \rScol)}} \mi{\rBspec_k}{\rProt_1 \mid \rZ, \rSrow = \Srow, \rScol = \Scol, \rrhor(b_r \cdot \alpha + k), \rrhoc(b_r\cdot \alpha + k)} \tag{by  definition of conditional mutual information} \\
&\Leq{(1)} \frac1{b_r^2 \cdot \alpha^2} \Exp_{\substack{(\Srow,\Scol) \\ \sim (\rSrow , \rScol)}} \sum_{\substack{i \in [b_r] \setminus \rSrow, \\ j \in [b_r] \setminus \rScol }} \mi{\rB_{i,j}}{\rProt_1 \mid \rZ, \Srow, \Scol, \rrhor(b_r \cdot \alpha + k) = i, \rrhoc(b_r\cdot \alpha + k) = j} \\
&\Eq{(2)} \frac1{b_r^2 \cdot \alpha^2} \Exp_{\substack{(\Srow,\Scol) \\ \sim (\rSrow , \rScol)}} \sum_{i \in [b_r] \setminus \rSrow, j \in [b_r] \setminus \rScol }\mi{\rB_{i,j}}{\rProt_1 \mid \rZ, \rSrow = \Srow, \rScol = \Scol} \\
&\Eq{(3)} \frac1{b_r^2 \cdot \alpha^2}  \sum_{i \in [b_r] \setminus \rSrow, j \in [b_r] \setminus \rScol }\mi{\rB_{i,j}}{\rProt_1 \mid \rZ, \rSrow, \rScol}
\end{align*}
where 
$(1)$ holds by the definition of conditional mutual information since, conditioned on any $\rSrow = \Srow$ and $\rScol = \Scol$, the values of $\rrhor(b_r \cdot \alpha + k)$ and $\rrhoc(b_r \cdot \alpha + k)$ are uniformly chosen from $[b_r]\backslash \Srow$ and $[b_r]\backslash \Scol$ (each of size $b_r \cdot \alpha + 1$), respectively; 
$(2)$ holds because, conditioned on any choice of $\rSrow = \Srow$ and $\rScol = \Scol$,  the event $\rrhoc(b_r\cdot \alpha + k) = j, \rrhor(b_r \cdot \alpha + k) = i$ is independent of $(\rB_{i,j}, \rProt_1, \rZ)$; and 
$(3)$ holds by the linearity of expectation and the definition of conditional mutual information.

We re-index both the sets $[b_r] \setminus \rSrow$ and $[b_r] \setminus \rScol$ as $[b_r \cdot \alpha + 1]$ for ease of exposition just for this final part of the proof. 
We know that $\rB_{i,j}$ for any $i, j \in [b_r]$ are sampled independently of each other in distribution $\DD_r$.
Hence, using \Cref{prop:info-increase}, we have that 
\begin{align*}
 &\frac1{b_r^2 \cdot \alpha^2}  \sum_{i \in [b_r] \setminus \rSrow, j \in [b_r] \setminus \rScol }\mi{\rB_{i,j}}{\rProt_1 \mid \rZ, \rSrow, \rScol} \\
 &\leq  \frac1{b_r^2 \cdot \alpha^2}  \sum_{i \in [b_r \cdot \alpha + 1], j \in [b_r \cdot \alpha + 1] }\mi{\rB_{i,j}}{\rProt_1 \mid  \rB_{i, >j}, \rB_{>i},  \rZ, \rSrow, \rScol} \\
 &= \frac1{b_r^2 \cdot \alpha^2} \cdot \mi{\rB_{1,1}, \ldots, \rB_{1, b_r\cdot \alpha + 1},\rB_{2,1}, \ldots, \rB_{b_r\cdot\alpha+1, b_r\cdot\alpha+1}}{\rProt_1 \mid \rZ, \rSrow, \rScol} \tag{by chain rule of mutual information, \itfacts{chain-rule}} \\
 &\leq \frac{1}{b_r^2 \cdot \alpha^2} \cdot \mi{\rBb{r}}{\rProt_1 \mid \rZ, \rSrow, \rScol} \tag{by the chain rule (\itfacts{chain-rule}) and non-negativity of mutual information (\itfacts{info-zero})} \\
 &\leq \frac{s}{b_r^2 \cdot \alpha^2}. \tag{as $\mi{\rBb{r}}{\rProt_1 \mid \rZ, \rSrow, \rScol} \leq \en{\rProt_1} \leq s$ by \itfacts{uniform}}
\end{align*}
This proves the claim.
\end{proof}

We now use these claims to show that the information revealed by random variable $\rRa$ about the input $(\rAstar, \rBstar)$ is not too large in $\DDreal$.
\begin{claim}\label{clm:tvd-mi-bound}
\[
\mi{\rAspec_{\rkstar}, \rBspec_{\rkstar}}{\rRa} \leq \frac{s}{ b_r^2 \cdot \alpha^2}.
\]
\end{claim}

\begin{proof}
	First, by chain rule of mutual information (\itfacts{chain-rule}), we have that
	\begin{align*}
	\mi{\rAspec_{\rkstar}, \rBspec_{\rkstar}}{\rRa} &= \mi{\rBspec_{\rkstar}}{\rRa} + \mi{\rAspec_{\rkstar}}{\rRa \mid \rBspec_{\rkstar}}.
	\end{align*}
We can prove that the second term is zero as follows: Let $\rRcomp$ represent the independent source of public randomness used in the message compression in step \ref{step:ii}. Then,  
\begin{align*}
	&\mi{\rAspec_{\rkstar}}{\rRa \mid \rBspec_{\rkstar}} \\
	&= \mi{\rAspec_{\rkstar}}{\rProt_1, \rsigmar, \rsigmac, \rBcommon, \rkstar, \rBspec_{>\rkstar}, \rAspec_{<\rkstar} \mid \rBspec_{\rkstar}} \tag{by expanding $\rRa$}  \\
	&= \mi{\rAspec_{\rkstar}}{\rsigmar, \rsigmac, \rBcommon, \rkstar, \rBspec_{>\rkstar}, \rAspec_{<\rkstar} \mid \rBspec_{\rkstar}} + \mi{\rAspec_{\rkstar}}{\rProt_1 \mid \rBspec_{\rkstar}, \rsigmar, \rsigmac, \rBcommon, \rkstar, \rBspec_{>\rkstar}, \rAspec_{<\rkstar}} \tag{by  chain rule in \itfacts{chain-rule}}  \\
	&= 0 + \mi{\rAspec_{\rkstar}}{\rProt_1 \mid \rBspec_{\rkstar}, \rsigmar, \rsigmac, \rBcommon, \rkstar, \rBspec_{>\rkstar}, \rAspec_{<\rkstar}}. \tag{by~\itfacts{info-zero} and construction of $\DD_r$} \\
	&\leq \mi{\rAspec_{\rkstar}}{\rBrest, \rBspec_{<\rkstar}, \rRcomp \mid \rBspec_{\rkstar}, \rsigmar, \rsigmac, \rBcommon, \rkstar, \rBspec_{>\rkstar}, \rAspec_{<\rkstar}}. \tag{by data processing inequality \itfacts{data-processing} as $\rProt_1$ is fixed by $\rBspec, \rBcommon, \rBrest, \rRcomp$} \\
	&= \mi{\rAspec_{\rkstar}}{\rBrest, \rBspec_{<\rkstar} \mid \rBspec_{\rkstar}, \rsigmar, \rsigmac, \rBcommon, \rkstar, \rBspec_{>\rkstar}, \rAspec_{<\rkstar}}. \tag{since $\rRcomp$ is an entirely independent source of randomness} \\
	&= 0 \tag{by~\itfacts{info-zero} and construction of $\DD_r$}.
\end{align*}
Therefore, an upper bound on $\mi{\rBspec_{\rkstar}}{\rRa}$ is an upper bound on $\mi{\rAspec_{\rkstar}, \rBspec_{\rkstar}}{\rRa}$. We have, 
\begin{align*}
&\mi{\rBspec_{\rkstar}}{\rRa} \\
&\hspace{20pt}\leq \mi{\rBspec_{\rkstar}}{\rRa,\rBspec_{<\rkstar}} \tag{by the chain rule (\itfacts{chain-rule}) and non-negativity of mutual information (\itfacts{info-zero})} \\
&\hspace{20pt}= \mi{\rBspec_{\rkstar}}{\rProt_1, \rsigmar, \rsigmac, \rBcommon, \rkstar, \rBspec_{>\rkstar}, \rAspec_{<\rkstar},\rBspec_{<\rkstar}} \tag{by expanding $\rRa$} \\
&\hspace{20pt}= \mi{\rBspec_{\rkstar}}{\rsigmar, \rsigmac, \rBcommon, \rkstar, \rBspec_{>\rkstar}, \rAspec_{<\rkstar},\rBspec_{<\rkstar}} \\
&\hspace{100pt} + \mi{\rBspec_{\rkstar}}{\rProt_1 \mid \rsigmar, \rsigmac, \rBcommon, \rkstar, \rBspec_{>\rkstar}, \rAspec_{<\rkstar},\rBspec_{<\rkstar}}  \tag{by  chain rule in \itfacts{chain-rule}} \\
&\hspace{20pt}= 0 + \mi{\rBspec_{\rkstar}}{\rProt_1 \mid \rsigmar, \rsigmac, \rBcommon, \rkstar, \rBspec_{>\rkstar}, \rAspec_{<\rkstar},\rBspec_{<\rkstar}}  \tag{by~\itfacts{info-zero} and construction of $\DD_r$} \\
&\hspace{20pt}\leq \mi{\rBspec_{\rkstar}}{\rProt_1 \mid \rsigmar, \rsigmac, \rBcommon, \rkstar, \rBspec_{>\rkstar},\rBspec_{<\rkstar}} \tag{removing $\rAspec_{<\rkstar}$ in the conditioning by \Cref{prop:info-decrease} and \Cref{clm:tvd-bound-help-1}} \\
&\hspace{20pt}= \frac1{k_r}\cdot\sum_{k \in [k_r]}  \mi{\rBspec_{k}}{\rProt_1 \mid \rsigmar, \rsigmac, \rBcommon, \rkstar=k, \rBspec_{>k}, \rBspec_{<k}} \tag{by the definition of conditional mutual information} \\
&\hspace{20pt}= \frac1{k_r}\cdot\sum_{k \in [k_r]}  \mi{\rBspec_{k}}{\rProt_1 \mid \rsigmar, \rsigmac, \rBcommon, \rBspec_{>k}, \rBspec_{<k}}  \tag{as event $\rkstar = k$ is independent of $(\rBspec, \rBcommon, \rProt_1, \rsigmar, \rsigmac)$ (see proof of~\Cref{clm:ic-lower-1})} \\
&\hspace{20pt}\leq \frac{s}{b_r^2 \cdot \alpha^2}, 
\end{align*}
by applying \Cref{clm:dir-sum-B} to bound each term, completing the proof. 
\end{proof}

We are now ready to show that $\DDfake$ and $\DDreal$ are statistically close. 
\begin{lemma}\label{lem:tvd-fake-real}
	\[
	\tvd{\DDreal}{\DDfake} \leq \sqrt{\frac{s}{2 \cdot b_r^2 \cdot \alpha^2}}.
	\]
\end{lemma}

\begin{proof}
Firstly, we use weak chain rule over total variation distance (\Cref{fact:tvd-chain-rule}) using the definitions of distributions $\DDfake$ and $\DDreal$ established in \Cref{clm:DDrealfake-input} to obtain
\begin{align*}
\tvd{\DDreal}{\DDfake} &\leq \Exp_{R \sim \rRa} \tvd{\distribution{\rAspec_{\rkstar}, \rBspec_{\rkstar} \mid \rRa = R}}{\distribution{\rAspec_{\rkstar}, \rBspec_{\rkstar}}} \\
&\leq \Exp_{R \sim \rRa}\sqrt{\frac12 \cdot \kl{\distribution{\rAspec_{\rkstar}, \rBspec_{\rkstar} \mid \rRa=R}}{\distribution{\rAspec_{\rkstar}, \rBspec_{\rkstar}}}} \tag{by Pinsker's inequality in \Cref{fact:pinskers}} \\
&\leq \frac1{\sqrt{2}} \cdot \sqrt{\Exp_{R \sim \rRa}  \kl{\distribution{\rAspec_{\rkstar}, \rBspec_{\rkstar}\mid \rRa=R}}{\distribution{\rAspec_{\rkstar}, \rBspec_{\rkstar}}}} \tag{by Jensen's inequality as the function $\sqrt{\cdot}$ is concave}\\
&= \frac1{\sqrt{2}} \cdot \sqrt{\mi{\rAspec_{\rkstar}, \rBspec_{\rkstar}}{\rRa}} \tag{by the relation between mutual information and KL-divergence in \Cref{fact:kl-info}} \\
&\leq \sqrt{\frac{s}{2 \cdot b_r^2 \cdot \alpha^2}} \tag{by \Cref{clm:tvd-mi-bound}},
\end{align*}
proving the lemma.
\end{proof}

Finally, we complete the proof of the main round-elimination argument.

\begin{proof}[Proof of \Cref{lem:round-elim}]
	By construction, we know that the total communication of protocol $\prot_{r-1}$ is upper bounded by the communication of protocol $\protnum{2}_r$ (since the messages of the former are a subset of the ones for the latter). 
	Using \Cref{cor:prot2-details}, we get that the communication of $\prot_{r-1}$ is at most 
	\[
		\cc{\prot_{r-1}} \leq \jpyconst \cdot (r/\eps) \cdot \paren{\frac{s}{k_r} + r}.
	\]
	Furthermore, by construction, protocol $\prot_{r-1}$ is an $(r-1)$-round protocol where Alice communicates first.
	
	We know, again from \Cref{cor:prot2-details}, that the probability of success of $\protnum{2}_r$ is at least $\delta - \eps$ when the simulated input $(\alicepart{r}, \bobpart{r})$ is sampled from $\DD_r$, i.e., the random variables follow distribution $\DDreal$. 
	However, in $\prot_{r-1}$, due to the fact that we sample the first round message at random, the simulated input is only statistically close to being sampled from $\DD_r$, i.e., the random variables follow distribution $\DDfake$.
	Therefore, the probability of success of $\prot_{r-1}$, using~\Cref{fact:tvd-small}, is bounded as follows:
	\begin{align*}
	\suc{\prot_{r-1}} &\geq \Prob\paren{\protnum{2}_r \textnormal{~succeeds on input $\DDreal$}} - \tvd{\DDreal}{\DDfake} \\
	&\geq \delta-\eps -\sqrt{\frac{s}{2\cdot b_r^2 \cdot \alpha^2}}.
	\end{align*}

Lastly, to get a deterministic protocol out of $\prot_{r-1}$, we use an averaging argument and fix the random bits of $\prot_{r-1}$ so that we have the same performance guarantee over the input distribution $\DD_{r-1}(n_{r-1}, \alpha)$. This 
gives us the desired $(r-1)$-round deterministic protocol for $\ahm_{r-1}(n_r, \alpha)$ when the players' inputs are sampled from $\DD_{r-1}(n_{r-1}, \alpha)$, and the search sequence is uniformly random. 
\end{proof}


\subsection{Reduction to Bipartite Matching in Dynamic Streams} \label{subsec:red-to-graph}
In this section, we prove the connection between maximum bipartite matching in the dynamic streaming model and the Augmented Hidden Matrices ($\ahm_r$) problem, which is presented as \Cref{lem:redToMatching}.
We repeat this lemma for the convenience of the reader here. 

\begin{restate}[\Cref{lem:redToMatching}]
		For any sufficiently large $n_r \in \IN$, integers $p,s \geq 1$, and number $\apx \geq 1$, let $\alg$ by any $p$-pass $s$-space
		dynamic streaming algorithm that computes a $\apx$-approximate maximum matching on a $(4n_r)$-vertex bipartite graphs with probability of success at least $1-1/\poly{(n)}$.
		Then, for $r = 2p-1$ and $\alpha = 1/(4 \cdot \apx \cdot r)$, there exists an $r$-round protocol $\prot$ for $\ahm_{r}(n_r, \alpha)$ with 
		probability of success 
		\[
		\suc{\prot} \geq \frac{1}{2} \cdot \left(1 + \frac{1}{3\apx}\right)
		\]
		and
		communication cost 
		\[
		\cc{\prot} \leq r \cdot s + O(r \cdot n_r \log{(n_r)}). 
		\]
\end{restate}

The key to proving~\Cref{lem:redToMatching} is to construct a bipartite graph $G = (L \cup R, \Eadd \backslash \Edel)$, which is defined by edge insertions $\Eadd$ and deletions $\Edel$, 
from any input instance for $\ahm_{r}$ with the following property. An $O(1)$-approximate maximum bipartite matching in $G$ can be used to solve $\ahm_{r}$ on the
search sequence with a non-trivial advantage over randomly guessing the answer.

\subsubsection{The Bipartite Graph Construction using $\ahm_{r}$}

Let us first recall the important aspects of $\ahm_r$. 
It is defined recursively using $b_r^2$ many $\ahm_{r-1}$ sub-instances, which are denoted $(\Aa{r-1}_{i,j}, \Bb{r-1}_{i,j})$ for $i,j \in [b_r]$ (see \Cref{fig:insdel_harddist} in \Cref{subsec:ahm-def}).
An instance of $\ahm_r$ identifies $k_r$ many special sub-instances $(\Aspec,\Bspec)$. 
Each of these has its own $k_{r-1}$ many special (sub-)sub-instances and, continuing this, ultimately corresponds to the $k_r \cdot k_{r-1} \cdot \ldots \cdot k_1$ many special base instances, where each one corresponds to a single bit as they are instances of $\ahm_0$. 
Then, given the uniform random search sequence, which is provided at the end of a protocol, solving $\ahm_r$ requires solving a uniform random special base instance.

We begin by showing how the base instances (or bits) are represented in our graph.
We will use the following basic bipartite graph construction on four vertices, called a \textbf{bit graph}, to encode bits in our construction of $G$. 

\bigskip

\begin{Definition}[Bit Graph]\label{def:bit-graph}
	Given a bit $x \in \{0,1 \}$ and a graph $G = (L \cup R, E)$ with $L = \{\ell_1, \ell_2\}$ and $R = \{r_1, r_2\}$, graph $G$ is said to be the bit graph of $x$ if
	\begin{itemize}
		\item $ E = \{(\ell^1, r^1), (\ell^2, r^2)\}$ when $x = 0$, and $E = \{(\ell^1, r^2), (\ell^2, r^1)\}$ when $x = 1$.
	\end{itemize}
	The bit $x$ can be identified by either of the edges in $E$. See \Cref{fig:bitgraph} for an illustration.
\end{Definition}

\bigskip

\begin{figure}[h!]
	\begin{subfigure}[t]{0.49\textwidth}
		\centering
		\begin{tikzpicture}[yscale=1.5]
    \tikzset{vertexgroup/.style={circle, draw=black, minimum size=17pt,inner sep=0pt, font=\tiny}}
    \tikzstyle{edgeg} = [draw=black,-, line width=3]
    \tikzstyle{edge} = [draw=black,-]
    \tikzstyle{vertex}=[circle,fill=black,minimum size=4pt,inner sep=0pt, font=\tiny]
    \tikzstyle{vertexnone}=[circle,fill=none,minimum size=5pt,inner sep=0pt, font=\tiny]
    \tikzstyle{plaintext} = [fill=none, font=\small]

      \foreach \i in {1,2}{
        \node[vertex] (l_{\i}) at (1,\i) {};
        \node[vertex] (r_{\i}) at (4, \i) {};
      }

      \node[plaintext, anchor=east] () at (0.75,1) {$\ell^2$};
      \node[plaintext, anchor=east] () at (0.75,2) {$\ell^1$};

      \node[plaintext, anchor=west] () at (4.25,1) {$r^2$};
      \node[plaintext, anchor=west] () at (4.25,2) {$r^1$};   

      \path[edge](l_{1}) -- (r_{1});
      \path[edge](l_{2}) -- (r_{2});

\end{tikzpicture}
\vspace{0.5em}
		\caption{When the bit $x = 0$.}
	\end{subfigure}
	\begin{subfigure}[t]{0.49\textwidth}
		\centering
		\begin{tikzpicture}[yscale=1.5]
    \tikzset{vertexgroup/.style={circle, draw=black, minimum size=17pt,inner sep=0pt, font=\tiny}}
    \tikzstyle{edgeg} = [draw=black,-, line width=3]
    \tikzstyle{edge} = [draw=black,-]
    \tikzstyle{vertex}=[circle,fill=black,minimum size=4pt,inner sep=0pt, font=\tiny]
    \tikzstyle{vertexnone}=[circle,fill=none,minimum size=5pt,inner sep=0pt, font=\tiny]
    \tikzstyle{plaintext} = [fill=none, font=\small]

      \foreach \i in {1,2}{
        \node[vertex] (l_{\i}) at (1,\i) {};
        \node[vertex] (r_{\i}) at (4, \i) {};
      }

      \node[plaintext, anchor=east] () at (0.75,1) {$\ell^2$};
      \node[plaintext, anchor=east] () at (0.75,2) {$\ell^1$};

      \node[plaintext, anchor=west] () at (4.25,1) {$r^2$};
      \node[plaintext, anchor=west] () at (4.25,2) {$r^1$};  

      \path[edge](l_{1}) -- (r_{2});
      \path[edge](l_{2}) -- (r_{1});

\end{tikzpicture}
\vspace{0.5em}
		\caption{When the bit $x = 1$.}
	\end{subfigure}
	\caption{An illustration of the bit graph on vertices corresponding to a bit $x \in \{0,1\}$.} \label{fig:bitgraph}
\end{figure}
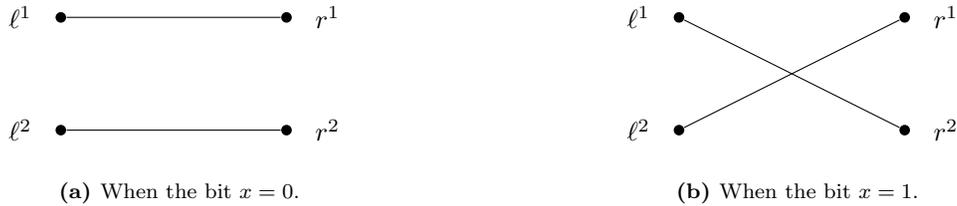

\bigskip

With that, given an instance $(\alicepart{r}, \bobpart{r})$ of $\ahm_r(n_r, \alpha)$, we give the following recursive procedure that separately encodes $\alicepart{r}$ and $\bobpart{r}$ as the edges of a bipartite graph on the vertex sets $L$ and $R$, where,
\begin{align*}
	L = \{\leftv{t}{1}, \leftv{t}{2} \mid t \in [n_r]\} \qquad \text{and} \qquad R = \{\rightv{t}{1}, \rightv{t}{2} \mid t \in [n_r]\}. \label{eq:LandR}
\end{align*}
For ease of notation, we let 
\[
L[i] = \set{\leftv{t}{1}, \leftv{t}{2} \mid (i-1) \cdot n_{r-1} < t \leq i \cdot n_{r-1} } \subseteq L
\]
 for $i \in [b_r]$, and we let $R[j] \subseteq R$ for $j \in [b_r]$ be defined similarly. 
The edges we add corresponding to the sub-instance $(\Aa{r-1}_{i,j}, \Bb{r-1}_{i,j})$ will be only across vertices in $L[i] $ and $R[j]$. 

We also use $\allone(t)$ and $\allzero(t)$ to denote the $t \times t$ matrix of all ones and all zeroes, respectively.
The construction of the graph is as follows.

\begin{ourbox}
	\vspace{0.5em}
	{$\edgeset(r, Z, L, R)$ where $Z \in \{\alicepart{r}, \bobpart{r}, \allone(n_r), \allzero(n_r) \}$ with $\card{L} = \card{R} = 2n_r$}:
	\begin{itemize}
		\item \textbf{For $\mathbf{r=0}$.} 
		\begin{itemize}
			\item When $Z \in \{0,1\}$, return the edges of the bit graph corresponding to $Z$.
			\item When $Z  = \emptyset$, return the empty set of edges.
		\end{itemize}
		\item \textbf{For $\mathbf{r \geq 1}$.} 
		\begin{itemize}
			\item When $Z \in \{\allone(n_r), \allzero(n_r)\}$, return the edges \[\bigcup_{i,j \in [b_r]} \edgeset(r-1, Z', L[i], R[j])\]
			where $Z'$ is the $n_{r-1} \times n_{r-1}$ matrix with all its entries the same as the entries in $Z$.
			\item When ${Z = \alicepart{r}}$, return the edges \[\bigcup_{i,j \in [b_r]} \edgeset(r-1, \Bb{r-1}_{i,j}, L[i], R[j]).\]
			\item When ${Z = \bobpart{r}}$, return the edges\footnote{This procedure ignores $\Acommon = (\alicepart{r-1}_{i,j} : b_r \cdot \alpha < \sigmar(i) \neq \sigmac(j) \leq b_r)$.}
			\begin{align*}
				&\Big(\bigcup_{\substack{i,j \in [b_r] \\ b_r \cdot \alpha <\sigmar(i) \neq \sigmac(j) \leq b_r}}	
				\edgeset(r-1, \Bb{r-1}_{i,j}, L[i], R[j])\Big) \\
				& \cup \Big(\bigcup_{\substack{i,j \in [b_r] \\ b_r \cdot \alpha <\sigmar(i) = \sigmac(j) \leq b_r}} \edgeset(r-1, \Aa{r-1}_{i,j}, L[i], R[j])\Big).
			\end{align*}
			If $r$ is even, in addition to the above, return the edges
			\begin{align*}
				\bigcup_{\substack{i,j \in [b_r] \\ \sigmar(i) \leq b_r\cdot \alpha \text{~or~} \sigmac(j) \leq b_r\cdot \alpha}} \big(\edgeset(0, \allone(n_{r-1}), L[i], R[j]) \cup  \edgeset(0, \allzero(n_{r-1}), L[i], R[j]) \big).
			\end{align*}
		\end{itemize}
	\end{itemize}
See \Cref{fig:edgesfromAHM} for an illustration.
\end{ourbox}
\vspace{0em}

Note that this procedure does not yet define the graph $G$. However, it allows Alice and Bob to construct edges (insertions or deletions) from their respective inputs, which we now show can be done without any communication.

\begin{claim}\label{clm:find-sets}
	For any $r \geq 0$, given an instance $(\alicepart{r}, \bobpart{r})$ of $\ahm_r(n_r, \alpha)$, Alice and Bob can compute $\edgeset(r, \Aa{r}, L, R)$ and $\edgeset(r, \Bb{r}, L, R)$, respectively, without any communication.
\end{claim}

\begin{proof}
	The proof of the claim is by induction on $r$ where the base case when $r = 0$ is trivial, i.e., both $\edgeset(0, \Aa{0}, L, R)$ and $\edgeset(0, \Bb{0}, L, R)$ can be constructed without communication.
	
	For any $r \geq 1$, the player holding $\Aa{r}$ has access to all $\Bb{r-1}_{i,j}$ for $i, j \in [b_r]$. Thus, the player can construct the required set of edges $\edgeset(r-1, \Bb{r-1}_{i,j}, L[i], R[j])$ for all $i,j \in [b_r]$ without any communication by the induction hypothesis. 
	The player holding $\Bb{r}$ has access to permutations $\sigmar, \sigmac$ and the following:
	\begin{itemize}
		\item $\Bb{r-1}_{i,j}$ for $b_r \cdot \alpha <\sigmar(i) \neq \sigmac(j) \leq b_r$ and thus the player can construct construct the required 
		\[
		\edgeset(r-1, \Bb{r-1}_{i,j}, L[i], R[j])
		\]
		without any communication by the induction hypothesis; and
		\item $\Aa{r-1}_{i,j}$ when $b_r \cdot \alpha <\sigmar(i) = \sigmac(j) \leq b_r$ and thus the player can construct 
		\[
		\edgeset(r-1, \Aa{r-1}_{i,j}, L[i], R[j])
		\]
		 without any communication, again, by the induction hypothesis.
	\end{itemize}
Additionally, when $r$ is even, this player constructs 
\[
\edgeset(r-1, \allone(n_{r-1}), L[i], R[j])\quad \text{and} \quad \edgeset(r-1, \allzero(n_{r-1}), L[i], R[j])
\]
 when $\sigmar(i) \leq b_r \cdot \alpha$ or $\sigmac(j) \leq b_r \cdot \alpha$, which is trivial and does not need any communication since it only requires the knowledge of $\sigmar, \sigmac$.
\end{proof}

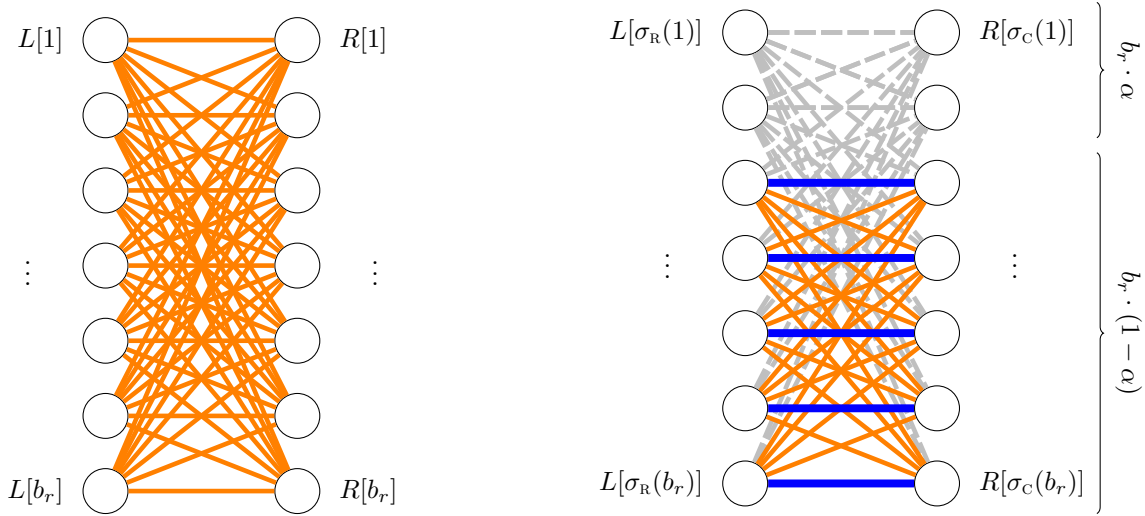
\begin{figure}[t!]
	\begin{subfigure}[t]{0.35\textwidth}
		\centering
		\begin{tikzpicture}[xscale=0.85]
  \tikzset{vertexgroup/.style={circle, draw=black, minimum size=17pt,inner sep=0pt, font=\tiny}}
  \tikzstyle{edgeg} = [draw=black,-, line width=1.75]
  \tikzstyle{plaintext} = [fill=none, font=\small]


    \foreach \i in {1, ..., 7}{
      \node[vertexgroup] (l_{i\i}) at (1,\i) {};
      \node[vertexgroup] (r_{i\i}) at (4, \i) {};
    }

    \node[plaintext, anchor=east] () at (0.5,7) {$L[1]$};
    \node[plaintext, anchor=east] () at (0,4) {$\vdots$};
    \node[plaintext, anchor=east] () at (0.5,1) {$L[b_r]$};

    \node[plaintext, anchor=west] () at (4.5,7) {$R[1]$};
    \node[plaintext, anchor=west] () at (5,4) {$\vdots$};
    \node[plaintext, anchor=west] () at (4.5,1) {$R[b_r]$};


    \foreach \i in {1, ..., 7}{
      \foreach \j in {1, ..., 7}{
        \path[edgeg, orange](l_{i\i}) -- (r_{i\j});
      }
    }



\end{tikzpicture}
\vspace{0.5em}
		\caption{An illustration of Alice's edges corresponding to $\Aa{r}$. An orange edge between $L[i]$ and $R[j]$ for each $i, j \in [b_r]$ represents the set of edges corresponding to the input $\Bb{r-1}_{i,j}$ held by Alice.} \label{fig:edgesfromAHM_A}
	\end{subfigure}
	\hfill
	\begin{subfigure}[t]{0.62\textwidth}
		\centering
		\begin{tikzpicture}[xscale=0.85]
  \tikzset{vertexgroup/.style={circle, draw=black, minimum size=17pt,inner sep=0pt, font=\tiny}}
  \tikzstyle{edgeg} = [draw=black,-, line width=1.75]
  \tikzstyle{plaintext} = [fill=none, font=\small]

    \foreach \i in {1, ..., 7}{
      \node[vertexgroup] (l_{i\i}) at (1,\i) {};
      \node[vertexgroup] (r_{i\i}) at (4, \i) {};
    }

    \node[plaintext, anchor=east] () at (0.5,7) {$L[\sigmar(1)]$};
    \node[plaintext, anchor=east] () at (0,4) {$\vdots$};
    \node[plaintext, anchor=east] () at (0.5,1) {$L[\sigmar(b_r)]$};

    \node[plaintext, anchor=west] () at (4.5,7) {$R[\sigmac(1)]$};
    \node[plaintext, anchor=west] () at (5,4) {$\vdots$};
    \node[plaintext, anchor=west] () at (4.5,1) {$R[\sigmac(b_r)]$};

    \draw[decorate, decoration={brace}] (6.5,5.4) -- (6.5,0.6);
    \node[plaintext, align=center, rotate=-90] () at (7, 3){$b_r \cdot (1-\alpha)$};
    \draw[decorate, decoration={brace}] (6.5,7.4) -- (6.5,5.6);
    \node[plaintext, align=center, rotate=-90] () at (7, 6.5){$b_r \cdot \alpha$};
    \node[plaintext, align=center, rotate=-90, white] () at (-2, 6.5){$b_r \cdot \alpha$};

    \foreach \i in {6, 7}{
      \foreach \j in {1, ..., 7}{
        \path[edgeg, gray!50, dash pattern = on 8pt off 2pt](l_{i\i}) -- (r_{i\j});
        \path[edgeg, gray!50, dash pattern = on 8pt off 2pt](l_{i\j}) -- (r_{i\i});
      }
    }

    \foreach \i in {1, ..., 5}{
      \foreach \j in {1, ..., 5}
      \path[edgeg, orange](l_{i\i}) -- (r_{i\j});
    }

    \foreach \i in {1, ..., 5}{
      \path[edgeg, blue, line width=3](l_{i\i}) -- (r_{i\i});
    }


\end{tikzpicture}
\vspace{0.5em}
		\caption{An illustration of Bob's edges corresponding to $\Bb{r}$. The vertex groups are represented after the permutations $\sigmar,\sigmac$ held by Bob. Each orange edge represents the edges corresponding to each off-diagonal sub-instances $\Bcommon$ held by Bob. Each blue edge represents the edges corresponding to the special sub-instances $\Aspec$ held by Bob. Each dashed gray edge represents the set of no edges when $r$ is odd and the set of all possible edges when $r$ is even.} \label{fig:edgesfromAHM_B}
	\end{subfigure}
	\caption{Illustrations of Alice's edges $\edgeset(r, \Aa{r}, L, R)$ (\Cref{fig:edgesfromAHM_A}) and Bob's edges $\edgeset(r, \Bb{r}, L, R)$ (\Cref{fig:edgesfromAHM_B}). In both illustrations, the nodes on the left and right corresponds to vertex groups $L[i]$ and $R[j]$ for each $i,j \in [b_r]$, respectively.
	An edge between vertex groups $L[i]$ and $R[j]$ represents the set of edges corresponding to the instance $(\Aa{r-1}_{i,j}, \Bb{r-1}_{i,j})$.
	} \label{fig:edgesfromAHM}
\end{figure}

Using the above recursive definition, we obtain the main procedure for constructing a bipartite graph $G$  from any instance $(\Aa{r}, \Bb{r})$ of $\ahm_r(n_r, \alpha)$.

One player will be solely responsible for adding edges to $G$, and the other player only deletes edges from $G$. The particular roles of Alice and Bob depends on the parity of $r$. 

\begin{ourbox}
\vspace{0.5em}
$\graphed(\Aa{r}, \Bb{r})$ for instance $(\Aa{r}, \Bb{r})$ of $\ahm_r(n_r, \alpha)$:
\begin{itemize}
	\item When $r$ is \textbf{odd}, Alice inserts the edges $\Eadd = \edgeset(r, \Aa{r}, L, R)$ and Bob deletes the edges $\Edel = \edgeset(r, \Bb{r}, L, R)$;
	\item When $r$ is \textbf{even}, Bob inserts the edges $\Eadd = \edgeset(r, \Bb{r}, L, R)$ and Alice deletes the edges $\Edel = \edgeset(r, \Aa{r}, L, R)$. 
\end{itemize}
Return graph $G = (L \cup R, \Eadd \backslash \Edel)$. See \Cref{fig:hardgraph} for an illustration of $G$.
\end{ourbox}

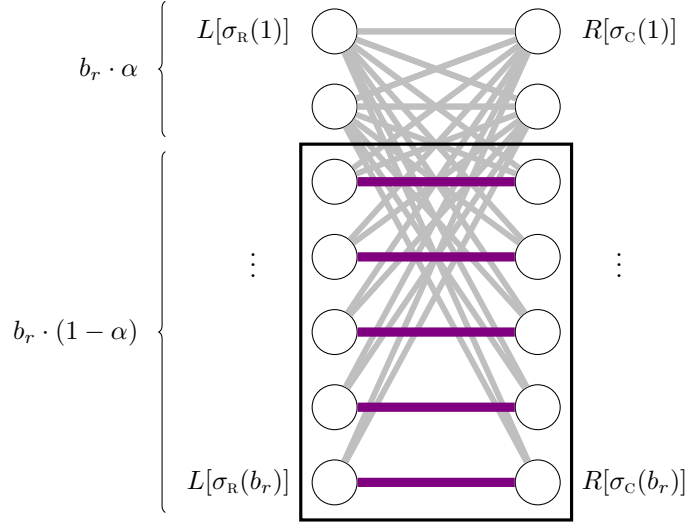
\begin{figure}[h!]
	\centering
	\begin{tikzpicture}[xscale=0.9]
  \tikzset{vertexgroup/.style={circle, draw=black, minimum size=17pt,inner sep=0pt, font=\tiny}}
  \tikzstyle{edgeg} = [draw=black,-, line width=2.25]
  \tikzstyle{plaintext} = [fill=none, font=\small]

    \foreach \i in {1, ..., 7}{
      \node[vertexgroup] (l_{i\i}) at (1,\i) {};
      \node[vertexgroup] (r_{i\i}) at (4, \i) {};
    }

    \node[plaintext, anchor=east] () at (0.5,7) {$L[\sigmar(1)]$};
    \node[plaintext, anchor=east] () at (0,4) {$\vdots$};
    \node[plaintext, anchor=east] () at (0.5,1) {$L[\sigmar(b_r)]$};

    \node[plaintext, anchor=west] () at (4.5,7) {$R[\sigmac(1)]$};
    \node[plaintext, anchor=west] () at (5,4) {$\vdots$};
    \node[plaintext, anchor=west] () at (4.5,1) {$R[\sigmac(b_r)]$};

    \draw[decorate, decoration={brace}] (-1.5,0.6) -- (-1.5,5.4);
    \node[plaintext, anchor=east] () at (-1.75, 3){$b_r \cdot (1-\alpha)$};
    \node[plaintext, anchor=west, white] () at (6.75, 3){$b_r \cdot (1-\alpha)$}; 
    \draw[decorate, decoration={brace}] (-1.5,5.6) -- (-1.5,7.4);
    \node[plaintext, anchor=east] () at (-1.75, 6.5){$b_r \cdot \alpha$};

    \foreach \i in {6, 7}{
      \foreach \j in {1, ..., 7}{
        \path[edgeg, gray!50](l_{i\i}) -- (r_{i\j});
        \path[edgeg, gray!50](l_{i\j}) -- (r_{i\i});
      }
    }

    \foreach \i in {1, ..., 5}{
      \path[edgeg, blue!50!red, line width=3.5](l_{i\i}) -- (r_{i\i});
    }

    \draw[very thick] (0.5,5.5) rectangle (4.5,0.5);

\end{tikzpicture}
\vspace{0.5em}
	\caption{An illustration of a graph $G = \graphed(\Aa{r}, \Bb{r})$ constructed from an instance $(\Aa{r}, \Bb{r})$ of $\ahm_r(n_r, \alpha)$. 
	The nodes on the left and right corresponds to vertex groups $L[i]$ and $R[j]$ for each $i,j \in [b_r]$ and are ordered after applying the permutations $\sigmar$ and $\sigmac$, respectively.
	An edge between vertex groups $L[i]$ and $R[j]$ represents the set of edges corresponding to the instance $(\Aa{r-1}_{i,j}, \Bb{r-1}_{i,j})$, where a lack of edge implies that all edges inserted have also been deleted, i.e., the off-diagonal sub-instances $(\Acommon, \Bcommon)$.
	Each purple edge represents the edges that correspond to each special sub-instance in $(\Aspec,\Bspec)$. The remaining thick gray edges represent the edges corresponding to sub-instances in $\Brest$. 
	The box represents the subgraph of $G$ induced by the vertices $\Lspec \cup \Rspec$. See \Cref{fig:insdel_harddist} for a comparison to $\ahm_r$.
	} \label{fig:hardgraph}
\end{figure}

Now, let us show that the instances we create are indeed valid graphs. 
\begin{claim}
	For any $r \geq 0$, the graph $G = (L \cup R, \Eadd \backslash \Edel) = \graphed(\Aa{r}, \Bb{r})$ is a valid graph, namely, $\Edel \subseteq \Eadd$.
\end{claim}

\begin{proof}
	By construction, the player that inserts the edges and the player that deletes the edges depends on the parity of $r$.
	When $r$ is odd, $\Eadd = \edgeset(r, \Aa{r}, L, R)$ and $\Edel = \edgeset(r, \Bb{r}, L, R)$, and when $r$ is even, the roles are reversed.
	Hence, we prove the claim by induction on $r$ where, when $r$ is odd, \[
	\edgeset(r, \Bb{r}, L, R) \subset \edgeset(r, \Aa{r}, L, R),
	\], and when $r$ is even, the  vice-versa holds. 
	
	When $r = 0$ (even case), there are no edge deletions since $\edgeset(0, \Aa{0}, L, R) = \emptyset$, so the claim vacuously holds. 
	When $r \geq 1$, we consider the cases where $r$ is odd and even separately.
	\subparagraph*{Odd $r$.}
	We consider the following groupings of the edges separately:
	\begin{itemize}
		\item  When $b_r \cdot \alpha <\sigmar(i) \neq \sigmac(j) \leq b_r$, the edge set $\edgeset(r-1, \Bb{r-1}_{i,j}, L[i], R[j])$ held by Bob is directly contained in the edges $\edgeset(r, \Aa{r}, L, R)$ held by Alice.
		\item When $b_r \cdot \alpha <\sigmar(i) = \sigmac(j) \leq b_r$, Bob holds the edge set $\edgeset(r-1, \Aa{r-1}_{i,j}, L[i], R[j])$ and Alice holds the edge set $\edgeset(r-1, \Bb{r-1}_{i,j}, L[i], R[j])$.
		Immediately by the induction hypothesis, since $r-1$ is even, we have that $\edgeset(r-1, \Aa{r-1}_{i,j}, L[i], R[j]) \subseteq \edgeset(r-1, \Bb{r-1}_{i,j}, L[i], R[j])$.
	\end{itemize}

	\subparagraph*{Even $r$.} 
Here, we have three different cases:
	\begin{itemize}
		\item When $b_r \cdot \alpha < \sigmar(i) \neq \sigmac(j) \leq b_r$, the edges $\edgeset(r-1, \Bb{r-1}_{i,j}, L[i], R[j])$ held by Alice are directly contained in the edges $\edgeset(r, \Bb{r}, L, R)$ by Bob. 
		\item When $b_r \cdot \alpha < \sigmar(i) = \sigmac(j) \leq b_r$, by the induction hypothesis as $r-1$ is odd, we have that $\edgeset(r-1, \Bb{r-1}_{i,j}, L[i], R[j]) \subseteq \edgeset(r-1, \Aa{r-1}_{i,j}, L[i], R[j])$ and are thus contained in $\edgeset(r, \Bb{r}, L, R)$.
		\item When $\sigmar(i) \leq b_r \cdot \alpha$ or $\sigmac(j) \leq b_r \cdot \alpha$, the edges $\edgeset(r-1, \Bb{r-1}_{i,j}, L[i], R[j])$ are trivially a subset of $\edgeset(r-1, \allone(n_{r-1}), L[i], R[j]) \cup \edgeset(r-1, \allzero(n_{r-1}), L[i], R[j])$, which are all included in $\edgeset(r, \Bb{r}, L, R)$ since $r$ is even.
	\end{itemize}

Overall, we have that $\Edel \subseteq \Eadd$ for any choice of $r \geq 0$ and thus $G = (L \cup R, \Eadd \backslash \Edel) = \graphed(\Aa{r}, \Bb{r})$ is a valid graph.
\end{proof}

With that, we have that the players can take an instance $(\Aa{r}, \Bb{r})$ of $\ahm_r$ and construct a valid graph $G = (L \cup R, \Eadd \backslash \Edel)$ without any communication. In the next section, we will show some structural properties of $G$ that relate it to $(\Aa{r}, \Bb{r})$.

\subsubsection{Properties of the Constructed Graph}
We need some notation before we begin proving the desirable properties of our construction. 
We define the following sets of special vertices: For $i, j \in [b_r]$,
\begin{align*}
	\Lspec_k := L[i] \text{~when $\sigmar(i) = b_r \cdot \alpha  + k$}\qquad \text{and} \qquad
	\Rspec_k := R[j] \text{~when $\sigmac(j) = b_r \cdot \alpha  + k$}
\end{align*}
for all $k \in [k_r]$, and
\begin{align*}
	\Lspec := \bigcup_{k \in [k_r]} \Lspec_k \qquad \text{and} \qquad	\Rspec := \bigcup_{k \in [k_r]} \Rspec_k.
\end{align*}
See \Cref{fig:hardgraph} for an illustration of the special vertices.
In the following claim, we show that the special vertices do, in fact, correspond to the special sub-instances $(\Aspec, \Bspec)$ of $(\alicepart{r}, \bobpart{r})$.

\begin{claim}\label{clm:specind}
	For $r \geq 1$, the subgraph of $G = \graphed( \alicepart{r}, \bobpart{r})$ induced by the vertices $\Lspec \cup \Rspec$ is exactly the vertex-disjoint union of graphs  $\graphed( \Aspec_k, \Bspec_k)$ for all $k \in [k_r]$.
\end{claim}
\begin{proof}
	We only prove the claim when $r$ is odd, as the case when $r$ is even follows similarly. 
	
	Observe first that the only possible edges between $\Lspec$ and $\Rspec$ are the edges between $L[i] $ and $R[j]$ corresponding to $(\Aa{r-1}_{i,j}, \Bb{r-1}_{i,j})$ for any $i , j \in [b_r]$ with $b_r \cdot \alpha <\sigmar(i), \sigmac(j) \leq b_r$.
	Considering only these edges, we split our analysis into two cases:
	\begin{itemize}
		\item 
		When $b_r \cdot \alpha < \sigmar(i) \neq \sigmac(j) \leq b_r$, Alice adds the edges corresponding to $\Bb{r-1}_{i,j}$,
		which is exactly the set of edges deleted by Bob. 
		Hence, none of these edges are present in $G$. 
		\item
		When $\sigmar(i) = \sigmac(j) = b_r \cdot \alpha + k$ for $k \in [k_r]$, Alice adds the edges corresponding to $\Bb{r-1}_{i,j}$,
		whereas Bob deletes the edges corresponding to $\Aa{r-1}_{i,j}$.
		As $r-1$ is even, these edge insertions and deletions exactly correspond to $\graphed( \Aa{r-1}_{i,j}, \Bb{r-1}_{i,j})$.
	It then follows that the subgraph of $G$ induced by $L[i] \cup R[j] = \Lspec_k \cup \Rspec_k$ is an exact copy of $\graphed( \Aspec_k, \Bspec_k)$.
	\end{itemize}
	
	Finally, for any $k_1, k_2 \in [k_r]$, $\Lspec_{k_1} \cup \Rspec_{k_1} $ is disjoint from $\Lspec_{k_2} \cup \Rspec_{k_2}$. Hence, the subgraph induced by $\Lspec \cup \Rspec$ is a vertex-disjoint union of the graphs $\graphed(\Aspec_k, \Bspec_k)$ for $k \in [k_r]$, proving the claim.
\end{proof}

To highlight another key structural property in our construction of $G$, recall from \Cref{obs:numspecial} that any instance $(\Aa{r}, \Bb{r})$ of $\ahm_r$ defines $k_r \cdot k_{r-1} \cdot \ldots \cdot k_1$ many special base instances.
We identify these special base instances and their corresponding sets of special base vertices in $L$ and $R$ using the tuples in $ [k_r] \times [k_{r-1}] \times \ldots \times [k_1]$. Formally, we have the following notation:

\paragraph{Notation.}
Any tuple $s = (s_r, s_{r-1}, \ldots, s_1)$ identifies the special sub-instance $(\Aspec_{s_r}, \Bspec_{s_r})$ of $(\alicepart{r}, \bobpart{r})$, then the special sub-instance $(\Aspec_{s_{r-1}}, \Bspec_{s_{r-1}})$ inside $(\Aspec_{s_r}, \Bspec_{s_r})$, and so on until finally pointing to a single special base instance $(\Aspec_{s_{1}}, \Bspec_{s_{1}})$.
We denote this special base instance as $(\Abase_s, \Bbase_s)$. 
Furthermore, we use $(\Abase, \Bbase)$ to denote the collection of special base instances $(\Abase_s, \Bbase_s)$ for all $s = (s_r, s_{r-1}, \ldots, s_1) \in [k_r] \times [k_{r-1}] \times \ldots \times [k_1]$.

Any tuple $s = (s_r, s_{r-1}, \ldots, s_1)$ also identifies the sets of vertices $\Lspec_{s_r} \subseteq L$ and $\Rspec_{s_r} \subseteq R$, then the sets $\Lspec_{s_{r-1}} \subseteq \Lspec_{s_r}$ and $\Rspec_{s_{r-1}} \subseteq \Rspec_{s_r}$, and so on until identifying the sets of special base vertices $\Lspec_{s_1} \subseteq \Lspec_{s_2}$ in $L$ and $\Rspec_{s_1} \subseteq \Rspec_{s_2}$ in $R$.
We denote these two sets of special base vertices as $\Lbase_s$ and $\Rbase_s$. We further define
\[\Lbase := \bigcup_{\substack{s \in [k_r] \times [k_{r-1}] \times \ldots \times [k_1]}} \Lbase_s \qquad \text{and} \qquad	\Rbase := \bigcup_{\substack{s \in [k_r] \times [k_{r-1}] \times \ldots \times [k_1]}} \Rbase_s.\]

With this notation, we can now show that the special base vertices $\Lbase \cup \Rbase$ correspond to the special base instances $(\Abase, \Bbase)$ using \Cref{clm:specind}.

\begin{corollary}\label{cor:baseindmatch}
	For $r \geq 1$, the subgraph of $G = \graphed(\alicepart{r}, \bobpart{r})$ induced by the vertices $\Lbase \cup \Rbase$ is exactly the vertex-disjoint union of the bit graphs $\graphed( \Abase_s, \Bbase_s)$ for all $s \in [k_r] \times [k_{r-1}] \times \ldots \times [k_1]$. 	
\end{corollary}
\begin{proof}
	We prove this by induction on $r$.
	
	\textbf{Base case for $\mathbf{r=1}$.} 
	For every $s_1 \in [k_1]$, we have the following directly from the above definitions:
	\[
		\Lbase_s = \Lspec_{s_1}, \Rbase_s = \Rspec_{s_1}, \Abase_s = \Aspec_{s_1}, \Bbase_s = \Bspec_{s_1}.
	\] 
	Therefore, $\Lbase = \Lspec$ and $\Rbase = \Rspec$. Then, by \Cref{clm:specind}, we have that the subgraph induced by the vertices $\Lbase$ and $\Rbase$ is a vertex-disjoint union of graphs 
	\[\graphed(\Aspec_{s_1}, \Bspec_{s_1}) = \graphed(\Abase_{s_1}, \Bbase_{s_1}) \] 
	for all $s_1 \in [k_1]$. 
	Furthermore, for each $s_1 \in [k_1]$, $\graphed(\Abase_{s_1}, \Bbase_{s_1})$ is a bit graph since it corresponds to the special base instance $(\Abase_{s_1}, \Bbase_{s_1})$, which is an instance of $\ahm_0$.

	\textbf{Inductive step for $\mathbf{r \geq 2}$.}
	Let $\Hbase := G[\Lbase \cup \Rbase]$ be the subgraph in the statement.
	It will also be useful for us to consider the subgraph $\Hspec := G[\Lspec \cup \Rspec]$.
	Since $\Lbase \subseteq \Lspec$ and $\Rbase \subseteq \Rspec$, we have that,
	\[
		\Hbase = G[\Lbase \cup \Rbase] = \Hspec[\Lbase \cup \Rbase].
	\]

	Now, by \Cref{clm:specind}, $\Hspec$ is a vertex-disjoint union of subgraphs 
	$\Hspec_{s_r} := \graphed(\Aspec_{s_r}, \Bspec_{s_r})$ for ${s_r} \in [k_r]$.
	Therefore, we only need to show that, for each ${s_r} \in [k_r]$, the subgraph $\Hspec_{s_r}$, when restricted to the vertices in $\Lbase \cup \Rbase$ is exactly a vertex-disjoint union of the bit graphs $\graphed(\Abase_{s}, \Bbase_{s})$ for all $s \in \{{s_r}\} \times [k_{r-1}] \times \ldots \times [k_1]$.

	Since the only vertices of $\Lbase \cup \Rbase$ that are in $\Hspec_{s_r}$ are exactly its special base vertices, what we want immediately follows from the induction hypothesis applied to $\Hspec_{s_r}$ for each ${s_r} \in [k_r]$.
\end{proof}

We can also show that these graphs have perfect matchings. 

\begin{claim}\label{lem:perfmatch}
	For $r \geq 0$ and any $(\Aa{r}, \Bb{r})$, the graph $G = \graphed(\alicepart{r}, \bobpart{r})$ contains a perfect matching, i.e., of size $2n_r$.
\end{claim}
\begin{proof}
	We prove this claim by induction on $r$ where the base case when $r= 0$ is trivial, i.e., we have that $\graphed(\alicepart{0}, \bobpart{0})$ is a single bit graph, which has a perfect matching by \Cref{def:bit-graph}.

	When $r \geq 1$, let $H = G[\Lspec \cup \Rspec]$ be its subgraph induced on the special vertices. By \Cref{clm:specind} and the induction hypothesis, we immediately have that $H$ contains a perfect matching. Hence, it remains to argue that there is a perfect matching in the subgraph 
	\[
	H' = G[L \backslash \Lspec \cup R \backslash \Rspec].
	\]

By definition of the special vertices, $L \backslash \Lspec$ corresponds to $L[i]$ where $\sigmar(i) \leq b_r \cdot \alpha$ and, similarly, $R \backslash \Rspec$ corresponds to $R[j]$ where $\sigmac(j) \leq b_r \cdot \alpha$ for $i,j \in [b_r]$.
	Then, to show that there is a perfect matching in $H'$, it is sufficient to argue that 
	\[G[L[i] \cup R[j]] \text{~for each~} i,j \in [b_r] \text{~where~} \sigmar(i) = \sigmac(j) \leq b_r \cdot \alpha\] 
	has a perfect matching since these are vertex disjoint graphs that cover all of $L \backslash \Lspec$ and $R \backslash \Rspec$.
We consider the case when $r$ is odd and even separately.

	\subparagraph*{Odd r.} 
	By construction, Alice adds edges corresponding to $\Bb{r-1}_{i,j}$. 
	In order to obtain $\graphed(\Aa{r-1}_{i,j}, \Bb{r-1}_{i,j})$, which has a perfect matching by the induction hypothesis, Bob would need to delete the edges corresponding to $\Aa{r-1}_{i,j}$ since $r-1$ is even.
	Instead, Bob deletes no edges so $G[L[i] \cup R[j]]$ is a superset of $\graphed(\Aa{r-1}_{i,j}, \Bb{r-1}_{i,j})$, which implies that it also has a perfect matching.

	\subparagraph*{Even r.} By construction, Alice deletes the edges corresponding to $\Bb{r-1}_{i,j}$. 
	In order to obtain $\graphed(\Aa{r-1}_{i,j}, \Bb{r-1}_{i,j})$, which has a perfect matching by the induction hypothesis, Bob would need to add the edges corresponding to $\Aa{r-1}_{i,j}$ since $r-1$ is odd.
	Instead, Bob adds a superset of these edges since he adds all possible edges 
	\[(\edgeset(r-1, \allone(n_{r-1}), L[i], R[j]) \cup \edgeset(r-1, \allzero(n_{r-1}), L[i], R[j])).\] 
	Therefore, $\graphed(\Aa{r-1}_{i,j}, \Bb{r-1}_{i,j})$  is a subgraph of $G[L[i] \cup R[j]]$ and has a perfect matching. 
\end{proof}

Finally, using these key structural properties, we show that any graph $G$ has a perfect matching and that any $\apx$-approximate maximum matching in $G$ includes many edges that correspond to the bit graphs of the special base instances.

\begin{lemma}\label{lem:apxhasMbase}
	For $r \geq 1$, any $\apx$-approximate maximum matching $M$ in $G = \graphed(\alicepart{r}, \bobpart{r})$ identifies the bits of at least 
	\[n_r \cdot (1/\apx - 2\alpha r)\] 
	many special base instances in $(\alicepart{r}, \bobpart{r})$.
\end{lemma}
\begin{proof}
	By \Cref{cor:baseindmatch}, we know that the edges in $H = G[\Lbase \cup \Rbase]$ are able to identify the bits of the special base instance.
	Therefore, we first obtain a bound on the number of edges in $M$ that must have both endpoints in $\Lbase \cup \Rbase$, i.e., the number of its edges in $H$.

	Observe that the number of edges in $M$ that are not in $H$ have at least one endpoint \emph{not} in $\Lbase \cup \Rbase$, i.e., there are at most $|(L \cup R) \backslash (\Lbase \cup \Rbase)|$ many such edges.
	Therefore, $M$ has at least 
	\[|M| - |(L \cup R) \backslash (\Lbase \cup \Rbase)|\]
	edges in $H$.
	Then, by definition of the special base vertices $\Lbase \cup \Rbase$ and the choice of parameters in \Cref{eq:params-nb} and \Cref{obs:numspecial}, we have that 
	\[
	|\Lbase| = |\Rbase| = 2 \cdot k_r \cdot k_{r-1} \cdot \ldots \cdot k_1 = 2 n_r \cdot (1-\alpha)^r
	\]
	and 
	\[
	|(L \cup R) \backslash (\Lbase \cup \Rbase)| = 4n_r - 4 n_r \cdot (1-\alpha)^r = 4n_r (1 - (1- \alpha)^r).
	\]
	Now, by \Cref{lem:perfmatch} and since $M$ is a $\apx$-approximation, we have that $M$ is of size at least $2n_r/\apx$, which gives us the following bound:
	\begin{align*}
	|M| - |(L \cup R) \backslash (\Lbase \cup \Rbase)| &\geq 2n_r/\apx -  4n_r (1 - (1- \alpha)^r) \\
	&= 4n_r(1/2\apx - 1 + (1 - \alpha)^r).
	\end{align*}

	Finally, by \Cref{cor:baseindmatch}, the only edges in $H$ are the edges of the vertex-disjoint union of the bit graphs $\graphed(\Abase_s, \Bbase_s)$ for $s \in [k_r] \times [k_{r-1}] \times \ldots \times [k_1]$. 
	By \Cref{def:bit-graph}, each bit graph is a matching of exactly two edges, each of which can be used to identify the underlying bit of the special base instance. Therefore, in the worst case, every two edges of $M$ in $H$ identifies a single bit of a special base instance and thus at least 
	\begin{align*}
		2n_r(1/2\apx - 1 + (1 - \alpha)^r) &\geq 2n_r (1/2\apx - 1 + 1 - \alpha \cdot r) \tag{by Bernoulli's inequality} \\
		& = n_r \cdot (1/\apx - 2\alpha \cdot r)
	\end{align*}
	many such bits are identified by $M$.
\end{proof}

\subsubsection{A Protocol for $\ahm_{r}(n_r, \alpha)$ using a Dynamic Streaming Algorithm}\label{subsec:proof-ahm-to-graph}

Using the above graph construction and its key properties, we can now construct a protocol $\prot$ for any instance $(\Aa{r}, \Bb{r})$ of $\ahm_r(n_r, \alpha)$ that simulates a run of any $p$-pass $s$-space dynamic streaming algorithm $\alg$ for $\apx$-approximate maximum matching on the bipartite graph $G = \graphed(\alicepart{r}, \bobpart{r})$.
We then use this protocol to prove \Cref{lem:redToMatching}.

\begin{ourbox}
	\vspace{0.5em}
	\textbf{A protocol $\prot$ for $\ahm_r(n_r, \alpha)$ on any input $(\Aa{r},\Bb{r})$ using the $p$-pass $s$-space dynamic streaming algorithm $\alg$ where $r = 2p-1$}:
	\begin{enumerate}[label=$(\alph*)$]
		\item Alice computes the edge insertions $E_A = \edgeset(r, \Aa{r}, L, R)$ and then constructs an arbitrary ordering $\sigma_A$ of the edges $E_A$.
		Similarly, Bob computes the edge deletions $E_B = \edgeset(r, \Bb{r}, L, R)$ and then constructs an arbitrary ordering $\sigma_B$ of the edges $E_B$.
		\item The players simulate each pass of $\alg$ on the stream of insertions $\sigma_A$ followed by deletions $\sigma_B$ in the usual manner: To simulate a single pass, Alice runs $\alg$ on the insertions in $\sigma_A$ then sends the memory state to Bob, who continues $\alg$ on the deletions  in $\sigma_B$ and sends the memory state back to Alice. In only the final pass, Bob computes the output matching $M$ of $\alg$ instead of sending its memory state.

		\item In parallel to the simulation, Alice and Bob also exchange messages to identify the labels of all the special base instances in $(\Abase, \Bbase$), each of which is identified by a tuple $(s_r, s_{r-1}, \ldots, s_1) \in [k_r] \times [k_{r-1}] \times \ldots \times [k_1]$ as follows. 
		\begin{itemize}
			\item In the first round, Alice does not send anything. In the second round, Bob sends the two input permutations $\sigmar, \sigmac$ of $[b_r]$ along with the message to Alice, using which Alice can identify $k_r$ many special sub-instances of $\ahm_r$. 
			\item For any round $2 < t \leq r$, the player who receives the message of round $t-1$ identifies the $k_r \cdot k_{r-1} \cdot \ldots \cdot k_{r-t+3}$ many special instances of $\ahm_{r-t+2}$ (sub-instances of $\ahm_{r-t+3}$). In round $t$, this player sends $2 \cdot k_r \cdot k_{r-1} \cdot \ldots \cdot k_{r-t+3}$ many permutations of $[b_{r-t+2}]$, corresponding to the two input permutations of each of these instances of $\ahm_{r-t+2}$. 
			\item After receiving message of round $r$, Bob is able to identify the $ k_r \cdot k_{r-1} \cdot \ldots \cdot k_{2}$ many special instances of $\ahm_{1}$. Bob already knows the two input permutations of $[b_1]$ associated with all these special instances of $\ahm_1$, and thus knows the labels of all the special base instances.
		\end{itemize}
		\item At the end of the protocol, the final player receives a uniform random search sequence $s^* = (\speckstar{r}, \speckstar{r-1}, \ldots, \speckstar{1})$. Then, using the information obtained in the previous step, the player identifies the special base instance $(\Abase_{s^*}, \Bbase_{s^*})$ and corresponding special base vertices $\Lbase_{s^*} = \{\ell_1, \ell_2\}$ and $\Rbase_{s^*} = \{ r_1, r_2 \}$. 
		The player returns $x_{s^*}$ as the output of the protocol, which is determined as follows:
		\begin{itemize}
			\item If $M$ contains the edge $(\ell_1, r_1)$ or $(\ell_2, r_2)$, $x_{s^*} = 0$;
			\item If $M$ contains the edge $(\ell_1, r_2)$ or $(\ell_2, r_1)$, $x_{s^*} = 1$;
			\item Otherwise, $x_{s^*}$ is a uniform random bit.
		\end{itemize}
	\end{enumerate}
	\vspace{0em}
\end{ourbox}

We now use this constructed protocol to prove the desired connection between maximum bipartite matching in the dynamic streaming model and the $\ahm_r$ problem, 
thus proving~\Cref{lem:redToMatching}. 

\begin{proof}[Proof of \Cref{lem:redToMatching}]
	To prove this lemma, we show that $\prot$ is an $r$-round protocol with 
	\[
	\cc{\prot} \leq r \cdot s + O(r^2 \cdot n_r \log n_r) \quad \text{and}\quad  \suc{\prot} \geq 1/2 + 1/6\apx
	\]
	where, as in the statement of the lemma, $r = 2p-1$, $\alpha = 1/(4 \apx r)$, and $\apx \geq 1$ is the approximation guarantee of the maximum matching returned by $\alg$.

	First, we argue the number of rounds of communication required by each step separately: 
	\begin{itemize}
		\item 
		By \Cref{clm:find-sets}, Alice and Bob can compute their edges $E_A$ and $E_B$ without any communication. 
		\item 
		In the simulation of $\alg$, each pass of the algorithm is simulated using two rounds of communication, one message from Alice and one message from Bob, except for the final pass, which only requires one round, i.e., one message from Alice. Since $\alg$ has $p$ passes, simulating it requires $r = 2p -1$ rounds of communication. 
		\item
		Identifying the labels of the special base instances (which is only required at the end to output the solution) also requires $r$ rounds of communication since, at the end of each round $t \in [r]$, $k_r \cdot k_{r-1} \cdot \ldots \cdot k_{r-t+1}$ many special instances of $\ahm_{r-t}$ are identifiable by the player that receives the message.
		\item
		Returning the output of protocol $\prot$ is solely computed from the output of $\alg$ and the identified labels, which can be done without any further communication.
	\end{itemize}
	Since the simulation of $\alg$ and identifying the labels is done in parallel, protocol $\prot$ is an $r$-round protocol.

	Next, we argue the communication cost of $\prot$ by considering separately the only two steps that require communication.
	To simulate $\alg$, the players only exchange the memory state of $\alg$. Since this requires at most $s$ bits for each of the $r$ messages, the overall communication required in the simulation is at most $r \cdot s$ bits.
	
		Next, we give an upper bound on the communication required to communicate the labels. In round $t$ for $2 \leq t \leq r$, the player sends $2 \cdot k_r \cdot k_{r-1} \cdot \ldots \cdot k_{r-t+3}$ many permutations of $[b_{r-t+2}]$. 
	Using \Cref{eq:params-nb}, we can say, in total,
		\begin{align*}
			2 \cdot k_r \cdot k_{r-1} \cdot \ldots \cdot k_{r-t+3} \cdot b_{r-t+2} \cdot \log (b_{r-t+2}) &\leq 2\cdot b_r \cdot b_{r-1} \cdot \ldots \cdot b_{r-t+3} \cdot b_{r-t+2} \cdot \log(b_{r-t+2}) \\
			&=  O(n_r \cdot \log (n_r))
		\end{align*}
		bits are sent. In $r$ rounds, the total number of bits to find the labels of all special instances is $O(r \cdot n_r \cdot \log (n_r))$.
	Therefore, $\cc{\prot} \leq r \cdot s + O(r \cdot n_r \log n_r)$.

	Finally, we argue the probability of success of $\prot$. The goal for solving the instance $(\Aa{r}, \Bb{r})$ of $\ahm_r(n_r, \alpha)$ is to output the bit of the special base instance $(\Abase_{s^*}, \Bbase_{s^*})$ where $s^* = (\speckstar{r}, \speckstar{r-1}, \ldots, \speckstar{1})$ is the uniform random search sequence given to the final player, i.e., Bob since $r = 2p-1$ is odd, at the end of the protocol. 
	In protocol $\prot$, the predictor $x_{s^*}$ of the solution is obtained from the matching $M$ returned by $\alg$ or is a uniform random guess.

	When $\alg$ succeeds, which occurs with $1-1/\poly{(n)}$ probability, it is guaranteed to output a $\apx$-approximate maximum matching $M$ in the bipartite graph $\graphed(\alicepart{r}, \bobpart{r})$.
	By \Cref{lem:apxhasMbase}, the edges of $M$ identify the bits of at least $n_r \cdot (1/\apx - 2\alpha r)$ many special base instances, which are $n_r \cdot (1- \alpha)^r$ many in total by \Cref{obs:numspecial}.
	Since $s^*$ uniformly selects one of the special base instances to correspond to the solution bit, $M$ identifies it with probability at least 
	\begin{align*}
		\frac{1/\apx - 2\alpha r}{(1- \alpha)^r} \geq 1/\apx - 2\alpha r = 1/\apx - 1/2\apx = 1/2\apx
	\end{align*}
	since $\alpha = 1/(4 \apx r) \in (0,1)$.

	Thus, we have, 
	\begin{itemize}
		\item When $\alg$ succeeds and $M$ correctly identifies the solution bit, the predictor $x_{s^*}$ returned by $\prot$ is correct (with probability $1$);
		\item When $\alg$ succeeds but $M$ does not identify the solution bit, the predictor $x_{s^*}$ returned by $\prot$ is a random guess and thus is correct with probability $1/2$;
		\item When $\alg$ fails, there is no guarantee on the predictor bit, and we might as well assume the answer is wrong (with probability $1$); 
	\end{itemize}
	As such, we have that the protocol succeeds with probability at least 
	\[
		\paren{1-1/\poly{(n)}} \cdot \paren{\frac{1}{2\apx} \cdot 1 + \paren{1-\frac{1}{2\apx}} \cdot \frac{1}{2}} + 1/\poly{(n)} \cdot 0 \geq \frac{1}{2} \cdot \paren{1+\frac{1}{3\apx}},
	\]
	since $\beta$ is a constant. This concludes the proof. 
\end{proof}

\begin{Remark}\label{rem:add-before-delete}
	In the construction of $\graphed(\Aa{r}, \Bb{r})$, when $r$ is even, Alice deletes the edges before Bob adds them: for the corresponding streaming problem, this corresponds to deleting edges that have not been inserted (although the final stream
	still ensures that any edge that is deleted will be inserted); this is \underline{not} consistent with the definition of dynamic graph streams. 
	Nevertheless, in our construction, we \emph{only} use the $\ahm_r$ lower bound when $r = 2p-1$ and thus is odd. Here, Alice adds edges to the graph, and Bob deletes the edges after all the edges are added. As a result, in the corresponding
	streaming problem, no edge is deleted before it is inserted, thus adhering to the restriction of the dynamic graph streams. 
\end{Remark}

\bibliographystyle{halpha-abbrv}
\bibliography{general}

\newcommand{\etalchar}[1]{$^{#1}$}
\begin{thebibliography}{CKP{\etalchar{+}}21b}
\expandafter\ifx\csname url\endcsname\relax
  \def\url#1{\texttt{#1}}\fi
\expandafter\ifx\csname doi\endcsname\relax
  \def\doi#1{\burlalt{doi:#1}{http://dx.doi.org/#1}}\fi
\expandafter\ifx\csname urlprefix\endcsname\relax\def\urlprefix{URL }\fi
\expandafter\ifx\csname href\endcsname\relax
  \def\href#1#2{#2}\fi
\expandafter\ifx\csname burlalt\endcsname\relax
  \def\burlalt#1#2{\href{#2}{#1}}\fi

\bibitem[A22]{Assadi22}
S.~Assadi.
\newblock A two-pass (conditional) lower bound for semi-streaming maximum
  matching.
\newblock In J.~S. Naor and N.~Buchbinder, editors, {\em Proceedings of the
  2022 {ACM-SIAM} Symposium on Discrete Algorithms, {SODA} 2022, Virtual
  Conference / Alexandria, VA, USA, January 9 - 12, 2022}, pages 708--742.
  {SIAM}, 2022.

\bibitem[A23]{Assadi23}
S.~Assadi.
\newblock Recent advances in multi-pass graph streaming lower bounds.
\newblock {\em {SIGACT} News}, 54(3):48--75, 2023.

\bibitem[A24]{Assadi24}
S.~Assadi.
\newblock A simple {(1} - \emph{{\(\epsilon\)}})-approximation semi-streaming
  algorithm for maximum (weighted) matching.
\newblock In M.~Parter and S.~Pettie, editors, {\em 2024 Symposium on
  Simplicity in Algorithms, {SOSA} 2024, Alexandria, VA, USA, January 8-10,
  2024}, pages 337--354. {SIAM}, 2024.

\bibitem[AAD{\etalchar{+}}23]{AshvinkumarADGW23}
V.~Ashvinkumar, S.~Assadi, C.~Deng, J.~Gao, and C.~Wang.
\newblock Evaluating stability in massive social networks: Efficient streaming
  algorithms for structural balance.
\newblock In N.~Megow and A.~D. Smith, editors, {\em Approximation,
  Randomization, and Combinatorial Optimization. Algorithms and Techniques,
  {APPROX/RANDOM} 2023, September 11-13, 2023, Atlanta, Georgia, {USA}}, volume
  275 of {\em LIPIcs}, pages 58:1--58:23. Schloss Dagstuhl - Leibniz-Zentrum
  f{\"{u}}r Informatik, 2023.

\bibitem[ABB{\etalchar{+}}19]{AssadiBBMS19}
S.~Assadi, M.~Bateni, A.~Bernstein, V.~S. Mirrokni, and C.~Stein.
\newblock Coresets meet {EDCS:} algorithms for matching and vertex cover on
  massive graphs.
\newblock In T.~M. Chan, editor, {\em Proceedings of the Thirtieth Annual
  {ACM-SIAM} Symposium on Discrete Algorithms, {SODA} 2019, San Diego,
  California, USA, January 6-9, 2019}, pages 1616--1635. {SIAM}, 2019.

\bibitem[ACG{\etalchar{+}}15]{AhnCGMW15}
K.~J. Ahn, G.~Cormode, S.~Guha, A.~McGregor, and A.~Wirth.
\newblock Correlation clustering in data streams.
\newblock In F.~R. Bach and D.~M. Blei, editors, {\em Proceedings of the 32nd
  International Conference on Machine Learning, {ICML} 2015, Lille, France,
  6-11 July 2015}, volume~37 of {\em {JMLR} Workshop and Conference
  Proceedings}, pages 2237--2246. JMLR.org, 2015.

\bibitem[ACK19]{AssadiCK19a}
S.~Assadi, Y.~Chen, and S.~Khanna.
\newblock Polynomial pass lower bounds for graph streaming algorithms.
\newblock In {\em Proceedings of the 51st Annual {ACM} {SIGACT} Symposium on
  Theory of Computing, {STOC} 2019, Phoenix, AZ, USA, June 23-26, 2019}, pages
  265--276, 2019.

\bibitem[AG11]{AhnG11}
K.~J. Ahn and S.~Guha.
\newblock Linear programming in the semi-streaming model with application to
  the maximum matching problem.
\newblock In L.~Aceto, M.~Henzinger, and J.~Sgall, editors, {\em Automata,
  Languages and Programming - 38th International Colloquium, {ICALP} 2011,
  Zurich, Switzerland, July 4-8, 2011, Proceedings, Part {II}}, volume 6756 of
  {\em Lecture Notes in Computer Science}, pages 526--538. Springer, 2011.

\bibitem[AG15]{AhnG15}
K.~J. Ahn and S.~Guha.
\newblock Access to data and number of iterations: Dual primal algorithms for
  maximum matching under resource constraints.
\newblock In G.~E. Blelloch and K.~Agrawal, editors, {\em Proceedings of the
  27th {ACM} on Symposium on Parallelism in Algorithms and Architectures,
  {SPAA} 2015, Portland, OR, USA, June 13-15, 2015}, pages 202--211. {ACM},
  2015.

\bibitem[AGL{\etalchar{+}}24]{AssadiGLMM24}
S.~Assadi, P.~Ghosh, B.~Loff, P.~Mittal, and S.~Mukhopadhyay.
\newblock Polynomial pass semi-streaming lower bounds for k-cores and
  degeneracy.
\newblock {\em CoRR}, abs/2405.14835. To appear in CCC 2024, 2024.

\bibitem[AGM12]{AhnGM12a}
K.~J. Ahn, S.~Guha, and A.~McGregor.
\newblock Analyzing graph structure via linear measurements.
\newblock In {\em Proceedings of the Twenty-third Annual ACM-SIAM Symposium on
  Discrete Algorithms}, SODA '12, pages 459--467. SIAM, 2012.

\bibitem[AJJ{\etalchar{+}}22]{AssadiJJST22}
S.~Assadi, A.~Jambulapati, Y.~Jin, A.~Sidford, and K.~Tian.
\newblock Semi-streaming bipartite matching in fewer passes and optimal space.
\newblock In J.~S. Naor and N.~Buchbinder, editors, {\em Proceedings of the
  2022 {ACM-SIAM} Symposium on Discrete Algorithms, {SODA} 2022, Virtual
  Conference / Alexandria, VA, USA, January 9 - 12, 2022}, pages 627--669.
  {SIAM}, 2022.

\bibitem[AKL16]{AssadiKL16}
S.~Assadi, S.~Khanna, and Y.~Li.
\newblock The stochastic matching problem with (very) few queries.
\newblock In V.~Conitzer, D.~Bergemann, and Y.~Chen, editors, {\em Proceedings
  of the 2016 {ACM} Conference on Economics and Computation, {EC} '16,
  Maastricht, The Netherlands, July 24-28, 2016}, pages 43--60. {ACM}, 2016.

\bibitem[AKL17]{AssadiKL17}
S.~Assadi, S.~Khanna, and Y.~Li.
\newblock On estimating maximum matching size in graph streams.
\newblock In P.~N. Klein, editor, {\em Proceedings of the Twenty-Eighth Annual
  {ACM-SIAM} Symposium on Discrete Algorithms, {SODA} 2017, Barcelona, Spain,
  Hotel Porta Fira, January 16-19}, pages 1723--1742. {SIAM}, 2017.

\bibitem[AKLY16]{AssadiKLY16}
S.~Assadi, S.~Khanna, Y.~Li, and G.~Yaroslavtsev.
\newblock Maximum matchings in dynamic graph streams and the simultaneous
  communication model.
\newblock In R.~Krauthgamer, editor, {\em Proceedings of the Twenty-Seventh
  Annual {ACM-SIAM} Symposium on Discrete Algorithms, {SODA} 2016, Arlington,
  VA, USA, January 10-12, 2016}, pages 1345--1364. {SIAM}, 2016.

\bibitem[AKM23]{AssadiKM23}
S.~Assadi, P.~Kumar, and P.~Mittal.
\newblock Brooks' theorem in graph streams: {A} single-pass semi-streaming
  algorithm for {\(\Delta\)}-coloring.
\newblock {\em TheoretiCS}, 2, 2023.

\bibitem[AKNS24]{AssadiKNS24}
S.~Assadi, C.~Konrad, K.~K. Naidu, and J.~Sundaresan.
\newblock O(log log {\(n\)}) passes is optimal for semi-streaming maximal
  independent set.
\newblock In {\em Proceedings of the 56th Annual {ACM} Symposium on Theory of
  Computing, {STOC} 2024, Vancouver, British Columbia, Canada, June 24-28,
  2024}. {ACM}, 2024.

\bibitem[AKZ24]{AssadiKZ24}
S.~Assadi, G.~Kol, and Z.~Zhang.
\newblock Optimal multi-pass lower bounds for {MST} in dynamic streams.
\newblock In B.~Mohar, I.~Shinkar, and R.~O'Donnell, editors, {\em Proceedings
  of the 56th Annual {ACM} Symposium on Theory of Computing, {STOC} 2024,
  Vancouver, BC, Canada, June 24-28, 2024}, pages 835--846. {ACM}, 2024.

\bibitem[ALT21]{AssadiLT21}
S.~Assadi, S.~C. Liu, and R.~E. Tarjan.
\newblock An auction algorithm for bipartite matching in streaming and
  massively parallel computation models.
\newblock In H.~V. Le and V.~King, editors, {\em 4th Symposium on Simplicity in
  Algorithms, {SOSA} 2021, Virtual Conference, January 11-12, 2021}, pages
  165--171. {SIAM}, 2021.

\bibitem[AMS96]{AlonMS96}
N.~Alon, Y.~Matias, and M.~Szegedy.
\newblock The space complexity of approximating the frequency moments.
\newblock In {\em STOC}, pages 20--29. ACM, 1996.

\bibitem[AOSS19]{AssadiOSS19}
S.~Assadi, K.~Onak, B.~Schieber, and S.~Solomon.
\newblock Fully dynamic maximal independent set with sublinear in n update
  time.
\newblock In T.~M. Chan, editor, {\em Proceedings of the Thirtieth Annual
  {ACM-SIAM} Symposium on Discrete Algorithms, {SODA} 2019, San Diego,
  California, USA, January 6-9, 2019}, pages 1919--1936. {SIAM}, 2019.

\bibitem[AR20]{AssadiR20}
S.~Assadi and R.~Raz.
\newblock Near-quadratic lower bounds for two-pass graph streaming algorithms.
\newblock In S.~Irani, editor, {\em 61st {IEEE} Annual Symposium on Foundations
  of Computer Science, {FOCS} 2020, Durham, NC, USA, November 16-19, 2020},
  pages 342--353. {IEEE}, 2020.

\bibitem[AS22]{AssadiS22}
S.~Assadi and V.~Shah.
\newblock An asymptotically optimal algorithm for maximum matching in dynamic
  streams.
\newblock In M.~Braverman, editor, {\em 13th Innovations in Theoretical
  Computer Science Conference, {ITCS} 2022, January 31 - February 3, 2022,
  Berkeley, CA, {USA}}, volume 215 of {\em LIPIcs}, pages 9:1--9:23. Schloss
  Dagstuhl - Leibniz-Zentrum f{\"{u}}r Informatik, 2022.

\bibitem[AS23]{AssadiS23}
S.~Assadi and J.~Sundaresan.
\newblock Hidden permutations to the rescue: Multi-pass streaming lower bounds
  for approximate matchings.
\newblock In {\em 64th {IEEE} Annual Symposium on Foundations of Computer
  Science, {FOCS} 2023, Santa Cruz, CA, USA, November 6-9, 2023}, pages
  909--932. {IEEE}, 2023.

\bibitem[BBCR10]{BarakBCR10}
B.~Barak, M.~Braverman, X.~Chen, and A.~Rao.
\newblock How to compress interactive communication.
\newblock In L.~J. Schulman, editor, {\em Proceedings of the 42nd {ACM}
  Symposium on Theory of Computing, {STOC} 2010, Cambridge, Massachusetts, USA,
  5-8 June 2010}, pages 67--76. {ACM}, 2010.

\bibitem[BDL21]{BernsteinDL21}
A.~Bernstein, A.~Dudeja, and Z.~Langley.
\newblock A framework for dynamic matching in weighted graphs.
\newblock In S.~Khuller and V.~V. Williams, editors, {\em {STOC} '21: 53rd
  Annual {ACM} {SIGACT} Symposium on Theory of Computing, Virtual Event, Italy,
  June 21-25, 2021}, pages 668--681. {ACM}, 2021.

\bibitem[BHH19]{BehnezhadHH19}
S.~Behnezhad, M.~Hajiaghayi, and D.~G. Harris.
\newblock Exponentially faster massively parallel maximal matching.
\newblock In D.~Zuckerman, editor, {\em 60th {IEEE} Annual Symposium on
  Foundations of Computer Science, {FOCS} 2019, Baltimore, Maryland, USA,
  November 9-12, 2019}, pages 1637--1649. {IEEE} Computer Society, 2019.

\bibitem[BRWY13]{BravermanRWY13}
M.~Braverman, A.~Rao, O.~Weinstein, and A.~Yehudayoff.
\newblock Direct products in communication complexity.
\newblock In {\em 54th Annual {IEEE} Symposium on Foundations of Computer
  Science, {FOCS} 2013, 26-29 October, 2013, Berkeley, CA, {USA}}, pages
  746--755. {IEEE} Computer Society, 2013.

\bibitem[BV10]{BogdanovV10}
A.~Bogdanov and E.~Viola.
\newblock Pseudorandom bits for polynomials.
\newblock {\em {SIAM} J. Comput.}, 39(6):2464--2486, 2010.

\bibitem[CCE{\etalchar{+}}16]{ChitnisCEHMMV16}
R.~Chitnis, G.~Cormode, H.~Esfandiari, M.~Hajiaghayi, A.~McGregor,
  M.~Monemizadeh, and S.~Vorotnikova.
\newblock Kernelization via sampling with applications to finding matchings and
  related problems in dynamic graph streams.
\newblock In R.~Krauthgamer, editor, {\em Proceedings of the Twenty-Seventh
  Annual {ACM-SIAM} Symposium on Discrete Algorithms, {SODA} 2016, Arlington,
  VA, USA, January 10-12, 2016}, pages 1326--1344. {SIAM}, 2016.

\bibitem[CCHM15]{ChitnisCHM15}
R.~H. Chitnis, G.~Cormode, M.~T. Hajiaghayi, and M.~Monemizadeh.
\newblock Parameterized streaming: Maximal matching and vertex cover.
\newblock In P.~Indyk, editor, {\em Proceedings of the Twenty-Sixth Annual
  {ACM-SIAM} Symposium on Discrete Algorithms, {SODA} 2015, San Diego, CA, USA,
  January 4-6, 2015}, pages 1234--1251. {SIAM}, 2015.

\bibitem[CDK19]{CormodeDK19}
G.~Cormode, J.~Dark, and C.~Konrad.
\newblock Independent sets in vertex-arrival streams.
\newblock In C.~Baier, I.~Chatzigiannakis, P.~Flocchini, and S.~Leonardi,
  editors, {\em 46th International Colloquium on Automata, Languages, and
  Programming, {ICALP} 2019, July 9-12, 2019, Patras, Greece}, volume 132 of
  {\em LIPIcs}, pages 45:1--45:14. Schloss Dagstuhl - Leibniz-Zentrum f{\"{u}}r
  Informatik, 2019.

\bibitem[CGMV20]{ChakrabartiGMV20}
A.~Chakrabarti, P.~Ghosh, A.~McGregor, and S.~Vorotnikova.
\newblock Vertex ordering problems in directed graph streams.
\newblock In {\em Proceedings of the 2020 {ACM-SIAM} Symposium on Discrete
  Algorithms, {SODA} 2020, Salt Lake City, UT, USA, January 5-8, 2020}, pages
  1786--1802, 2020.

\bibitem[CKP{\etalchar{+}}21a]{ChenKPSSY21a}
L.~Chen, G.~Kol, D.~Paramonov, R.~R. Saxena, Z.~Song, and H.~Yu.
\newblock Almost optimal super-constant-pass streaming lower bounds for
  reachability.
\newblock In S.~Khuller and V.~V. Williams, editors, {\em {STOC} '21: 53rd
  Annual {ACM} {SIGACT} Symposium on Theory of Computing, Virtual Event, Italy,
  June 21-25, 2021}, pages 570--583. {ACM}, 2021.

\bibitem[CKP{\etalchar{+}}21b]{ChenKPSSY21b}
L.~Chen, G.~Kol, D.~Paramonov, R.~R. Saxena, Z.~Song, and H.~Yu.
\newblock Near-optimal two-pass streaming algorithm for sampling random walks
  over directed graphs.
\newblock In N.~Bansal, E.~Merelli, and J.~Worrell, editors, {\em 48th
  International Colloquium on Automata, Languages, and Programming, {ICALP}
  2021, July 12-16, 2021, Glasgow, Scotland (Virtual Conference)}, volume 198
  of {\em LIPIcs}, pages 52:1--52:19. Schloss Dagstuhl - Leibniz-Zentrum
  f{\"{u}}r Informatik, 2021.

\bibitem[CLM{\etalchar{+}}18]{CzumajLMMOS18}
A.~Czumaj, J.~Lacki, A.~Madry, S.~Mitrovic, K.~Onak, and P.~Sankowski.
\newblock Round compression for parallel matching algorithms.
\newblock In I.~Diakonikolas, D.~Kempe, and M.~Henzinger, editors, {\em
  Proceedings of the 50th Annual {ACM} {SIGACT} Symposium on Theory of
  Computing, {STOC} 2018, Los Angeles, CA, USA, June 25-29, 2018}, pages
  471--484. {ACM}, 2018.

\bibitem[CSWY01]{ChakrabartiSWY01}
A.~Chakrabarti, Y.~Shi, A.~Wirth, and A.~C. Yao.
\newblock Informational complexity and the direct sum problem for simultaneous
  message complexity.
\newblock In {\em 42nd Annual Symposium on Foundations of Computer Science,
  {FOCS} 2001, 14-17 October 2001, Las Vegas, Nevada, {USA}}, pages 270--278.
  {IEEE} Computer Society, 2001.

\bibitem[CT06]{CoverT06}
T.~M. Cover and J.~A. Thomas.
\newblock {\em \textnormal{\textbf{Elements of information theory}} {(2.}
  ed.)}.
\newblock Wiley, 2006.

\bibitem[DK20]{DarkK20}
J.~Dark and C.~Konrad.
\newblock {Optimal Lower Bounds for Matching and Vertex Cover in Dynamic Graph
  Streams}.
\newblock In {\em 35th Computational Complexity Conference (CCC 2020)}, Leibniz
  International Proceedings in Informatics (LIPIcs), pages 30:1--30:14, 2020.

\bibitem[DNO14]{DobzinskiNO14}
S.~Dobzinski, N.~Nisan, and S.~Oren.
\newblock Economic efficiency requires interaction.
\newblock In D.~B. Shmoys, editor, {\em Symposium on Theory of Computing,
  {STOC} 2014, New York, NY, USA, May 31 - June 03, 2014}, pages 233--242.
  {ACM}, 2014.

\bibitem[DP09]{DubhashiP09}
D.~P. Dubhashi and A.~Panconesi.
\newblock {\em \textnormal{\textbf{Concentration of measure for the analysis of
  randomized algorithms}}}.
\newblock Cambridge University Press, 2009.

\bibitem[DV13]{DasS13}
A.~K. Das and S.~Vishwanath.
\newblock On finite alphabet compressive sensing.
\newblock In {\em 2013 IEEE International Conference on Acoustics, Speech and
  Signal Processing}, pages 5890--5894. IEEE, 2013.

\bibitem[FHS17]{FoxHS17}
J.~Fox, H.~Huang, and B.~Sudakov.
\newblock On graphs decomposable into induced matchings of linear sizes.
\newblock {\em Bulletin of the London Mathematical Society}, 49(1):45--57,
  2017.

\bibitem[FKM{\etalchar{+}}08]{FeigenbaumKMSZ08}
J.~Feigenbaum, S.~Kannan, A.~McGregor, S.~Suri, and J.~Zhang.
\newblock Graph distances in the data-stream model.
\newblock {\em {SIAM} J. Comput.}, 38(5):1709--1727, 2008.

\bibitem[GGK{\etalchar{+}}18]{GhaffariGKMR18}
M.~Ghaffari, T.~Gouleakis, C.~Konrad, S.~Mitrovic, and R.~Rubinfeld.
\newblock Improved massively parallel computation algorithms for mis, matching,
  and vertex cover.
\newblock In C.~Newport and I.~Keidar, editors, {\em Proceedings of the 2018
  {ACM} Symposium on Principles of Distributed Computing, {PODC} 2018, Egham,
  United Kingdom, July 23-27, 2018}, pages 129--138. {ACM}, 2018.

\bibitem[GKK12]{GoelKK12}
A.~Goel, M.~Kapralov, and S.~Khanna.
\newblock On the communication and streaming complexity of maximum bipartite
  matching.
\newblock In Y.~Rabani, editor, {\em Proceedings of the Twenty-Third Annual
  {ACM-SIAM} Symposium on Discrete Algorithms, {SODA} 2012, Kyoto, Japan,
  January 17-19, 2012}, pages 468--485. {SIAM}, 2012.

\bibitem[GKMS19]{GamlathKMS19}
B.~Gamlath, S.~Kale, S.~Mitrovic, and O.~Svensson.
\newblock Weighted matchings via unweighted augmentations.
\newblock In P.~Robinson and F.~Ellen, editors, {\em Proceedings of the 2019
  {ACM} Symposium on Principles of Distributed Computing, {PODC} 2019, Toronto,
  ON, Canada, July 29 - August 2, 2019}, pages 491--500. {ACM}, 2019.

\bibitem[GO13]{GuruswamiO13}
V.~Guruswami and K.~Onak.
\newblock Superlinear lower bounds for multipass graph processing.
\newblock In {\em Proceedings of the 28th Conference on Computational
  Complexity, {CCC} 2013, K.lo Alto, California, USA, 5-7 June, 2013}, pages
  287--298, 2013.

\bibitem[HJMR07]{HarshaJMR07}
P.~Harsha, R.~Jain, D.~A. McAllester, and J.~Radhakrishnan.
\newblock The communication complexity of correlation.
\newblock In {\em 22nd Annual {IEEE} Conference on Computational Complexity
  {(CCC} 2007), 13-16 June 2007, San Diego, California, {USA}}, pages 10--23.
  {IEEE} Computer Society, 2007.

\bibitem[JPY16]{JainPY16}
R.~Jain, A.~Pereszlényi, and P.~Yao.
\newblock A direct product theorem for two-party bounded-round public-coin
  communication complexity.
\newblock {\em Algorithmica}, 76(3):720--748, 12 2016.
\newblock A preliminary version of this article has appeared in the Proceedings
  of the 53rd Annual IEEE Symposium on Foundations of Computer Science, FOCS
  2012.

\bibitem[JST11]{JowhariST11}
H.~Jowhari, M.~Saglam, and G.~Tardos.
\newblock Tight bounds for lp samplers, finding duplicates in streams, and
  related problems.
\newblock In M.~Lenzerini and T.~Schwentick, editors, {\em Proceedings of the
  30th {ACM} {SIGMOD-SIGACT-SIGART} Symposium on Principles of Database
  Systems, {PODS} 2011, June 12-16, 2011, Athens, Greece}, pages 49--58. {ACM},
  2011.

\bibitem[Kap13]{Kapralov13}
M.~Kapralov.
\newblock Better bounds for matchings in the streaming model.
\newblock In S.~Khanna, editor, {\em Proceedings of the Twenty-Fourth Annual
  {ACM-SIAM} Symposium on Discrete Algorithms, {SODA} 2013, New Orleans,
  Louisiana, USA, January 6-8, 2013}, pages 1679--1697. {SIAM}, 2013.

\bibitem[Kap21]{Kapralov21}
M.~Kapralov.
\newblock Space lower bounds for approximating maximum matching in the edge
  arrival model.
\newblock In D.~Marx, editor, {\em Proceedings of the 2021 {ACM-SIAM} Symposium
  on Discrete Algorithms, {SODA} 2021, Virtual Conference, January 10 - 13,
  2021}, pages 1874--1893. {SIAM}, 2021.

\bibitem[KN97]{KushilevitzN97}
E.~Kushilevitz and N.~Nisan.
\newblock {\em \textnormal{\textbf{Communication complexity}}}.
\newblock Cambridge University Press, 1997.

\bibitem[KN21]{KonradN21}
C.~Konrad and K.~K. Naidu.
\newblock On two-pass streaming algorithms for maximum bipartite matching.
\newblock In M.~Wootters and L.~Sanit{\`{a}}, editors, {\em Approximation,
  Randomization, and Combinatorial Optimization. Algorithms and Techniques,
  {APPROX/RANDOM} 2021, August 16-18, 2021, University of Washington, Seattle,
  Washington, {USA} (Virtual Conference)}, volume 207 of {\em LIPIcs}, pages
  19:1--19:18. Schloss Dagstuhl - Leibniz-Zentrum f{\"{u}}r Informatik, 2021.

\bibitem[KN24]{KonradN24}
C.~Konrad and K.~K. Naidu.
\newblock An unconditional lower bound for two-pass streaming algorithms for
  maximum matching approximation.
\newblock In D.~P. Woodruff, editor, {\em Proceedings of the 2024 {ACM-SIAM}
  Symposium on Discrete Algorithms, {SODA} 2024, Alexandria, VA, USA, January
  7-10, 2024}, pages 2881--2899. {SIAM}, 2024.

\bibitem[Kon15]{Konrad15}
C.~Konrad.
\newblock Maximum matching in turnstile streams.
\newblock In N.~Bansal and I.~Finocchi, editors, {\em Algorithms - {ESA} 2015 -
  23rd Annual European Symposium, Patras, Greece, September 14-16, 2015,
  Proceedings}, volume 9294 of {\em Lecture Notes in Computer Science}, pages
  840--852. Springer, 2015.

\bibitem[Kon18]{Konrad18}
C.~Konrad.
\newblock A simple augmentation method for matchings with applications to
  streaming algorithms.
\newblock In I.~Potapov, P.~G. Spirakis, and J.~Worrell, editors, {\em 43rd
  International Symposium on Mathematical Foundations of Computer Science,
  {MFCS} 2018, August 27-31, 2018, Liverpool, {UK}}, volume 117 of {\em
  LIPIcs}, pages 74:1--74:16. Schloss Dagstuhl - Leibniz-Zentrum f{\"{u}}r
  Informatik, 2018.

\bibitem[KW20]{KunKoW20}
Y.~Kun{-}Ko and O.~Weinstein.
\newblock An adaptive step toward the multiphase conjecture.
\newblock In S.~Irani, editor, {\em 61st {IEEE} Annual Symposium on Foundations
  of Computer Science, {FOCS} 2020, Durham, NC, USA, November 16-19, 2020},
  pages 752--761. {IEEE}, 2020.

\bibitem[LMSV11]{LattanziMSV11}
S.~Lattanzi, B.~Moseley, S.~Suri, and S.~Vassilvitskii.
\newblock Filtering: a method for solving graph problems in mapreduce.
\newblock In {\em {SPAA} 2011: Proceedings of the 23rd Annual {ACM} Symposium
  on Parallelism in Algorithms and Architectures, San Jose, CA, USA, June 4-6,
  2011 (Co-located with {FCRC} 2011)}, pages 85--94, 2011.

\bibitem[Lov09]{Lovett09}
S.~Lovett.
\newblock Unconditional pseudorandom generators for low degree polynomials.
\newblock {\em Theory Comput.}, 5(1):69--82, 2009.

\bibitem[LP09]{LovaszP09}
L.~Lov{\'a}sz and M.~D. Plummer.
\newblock {\em Matching theory}, volume 367.
\newblock American Mathematical Soc., 2009.

\bibitem[McG05]{McGregor05}
A.~McGregor.
\newblock Finding graph matchings in data streams.
\newblock In C.~Chekuri, K.~Jansen, J.~D.~P. Rolim, and L.~Trevisan, editors,
  {\em Approximation, Randomization and Combinatorial Optimization. Algorithms
  and Techniques}, 2005.

\bibitem[MNSW95]{MiltersenNSW95}
P.~B. Miltersen, N.~Nisan, S.~Safra, and A.~Wigderson.
\newblock On data structures and asymmetric communication complexity.
\newblock In F.~T. Leighton and A.~Borodin, editors, {\em Proceedings of the
  Twenty-Seventh Annual {ACM} Symposium on Theory of Computing, 29 May-1 June
  1995, Las Vegas, Nevada, {USA}}, pages 103--111. {ACM}, 1995.

\bibitem[MR95]{MotwaniR95}
R.~Motwani and P.~Raghavan.
\newblock {\em Randomized Algorithms}.
\newblock Cambridge University Press, 1995.

\bibitem[NY19]{NelsonY19}
J.~Nelson and H.~Yu.
\newblock Optimal lower bounds for distributed and streaming spanning forest
  computation.
\newblock In T.~M. Chan, editor, {\em Proceedings of the Thirtieth Annual
  {ACM-SIAM} Symposium on Discrete Algorithms, {SODA} 2019, San Diego,
  California, USA, January 6-9, 2019}, pages 1844--1860. {SIAM}, 2019.

\bibitem[PS97]{PanconesiS97}
A.~Panconesi and A.~Srinivasan.
\newblock Randomized distributed edge coloring via an extension of the
  chernoff-hoeffding bounds.
\newblock {\em {SIAM} J. Comput.}, 26(2):350--368, 1997.

\bibitem[RS78]{RuzsaS78}
I.~Z. Ruzsa and E.~Szemer{\'e}di.
\newblock Triple systems with no six points carrying three triangles.
\newblock {\em Combinatorics (Keszthely, 1976), Coll. Math. Soc. J. Bolyai},
  18:939--945, 1978.

\bibitem[RY20]{RaoY20}
A.~Rao and A.~Yehudayoff.
\newblock {\em \textnormal{\textbf{Communication Complexity: and
  Applications}}}.
\newblock Cambridge University Press, 2020.

\bibitem[SSS95]{SchmidtSS95}
J.~P. Schmidt, A.~Siegel, and A.~Srinivasan.
\newblock Chernoff–hoeffding bounds for applications with limited
  independence.
\newblock {\em SIAM Journal on Discrete Mathematics}, 8(2):223--250, 1995.

\bibitem[Tir18]{Tirodkar18}
S.~Tirodkar.
\newblock Deterministic algorithms for maximum matching on general graphs in
  the semi-streaming model.
\newblock In S.~Ganguly and P.~K. Pandya, editors, {\em 38th {IARCS} Annual
  Conference on Foundations of Software Technology and Theoretical Computer
  Science, {FSTTCS} 2018, December 11-13, 2018, Ahmedabad, India}, volume 122
  of {\em LIPIcs}, pages 39:1--39:16. Schloss Dagstuhl - Leibniz-Zentrum
  f{\"{u}}r Informatik, 2018.

\bibitem[Vel24]{Veldt24}
N.~Veldt.
\newblock Growing a random maximal independent set produces a 2-approximate
  vertex cover.
\newblock In M.~Parter and S.~Pettie, editors, {\em 2024 Symposium on
  Simplicity in Algorithms, {SOSA} 2024, Alexandria, VA, USA, January 8-10,
  2024}, pages 355--362. {SIAM}, 2024.

\bibitem[Wei15]{Weinstein15}
O.~Weinstein.
\newblock Information complexity and the quest for interactive compression.
\newblock {\em {SIGACT} News}, 46(2):41--64, 2015.

\end{thebibliography}

\clearpage
\appendix

\part*{Appendix}

\section{Background on Information Theory}\label{app:info}

We now briefly introduce some definitions and facts from information theory that are needed in this thesis. We refer the interested reader to the text by Cover and Thomas~\cite{CoverT06} for an excellent introduction to this field, 
and the proofs of the statements used in this Appendix. 

For a random variable $\rA$, we use $\supp{\rA}$ to denote the support of $\rA$ and $\distribution{\rA}$ to denote its distribution. 
When it is clear from the context, we may abuse the notation and use $\rA$ directly instead of $\distribution{\rA}$, for example, write 
$A \sim \rA$ to mean $A \sim \distribution{\rA}$, i.e., $A$ is sampled from the distribution of random variable $\rA$. 

\begin{itemize}
\item We denote the \emph{Shannon Entropy} of a random variable $\rA$ by
$\en{\rA}$, which is defined as: 
\begin{align}
	\en{\rA} := \sum_{A \in \supp{\rA}} \Pr\paren{\rA = A} \cdot \log{\paren{1/\Pr\paren{\rA = A}}} \label{eq:entropy}
\end{align} 

\item The \emph{conditional entropy} of $\rA$ conditioned on $\rB$ is denoted by $\en{\rA \mid \rB}$ and defined as:
\begin{align}
\en{\rA \mid \rB} := \Ex_{B \sim \rB} \bracket{\en{\rA \mid \rB = B}}, \label{eq:cond-entropy}
\end{align}
where 
$\en{\rA \mid \rB = B}$ is defined in a standard way by using the distribution of $\rA$ conditioned on the event $\rB = B$ in Eq~(\ref{eq:entropy}).

\item The \emph{mutual information} of two random variables $\rA$ and $\rB$ is denoted by
$\mi{\rA}{\rB}$ and is defined:
\begin{align}
\mi{\rA}{\rB} := \en{A} - \en{A \mid  B} = \en{B} - \en{B \mid  A}. \label{eq:mi}
\end{align}

\item The \emph{conditional mutual information} is defined as $\mi{\rA}{\rB \mid \rC}:= \en{\rA \mid \rC} - \en{\rA \mid \rB,\rC}$.
\end{itemize}

\subsection*{Useful Properties of Entropy and Mutual Information}\label{sec:prop-en-mi}

We shall use the following basic properties of entropy and mutual information throughout. 
Proofs of these properties mostly follow from convexity of the entropy function
and Jensen's inequality and can be found in~\cite[Chapter~2]{CoverT06}. 

\begin{fact}\label{fact:it-facts}
  Let $\rA$, $\rB$, $\rC$, and $\rD$ be four (possibly correlated) random variables.
   \begin{enumerate}
  \item \label{part:uniform} $0 \leq \en{\rA} \leq \log{\card{\supp{\rA}}}$. The right equality holds
    iff $\distribution{\rA}$ is uniform.
  \item \label{part:info-zero} $\mi{\rA}{\rB \mid \rC} \geq 0$. The equality holds iff $\rA$ and
    $\rB$ are \emph{independent} conditioned on $\rC$.
  \item \label{part:cond-reduce} \emph{Conditioning on a random variable reduces entropy}:
    $\en{\rA \mid \rB,\rC} \leq \en{\rA \mid  \rB}$.  The equality holds iff $\rA \perp \rC \mid \rB$.
  \item \label{part:chain-rule} \emph{Chain rule for mutual information}: $\mi{\rA,\rB}{\rC \mid \rD} = \mi{\rA}{\rC \mid \rD} + \mi{\rB}{\rC \mid  \rA,\rD}$.
  \item \label{part:data-processing} \emph{Data processing inequality}: for a function $f(\rA)$ of $\rA$, $\mi{f(\rA)}{\rB \mid \rC} \leq \mi{\rA}{\rB \mid \rC}$. 
   \end{enumerate}
\end{fact}


\noindent
We also use the following two standard propositions on effect of conditioning on mutual information.

\begin{proposition}\label{prop:info-increase}
  For random variables $\rA, \rB, \rC, \rD$, if $\rA \perp \rD \mid \rC$, then, 
  \[\mi{\rA}{\rB \mid \rC} \leq \mi{\rA}{\rB \mid  \rC,  \rD}.\]
\end{proposition}
 \begin{proof}
  Since $\rA$ and $\rD$ are independent conditioned on $\rC$, by
  \itfacts{cond-reduce}, $\HH(\rA \mid  \rC) = \HH(\rA \mid \rC, \rD)$ and $\HH(\rA \mid  \rC, \rB) \ge \HH(\rA \mid  \rC, \rB, \rD)$.  We have,
	 \begin{align*}
	  \mi{\rA}{\rB \mid  \rC} &= \HH(\rA \mid \rC) - \HH(\rA \mid \rC, \rB) = \HH(\rA \mid  \rC, \rD) - \HH(\rA \mid \rC, \rB) \\
	  &\leq \HH(\rA \mid \rC, \rD) - \HH(\rA \mid \rC, \rB, \rD) = \mi{\rA}{\rB \mid \rC, \rD}. \qed
	\end{align*}
	
\end{proof}

\begin{proposition}\label{prop:info-decrease}
  For random variables $\rA, \rB, \rC,\rD$, if $ \rA \perp \rD \mid \rB,\rC$, then, 
  \[\mi{\rA}{\rB \mid \rC} \geq \mi{\rA}{\rB \mid \rC, \rD}.\]
\end{proposition}
 \begin{proof}
 Since $\rA \perp \rD \mid \rB,\rC$, by \itfacts{cond-reduce}, $\HH(\rA \mid \rB,\rC) = \HH(\rA \mid \rB,\rC,\rD)$. Moreover, since conditioning can only reduce the entropy (again by \itfacts{cond-reduce}), 
  \begin{align*}
 	\mi{\rA}{\rB \mid  \rC} &= \HH(\rA \mid \rC) - \HH(\rA \mid \rB,\rC) \geq \HH(\rA \mid \rD,\rC) - \HH(\rA \mid \rB,\rC) \\
	&= \HH(\rA \mid \rD,\rC) - \HH(\rA \mid \rB,\rC,\rD) = \mi{\rA}{\rB \mid \rC,\rD}. \qed
 \end{align*}

\end{proof}

\subsection*{Measures of Distance Between Distributions}\label{sec:prob-distance}

We use two main measures of distance (or divergence) between distributions, namely the \emph{Kullback-Leibler divergence} (KL-divergence) and the \emph{total variation distance}. 

\paragraph{KL-divergence.} For two distributions $\mu$ and $\nu$ over the same probability space, the \textbf{Kullback-Leibler (KL) divergence} between $\mu$ and $\nu$ is denoted by $\kl{\mu}{\nu}$ and defined as: 
\begin{align}
\kl{\mu}{\nu}:= \Ex_{a \sim \mu}\Bracket{\log\frac{\mu(a)}{{\nu}(a)}}. \label{eq:kl}
\end{align}
We also have the following relation between mutual information and KL-divergence. 
\begin{fact}\label{fact:kl-info}
	For random variables $\rA,\rB,\rC$, 
	\[\mi{\rA}{\rB \mid \rC} = \Ex_{(B,C) \sim {(\rB,\rC)}}\Bracket{ \kl{\distribution{\rA \mid \rB=B,\rC=C}}{\distribution{\rA \mid \rC=C}}}.\] 
\end{fact}


\paragraph{Total variation distance.} We denote the \textbf{total variation distance} between two distributions $\mu$ and $\nu$ on the same 
support $\Omega$ by $\tvd{\mu}{\nu}$, defined as: 
\begin{align}
\tvd{\mu}{\nu}:= \max_{\Omega' \subseteq \Omega} \paren{\mu(\Omega')-\nu(\Omega')} = \frac{1}{2} \cdot \sum_{x \in \Omega} \card{\mu(x) - \nu(x)}.  \label{eq:tvd}
\end{align}
\noindent
We use the following basic properties of total variation distance. 
\begin{fact}\label{fact:tvd-small}
	Suppose $\mu$ and $\nu$ are two distributions for $\event$, then, 
	$
	{\mu}(\event) \leq {\nu}(\event) + \tvd{\mu}{\nu}.
$
\end{fact}

We also have the following (chain-rule) bound on the total variation distance of joint variables.

\begin{fact}\label{fact:tvd-chain-rule}
	For any distributions $\mu$ and $\nu$ on $n$-tuples $(X_1,\ldots,X_n)$, 
	\[
		\tvd{\mu}{\nu} \leq \sum_{i=1}^{n} \Exp_{X_{<i} \sim \mu} \tvd{\mu(X_i \mid X_{<i})}{\nu(X_i \mid X_{<i})}. 
	\]
\end{fact}

The following Pinsker's inequality bounds the total variation distance between two distributions based on their KL-divergence, 

\begin{fact}[Pinsker's inequality]\label{fact:pinskers}
	For any distributions $\mu$ and $\nu$, 
	$
	\tvd{\mu}{\nu} \leq \sqrt{\frac{1}{2} \cdot \kl{\mu}{\nu}}.
	$ 
\end{fact}



\end{document}